\pdfoutput=1
\documentclass[12pt,a4paper]{article}
\usepackage{jheppub}
\usepackage[usenames,dvipsnames]{xcolor}
 
\usepackage{microtype}
\usepackage{amssymb} 
\usepackage{amsmath}
\usepackage{mathtools}
\usepackage{amsfonts}    
\usepackage{dsfont}
\usepackage{pdfpages}
\usepackage{verbatim}
\usepackage{braket}
\usepackage{bm}
\usepackage{tensor}
\usepackage{upgreek}
\usepackage{subcaption}
\usepackage{simpler-wick}
\usepackage{enumitem}
\usepackage[hang,flushmargin]{footmisc}
\usepackage{amsthm}
\usepackage{tcolorbox}
\usepackage{nicematrix}
\usepackage[table]{xcolor}

\usepackage{tikz}
\usepackage{listings}
\usepackage{color}
\usetikzlibrary{decorations.markings}

\usepackage{pgfplots}
\usepgfplotslibrary{fillbetween}
\pgfplotsset{compat=1.18}

\definecolor{rosewood}{rgb}{0.4, 0.0, 0.04}
\definecolor{pyblue}{RGB}{31, 119, 180}
\definecolor{pyorange}{RGB}{255, 127, 14}
\definecolor{pygreen}{RGB}{44, 160, 44}
\definecolor{pyred}{RGB}{214, 39, 40}
\definecolor{lightgray}{gray}{0.9}

\usepackage{hyperref}
\hypersetup{
    pdfencoding=unicode,
    colorlinks=true,
    urlcolor=rosewood,
    linkcolor=pyblue,
    citecolor=pygreen,
    pdftitle={Massive spectral gluing},
    pdfauthor={Denis Werth},
    pdfdisplaydoctitle=true,
    pdfstartview=FitH,
    linktocpage=true
}

\def \d {\mathrm{d}}

\def \k {\bm{k}}

\def \u {\bm{u}}

\def \A {\mathcal{A}}

\def \C {\mathcal{C}}
\def \D {\mathcal{D}}

\def \F {\mathcal{F}}
\def \G {\mathcal{G}}

\def \I {\mathcal{I}}

\def \N {\mathcal{N}}
\def \O {\mathcal{O}}
\def \P {\mathcal{P}}

\def \R {\mathcal{R}}
\def \T {\mathcal{T}}

\def \V {\mathcal{V}}
\def \W {\mathcal{W}}

\def \Y {\bm{Y}}

\def \aa {{\sf{a}}}
\def \bb {{\sf{b}}}
\def \cc {{\sf{c}}}

\def \Re {\mathrm{Re}}

\def \Res {\mathrm{Res}}

\newtheorem{theorem}{Theorem}[section]
\newtheorem{corollary}{Corollary}[theorem]
\newtheorem{lemma}[theorem]{Lemma}
\newtheorem{proposition}{Proposition}
\newtheorem{definition}{Definition}

\title{All Tree-Level Massive Cosmological Correlators via Spectral Gluing}

\author[1, 2]{Jonathan Gr\"afe}\emailAdd{graefe@mpp.mpg.de}
\author[1, 2]{and Denis Werth}\emailAdd{werth@mpp.mpg.de}
\affiliation[1]{Max Planck Institute for Physics, Werner-Heisenberg-Institut, Munich, D-85748, Germany}
\affiliation[2]{Max Planck-IAS-NTU Center for Particle Physics, Cosmology and Geometry}

\abstract{Massive cosmological correlators exhibit a rich hypergeometric structure already at tree level, reflecting the distorted propagation of particles in de Sitter spacetime. In this paper, we reveal that this apparent complexity conceals a remarkably simple underlying mathematical structure. 
Using the spectral representation, we compute arbitrary tree-level correlators of scalar fields with generic masses and show that they are constructed from fundamental building blocks belonging to the family of Lauricella generalised hypergeometric functions, glued together by spectral integrals. We develop a spectral gluing algorithm that evaluates these integrals through elementary graph combinatorics, yielding explicit series representations that resum the dependence on internal energies away from soft limits. 
This algorithm naturally generates solutions to the differential equations satisfied by massive correlators as expansions in the corresponding eigenfunctions. Acting with a set of graph annihilators, we uncover a new class of magical identities among generalised hypergeometric functions, revealing an unexpected simplification: once the dynamical propagators are stripped away, the remaining hypergeometric kinematic dependence collapses to rational functions.
Our results expose a hidden simplicity in the rigid hypergeometric analytic structure dictated by graph combinatorics, and hint at an intrinsic geometric principle from which properties of massive correlators naturally emerge.
}

\begin{document}

\setcounter{tocdepth}{3}
\maketitle
\setcounter{page}{1}

\newpage
\section{Introduction}

Cosmological correlators are the fundamental observables of the primordial Universe, providing the bridge between quantum field theory in curved spacetime and the large-scale structure observed today. Unlike scattering amplitudes, which relate asymptotic states in flat space and are tightly constrained by Lorentz invariance, cosmological correlators have their analytic structure significantly altered by the time-dependent background: even at tree level, cosmological correlators are no longer rational functions of kinematic invariants but instead exhibit a rich hierarchy of transcendental functions. Understanding the origin of this complexity---and whether it conceals an underlying organising principle---has become one of the central questions in modern study of cosmological observables.

\vskip 4pt
Over the past decade, significant progress has been made in uncovering a wealth of unexpected mathematical properties, especially in the simplified setting of conformally coupled scalars in power-law cosmologies, where cosmological correlators can be represented as twisted integrals over rational functions, closely resembling the structure of Feynman integrals. While this model does not capture the full complexity of realistic inflationary observables, it provides a rich and tractable laboratory in which the analytic properties of correlators become manifest. These correlators have been shown to satisfy canonical differential equations which can be derived from simple graphical rules~\cite{Hillman:2019wgh, De:2023xue, Arkani-Hamed:2023kig, Baumann:2024mvm, Grimm:2024mbw, He:2024olr, Grimm:2024tbg, Hang:2024xas, De:2024zic, Baumann:2025qjx, Capuano:2025ehm, Glew:2025ypb}, revealing a deep connection with twisted cohomology and intersection theory~\cite{McLeod:2026jpz, McLeod:2026kpo}, algebraic geometry~\cite{Fevola:2024nzj}, cluster algebra~\cite{Capuano:2025myy, Mazloumi:2025pmx, Paranjape:2026htn, Capuano:2026pgq, Ferro:2026oph}, and tree-graph combinatorics~\cite{Fan:2024iek, Fan:2025scu}. They are also known to be closely related to the corresponding scattering amplitudes through dressing rules~\cite{Chowdhury:2023arc, Chowdhury:2025ohm, Chowdhury:2025ohm, Chowdhury:2026upp, Chowdhury:2026dwm, Das:2026vfv}. Moreover, these correlators admit a geometric interpretation: flat-space wavefunction coefficients organise into canonical forms of cosmological polytopes~\cite{Arkani-Hamed:2017fdk, Arkani-Hamed:2018bjr, Benincasa:2018ssx, Benincasa:2019vqr, Benincasa:2020aoj, Albayrak:2023hie} (see~\cite{Benincasa:2022gtd} for a review), while summing over channels led to the discovery of cosmohedra~\cite{Arkani-Hamed:2024jbp, Glew:2025otn, Figueiredo:2025daa, Forcey:2025voc, Ardila-Mantilla:2026cbo} (see e.g.~\cite{Telen:2025zsz} for a review on positive geometry). These developments eventually have uncovered hidden factorisation properties~\cite{De:2025bmf, Li:2026gns}, and suggest that the apparent complexity of cosmological correlators may be governed by a much simpler underlying mathematical structure.

\vskip 4pt
For massive particle exchanges, however, the picture remains considerably less complete. Already the simplest tree-level single-exchange diagram gives rise to non-trivial hypergeometric functions~\cite{Arkani-Hamed:2018kmz, Qin:2022fbv, Qin:2023ejc, Werth:2024mjg}, reflecting the richer analytic structure associated with massive propagation in de Sitter space. As the number of internal exchanges increase, these functions generalise to multivariable hypergeometric structures whose analytic properties quickly become rather intractable~\cite{Xianyu:2023ytd, Liu:2024str}. A variety of complementary approaches---e.g.~solving boundary differential equations~\cite{Arkani-Hamed:2015bza, Arkani-Hamed:2018kmz, Pimentel:2022fsc, Jazayeri:2022kjy, Qin:2022fbv, Chen:2023iix, Qin:2023ejc, Gasparotto:2024bku, Chen:2024glu, Aoki:2024uyi, Liu:2024str, Xianyu:2025lbk, Chen:2026dqp, Baumann:2026atn}, the use of Mellin-Barnes representations~\cite{Sleight:2019mgd, Sleight:2019hfp, Sleight:2021plv, Qin:2022lva, Xianyu:2023ytd}, dispersive~\cite{Xianyu:2022jwk, Liu:2024xyi, Liu:2026jzn} and spectral methods~\cite{Werth:2024mjg, Grafe:2026qsm, Belrhali:2026ktb, Belrhali:2026rkn} (also see~\cite{Melville:2023kgd, Melville:2024ove} for a massive de Sitter $S$-matrix definition), introducing the (orthogonal) Grassmannian as a new kinematic space~\cite{Arundine:2026fbr, Arundine:2026myr}, and general constraints from unitarity and locality~\cite{Melville:2021lst, Cespedes:2020xqq, AguiSalcedo:2023nds, Goodhew:2020hob, Das:2025qsh}---have been developed to tame this complexity. Nevertheless, unlike the conformally coupled case, where many hidden structures are now understood, the organisation of massive correlators remains largely unexplored. We still lack a general understanding of why these hypergeometric functions appear, how they organise themselves for arbitrary graphs, and whether they obey a deeper mathematical principle. This raises the question of whether a universal construction exists that can generate all tree-level massive correlators at once. 

\vskip 4pt
In this paper, we show that the complexity of massive correlators is largely an illusion, and that the hypergeometric structure is graph combinatorics in disguise. We demonstrate that arbitrary tree-level massive cosmological correlators are governed by a rigid and universal mathematical structure. We show that every tree graph factorises into fundamental building blocks, and reconstructing the full correlator reduces to an elementary gluing procedure that follows simple graphical rules. Remarkably, after stripping away the genuinely dynamical propagators of massive correlators by acting on the constructed solutions with a set of graph annihilators, the remaining hypergeometric structure collapses dramatically. This reveals unexpected ``magical identities'' among multivariable hypergeometric functions that reduce seemingly complicated expressions to simple rational functions.

\paragraph{Vertex functions as fundamental building blocks.} Using a spectral representation of bulk-to-bulk massive propagators, every massive correlator factorises into elementary single-time-integral objects, which we call {\it vertex functions}. These objects are universal: they only depend on the local data associated with a single interaction vertex (its external energy, twist and attached massive legs), and are completely independent of the global graph in which they appear. We compute the most general vertex function and show that it is given by a type-$C$ {\it Lauricella} function, providing a natural multivariable generalisation of the familiar hypergeometric function. This generalised hypergeometric structure is fully constrained by conformal symmetry on the boundary, and arises as one moves away from soft limits:
\begin{center}
\begin{tikzpicture}[line width=1. pt, scale=2]

    \foreach \i in {0,...,5}
    \coordinate (V\i) at ({60*\i}:0.8);

    \draw[pyblue] (V5) -- node[midway, right, yshift=-5pt, xshift=3pt] {$Y_1$} (V0);
    \draw[pyblue] (V0) -- node[midway, above, yshift=-5pt, xshift=10pt] {$Y_2$} (V1);
    \node[pyblue] at ($(V1)!1/2!(V2)$) {$\cdots$};
    \draw[pyblue] (V2) -- node[midway, left, yshift=5pt, xshift=-3pt] {$Y_n$} (V3);

    \draw[black] (V3) -- (V4) -- (V5);
    
    \node[black] at (-0.3, -0.3) {$X$};

    \draw[black, ->] (1.2, 0) -- (1.7, 0);

    \node[black] at (3.9, 0) {$\displaystyle F_C^{(n)}\left[\,\, \# \,\, \middle| \, \left(\frac{\textcolor{pyblue}{Y_1}}{X}\right)^2, \left(\frac{\textcolor{pyblue}{Y_2}}{X}\right)^2, \ldots, \left(\frac{\textcolor{pyblue}{Y_n}}{X}\right)^2 \, \right]\,.$};

\end{tikzpicture} 
\end{center}
In particular, the dependence on the internal energy variables in \textcolor{pyblue}{blue} flowing into internal lines of massive correlators is completely resummed. These vertex functions encode the entire kinematic information of massive graphs at the local level. 

\paragraph{Spectral gluing algorithm from graph combinatorics.} Once every tree graph is decomposed into vertex functions, reconstructing the full correlator reduces to evaluating spectral integrals using the residue theorem. We develop a {\it spectral gluing algorithm} that performs this task for arbitrary tree graphs through elementary graph combinatorics. The poles contributing to each spectral integration are determined solely by simple combinatorial data encoded in the graph incidence matrix $Q$, e.g.
\begin{equation*}
\vcenter{\hbox{
\begin{tikzpicture}
[line width=1. pt, scale=2]
    [line width=1. pt, scale=2, baseline={(current bounding box.center)}]
    
    \draw[black, postaction={decorate}, decoration={markings, mark=at position 0.6 with {\arrow[pyblue]{>}}}] (0, 0) -- node[left] {\scriptsize{$\textcolor{gray}{x_{14}}$}} (0, 0.7);
    \draw[black, postaction={decorate}, decoration={markings, mark=at position 0.6 with {\arrow[pyblue]{<}}}] (0, 0) -- node[above left] {\scriptsize{$\textcolor{gray}{x_{24}}$}} (-0.5, -0.5);
    \draw[black, postaction={decorate}, decoration={markings, mark=at position 0.6 with {\arrow[pyblue]{>}}}] (0, 0) -- node[above right] {\scriptsize{$\textcolor{gray}{x_{34}}$}} (0.5, -0.5);
        
    \draw[fill=black] (0, 0.7) circle (.05cm) node[above, yshift=3pt] {\scriptsize{$\textcolor{gray}{(1)}$}};
    \draw[pyblue, fill=pyblue] (-0.5, -0.5) circle (.05cm) node[below left] {\scriptsize{$\textcolor{gray}{(2)}$}};
    \draw[fill=black] (0.5, -0.5) circle (.05cm) node[below right] {\scriptsize{$\textcolor{gray}{(3)}$}};
    \draw[fill=black] (0, 0) circle (.05cm) node[above right] {\scriptsize{$\textcolor{gray}{(4)}$}};
\end{tikzpicture} 
    }}\,, \quad Q = 
    \begin{bNiceMatrix}[first-row,code-for-first-row=\scriptstyle, first-col,code-for-first-col=\scriptstyle,]
        & \textcolor{gray}{x_{14}} & \textcolor{gray}{x_{24}} & \textcolor{gray}{x_{34}} \\
        \textcolor{gray}{(1)} & +1 & 0 & 0 \\
        \textcolor{gray}{(2)} & 0 & -1 & 0 \\
        \textcolor{gray}{(3)} & 0 & 0 & +1 \\
        \textcolor{gray}{(4)} & -1 & +1 & -1 
    \end{bNiceMatrix}\,,
\end{equation*}
where the choice of a root vertex in \textcolor{pyblue}{blue} (which uniquely fixes the tree structure) determines the subsequent convergence region of the correlator in kinematic space. This allows residues to be collected algorithmically without ever evaluating nested time integrals. The resulting expressions are explicit partially resummed multiple series that provide general solutions to the differential equations satisfied by massive correlators.

\paragraph{Magical identities from graph annihilators.} The series representations obtained with the spectral gluing algorithm make manifest a remarkable hidden simplicity. Acting on the resulting correlators with a set of differential {\it graph annihilators} removes the dynamical propagator structure while preserving the underlying hypergeometric dependence. The outcome is an infinite family of previously unknown identities among Lauricella functions: multi-nested series of multivariable hypergeometric functions collapse to simple rational functions. An example is given by (Eq.~\eqref{eq: two-site chain magic identity} for $\tilde{p}_1=\tilde{p}_2=1/2$)
\begin{equation*}
    \begin{aligned}
    \frac{X_1}{X_1+X_2} &= \sum\limits_{m=0}^\infty (-1)^m  \left(\frac{X_2}{X_1}\right)^m  \textcolor{pyblue}{_2F_1\left[\left.\begin{matrix} \frac{1+m}{2},\,  \frac{2+m}{2}\\ \tfrac{3}{2}+m \end{matrix}\right\vert \left(\frac{Y_{12}}{X_1}\right)^2\right]
    {}_2F_1\left[\left.\begin{matrix} \frac{-m}{2},\,  \frac{1-m}{2}\\ \tfrac{1}{2}-m \end{matrix}\right\vert \left(\frac{Y_{12}}{X_2}\right)^2\right] }\,.
    \end{aligned}
\end{equation*}
Remarkably, this identity remains true even after stripping off the hypergeometric structure in \textcolor{pyblue}{blue}. This corresponds to the correlator soft limit, $Y_{12}\to0$, and boils down to the well-known geometric series in $X_2/X_1$. These {\it magical identities} reveal that much of the apparent transcendental complexity of massive correlators is purely kinematical. 

\vskip 4pt
Taken together, these results point towards a different way of thinking about massive cosmological correlators. Their complexity is not an intrinsic property of ever more elaborate special functions, but rather the result of gluing a set of universal hypergeometric objects according to the combinatorial data of the graph. 

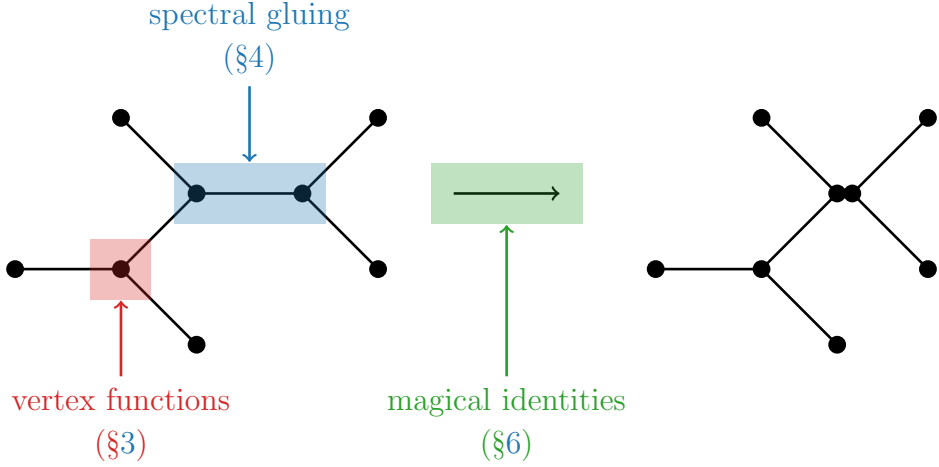
\begin{figure}[h!]
    \centering
\begin{tikzpicture}[line width=1. pt, scale=2, baseline={(0,0)}]
    \draw[black] (0, 0) -- (0.7, 0);
    \draw[black] (0.7, 0) -- (1.2, 0.5);
    \draw[black] (0.7, 0) -- (1.2, -0.5);
    \draw[black] (0, 0) -- (-0.5, 0.5);
    \draw[black] (0, 0) -- (-0.5, -0.5);
    \draw[black] (-0.5, -0.5) -- (0, -1);
    \draw[black] (-0.5, -0.5) -- (-1.2, -0.5);
    
    \draw[fill=black] (0, 0) circle (.05cm);
    \draw[fill=black] (0.7, 0) circle (.05cm);
    \draw[fill=black] (1.2, 0.5) circle (.05cm);
    \draw[fill=black] (1.2, -0.5) circle (.05cm);
    \draw[fill=black] (-0.5, 0.5) circle (.05cm);
    \draw[fill=black] (-0.5, -0.5) circle (.05cm);
    \draw[fill=black] (0, -1) circle (.05cm);
    \draw[fill=black] (-1.2, -0.5) circle (.05cm);
    
    \draw[black, ->] (1.7, 0) -- (2.4, 0);

    \node[draw=none, fill=pyred, fill opacity=0.3, minimum width=0.8cm, minimum height=0.8cm, inner sep=0pt] (box) at (-0.5,-0.5) {};
    \draw[<-, pyred] (box.south) -- ++(0,-0.5);
    \node[pyred, below, align=center] at ($(box.south)+(0,-0.5)$) {vertex functions\\ 
    (\S\ref{sec: vertex functions})};
    
    \node[draw=none, fill=pyblue, fill opacity=0.3, minimum width=2cm, minimum height=0.8cm, inner sep=0pt] (box) at (0.35,0) {};
    \draw[<-, pyblue] (box.north) -- ++(0,0.5);
    \node[pyblue, above, align=center] at ($(box.north)+(0,0.5)$) {spectral gluing\\ 
    (\S\ref{sec: spectral gluing})};
    
    \node[draw=none, fill=pygreen, fill opacity=0.3, minimum width=2cm, minimum height=0.8cm, inner sep=0pt] (box) at (2.05,0) {};
    \draw[<-, pygreen] (box.south) -- ++(0,-1);
    \node[pygreen, below, align=center] at ($(box.south)+(0,-1)$) {magical identities\\ 
    (\S\ref{sec: hidden simplicity})};
\end{tikzpicture} 
\begin{tikzpicture}[line width=1. pt, scale=2, baseline={(0,0)}]
    \draw[black] (0.1, 0) -- (0.6, 0.5);
    \draw[black] (0.1, 0) -- (0.6, -0.5);
    \draw[black] (0, 0) -- (-0.5, 0.5);
    \draw[black] (0, 0) -- (-0.5, -0.5);
    \draw[black] (-0.5, -0.5) -- (0, -1);
    \draw[black] (-0.5, -0.5) -- (-1.2, -0.5);
  
    \draw[fill=black] (0, 0) circle (.05cm);
    \draw[fill=black] (0.1, 0) circle (.05cm);
    \draw[fill=black] (0.6, 0.5) circle (.05cm);
    \draw[fill=black] (0.6, -0.5) circle (.05cm);
    \draw[fill=black] (-0.5, 0.5) circle (.05cm);
    \draw[fill=black] (-0.5, -0.5) circle (.05cm);
    \draw[fill=black] (0, -1) circle (.05cm);
    \draw[fill=black] (-1.2, -0.5) circle (.05cm);
\end{tikzpicture} 
    \caption{Schematic roadmap of the paper.}
    \label{fig: roadmap}
\end{figure}

\paragraph{Outline.} The outline of the paper is summarised in Fig.~\ref{fig: roadmap}. In Sec.~\ref{sec: Hypergeometric Structure of Correlators}, we introduce master integrals for tree-level massive cosmological correlators. We then show that a spectral representation of massive bulk-to-bulk propagators completely factorises the time integrals of any tree graph, revealing a set of elementary building blocks: vertex functions. This factorisation provides a natural strategy for computing general tree-level massive correlators, whereby graphs are decomposed into vertex functions that are subsequently glued together through spectral integration. In Sec.~\ref{sec: vertex functions}, we study these vertex functions and show that they belong to the class of type-$C$ Lauricella generalised hypergeometric functions. We then discuss their physical regions in kinematic space, differential annihilators and soft limits. In Sec.~\ref{sec: spectral gluing}, we develop a spectral gluing algorithm to perform the spectral integrals for arbitrary tree-level massive correlators, which relies entirely on elementary graph combinatorics. In Sec.~\ref{sec: selected examples}, we illustrate the power of this method through explicit examples, including closed-form expressions for the maximally-nested contributions of all $N$-site chains and star graphs. In Sec.~\ref{sec: hidden simplicity}, we show that our series representations provide solutions to boundary differential equations satisfied by massive correlators. By acting on these solutions with graph annihilators, we uncover a new class of identities among generalised hypergeometric functions. Our conclusions are summarised in Sec.~\ref{sec: conclusions} where we also give an outlook on future directions.

\vskip 4pt
Several appendices provide additional technical details that complement the main text. In App.~\ref{sec: special functions}, we review the special functions that play a central role throughout the paper, summarising their key properties, including differential equations they satisfy, analytic continuation formulae, and reflexion identities. In App.~\ref{sec: details on spectral gluing algorithm}, we present the explicit computation of the spectral integrals for the two- and three-site chains. In App.~\ref{sec: diff eq derivation}, we derive the complete set of differential equations satisfied by tree-level massive graphs.

\paragraph{Notation \& conventions.} We use the mostly-plus signature for the metric $(-, +, \ldots, +)$. Spatial $d$-dimensional vectors are written in boldface $\bm{k}$. We use \textsf{sans serif} letters $(\sf{a, b, \ldots})$ for Schwinger-Keldysh indices. We work in the Poincaré patch of de Sitter space with the metric $\d s^2 = a^2(\tau)(-\d\tau^2+\d \bm{x}^2)$, where $a(\tau) = -(H\tau)^{-1}$ is the scale factor, $H$ is the Hubble parameter, and $\tau$ is conformal time. For a given tree graph with $V$ vertices (and $I=V-1$ internal lines), external and internal energies are denoted by $X_i$ ($i=1, \ldots, V$) and $Y_{ij}$ ($(i, j)$ running over all $I$ incident vertex-edge pairs), respectively, $p_i$ denotes the vertex twist (with $\tilde{p}\equiv p+\tfrac{n-2}{2}d$), and we define energy ratios as $u_{ij} \equiv Y_{ij}/X_i$. An example is given by:
\begin{center}
\begin{tikzpicture}[line width=1. pt, scale=2]
    
    \draw[black] (0, 0) -- node[left] {$Y_{14}$} (0, 0.7);
    \draw[black] (0, 0) -- node[above left] {$Y_{24}$} (-0.5, -0.5);
    \draw[black] (0, 0) -- node[above right] {$Y_{34}$} (0.5, -0.5);
        
    \draw[fill=black] (0, 0.7) circle (.05cm) node[above] {$X_1$};
    \draw[fill=black] (-0.5, -0.5) circle (.05cm) node[below left] {$X_2$};
    \draw[fill=black] (0.5, -0.5) circle (.05cm) node[below right] {$X_3$};
    \draw[fill=black] (0, 0) circle (.05cm) node[above right] {$X_4$};
\end{tikzpicture} 
\end{center}
for which the dimensionless kinematic variables are
\begin{equation}
    \begin{aligned}
        u_{14} &\equiv \frac{Y_{14}}{X_1} \,, \quad u_{41} \equiv \frac{Y_{14}}{X_4} \,, \quad u_{24} \equiv \frac{Y_{24}}{X_2} \,, \\
        u_{42} &\equiv \frac{Y_{24}}{X_4} \,, \quad u_{34} \equiv \frac{Y_{34}}{X_3} \,, \quad u_{43} \equiv \frac{Y_{34}}{X_4} \,.
    \end{aligned}
\end{equation}
Additional definitions will be introduced as needed in the main text, and further notations are given in App.~\ref{sec: special functions}.

\newpage
\section{Hypergeometric Structure of Correlators}
\label{sec: Hypergeometric Structure of Correlators}

We begin by introducing the central objects of interest: multiple-exchange massive cosmological correlators. We first review their corresponding diagrammatic rules and establish the notations used throughout. We then motivate the spectral approach for computing these observables by highlighting the rich structure of the associated master integrals at tree level. In particular, we show that these integrals can be decomposed into elementary building blocks belonging to the broad class of generalised multivariable hypergeometric functions, which are subsequently glued together through spectral integration. This approach makes the analytic structure as well as the factorisation property of massive tree graphs manifest.

\subsection{Master integrals for tree correlators}
\label{subsec: master integrals for tree correlators}

Let us consider generic $N$-point (connected) correlators of a conformally coupled (with mass $m^2/H^2=(d^2-1)/4$ in $d$ spatial dimensions) bulk scalar field $\varphi$, which at equal time and in the late-time limit read
\begin{equation}
    \lim_{\tau\to0^-} \braket{\Omega| \varphi_{\k_1}(\tau) \cdots \varphi_{\k_N}(\tau)|\Omega} = (2\pi)^d \delta^{(d)}(\k_1 + \cdots + \k_N) \T_N(\k_1, \ldots, \k_N) \,,
\end{equation}
where $\varphi_{\k_i}$ ($i=1, \ldots, N$) are operators in $d$-dimensional spatial Fourier space, and $\tau$ is conformal time. The expectation value of products of such operators is taken with respect to the Bunch-Davies vacuum state $\ket{\Omega}$, defined in the infinite past limit $\tau \to -\infty$. We will assume that the bulk theory is a weakly coupled local quantum field theory in a rigid de Sitter background. After stripping off the overall momentum conservation delta function, owing to spatial translation invariance, the cosmological correlators $\T_N(\k_1, \ldots, \k_N)$ can be expanded into connected graphs in perturbation theory. In this work, we focus on the leading contributions coming from tree graphs.

\subsubsection*{Cosmological graphs}

We use the bulk Schwinger-Keldysh diagrammatic rules to compute cosmological correlators (see e.g.~\cite{Chen:2017ryl}), where two copies $\varphi_\pm$ of the field $\varphi$ are introduced on two branches of the in-in path integral contour that are sewn at the time $\tau=0$. The computation of correlators reduces to the evaluation of a set of nested time integrals, with the level of nesting determined by the number of vertices. To streamline this problem, we instead introduce graphs $\G(\{X_i\}, \{Y_j\})$ with $V$ vertices and $I$ internal lines, characterised by the set of external energies $\{X_i\}$ ($i = 1, \ldots, V$) and internal energies $\{Y_j\}$ ($j = 1, \ldots, I$). In practice, these internal energies are the magnitudes of $d$-dimensional momenta flowing through internal lines, that are found after enforcing momentum conservation at each vertex. For tree graphs, we have the relation $I=V-1$. We further allow for an arbitrary number of conformally coupled fields, possibly none if the vertex is internal, with momenta $k_1, \ldots, k_n$ attached to any given vertex. In this case, the corresponding vertex energy is given by the sum of the magnitudes of these external momenta, $X = k_1 + \cdots + k_n$. This simplification is enabled by the particularly simple form of the mode functions for conformally coupled fields.

\vskip 4pt
Following standard notations, we represent a $(+)$ (resp.~$(-)$) vertex with an operator insertion $\varphi_+$ (resp.~$\varphi_-$) on the $(+)$ (resp.~$(-)$) in-in branch by a black (resp.~white) dot $\raisebox{0pt}{
\begin{tikzpicture}[line width=1. pt, scale=2]
\draw[fill=black] (0, 0) circle (.05cm);
\end{tikzpicture} 
}$  (resp.~$\raisebox{0pt}{
\begin{tikzpicture}[line width=1. pt, scale=2]
\draw[fill=white] (0, 0) circle (.05cm);
\end{tikzpicture} 
}$). Every graph is the sum of all coloured combinations of black and white dots:
\begin{equation}
    \G(\{X_i\}, \{Y_j\}) = \sum_{\aa_1, \ldots, \aa_V=\pm} (i \aa_1) \cdots (i \aa_V) \, \I_{\aa_1 \cdots \aa_V}(\{X_i\}, \{Y_j\}) \,,
\end{equation}
with $\aa_1, \ldots, \aa_V$ being Schwinger-Keldysh indices (the sum contains $2^V$ terms), and where we have defined the master integrals
\begin{equation}
\label{eq: def master integrals}
    \I_{\aa_1 \cdots \aa_V}(\{X_i\}, \{Y_j\}) \equiv \int_{-\infty}^0 \prod_{i=1}^V \left[ \frac{\d\tau_i}{(-\tau_i)^{d+1}} (-\tau_i)^{p_i} e^{i\aa_i X_i \tau_i} \right] \prod_{j=1}^I D_{\aa_j \bb_j}(Y_j; \tau_j, \tau_j') \,.
\end{equation}
Naturally, both time variables $\tau_j$ and $\tau_j'$ (and the corresponding indices $\aa_j$ and $\bb_j$) should be identified with the corresponding times at the two vertices of the internal line, with internal energy $Y_j$, to which the bulk-to-bulk propagator $D_{\aa_j \bb_j}$ is assigned. Notice that we have stripped off all factors $(-i \aa_i)$ and trivial multiplicative coupling constants associated with each vertex $i$. For de Sitter covariant non-derivative interactions (of the form $\varphi^n\sigma^m$ where $\sigma$ is some massive bulk scalar field), the twist $p_i$ is related to the number of external fields $n_i$ attached to the vertex by the relation $p_i = n_i(d-1)/2$. In what follows, we will keep this twist general though, $p_i \in \mathbb{C}$, as it allows to accommodate for (spatial or time) derivative interactions, with complex twists describing models with resonant background. The only restriction we place on $p_i$ is to have IR-finite time integrals, which in practice imposes a lower-bound on $\Re(p_i)$ that depends on the number of legs attached to this vertex. From a phenomenological perspective, we emphasize that any tree-level correlator involving massless external fields, coupled to additional fields through arbitrary interactions, can in principle be expressed as a linear combination of the master integrals~\eqref{eq: def master integrals}, upon which an appropriate set of weight-shifting operators act, see e.g.~\cite{Baumann:2019oyu}.

\subsubsection*{Bulk-to-bulk propagators}

For definiteness, we consider that each internal line $j$ corresponds to the propagation of a massive scalar field in the principal series of de Sitter unitary irreducible representations (with mass $m_j/H>d/2$ in Hubble units), for which we introduce the mass parameter $\mu_j\in\mathbb{R}$, defined by $\mu_j^2 = m_j^2/H^2-d^2/4$. Notice that the exchange of complementary-series fields ($m_j/H<d/2$) can be accessed via analytic continuation after the replacement $\mu_j \to -i\nu_j$, with $\nu_j^2 = d^2/4 - m_j^2/H^2$, in the final result. In this case, potential spurious late-time IR divergences cancel against each other when summing all Schwinger-Keldysh contributions to a specific graph. Generalisation to other cases, e.g.~dS-breaking dispersion relations or spinning fields should be straightforward. The four bulk-to-bulk propagators then read
\begin{equation}
\label{eq: bulk-to-bulk propagators}
    \begin{aligned}
        D^{(\mu)}_{-+}(Y; \tau, \tau') &= \tfrac{\pi}{4} \, (\tau \tau')^{d/2} H_{i\mu}^{(1)}(-Y\tau) H_{i\mu}^{(2)}(-Y\tau') \,, \\
        D^{(\mu)}_{+-}(Y; \tau, \tau') &= \left[D^{(\mu)}_{-+}(Y; \tau, \tau')\right]^* \,, \\
        D^{(\mu)}_{\pm\pm}(Y; \tau, \tau') &= D^{(\mu)}_{\mp \pm}(Y; \tau, \tau') \Theta(\tau-\tau') + D^{(\mu)}_{\pm \mp}(Y; \tau, \tau') \Theta(\tau'-\tau) \,,
    \end{aligned}
\end{equation}
where $H_{i\mu}^{(1)}$ and $H_{i\mu}^{(2)}$ are the Hankel functions of the first and second kind, respectively. The difficulty in evaluating the master integrals~\eqref{eq: def master integrals} becomes clear: it requires integrating products of time-ordered Hankel functions. 

\subsubsection*{Feynman rules}

As discussed above, the master integrals~\eqref{eq: def master integrals} are naturally associated with specific tree graphs, in close analogy with the Feynman rules used to compute scattering amplitudes in flat spacetime. For the reader’s convenience, we now make these rules explicit. In particular, the prescription for constructing the integral $\I_{\aa_1 \cdots \aa_V}(\{X_i\}, \{Y_j\})$ corresponding to a tree graph with $V$ vertices and $I$ internal lines is as follows:

\begin{itemize}
    \item Colour the tree-graph vertices in black or white according to the Schwinger-Keldysh index $\aa = \pm$ of the corresponding vertices.
    \item Assign an external energy $X_i$ and a twist $p_i$ to each vertex $i=1, \ldots, V$, and an internal energy $Y_j$ and a mass parameter $\mu_j$ to each internal line $j=1, \ldots, I$.
    \item Assign the vertex factor $(-\tau_i)^{p_i} e^{i\aa_i X_i \tau_i}$ to each bulk interaction. 
    \item Assign a bulk-to-bulk propagator~\eqref{eq: bulk-to-bulk propagators} to each internal line, according to the Schwinger-Keldysh indices. Propagators are (anti)-time ordered when both indices are the same.
    \item Finally, integrate over the time insertions of all bulk vertices, $\int_{-\infty}^0 \d\tau_i/(-\tau_i)^{d+1}$, which includes the de Sitter-invariant volume element in $d$ spatial dimensions.
\end{itemize}

\subsubsection*{Dimensionless master integrals}

For later convenience, we define {\it dimensionless} master integrals $\widehat{\I}_{\aa_1 \cdots \aa_V}$ in the following way:
\begin{equation}
    \I_{\aa_1 \cdots \aa_V}(\{X_i\}, \{Y_j\}) = \left[\prod_{i=1}^V \frac{1}{X_i^{p_i-d}}\right]\left[\prod_{j=1}^I \frac{\pi/4}{(X_j X_j')^{d/2}}\right] \, \widehat{\I}_{\aa_1 \cdots \aa_V}(\{X_i\}, \{Y_j\}) \,,
\end{equation}
where we have factored out overall powers of external energies. The additional stripped-off energies (together with the factor $\pi/4$) for every internal line is for later convenience. Here, $X_j$ and $X_j'$ denote both vertex energies attached to the internal line $j$. Since the master integrals $\widehat{\I}_{\aa_1 \cdots \aa_V}$ are dimensionless, they depend on the $V+I=2I+1$ variables only through their ratios, which reduces the number of variables to $2I$ dimensionless energy ratios. These ratios fully specify the kinematics of dimensionless master integrals. Indeed, by introducing the dimensionless time $z_i \equiv X_i \tau_i$ at each vertex, the dimensionless master integrals are given by
\begin{equation}
\label{eq: def dimensionless master integrals}
    \widehat{\I}_{\aa_1 \cdots \aa_V}(\{X_i\}, \{Y_j\}) \equiv \int_{-\infty}^0 \prod_{i=1}^V \left[ \frac{\d z_i}{(-z_i)^{d+1}} (-z_i)^{p_i} e^{i\aa_i z_i} \right] \prod_{j=1}^I \widehat{D}_{\aa_j \bb_j}(\tfrac{Y_j}{X_i} z_j, \tfrac{Y_j}{X_i'} z_j') \,,
\end{equation}
where the dimensionless bulk-to-bulk propagators $\widehat{D}_{\aa\bb}^{(\mu)}$ are constructed out of 
\begin{equation}
\label{eq: def dimensionless prop}
    \widehat{D}_{-+}^{(\mu)}(\tfrac{Y}{X}z, \tfrac{Y}{X'}z') \equiv (zz')^{d/2} H_{i\mu}^{(1)}(-\tfrac{Y}{X}z) H_{i\mu}^{(2)}(-\tfrac{Y}{X'}z') \,.
\end{equation}
In what follows, we will exclusively work with dimensionless master integrals, which will always be hatted to avoid any possible confusion.

\subsubsection*{Selected integrals}

Throughout this work, we will mostly focus directly on the master tree-graph integrals $\widehat{\I}_{\aa_1 \cdots \aa_V}(\{X_i\}, \{Y_j\})$. With a slight abuse of terminology, we will also refer to these objects as (massive) cosmological correlators. The simplest example consists of one site only, and corresponds to a contact interaction. Using the above Feynman rules, and labelling the vertex as $\aa=+$ for concreteness, we get
\begin{equation}
\label{eq: one-site master integral}
    \widehat{\I}_+ = \raisebox{0pt}{
\begin{tikzpicture}[line width=1. pt, scale=2]
\draw[fill=black] (0, 0) circle (.05cm) node[above=1mm] {$(X, p)$};
\end{tikzpicture} 
} = \int_{-\infty}^0 \frac{\d z}{(-z)^{d+1}} (-z)^p e^{iz} \,,
\end{equation}
which converges in the IR for $p>d$. The number of external legs is arbitrary, as the computation depends only on the total energy $X=\sum_i k_i$ of the vertex. As such, the master integral $\I_+(X)$ contributes universally to any $N$-point correlator $\T_N$. For aficionados, notice that we strip off overall factors $2k_i$ coming from the mode functions of the conformally coupled field. However, as the graph only depends on a single energy variable, the corresponding dimensionless integral $\widehat{\I}_+$ is energy independent. Being simple enough, this integral straightforwardly evaluates to
\begin{equation}
    \widehat{\I}_+ = e^{-\frac{i\pi}{2}(p-d)} \, \Gamma(p-d) \,.
\end{equation}
This closed-form expression enables an analytic continuation of the integral to any $p\in\mathbb{C}\backslash\{d-n|n\in\mathbb{N}\}$. The excluded discrete values correspond to singular points, in whose vicinity the integral $\I_+$ develops a logarithmic branch point in the energy $X$. In order to achieve convergence in the UV, the time-integral contour in~\eqref{eq: one-site master integral} has been slightly tilted in the complex $z$-plane, $-\infty \to -\infty^+ \equiv -\infty(1-i\epsilon)$, with $\epsilon>0$, as required by the Bunch-Davies vacuum. The $\aa=-$ contribution is simply the complex conjugate master integral, $\widehat{\I}_- = \widehat{\I}_+^*$. Summing over all Schwinger-Keldysh contributions, we obtain the full graph
\begin{equation}
    \widehat{\G}(E) = i\widehat{\I}_+ - i\widehat{\I}_- = 2\Re(i\widehat{\I}_+) \,,
\end{equation}
which is manifestly real. 

\vskip 4pt
The next-to-simplest example is given by the tree-level exchange of a single massive field, commonly referred to as the two-site graph, in which four master integrals contribute. Half of them are related to the remaining ones by complex conjugation: $\widehat{\I}_{--} = \widehat{\I}_{++}^*$ and $\widehat{\I}_{+-} = \widehat{\I}_{-+}^*$. Explicitly, the fully-nested (all-plus) master integral reads
\begin{equation}
\label{eq: two-site chain ++}
    \begin{aligned}
        &\widehat{\I}_{++}(X_1, X_2, Y) = \raisebox{-17pt}{
\begin{tikzpicture}
[line width=1. pt, scale=2]
\draw[fill=black] (0, 0) circle (.05cm) node[above=0.5mm] {$(X_1, p_1)$};
\draw[fill=black] (1, 0) circle (.05cm) node[above=0.5mm] {$(X_2, p_2)$};
\draw[black] (0.05, 0) -- node[below] {$(Y, \mu)$} (1, 0);
\end{tikzpicture} 
} \\
        &= \int_{-\infty}^0 \left[\frac{\d z_1}{(-z_1)^{d+1}} (-z_1)^{p_1} e^{iz_1}\right] \left[\frac{\d z_2}{(-z_2)^{d+1}} (-z_2)^{p_2} e^{iz_2}\right] \widehat{D}_{++}^{(\mu)}(\tfrac{Y}{X_1} z_1, \tfrac{Y}{X_2} z_2) \,,
    \end{aligned}
\end{equation}
where $Y$ and $\mu$ is the energy and the mass parameter of the internal line, respectively. This correlator has been extensively studied in recent years with a plethora of various computational methods, see, e.g.,~\cite{Arkani-Hamed:2018kmz, Qin:2022fbv, Qin:2023ejc, Liu:2024str, Werth:2024mjg} for full closed-form expressions. The resulting solutions can be expressed in terms of hypergeometric functions. Extensions of this master integral including a dS-breaking dispersion relation and spinning fields are treated in e.g.~\cite{Pimentel:2022fsc, Jazayeri:2022kjy, Qin:2022fbv, Qin:2025xct}. The associated fully-factorised master integral is given by 
\begin{equation}
    \begin{aligned}
        &\widehat{\I}_{-+}(X_1, X_2, Y) = \raisebox{-17pt}{
\begin{tikzpicture}
[line width=1. pt, scale=2]
\draw[fill=white] (0, 0) circle (.05cm) node[above=0.5mm] {$(X_1, p_1)$};
\draw[fill=black] (1, 0) circle (.05cm) node[above=0.5mm] {$(X_2, p_2)$};
\draw[black] (0.05, 0) -- node[below] {$(Y, \mu)$} (1, 0);
\end{tikzpicture} 
} \\
        &= \int_{-\infty}^0 \left[\frac{\d z_1}{(-z_1)^{d+1}} (-z_1)^{p_1} e^{-iz_1}\right] \left[\frac{\d z_2}{(-z_2)^{d+1}} (-z_2)^{p_2} e^{iz_2}\right] \widehat{D}_{-+}^{(\mu)}(\tfrac{Y}{X_1} z_1, \tfrac{Y}{X_2} z_2) \,,
    \end{aligned}
\end{equation}
where the two time integrals factorise due to the non-time-ordered bulk-to-bulk propagator $D_{-+}^{(\mu)}$. This master integral evaluates to a product of hypergeometric $_2F_1$ functions, see Sec.~\ref{app: two-site chain} for more details. Summing over all Schwinger-Keldysh contributions gives the full graph $\widehat{\G}(X_1, X_2, Y)$.

\subsubsection*{Singularity structure}

Cosmological correlators exhibit a rich singularity structure, which we now make explicit. In particular, all correlators develop singularities, more precisely, branch points in the case of massive exchange, whenever the energies flowing into a given subgraph sum to zero. When the energies associated with all external legs sum to zero, a configuration commonly referred to as the total energy singularity, the residue at this branch point reproduces the flat-space scattering amplitude for the corresponding process~\cite{Maldacena:2011nz, Raju:2012zr}. More generally, when the energies entering a subgraph vanish, the correlator factorises at this kinematic locus into a product involving a lower-point correlator. Further details can be found, e.g., in~\cite{Arkani-Hamed:2018kmz, Arkani-Hamed:2017fdk, Arkani-Hamed:2018bjr, Benincasa:2018ssx}.

\vskip 4pt
As a concrete example, consider the two-site graph discussed above. Schematically, we have
\begin{equation}
    \lim_{X_1+Y\to0} \widehat{\G}(X_1, X_2, Y) \propto \log\left(\frac{X_1+Y}{X_1}\right) \times \left[\raisebox{-11pt}{\begin{tikzpicture}
[line width=1. pt, scale=2]
\draw[fill=black] (0, 0) circle (.05cm) node[above=0.5mm] {$X_2$};
\draw[black] (0, 0) -- node[below] {$Y$} (0.5, 0);
\end{tikzpicture}} \,\, - \raisebox{-13.5pt}{
\begin{tikzpicture}
[line width=1. pt, scale=2]
\draw[fill=black] (0.5, 0) circle (.05cm) node[above=0.5mm] {$X_2$};
\draw[black] (0, 0) -- node[below] {$-Y$} (0.5, 0);
\end{tikzpicture} 
}\right] \,.
\end{equation}
In this limit, the residue at the partial-energy branch point $X_1+Y\to0$ is proportional to a shifted one-site graph with an attached massive leg. An analogous statement holds upon exchanging $X_1\leftrightarrow X_2$. The total-energy pole is obtained in the limit $X_1+X_2\to0$, and originates from the fully nested contribution $\widehat{\I}_{++}$ (together with its conjugate counterpart $\widehat{\I}_{--}$). This factorisation structure generalises to arbitrary tree correlators.

\vskip 4pt
Notably, these singularities do not lie within the physical region of kinematic space, as accessing them requires negative energies, i.e. sums of external momentum magnitudes, and is therefore only possible via analytic continuation. A central objective is thus to understand how correlators extend away from these singular loci. In this work, we show that the behaviour of massive correlators away from the singularities is governed by the hypergeometric structure of elementary vertex building blocks, with the singularities themselves corresponding to the associated hypergeometric singular points. An important additional observation is that, for cosmological correlators in the Bunch–Davies vacuum, the Euclidean region coincides with the physical region.

\subsection{Spectral propagators}
\label{subsec: spectral propagators}

The main difficulty in evaluating the master integrals~\eqref{eq: def dimensionless master integrals} stems from the nested structure of the time integrations, which originates from the (anti-)time-ordered bulk-to-bulk propagators. Using harmonic analysis, these propagators $\widehat{D}^{(\mu)}_{\pm\pm}$ can be expressed in a spectral representation as integrals over products of two Hankel functions, thereby achieving factorisation. In this representation, time ordering is effectively traded for an additional spectral integration. Explicitly, this split representation reads~\cite{Melville:2024ove, Werth:2024mjg}:\footnote{See~\cite{Belrhali:2026rkn} for a formal derivation of~\eqref{eq: split representation ++}.}
\begin{equation}
\label{eq: split representation ++}
    \widehat{D}_{++}^{(\mu)}(\tfrac{Y}{X}z, \tfrac{Y}{X'}z') =
    e^{-\frac{i\pi}{2}} (z z')^{\frac{d}{2}}
    \int\limits_{-\infty}^{+\infty} [\d\nu] \rho_\nu(\mu)\, e^{+\pi\nu} \, H_{i\nu}^{(2)}(-\tfrac{Y}{X} z) H_{i\nu}^{(2)}(-\tfrac{Y}{X} z') \,.
\end{equation}
where the integration measure includes the de Sitter density of states $[\d\nu] \equiv \d\nu\N_\nu$ with 
\begin{equation}
     \N_\nu \equiv \frac{\nu}{\pi} \sinh(\pi\nu) = \frac{1}{\Gamma[\pm i\nu]} \,,
\end{equation}
where we use the short notation for the product of $\Gamma$-function, $\Gamma[\pm a] \equiv \Gamma(+a)\Gamma(-a)$. Notice that it is shadow symmetric, $\N_\nu = \N_{-\nu}$. The tree-level spectral density includes the following pole prescription
\begin{equation}
\label{eq: iepsilon prescription}
    \rho_\nu(\mu) \equiv \frac{1}{(\nu^2-\mu^2)_{i\epsilon}} = \frac{1}{2\sinh(\pi\mu)}\left[\frac{e^{+\pi\mu}}{\nu^2-\mu^2+i\epsilon} - \frac{e^{-\pi\mu}}{\nu^2-\mu^2-i\epsilon}\right] \,,
\end{equation}
and generalises the usual flat-space pole prescription for the Feynman propagator by accounting for spontaneous particle production, as can be seen from the Boltzmann weights $e^{\pm\pi\mu}$ associated with the on-shell poles. Naturally, the anti-time ordered propagator is given by $\widehat{D}_{--}^{(\mu)} = [\widehat{D}_{++}^{(\mu)}]^*$, which explicitly reads
\begin{equation}
\label{eq: split representation --}
    \widehat{D}_{--}^{(\mu)}(\tfrac{Y}{X}z, \tfrac{Y}{X'}z') =
    e^{+\frac{i\pi}{2}} (z z')^{\frac{d}{2}} 
    \int\limits_{-\infty}^{+\infty} [\d\nu]\rho_\nu(\mu)\, e^{-\pi\nu} \, H_{i\nu}^{(1)}(-\tfrac{Y}{X} z) H_{i\nu}^{(1)}(-\tfrac{Y}{X} z') \,,
\end{equation}
where the Boltzmann weights associated with on-shell poles are inverted in the $i\epsilon$-prescription. The two dimensionless times $z$ and $z'$ appear manifestly factorised inside the spectral integral. 

\subsubsection*{On-shell poles}

As we will see in the following, it is convenient to rewrite the $i\epsilon$-prescription~\eqref{eq: iepsilon prescription} in a distributional sense:
\begin{equation}
\label{eq: iepsilon distribution}
    \frac{1}{(\nu^2-\mu^2)_{i\epsilon}} = 
    \delta\left(\frac{1}{\nu^2 - \mu^2}\right)+ \P\left(\frac{1}{\nu^2 - \mu^2}\right) \,,
\end{equation}
where $\delta$ denotes the on-shell contribution---which will be specified and discussed at length in Sec.~\ref{subsec: on-shell poles}---, and $\P$ denotes the Cauchy principal value. The first term is only non-vanishing at the on-shell poles $\nu=\pm\mu$, which we will sometimes refer to as the particle production (or on-shell) poles. 
Importantly, when computing correlators, the $\delta$-function encoded in this on-shell contribution trivialises the spectral integration and the graph factorises. This on-shell contribution contains non-analytic dependences on energy ratios. Specifically, it encodes the {\it local} signal (non-analytic in $X/X'$) and the {\it non-local} signal (non-analytic in the internal energy $Y$), see e.g.~\cite{Tong:2021wai} for more details. We provide extensive details on collecting these on-shell poles in Sec.~\ref{subsec: on-shell poles}.

\subsubsection*{Analytic poles}

The second term in~\eqref{eq: iepsilon distribution} instead generates a contribution to the correlator $\widehat{\I}_{++}^\P$ that is analytic in the energies, and schematically reads
\begin{equation}
    \widehat{\I}_{++}^\P
    \propto \int\limits_{-\infty}^{+\infty} [\d\nu]\rho_{\nu, \P}(\mu) \, \V^{(1)}_{+, \nu}(X_1, Y; p_1) \V^{(1)}_{+, \nu}(X_2, Y; p_2)\,.
\end{equation}
With a slight abuse of notation, the subscript $\P$ here denotes that we retain only the second term in the $i\epsilon$-prescription~\eqref{eq: iepsilon distribution}, corresponding to the Cauchy principal value. We have introduced the so-called ``vertex function'',
\begin{equation}
    \V^{(1)}_{+, \mu} (X, Y; p) \equiv e^{+\frac{\pi}{2}\mu}\int_{-\infty}^0 \frac{\d z}{(-z)^{d+1}} (-z)^{p+\frac{d}{2}} e^{+iz} H_{i\mu}^{(2)}(-\tfrac{Y}{X} z) \,.
\end{equation}
These objects will play a central role in the following and will be studied in detail in Sec.~\ref{sec: vertex functions}. The contribution $\widehat{\I}_{++}^\P$ is often referred to as the EFT part, as it persists when the exchanged particle is integrated out. In the context of scattering amplitudes, the analogous contribution corresponds to the sum of effective contact interactions obtained by expanding the massive propagator $(-\Box-m^2)^{-1}$ at large mass $m^2$ (see Sec.~\ref{subsec: large-mass expansion} for more details). For cosmological correlators, this term is also sometimes called the {\it background}, as it does not exhibit non-analytic signatures of the exchanged particle. 

\vskip 4pt
Evaluating the remaining spectral integration constitutes the core of the spectral gluing method developed in this work and will be discussed at length in what follows. 

\subsection{Factorisation \& gluing}

By introducing a spectral integral for each (anti-)time-ordered bulk-to-bulk propagator, it becomes clear that any tree graph can be {\it entirely} factorised in time, thereby eliminating the challenge to evaluate nested time integrals. Naturally, non-time-ordered bulk-to-bulk propagators do not require introducing a spectral integral, as they are already factorised in time. The first step is therefore to factorise (or cut) all internal lines, and set the mass-parameter $\mu_j$ of each (anti-)time-ordered internal line---connecting two dots with same colour---off shell. To illustrate this, an explicit example of the factorisation of a tree-level graph is:
\begin{center}
\begin{tikzpicture}[line width=1. pt, scale=2]
    \draw[black] (0, 0) -- (0.7, 0);
    \draw[black] (0.7, 0) -- (1.2, 0.5);
    \draw[black] (0.7, 0) -- (1.2, -0.5);
    \draw[black] (0, 0) -- (-0.5, 0.5);
    \draw[black] (0, 0) -- (-0.5, -0.5);
    \draw[black] (-0.5, -0.5) -- (0, -1);
    \draw[black] (-0.5, -0.5) -- (-1.2, -0.5);
    
    \draw[fill=black] (0, 0) circle (.05cm);
    \draw[fill=black] (0.7, 0) circle (.05cm);
    \draw[fill=black] (1.2, 0.5) circle (.05cm);
    \draw[fill=white] (1.2, -0.5) circle (.05cm);
    \draw[fill=white] (-0.5, 0.5) circle (.05cm);
    \draw[fill=white] (-0.5, -0.5) circle (.05cm);
    \draw[fill=black] (0, -1) circle (.05cm);
    \draw[fill=white] (-1.2, -0.5) circle (.05cm);
    \draw[black, ->] (1.7, 0) -- node[below] {factorisation} (2.8, 0);
\end{tikzpicture} 
\begin{tikzpicture}[line width=1. pt, scale=2]
    \draw[black] (0, 0) -- (0.7, 0);
    \draw[white, line width=2pt] 
  ($(0,0)!1/3!(0.7,0)$) -- ($(0,0)!2/3!(0.7,0)$);
    \node at ($(0,0)!1/2!(0.7,0)$) {$\times$};
    \draw[black] (0.7, 0) -- (1.2, 0.5);
    \draw[white, line width=2pt] 
  ($(0.7,0)!1/3!(1.2,0.5)$) -- ($(0.7,0)!2/3!(1.2,0.5)$);
    \node at ($(0.7,0)!1/2!(1.2,0.5)$) {$\times$};
    \draw[black] (0.7, 0) -- (1.2, -0.5);
    \draw[white, line width=2pt] 
  ($(0.7,0)!1/3!(1.2,-0.5)$) -- ($(0.7,0)!2/3!(1.2,-0.5)$);
    \node at ($(0.7,0)!1/2!(1.2,-0.5)$) {$\times$};
    \draw[black] (0, 0) -- (-0.5, 0.5);
    \draw[white, line width=2pt] 
  ($(0,0)!1/3!(-0.5, 0.5)$) -- ($(0,0)!2/3!(-0.5, 0.5)$);
    \node at ($(0,0)!1/2!(-0.5,0.5)$) {$\times$};
    \draw[black] (0, 0) -- (-0.5, -0.5);
    \draw[white, line width=2pt] 
  ($(0,0)!1/3!(-0.5, -0.5)$) -- ($(0,0)!2/3!(-0.5, -0.5)$);
    \node at ($(0,0)!1/2!(-0.5,-0.5)$) {$\times$};
    \draw[black] (-0.5, -0.5) -- (0, -1);
    \draw[white, line width=2pt] 
  ($(-0.5, -0.5)!1/3!(0, -1)$) -- ($(-0.5, -0.5)!2/3!(0, -1)$);
    \node at ($(-0.5, -0.5)!1/2!(0, -1)$) {$\times$};
    \draw[black] (-0.5, -0.5) -- (-1.2, -0.5);
    \draw[white, line width=2pt] 
  ($(-0.5, -0.5)!1/3!(-1.2, -0.5)$) -- ($(-0.5, -0.5)!2/3!(-1.2, -0.5)$);
    \node at ($(-0.5, -0.5)!1/2!(-1.2, -0.5)$) {$\times$};
  
    \draw[fill=black] (0, 0) circle (.05cm);
    \draw[fill=black] (0.7, 0) circle (.05cm);
    \draw[fill=black] (1.2, 0.5) circle (.05cm);
    \draw[fill=white] (1.2, -0.5) circle (.05cm);
    \draw[fill=white] (-0.5, 0.5) circle (.05cm);
    \draw[fill=white] (-0.5, -0.5) circle (.05cm);
    \draw[fill=black] (0, -1) circle (.05cm);
    \draw[fill=white] (-1.2, -0.5) circle (.05cm);
\end{tikzpicture} 
\end{center}
This example requires introducing three spectral integrals, and therefore setting three internal mass parameters off shell, to reconstruct the full graph. The complete factorisation of a tree-level graph naturally reveals a set of elementary building blocks: the vertex functions. These building blocks are defined by single time integrals. The remaining task is then to perform the spectral integrations in order to reconstruct the original graph. We refer to this procedure as spectral {\it gluing}, which is schematically represented by:
\begin{equation*}
    \vcenter{\hbox{
\begin{tikzpicture}
[line width=1. pt, scale=2]
    \draw[fill=gray!50] (0, 0) circle (0.3cm);
    \draw[fill=gray!50] (1.6, 0) circle (0.3cm);
    \draw[fill=black] (0.3, 0) circle (.05cm) node[above right=0.1mm] {};
    \draw[fill=black] (1.3, 0) circle (.05cm) node[above=0.5mm] {};
    \draw[black] (0.3, 0) -- (1.3, 0);
    \draw[white, line width=2pt] 
  ($(0.3, 0)!1/3!(1.3, 0)$) -- ($(0.3, 0)!2/3!(1.3, 0)$);
    \node at ($(0.3, 0)!1/2!(1.3, 0)$) {$\times$};
    \draw[black, ->] (2.2, 0) -- node[below] {gluing} (2.8, 0);
\end{tikzpicture} 
    }} \int\limits_{-\infty}^{+\infty} [\d\nu] \rho_\nu(\mu_j)
    \vcenter{\hbox{
\begin{tikzpicture}
[line width=1. pt, scale=2]
    \draw[fill=gray!50] (0, 0) circle (0.3cm);
    \draw[fill=gray!50] (1.6, 0) circle (0.3cm);
    \draw[fill=black] (0.3, 0) circle (.05cm) node[above right=0.1mm] {$\nu$};
    \draw[fill=black] (1.3, 0) circle (.05cm) node[above left=0.1mm] {$\nu$};
    \draw[black] (0.3, 0) -- (1.3, 0);
    \draw[white, line width=2pt] 
  ($(0.3, 0)!1/3!(1.3, 0)$) -- ($(0.3, 0)!2/3!(1.3, 0)$);
  \node at ($(0.3, 0)!1/2!(1.3, 0)$) {$\times$};
\end{tikzpicture} 
    }}
\end{equation*}
for a time-ordered internal line $j$ with mass parameter $\mu_j$, and where the grey circles denote the remain subgraphs. Remarkably, the elementary vertex functions are universally meromorphic in the whole complex $\nu$-plane and evaluating the spectral integrals amounts to collecting residues at an (infinite) set of poles.

\vskip 4pt
Having established the strategy for reconstructing multiple-exchange massive cosmological correlators, the remaining tasks are twofold: (i) to compute the vertex functions and determine their analytic structure in the complex $\nu$-plane, and (ii) to understand how to perform the spectral integrations in general settings. These issues will be addressed in the remainder of this paper.

\newpage
\section{Vertex Functions}
\label{sec: vertex functions}

After factorising tree-level massive cosmological correlators, the fundamental building blocks are vertex functions, each involving a single time integral. We start off with the simplest case, namely the single-massive-leg vertex, which is generated by a single massive mode function. Such vertices play a central role, as they generically arise as endpoints of massive diagrams. In particular, we show that their folded limit can be obtained straightforwardly, owing to the well-known analytic continuation properties of hypergeometric functions. In Sec.~\ref{subsec: two-massive-leg vertex}, we study in detail the two-massive-leg vertex function and determine its analytic continuation across the entire physical kinematic region. In Sec.~\ref{subsec: generalisation to arbitrary massive legs}, we extend this analysis to an arbitrary number of massive legs, showing that vertex functions universally belong to the class of Lauricella functions, i.e.~multivariable generalisations of hypergeometric functions.

\subsection{Single-massive-leg vertex \& degenerate limit}
\label{subsec: single-massive vertex function}

Consider the vertex function with a single massive leg carrying internal energy $Y$ and characterised by a mass parameter $\mu$ in the principal series, $\mu\in\mathbb{R}_{>0}$. Although we restrict to this case for concreteness, the results can be extended beyond the principal series with appropriate care in analytically continuing in $\mu$. The vertex also depends on the external energy $X$ and the twist $p$. Focusing on the $\aa=+$ case, the vertex function is defined as follows:
\begin{equation}
\label{eq: V1 definition}
    \V_{+, \mu}^{(1)}(X, Y; p) \equiv e^{+\frac{\pi}{2}\mu} \int_{-\infty}^0 \frac{\d z}{(-z)^{d+1}} (-z)^{p+\frac{d}{2}}e^{iz} H_{i\mu}^{(2)}(-\tfrac{Y}{X} z) \,,
\end{equation}
where the additional power $(-z)^{\frac{d}{2}}$ comes from the massive mode function. Schematically, we represent this vertex function as:
\begin{center}
\begin{tikzpicture}[line width=1. pt, scale=2]
    \draw[black] (-1, 0) -- node[below] {$(Y, \mu)$} (0, 0);

    \draw[gray!40] (0, 0) -- (0.5, 0.5) node[right] {$k_1$};
    \draw[gray!40] (0, 0) -- (0.5, 0.25) node[right] {$k_2$};
    \draw[gray!40] (0, 0) -- (0.5, 0) node[right] {$k_3$};
    \node[gray!40] at ($(0.5, 0)!1/3!(0.5, -0.5)$) {$\vdots$};
    \draw[gray!40] (0, 0) -- (0.5, -0.5) node[right] {$k_m$};
    
    \draw[fill=black] (0, 0) circle (.05cm) node[above] {$(X, p)$};
\end{tikzpicture} 
\end{center}
where $X = k_1 + \cdots + k_m$ is the sum of momentum magnitudes of conformally coupled scalar fields attached to the vertex. The generalised triangle inequality dictates the physical region in kinematic space: $Y\leq X$, with the degenerate (folded) case being $Y=X$. This case will be treated separately in the following. The associated $\aa=-$ counterpart $\V^{(1)}_{-, \mu}$ is found by the formal replacement $e^{iz} \to e^{-iz}$, $H^{(2)}_{i\mu} \to H^{(1)}_{i\mu}$ and $e^{+\frac{\pi}{2}\mu} \to e^{-\frac{\pi}{2}\mu}$.

\vskip 4pt
In order to compute this integral, we use the convenient representation of the Hankel function as a Mellin-Barnes integral, given by
\begin{equation}
\label{eq: Hankel Mellin integral}
    H_{i\mu}^{(2)}(z) = \frac{1}{\pi}\int_{-i\infty}^{+i\infty} \frac{\d s}{2i\pi} \left(\frac{z}{2}\right)^{-2s} e^{-\frac{i\pi}{2}(2s-i\mu-1)}\Gamma[s\pm \tfrac{i\mu}{2}] \,,
\end{equation}
where the integration contour runs along the imaginary axis, and where we again use the short notation $\Gamma[\pm a] \equiv \Gamma(+a)\Gamma(-a)$ for a product of $\Gamma$-functions. At the level of the Mellin integrand, the dependence on time $z$ is a simple power-law. Inserting Eq.~\eqref{eq: Hankel Mellin integral} into the vertex function yields the following Mellin representation of the vertex function
\begin{equation}
\label{eq: V1 Mellin representation}
    \V_{+, \mu}^{(1)}(X, Y; p) = \C_+^{(1)}(p)  \int_{-i\infty}^{+i\infty} \frac{\d s}{2i\pi} \left(\frac{Y}{2X}\right)^{-2s} \Gamma[s\pm \tfrac{i\mu}{2}, p-\tfrac{d}{2}-2s] \,,
\end{equation}
where we have used
\begin{equation}
\label{eq: elementary time integral}
    \int_{-\infty}^0 \d z (-z)^{\alpha-1}e^{iz} = e^{-\frac{i\pi\alpha}{2}} \Gamma(\alpha) \,,
\end{equation}
which formally converges for $\alpha>0$ and is defined elsewhere by analytic continuation. The overall prefactor, which includes a phase, reads
\begin{equation}
\label{eq: overall C1+}
    \C_+^{(1)}(p) \equiv \tfrac{1}{\pi} e^{-\frac{i\pi}{2}(p-\frac{d}{2}-1)} \,.
\end{equation}
The dimensionless vertex function only depends on the external kinematic ratio
\begin{equation}
    u \equiv \frac{Y}{X} \in [0, 1] \,.
\end{equation}
In Eq.~\eqref{eq: V1 Mellin representation}, the integration contour separates poles on the left, located at $s=\pm \tfrac{i\mu}{2}-n$, and poles on the right, located at $s=\tfrac{1}{2}(p-\tfrac{d}{2}+n)$, with $n\in\mathbb{N}$, as shown in Fig.~\ref{fig: Mellin space poles}. Closing the contour and picking the residues therefore depends on the magnitude of $u$. For $u<1$ in the physical region, we close the contour to the left as the arc at infinity vanishes exponentially, and we collect the residues of the two infinite towers of poles, giving
\begin{equation}
\label{eq: V1 expression}
    \boxed{
    \V_{+, \mu}^{(1)}(u; p) = \C_+^{(1)}(p)  \sum_{\alpha=\pm i\mu} \F^{(1)}_{\alpha}(u; p-\tfrac{d}{2}) \,,
    }
\end{equation}
where we have introduced the function
\begin{equation}
\label{eq: F1 function}
    \boxed{
    \F^{(1)}_\alpha(u; p) \equiv \Gamma[p+\alpha, -\alpha, 1+\alpha] \left(\frac{u}{2}\right)^{\alpha} \bar{F}^{(1)}_C\left[\left.\begin{matrix} \frac{p+\alpha}{2},\,  \frac{p+1+\alpha}{2}\\ 1+\alpha \end{matrix}\right\vert u^2\right] \,,
    }
\end{equation}
where $\bar{F}_C^{(1)}$ is the regularised type-$C$ Lauricella function of order $1$ (of a single variable) which reduces to Gau\ss'~hypergeometric function $F_C^{(1)} \equiv {}_2F_1$, see App.~\ref{sec: special functions} for more details. Importantly, the two leading power-law behaviours $u^{\pm i\mu}$ correspond to the positive- and negative-frequency modes of the massive internal leg near the future boundary $z\to0$. In Mellin space, these are associated with the two leading poles $s=\pm i\mu/2$. As will become clear, particularly in the context of spectral gluing, these contributions govern the large-$\mu$ behaviour of the vertex functions and determine how the spectral contour can be closed. In this way, the resummation over subleading poles effectively reconstructs the underlying hypergeometric structure of the vertex function. Eventually, the $\aa=-$ vertex function is given by Eq.~\eqref{eq: V1 expression} with the replacement $\C_+^{(1)}  \to \C_-^{(1)} $, where the overall prefactor is simply given by $\C_-^{(1)}(p) = \C_+^{(1)*}(p)$.

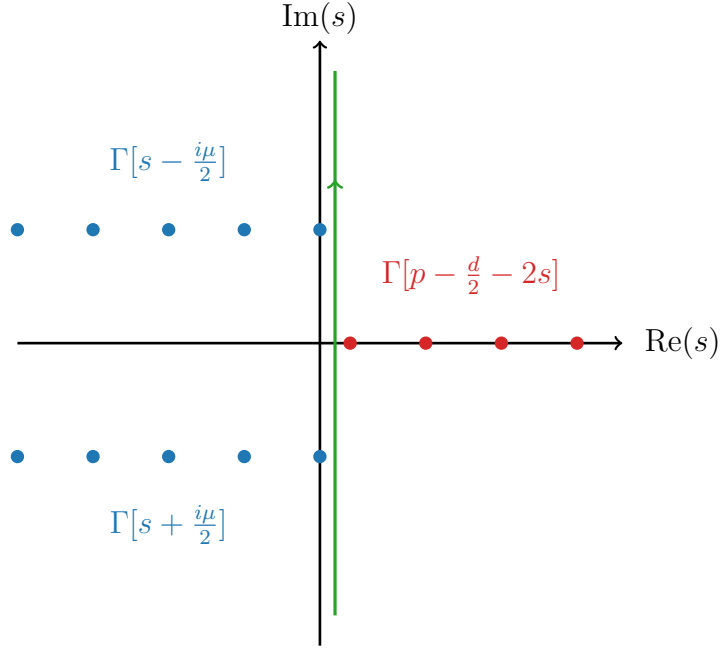
\begin{figure}[h!]
    \centering
    \hspace*{1cm}
    \begin{tikzpicture}[line width=1. pt, scale=2]
		
		\draw[->] (0, -2) -- (0, 2);
		\draw[->] (-2, 0) -- (2, 0);

        \draw[pygreen, line width=1.2pt,
                        postaction={decorate},
                        decoration={markings, mark=at position 0.8 with {\arrow{>}}}] (0.1, -1.8) -- (0.1, 1.8);

        \node[pyblue] at (-1, 1.2) {$\Gamma[s-\tfrac{i\mu}{2}]$};
        \filldraw[pyblue] (0, 0.75) circle (1pt);
        \filldraw[pyblue] (-0.5, 0.75) circle (1pt);
        \filldraw[pyblue] (-1, 0.75) circle (1pt);
        \filldraw[pyblue] (-1.5, 0.75) circle (1pt);
        \filldraw[pyblue] (-2, 0.75) circle (1pt);

        \node[pyblue] at (-1, -1.2) {$\Gamma[s+\tfrac{i\mu}{2}]$};
        \filldraw[pyblue] (0, -0.75) circle (1pt);
        \filldraw[pyblue] (-0.5, -0.75) circle (1pt);
        \filldraw[pyblue] (-1, -0.75) circle (1pt);
        \filldraw[pyblue] (-1.5, -0.75) circle (1pt);
        \filldraw[pyblue] (-2, -0.75) circle (1pt);

        \node[pyred] at (1, 0.45) {$\Gamma[p-\tfrac{d}{2}-2s]$};
        \filldraw[pyred] (0.2, 0) circle (1pt);
        \filldraw[pyred] (0.7, 0) circle (1pt);
        \filldraw[pyred] (1.2, 0) circle (1pt);
        \filldraw[pyred] (1.7, 0) circle (1pt);

		\node at (0, 2.15) {$\text{Im}(s)$};
		\node at (2.4, 0) {$\text{Re}(s)$};
	\end{tikzpicture}
    \caption{Integration contour for the Mellin-Barnes representation~\eqref{eq: V1 Mellin representation} of the vertex function $\V^{(1)}$, with the corresponding towers of poles, located at $s=\mp \tfrac{i\mu}{2}-n$ and $s=\tfrac{1}{2}(p-\tfrac{d}{2}+n)$, with $n\in\mathbb{N}$. The choice of whether to close the contour to the left ($u<1$) or to the right ($u>1$) is dictated by the external kinematic ratio $u$, which, in the physical region, takes values in the interval $[0, 1]$.}
    \label{fig: Mellin space poles}
\end{figure}

\subsubsection*{Analytic continuation}

The Mellin-Barnes representation of the vertex function~\eqref{eq: V1 Mellin representation} is suitable for exploring other regions in kinematic space, not necessarily physical, as we now illustrate. For $u>1$, the integration contour can be closed to the right in Mellin space, see Fig.~\ref{fig: Mellin space poles}, which gives
\begin{equation}
    \V_{+, \mu}^{(1)}(u; p) = \frac{\C_+^{(1)}(p)}{2} \,  \sum_{n=0}^{\infty} \frac{(-1)^n}{n!} \Gamma\left[\frac{p-\tfrac{d}{2}+n \pm i\mu}{2}\right] \left(\frac{2}{u}\right)^{n+p-\frac{d}{2}} \,,
\end{equation}
where the overall factor $\tfrac{1}{2}$ comes from the factor $\Gamma[\cdots-2z]$ when taking the residue. This series can once more be resummed into a Gau\ss'~hypergeometric ${}_2F_1$ function; however, the resulting expression is not particularly illuminating, and we therefore refrain from presenting it explicitly. Allowing the external fields to have a reduced speed of sound, $c_s<1$, makes this region accessible, since the kinematic ratio becomes $u=Y/(c_s X)<1/c_s$, and can therefore lie outside the unit disc.

\subsubsection*{Soft limit}

The vertex function $\V_{+, \mu}^{(1)}(u; p)$ has a branch point at $u=0$, with the corresponding cut from $0$ to $-\infty$ on the real $u$-axis. This branch cut lies deeply in the unphysical region and thus cannot be directly probed by physically measured correlators. However, this translates to the non-analyticity of the vertex function in the complex $u$-plane, where the leading non-analytic behaviour is given by
\begin{equation}
\label{eq: single-massive-leg soft limit}
    \lim_{u\to0^+}\V_{+, \mu}^{(1)}(u; p) = \C_+^{(1)}(p)  \sum_{\alpha=\pm i\mu} \Gamma[p-\tfrac{d}{2}+\alpha, -\alpha]\left(\frac{u}{2}\right)^{\alpha} \,,
\end{equation}
as one approaches the soft limit $u\to0$. Notice that the hypergeometric structure drops out completely, as ${}_2F_1[\cdots|0]=1$. For a real twist $p\in\mathbb{R}$, the sum over both positive and negative modes is manifestly real and generates oscillations in $\log(u)$ with frequency determined by $\mu$.

\subsubsection*{Partial-energy singularity}

The vertex function also becomes singular when approaching the partial-energy kinematic locus $u\to-1$, i.e.~$X+Y\to0$. This corresponds to the sum of the energies entering the vertex adding up to zero. To probe this limit, we come back to the time-integral representation~\eqref{eq: V1 definition} and use the early-time asymptotic form of the Hankel function $H_{i\mu}^{(2)}(-uz) \propto (-uz)^{-1/2}e^{iuz}$, where we neglect unimportant overall factors and phases. The vertex function therefore develops a logarithmic cut
\begin{equation}
    \lim_{u\to-1}\V_{+, \mu}^{(1)}(u; p) \propto \int_{-\infty}^0 \frac{\d z}{(-z)^{d+1}} (-z)^{p+\frac{d}{2}-\frac{1}{2}}e^{(1+u)iz} \propto \Gamma[p-\tfrac{d}{2}-\tfrac{1}{2}] (1+u)^{\tfrac{d}{2}-p+\tfrac{1}{2}} \,.
\end{equation}
We observe that the oscillatory phase precisely cancels for the special locus $u=-1$.

\subsubsection*{Degenerate limit}

A particularly relevant kinematic configuration arises when momentum conservation at the vertex enforces equality between the external energy $X$ and the internal energy $Y$, corresponding to $u=1$. Reaching this physical kinematic locus amounts to taking the folded limit: 
\begin{center}
\begin{tikzpicture}[line width=1. pt, scale=2]
    \draw[black, postaction={decorate},
                        decoration={markings, mark=at position 0.5 with {\arrow{<}}}] 
                        (0, 0) -- node[below] {$\Y$} (1.5, 0);
    \draw[gray!40, postaction={decorate},
                        decoration={markings, mark=at position 0.5 with {\arrow{>}}}] 
                        (0, 0) -- node[left] {$\k_1$} (0.2, 0.5);
    \draw[gray!40, postaction={decorate},
                        decoration={markings, mark=at position 0.5 with {\arrow{>}}}] 
                        (0.2, 0.5) -- node[above left] {$\k_2$} (0.5, 0.7);
    \draw[gray!40, postaction={decorate},
                        decoration={markings, mark=at position 0.5 with {\arrow{>}}}] 
                        (0.5, 0.7) -- node[above] {$\k_3$} (1, 0.7);
    \draw[gray!40, postaction={decorate},
                        decoration={markings, mark=at position 0.5 with {\arrow{>}}}] 
                        (1.2, 0.5) -- node[above right] {$\k_n$} (1.5, 0);
    \node[gray!40] at ($(1, 0.7)!1/3!(1.2, 0.5)$) {$\cdot$};
    \node[gray!40] at ($(1, 0.7)!2/3!(1.2, 0.5)$) {$\cdot$};
    \draw[black, ->] (2, 0.35) -- node[below] {folded limit} (3, 0.35);
\end{tikzpicture} 
\hspace*{0.2cm}
\begin{tikzpicture}[line width=1. pt, scale=2]
    \draw[black, postaction={decorate},
                        decoration={markings, mark=at position 0.5 with {\arrow{<}}}] 
                        (0, 0) -- node[below] {$\Y$} (1.5, 0);
    \draw[gray!40, postaction={decorate},
                        decoration={markings, mark=at position 0.5 with {\arrow{>}}}] 
                        (0, 0) -- node[above] {$\k_1$} (0.2, 0.05);
    \draw[gray!40, postaction={decorate},
                        decoration={markings, mark=at position 0.5 with {\arrow{>}}}] 
                        (0.2, 0.05) -- node[above] {$\k_2$} (0.5, 0.05);
    \draw[gray!40, postaction={decorate},
                        decoration={markings, mark=at position 0.5 with {\arrow{>}}}] 
                        (0.5, 0.05) -- node[above] {$\k_3$} (1, 0.05);
    \draw[gray!40, postaction={decorate},
                        decoration={markings, mark=at position 0.5 with {\arrow{>}}}] 
                        (1.2, 0.05) -- node[above] {$\k_n$} (1.5, 0);
    \node[gray!40] at ($(1, 0.05)!1/3!(1.2, 0.05)$) {$\cdot$};
    \node[gray!40] at ($(1, 0.05)!2/3!(1.2, 0.05)$) {$\cdot$};
\end{tikzpicture} 
\end{center}
where the sum of external momenta $\sum_{i=1}^n \k_i$ becomes opposite to the internal momentum $\Y$. This kinematic configuration is also known as the collinear limit, and is of phenomenological importance as it naturally appears in the presence of linear mixings. 

\vskip 4pt
The known analytic continuation property of the hypergeometric function allows to easily reach this limit. Yet, each individual frequency mode $\alpha=\pm i\mu$ of the vertex function~\eqref{eq: V1 expression} is singular at $u\to1$. Indeed, for $\Re(1/2-p+d/2)<0$, which covers the vast majority of phenomenologically relevant graphs, $\F_\alpha^{(1)}(u; p-\tfrac{d}{2})$ develops a power-law branch point at $u\to1$ and the limit is singular. The marginal case $\Re(1/2-p+d/2)=0$ gives rise to a logarithmic branch point, and the case $\Re(1/2-p+d/2)>0$ is finite. However, summing over positive- and negative-frequency modes of the vertex function yields a finite expression as the behaviour around the singularity does not depend on $\alpha$. Explicitly, using the connection formula
\begin{equation}
    {_2F_1}\left[\left.\begin{matrix} a, b\\ c \end{matrix}\right\vert 1-\epsilon\right] = \Gamma\left[\begin{matrix} c, c-a-b \\ c-a, c-b \end{matrix}\right] + \Gamma\left[\begin{matrix} c, a+b-c \\ a, b \end{matrix}\right] \epsilon^{c-a-b} + \O(\epsilon) \,,
\end{equation}
around $\epsilon\to0$ ($\epsilon>0$), each mode $\F_\alpha^{(1)}$ decomposes near $u\to1^-$ as
\begin{equation}
    \F_\alpha^{(1)}(u; p-\tfrac{d}{2}) = \R_\alpha + \mathcal{S}_\alpha (1-u^2)^{\frac{1}{2}-p+\frac{d}{2}} \,,
\end{equation}
where $\R_\alpha$ is the regular part and $\mathcal{S}_\alpha$ is the leading singular coefficient, which read
\begin{equation}
    \begin{aligned}
        \R_\alpha &= 2^{-\alpha} \Gamma\left[\begin{matrix} -\alpha, p-\tfrac{d}{2}+\alpha, 1+\alpha, \tfrac{1}{2}-p+\tfrac{d}{2} \\ 1+\tfrac{\alpha-p+d/2}{2}, \tfrac{1}{2}+\tfrac{\alpha-p+d/2}{2} \end{matrix}\right] \,, \\
        \mathcal{S}_\alpha &= 2^{-\alpha} \Gamma\left[\begin{matrix} -\alpha, p-\tfrac{d}{2}+\alpha, 1+\alpha, p-\tfrac{d}{2}-\tfrac{1}{2} \\ \tfrac{p-d/2+\alpha}{2}, \tfrac{p-d/2+1+\alpha}{2} \end{matrix}\right] \,.
    \end{aligned}
\end{equation}
We simplify the product of $\Gamma$-functions in the denominator of the singular part $\mathcal{S}_\alpha$ using the Legendre duplication formula
\begin{equation}
    \Gamma\left[\tfrac{p-d/2+\alpha}{2}, \tfrac{p-d/2+1+\alpha}{2}\right] = \frac{\sqrt{\pi}}{2^{p-d/2+\alpha-1}} \Gamma[p-\tfrac{d}{2}+\alpha] \,.
\end{equation}
Substituting into $\mathcal{S}_\alpha$ yields
\begin{equation}
    \mathcal{S}_\alpha = -\frac{2^{p-\tfrac{d}{2}-1} \sqrt{\pi} \, \Gamma[p-\tfrac{d}{2}-\tfrac{1}{2}]}{\sin(\pi\alpha)} \,.
\end{equation}
The total singular contribution to $\V_{+,\mu}^{(1)}$ is therefore proportional to 
\begin{equation}
    \sum_{\alpha=\pm i\mu} \mathcal{S}_\alpha = 0 \,.
\end{equation}
The singular terms cancel identically for any mass parameter $\mu$, and the limit $u\to 1^-$ is instead entirely given by the regular parts $\R_\alpha$, giving
\begin{equation}
    \lim_{u\to1^-} \V_{+, \mu}^{(1)}(u; p) = \C_+^{(1)}(p)  \sum_{\alpha=\pm i\mu} \R_\alpha \,.
\end{equation}
The physical vertex function is therefore regular, as imposed in the Bunch-Davies vacuum. As we will see later, though, taking degenerate limits in specific diagrams can be subtle. In particular, the spectral gluing procedure selects either positive- or negative-frequency vertex functions, which naively suggests the presence of branch point singularities. However, these singularities are in fact spurious, as required by physical consistency: certain parameters of the hypergeometric function become negative in this limit, causing it to truncate to a polynomial.

\subsection{Case study: two-massive-leg vertex}
\label{subsec: two-massive-leg vertex}

We now present the analysis for the two-massive-leg vertex as a case study before generalising the framework to arbitrary vertex functions in the following section, where many of the intermediate steps discussed here are bypassed. This example is particularly instructive, as it constitutes the first instance of a vertex function that necessitates analytic continuation in order to be evaluated in the physical region of kinematic space. We consider the vertex function with two massive legs attached carrying internal energies $Y_1$ and $Y_2$, and associated mass parameters $\mu_1, \mu_2 \in \mathbb{R}_{>0}$, respectively, defined as
\begin{equation}
    \V_{+, \{\mu_j\}}^{(2)}(X, \{Y_j\}; p) \equiv e^{+\frac{\pi}{2}\mu_{12}} \int_{-\infty}^0 \frac{\d z}{(-z)^{d+1}} (-z)^{p+d} e^{iz} H_{i\mu_1}^{(2)}(-\tfrac{Y_1}{X}z) H_{i\mu_2}^{(2)}(-\tfrac{Y_2}{X}z) \,,
\end{equation}
where $\mu_{12} \equiv \mu_1+\mu_2$. Schematically, we represent this vertex function as:
\begin{center}
\begin{tikzpicture}[line width=1. pt, scale=2]
    \draw[black] (-1, -0.6) -- node[above left] {$(Y_1, \mu_1)$} (0, 0);
    \draw[black] (0, 0) -- node[above right] {$(Y_2, \mu_2)$} (1, -0.6);

    \draw[gray!40] (0, 0) -- (-0.8, 0.5) node[above] {$k_1$};
    \draw[gray!40] (0, 0) -- (-0.4, 0.5) node[above] {$k_2$};
    \draw[gray!40] (0, 0) -- (0, 0.5) node[above] {$k_3$};
    \node[gray!40] at ($(0, 0.5)!1/2!(0.8, 0.5)$) {$\cdots$};
    \draw[gray!40] (0, 0) -- (0.8, 0.5) node[above] {$k_m$};
    
    \draw[fill=black] (0, 0) circle (.05cm) node[below, yshift=-6pt] {$(X, p)$};
\end{tikzpicture} 
\end{center}
where $X = k_1 + \cdots + k_m$, and $p$ is the vertex twist. The $\aa=-$ vertex function is found by replacing $e^{iz} \to e^{-iz}$, $H^{(2)}_{i\mu_j} \to H^{(1)}_{i\mu_j}$ (for $j=1, 2$) and $e^{+\frac{\pi}{2}\mu_{12}} \to e^{-\frac{\pi}{2}\mu_{12}}$. In order to compute this integral, we use the Mellin-Barnes representation for the Hankel function~\eqref{eq: Hankel Mellin integral} to render the integral over the (dimensionless) time $z$ elementary:
\begin{equation}
\label{eq: V2 MB integral rep}
    \begin{aligned}
        \V_{+, \{\mu_j\}}^{(2)}(X, \{Y_j\}; p) = \C_+^{(2)}(p) \int_{-i\infty}^{+i\infty} &\frac{\d s_1}{2i\pi} \frac{\d s_2}{2i\pi} \left(\frac{Y_1}{2X}\right)^{-2s_1} \left(\frac{Y_2}{2X}\right)^{-2s_2} \\
        &\times\Gamma[s_1\pm \tfrac{i\mu_1}{2}, s_2\pm \tfrac{i\mu_2}{2}, p-2s_{12}] \,,
    \end{aligned}
\end{equation}
where we have used~\eqref{eq: elementary time integral}. Similarly to the single-massive-leg vertex studied previously, we define the overall prefactor:
\begin{equation}
    \C_+^{(2)}(p) = \tfrac{1}{\pi^2} e^{-\frac{i\pi}{2}(p-2)} \,.
\end{equation}
The dimensionless vertex function $\V_+^{(2)}$ depends on two dimensionless energy ratios:
\begin{equation}
    u_1 \equiv \frac{Y_1}{X} \,, \quad u_2 \equiv \frac{Y_2}{X} \,.
\end{equation}
The admissible range of these variables and the definition of the Euclidean region will be discussed in due course. Around the double-soft limit $u_1, u_2 \to 0$, we close both Mellin-Barnes contours to the left, and collect the four towers of poles. The found series arrange in the following form:
\begin{equation}
\label{eq: V2 expression}
    \boxed{
    \V_{+, \mu_1, \mu_2}^{(2)}(u_1, u_2; p) = \C_+^{(2)}(p) \sum\limits_{\substack{\alpha_1 =\pm i\mu_1 \\ \alpha_2 = \pm i\mu_2}} \F_{\alpha_1, \alpha_2}^{(2)}(u_1, u_2; p) \,,
    }
\end{equation}
where we define the function
\begin{equation}
\label{eq: F2 function}
    \boxed{
    \begin{aligned}
        \F_{\alpha_1, \alpha_2}^{(2)}(u_1, u_2; p) &\equiv \Gamma[p+\alpha_{12}, -\alpha_1, 1+\alpha_1, -\alpha_2, 1+\alpha_2] \\
        &\times\left(\frac{u _1}{2}\right)^{\alpha_1} \left(\frac{u_2}{2}\right)^{\alpha_2} \bar{F}^{(2)}_C\left[\left.\begin{matrix} \frac{p+\alpha_{12}}{2},\,  \frac{p+1+\alpha_{12}}{2}\\ 1+\alpha_1, \, 1+\alpha_2 \end{matrix}\right\vert u_1^2, u_2^2\right] \,,
    \end{aligned}
    }
\end{equation}
with $\alpha_{12} \equiv \alpha_1+\alpha_2$, and where $\bar{F}_C^{(2)}$ is the regularised type-$C$ Lauricella function of order $2$ (of two variables) which reduces to the Appell function $F_C^{(2)} \equiv F_4$, see App.~\ref{sec: special functions} for more details. The sum in~\eqref{eq: V2 expression} contains four terms, each arising from selecting the positive- or negative-frequency mode of either massive mode. The $\aa=-$ vertex function takes the same form as~\eqref{eq: V2 expression} albeit performing the replacement $\C_+^{(2)}(p) \to \C_-^{(2)}(p) = \C_+^{(2)*}(p)$, since the the vertex function is manifestly shadow symmetric both in $\mu_1$ and $\mu_2$. Eventually, it is important to emphasize that the functions~\eqref{eq: F2 function} are meromorphic in the entire complex $\mu_1$- and $\mu_2$-planes, and has only poles at locations dictated by the overall $\Gamma$-functions, since the regularised $\bar{F}^{(2)}_C$ function is analytic.

\subsubsection*{Analytic continuation in the Euclidean region}

Since a two-massive-leg vertex always appears as an internal graph vertex, the variables $u_1$ and $u_2$ are not constrained by a generalised triangle inequality. By momentum conservation, the internal kinematic variables are arranged as follows:
\begin{center}
\begin{tikzpicture}[line width=1. pt, scale=2]
    \draw[black, postaction={decorate},
                        decoration={markings, mark=at position 0.5 with {\arrow{>}}}] 
                        (-1.2, 0) -- node[below] {$\Y_1$} (0, -0.5);
    \draw[black, postaction={decorate},
                        decoration={markings, mark=at position 0.5 with {\arrow{>}}}] 
                        (0, -0.5) -- node[below] {$\Y_2$} (1.2, 0);

    \draw[gray!40, postaction={decorate},
                        decoration={markings, mark=at position 0.5 with {\arrow{>}}}] 
                        (-1.2, 0) -- node[above left] {$\k_1$} (-0.8, 0.6);
    \draw[gray!40, postaction={decorate},
                        decoration={markings, mark=at position 0.5 with {\arrow{>}}}] 
                        (-0.8, 0.6) -- node[above left] {$\k_2$} (-0.4, 0.8);
    \draw[gray!40, postaction={decorate},
                        decoration={markings, mark=at position 0.5 with {\arrow{>}}}] 
                        (-0.4, 0.8) -- node[above] {$\k_3$} (0, 0.8);
    \draw[gray!40, postaction={decorate},
                        decoration={markings, mark=at position 0.5 with {\arrow{>}}}] 
                        (0.8, 0.6) -- node[above right] {$\k_m$} (1.2, 0);
    \node[gray!40] at ($(0, 0.8)!1/3!(0.8, 0.6)$) {$\cdot$};
    \node[gray!40] at ($(0, 0.8)!2/3!(0.8, 0.6)$) {$\cdot$};
\end{tikzpicture} 
\end{center}
It becomes clear that the Euclidean (physically allowed) region is defined by $Y_1 \leq Y_2+X$ and $Y_2 \leq Y_1+X$, which is equivalent, together with the condition of positive energies, to
\begin{equation}
\label{eq: Euclidean region for V2}
    \{0\leq u_1, 0\leq u_2, u_1 \leq u_2+1 \,, u_2 \leq u_1+1\} \,.
\end{equation}
Notice that the last two conditions can be encompassed in the compact condition $|u_1 - u_2| \leq 1$. The Euclidean region~\eqref{eq: Euclidean region for V2} is shown in Fig.~\ref{fig: kinematic space for V2} as the union $(\textbf{I})\cup(\textbf{IV})$. However, the Appell $F_4$ series representation~\eqref{eq: Appell F4 series rep}---equivalently the Lauricella $F_C^{(2)}$ function, see App.~\ref{sec: special functions}---only converges for $u_1+u_2<1$, hence does not cover the entire Euclidean region. We therefore need to find a suitable analytic continuation for the vertex function $\V^{(2)}_{\aa, \mu_1, \mu_2}$ to cover the region of interest. 
It is important to note that analytic continuations in kinematic space leave the analytic structure of the vertex function in the complex $\mu_1$- and $\mu_2$-planes unchanged.

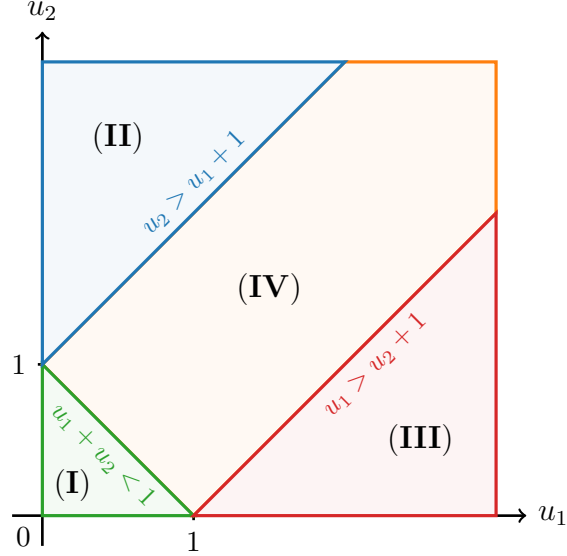
\begin{figure}[h!]
\centering
    \begin{tikzpicture}[line width=1. pt, scale=2]

    \draw[->] (-0.2,0) -- (3.2,0) node[right] {$u_1$};
    \draw[->] (0,-0.2) -- (0, 3.2) node[above] {$u_2$};

    \foreach \x in {1} {
        \draw (\x, 0.03) -- (\x, -0.03) node[below] {\small $\x$};
    }
    \node[below left] at (0,0) {\small $0$};
    
    \foreach \y in {1} {
        \draw (0.03, \y) -- (-0.03, \y) node[left] {\small $\y$};
    }

    \fill[pyorange!5] (0,1) -- (2,3) -- (3,3) -- (3,2) -- (1, 0) -- cycle;
    \draw[pyorange, line width=1.2pt] (0,1) -- (2,3) -- (3,3) -- (3,2) -- (1, 0) -- cycle;
    \node[black] at (1.5, 1.5) {(\textbf{IV})};

    \fill[pygreen!5] (0,0) -- (1,0) -- (0,1) -- cycle;
    \draw[pygreen, line width=1.2pt] (0,0) -- (1,0) -- (0,1) -- cycle;
    \node[pygreen, rotate=-45] at (0.4, 0.4) {\footnotesize $u_1+u_2<1$};
    \node[black] at (0.2, 0.2) {(\textbf{I})};
    
    \fill[pyblue!5] (0,1) -- (0,3) -- (2,3) -- cycle;
    \draw[pyblue, line width=1.2pt] (0,1) -- (0,3) -- (2,3) -- cycle;
    \node[pyblue, rotate=45] at (1, 2.2) {\footnotesize $u_2>u_1+1$};
    \node[black] at (0.5, 2.5) {(\textbf{II})};

    \fill[pyred!5] (1,0) -- (3,0) -- (3,2) -- cycle;
    \draw[pyred, line width=1.2pt] (1,0) -- (3,0) -- (3,2) -- cycle;
    \node[pyred, rotate=45] at (2.2, 1) {\footnotesize $u_1>u_2+1$};
    \node[black] at (2.5, 0.5) {(\textbf{III})};

    \end{tikzpicture}
    \caption{Kinematic space $(u_1, u_2)$ of the two-massive-leg, in the positive quadrant, with the following regions: (\textbf{I}) $\{u_1\leq0, u_2\leq0, u_1+u_2<1\}$, (\textbf{II}) $\{u_1\leq0, u_2\leq0, u_2>u_1+1\}$, (\textbf{III}) $\{u_1\leq0, u_2\leq0, u_1>u_2+1\}$, and (\textbf{IV}) $\{u_1+u_2>1, u_1\leq u_2+1, u_2\leq u_1+1\}$. The entire physical region $\{0\leq u_1, 0\leq u_2, u_1 \leq u_2+1 \,, u_2 \leq u_1+1\}$ is the union: $(\textbf{I})\cup(\textbf{IV})$. The origin is the double-soft limit $u_1, u_2\to0$.}
    \label{fig: kinematic space for V2}
\end{figure}

\vskip 4pt
The four functions $\F_{\pm i\mu_1, \pm i\mu_2}^{(2)}$ in~\eqref{eq: V2 expression} actually correspond to the four Frobenius solutions in the secondary fan of the GKZ system describing the vertex function $\V^{(2)}_{\aa, \mu_1, \mu_2}$, see e.g.~\cite{Stienstra:2005nr, Cattani2006THREELO}. These series solutions converge only in the vicinity of the singular points $(0, 0), (0, \infty)$ and $(\infty, 0)$, and can all be found from the two-fold Mellin-Barnes integral representation~\eqref{eq: V2 MB integral rep}. For instance, closing the $s_1$-integral to the left and the $s_2$-integral to the right (the order of integration does not matter) leads to the following series solution:
\begin{equation}
    \begin{aligned}
        \V_{+, \mu_1, \mu_2}^{(2)}(u_1, u_2; p) &= \frac{\C_+^{(2)}(p)}{2} \sum_{m_1, m_2=0}^\infty 
        \frac{(-1)^{m_{12}}}{m_1! m_2!} \left(\frac{u_1}{u_2}\right)^{2m_1+i\mu_1} \left(\frac{2}{u_2}\right)^{m_2+p} \\
        &\times \Gamma[-m_1-i\mu_1, \tfrac{p+m_2+i\mu_1+i\mu_2}{2}+m_1, \tfrac{p+m_2+i\mu_1-i\mu_2}{2}+m_1] \\
        &+ (\mu_1 \leftrightarrow -\mu_1) \,,
    \end{aligned}
\end{equation}
where the series in $m_2$ can be formally resummed into a hypergeometric series. This series representation around $(u_1, u_2)=(0, \infty)$ converges in the region $(\textbf{II})$ of Fig.~\ref{fig: kinematic space for V2}. Similarly, by simply exchanging the indices $1\leftrightarrow2$, we find the analogous analytic continuation in the region $(\textbf{III})$. Closing both Mellin-Barnes contours to the right leads to the same analytic continuations, with the order of integration selecting the region of convergence. We have rediscovered the known simple analytic continuations of the Appell $F_4$ function, see~\ref{subsec: Appell F4}. However, these analytic continuations do not cover the entire physical region $(\textbf{I})\cup(\textbf{IV})$.

\paragraph{Reduction formula.} A key observation to analytically continue the vertex function $\V^{(2)}_{\aa, \mu_1, \mu_2}$ in the entire physical region is to realise that the Appell $F_4$ function can be represented by a series of a product of two hypergeometric ${}_2F_1$ functions (see Eq.~(16.16.7) of~\cite{Olver:2010ouy}, also~\cite{Burchnall1940EXPANSIONSOA, nakagawa2023, Clingher_2017, Alim:2026gcl}):
\begin{equation}
\label{eq: Appell F4 analytic continuation}
    \begin{aligned}
        &F_4\left[\left.\begin{matrix} \frac{p+\alpha_{12}}{2},\,  \frac{p+1+\alpha_{12}}{2}\\ 1+\alpha_1, \, 1+\alpha_2 \end{matrix}\right\vert u_1^2, u_2^2\right] = \sum_{k=0}^\infty \frac{\left(\frac{p+\alpha_{12}}{2}\right)_k \left(\frac{p+1+\alpha_{12}}{2}\right)_k \left(\frac{2p-1}{2}\right)_k}{(1+\alpha_1)_k (1+\alpha_2)_k k!} \,\, \xi_1^k \xi_2^k \\
        &\times {}_2F_1\left[\left.\begin{matrix} \frac{p+\alpha_{12}}{2}+k,\,  \frac{p+1+\alpha_{12}}{2}+k\\ 1+\alpha_1+k \end{matrix}\right\vert \xi_1\right] {}_2F_1\left[\left.\begin{matrix} \frac{p+\alpha_{12}}{2}+k,\,  \frac{p+1+\alpha_{12}}{2}+k\\ 1+\alpha_2+k \end{matrix}\right\vert \xi_2\right] \,,
    \end{aligned}
\end{equation}
where the new kinematic variables $\xi_1$ and $\xi_2$ are the roots of the following algebraic equations:
\begin{equation}
    \begin{cases}
        \xi_1^2 - (1+u_1^2-u_2^2)\xi_1 + u_1^2 = 0 \,, \\
        \xi_2^2 - (1-u_1^2+u_2^2)\xi_2 + u_2^2 = 0 \,.
    \end{cases}
\end{equation}
The discriminants are found to be
\begin{equation}
    \begin{cases}
        \Delta_1 = (-1+u_1-u_2)(1+u_1-u_2)(-1+u_1+u_2)(1+u_1+u_2) \,, \\
        \Delta_2 = (-1-u_1+u_2)(1-u_1+u_2)(-1+u_1+u_2)(1+u_1+u_2) \,,
    \end{cases}
\end{equation}
where the four factors precisely delimitate the four boundaries of the $F_4$ singular locus chambers illustrated in Fig.~\ref{fig: kinematic space for V2}. The parameters $\xi_{1, 2}$ therefore appear to be the natural coordinates adapted to the $F_4$ geometry. Inside the simplex around the origin $(\textbf{I})$, one has $\Delta_{1, 2}>0$ and $\xi_{1, 2}\in\mathbb{R}_{>0}$ are real. In other regions, the parameters $\xi_{1, 2}$ are complex. Explicitly, they are given by
\begin{equation}
    \xi_1 \equiv \frac{1+u_1^2-u_2^2-\sqrt{\Delta_1}}{2} \,, \quad
    \xi_2 \equiv \frac{1-u_1^2+u_2^2-\sqrt{\Delta_2}}{2} \,,
\end{equation}
where we have retained the smaller root, for which the analytic continuation matches the numerical evaluation of the vertex function.
The reduction formula~\eqref{eq: Appell F4 analytic continuation} follows from solving the Appell $F_4$ system of partial differential equations (see App.~\ref{subsec: Appell F4}) with a separable ansatz, and then from  the superposition principle. Ultimately, it allows to analytically continue the Appell $F_4$ function in the physical region using {\it only} the known analytic continuations of the hypergeometric ${}_2F_1$ functions.
Eventually, notice that for $p=1/2$, we have $(\frac{2p-1}{2})_k=(0)_k$ which vanishes for $k\geq1$, and Eq.~\eqref{eq: Appell F4 analytic continuation} collapses onto the simple factorised form~\cite{Bailey1933}:
\begin{equation}
\label{eq: Appell F4 factorised reduction formula}
    \begin{aligned}
        F_4\left[\left.\begin{matrix} \frac{1/2+\alpha_{12}}{2},\,  \frac{3/2+\alpha_{12}}{2}\\ 1+\alpha_1, \, 1+\alpha_2 \end{matrix}\right\vert u_1^2, u_2^2\right] &= {}_2F_1\left[\left.\begin{matrix} \frac{1/2+\alpha_{12}}{2},\,  \frac{3/2+\alpha_{12}}{2}+k\\ 1+\alpha_1 \end{matrix}\right\vert \xi_1\right] \\
        &\times {}_2F_1\left[\left.\begin{matrix} \frac{1/2+\alpha_{12}}{2},\,  \frac{3/2+\alpha_{12}}{2}+k\\ 1+\alpha_2 \end{matrix}\right\vert \xi_2\right] \,.
    \end{aligned}
\end{equation}
This factorised reduction formula can be applied to the triple-exchange bispectrum of quasi-single field inflation, see~\cite{Chen:2009zp}, albeit only for $d=2$, in order to obtain a representation that converges in the entire physical region of kinematic space.

\begin{figure}[h!]
    \centering
    \begin{subfigure}{.5\textwidth}
        \centering
        \includegraphics[width=1\linewidth]{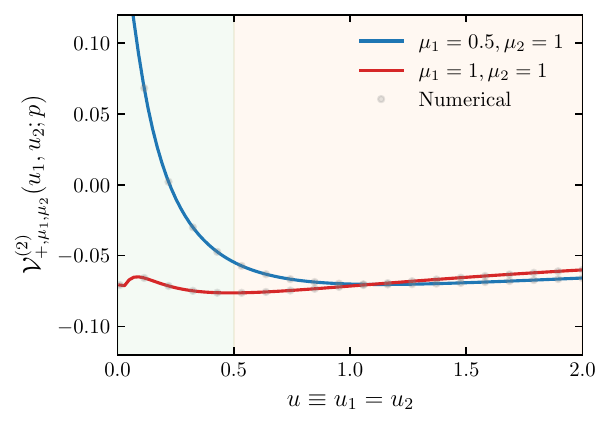}
    \end{subfigure}%
    \begin{subfigure}{.5\textwidth}
        \centering
        \includegraphics[width=1\linewidth]{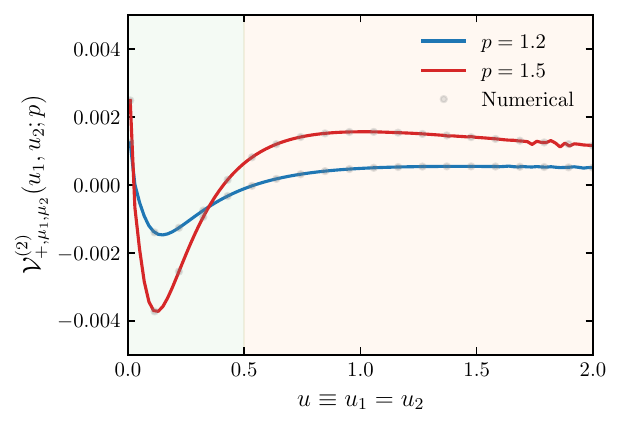}
    \end{subfigure}
   \caption{Vertex function $\V^{(2)}_{+, \mu_1, \mu_2}(u_1, u_2; p)$ along the diagonal line $u=u_1=u_2$ spanning both regions (\textbf{I}) and (\textbf{IV}). {\it Left panel}: we fix the twist $p=0.5$ and vary the mass parameter $\mu_1=0.5, 1$ with $\mu_2=1$. The analytic continuation is provided by the factorised reduction formula~\eqref{eq: Appell F4 factorised reduction formula}. {\it Right panel}: we fix the mass parameters $\mu_1=1$ and $\mu_2=2$, and we vary the twist parameter $p=1.2, 1.5$. Here, the analytic continuation is determined by the reduction formula~\eqref{eq: Appell F4 analytic continuation} on which we have applied a Levin $u$-transform to regularise and speed up the convergence. We have retained $N=20$ and $n=10$ terms in the truncated series. In both panels, only the real part of the vertex function is depicted. Grey dots correspond to numerical integration of the vertex functions.}
  \label{fig: vertex function 2}
\end{figure}

\paragraph{Series regularisation.} The series~\eqref{eq: Appell F4 analytic continuation} is actually not convergent in the region (\textbf{IV}). Indeed, for large $k$, the asymptotic behaviour of the hypergeometric function is:
\begin{equation}
    {}_2F_1\left[\left.\begin{matrix} a+k,\,  b+k\\ c+k \end{matrix}\right\vert \xi\right] \sim (1-\xi)^{-k}\,, \quad \text{as} \quad k\to\infty \,.
\end{equation}
Along the diagonal $u_1=u_2=u$, the $k$-th term of the outer series behaves asymptotically as
\begin{equation}
    \left(\frac{\xi_1 \xi_2}{(1-\xi_1)(1-\xi_2)}\right)^k \sim \left(\frac{1-\sqrt{1-4u^2}}{1+\sqrt{1-4u^2}}\right)^{2k} \,.
\end{equation}
For $u>1/2$, the quantity under the square root becomes negative, rendering the ratio a pure phase. As a result, the magnitude of the entire expression is unity, and the series is formally divergent. Owing to the alternative-sign structure of the series, we regularise the outer summation with a Levin $u$-transform~\cite{Olver:2010ouy}. Introducing the truncated partial sums to~\eqref{eq: Appell F4 analytic continuation} by writing $S_N(\xi_1, \xi_2) = \sum_{k=0}^N a_k(\xi_1, \xi_2)$, the transformed sequence is defined as
\begin{equation}
    L_N^{(n)}(\xi_1, \xi_2) = \frac{\sum_{k=0}^N \binom{n}{k} (-1)^k \frac{(k+n+1)^{N-1}}{a_{k+n}(\xi_1, \xi_2)} S_{k+n}(\xi_1, \xi_2)}{\sum_{k=0}^N \binom{n}{k} (-1)^k \frac{(k+n+1)^{N-1}}{a_{k+n}(\xi_1, \xi_2)}} \,.
\end{equation}
The Levin transform constitutes a nonlinear sequence transformation designed to accelerate convergence by effectively modelling and cancelling the leading asymptotic behaviour of the truncation error. In particular, it can assign stable numerical values to slowly convergent or even divergent series, making it especially well suited for the present context. In Fig.~\ref{fig: vertex function 2}, we display the vertex function $\V^{(2)}$ evaluated along the diagonal $u=u_1=u_2$, spanning both regions (\textbf{I}) and (\textbf{IV}), for various values of the mass parameters and the twist. Our approach yields a convergent representation of the vertex function throughout the entire physical domain, thereby providing a successful analytic continuation. We further observe that convergence becomes slower as the values of $u$ increase, requiring a larger number of terms in the series. This is illustrated in the right panel, where residual oscillations are visible for $p=1.5$; these oscillations are progressively suppressed as more terms are included in the summation.

\subsection{Generalisation to arbitrary number of massive legs}
\label{subsec: generalisation to arbitrary massive legs}

Let us now give the general formula for the vertex function $\V^{(n)}_{\aa, \mu_1, \ldots, \mu_n}$ with $n \geq 1$ massive legs, defined here for a $\aa=+$ vertex by
\begin{equation}
\label{eq: Vn definition}
    \V^{(n)}_{+, \{\mu_j\}}(\{u_j\}; p) \equiv e^{+\frac{\pi}{2}\mu_{1\cdots n}} \int_{-\infty}^0 \frac{\d z}{(-z)^{d+1}} (-z)^{p+n\tfrac{d}{2}} e^{iz} \prod_{j=1}^n H_{i\mu_j}^{(2)}(-u_j z) \,,
\end{equation}
where $\mu_{1\cdots n}\equiv \mu_1 + \cdots + \mu_n$. This vertex function depends on the dimensionless energy ratios $u_j \equiv Y_j/X$ for which we will determine the Euclidean region further below, and is schematically represented as follows:
\begin{center}
\begin{tikzpicture}[line width=1. pt, scale=2]
    \draw[black] (-1, 0) -- node[above left] {$(Y_1, \mu_1)$} (0, 0);
    \draw[black] (-0.5, -0.6) -- node[left] {$(Y_2, \mu_2)$} (0, 0);
    \node[black, rotate=30] at ($(-0.5, -0.6)!1/2!(1, 0)$) {$\cdots$};
    \draw[black] (0, 0) -- node[above right] {$(Y_n, \mu_n)$} (1, 0);

    \draw[gray!40] (0, 0) -- (-0.8, 0.5) node[above] {$k_1$};
    \draw[gray!40] (0, 0) -- (-0.4, 0.5) node[above] {$k_2$};
    \draw[gray!40] (0, 0) -- (0, 0.5) node[above] {$k_3$};
    \node[gray!40] at ($(0, 0.5)!1/2!(0.8, 0.5)$) {$\cdots$};
    \draw[gray!40] (0, 0) -- (0.8, 0.5) node[above] {$k_m$};
    
    \draw[fill=black] (0, 0) circle (.05cm) node[above] {$(X, p)$};
\end{tikzpicture} 
\end{center}
The corresponding $n$-fold Mellin-Barnes representation is given by
\begin{equation}
\label{eq: Vn MB rep}
    \begin{aligned}
        \V^{(n)}_{+, \{\mu_j\}}(\{u_j\}; p) = \C_+^{(n)}(p) \int_{-i\infty}^{+i\infty} \prod_{j=1}^n &\left[\frac{\d s_j}{2\pi i} \left(\frac{u_j}{2}\right)^{-2s_j} \Gamma[s_j \pm \tfrac{i\mu_j}{2}]\right] \\
        &\times \Gamma[p + \tfrac{n-2}{2}d-2s_{1\cdots n}] \,,
    \end{aligned}
\end{equation}
where the overall prefactor reads
\begin{equation}
\label{eq: Cn+ overall def}
    \C_+^{(n)}(p) = \tfrac{1}{\pi^n} e^{-\frac{i\pi}{2}(p+\frac{n-2}{2}d-n)} \,.
\end{equation}
Closing all integral contours to the left, which corresponds to expanding around the soft limit $u_j\to0$ for all massive legs, leads to a series representation in terms of a Lauricella $F_C^{(n)}$ function of $n$ variables. To make notations compact, we introduce the following multi-index vectors: 
\begin{equation}
    \bm{\alpha} \equiv (\alpha_1, \ldots, \alpha_n) \,, \quad \bm{u} \equiv (u_1, \ldots, u_n) \,, \quad \bm{\mu} \equiv (\mu_1, \ldots, \mu_n) \,,
\end{equation}
with the scalar shorthands $|\bm{\alpha}| \equiv \sum_{j=1}^n\alpha_j$ and $|\bm{\mu}| \equiv \sum_{j=1}^n\mu_j$. For component-wise products, we also define
\begin{equation}
    \left(\frac{\bm{u}}{2}\right)^{\bm{\alpha}} \equiv \prod_{j=1}^n \left(\frac{u_j}{2}\right)^{\alpha_j} \,, \quad \Gamma[-\bm{\alpha}] \equiv \prod_{j=1}^n \Gamma[-\alpha_j]\,, \quad \Gamma[\bm{1}+\bm{\alpha}] = \prod_{j=1}^n \Gamma[1+\alpha_j]\,.
\end{equation}
The resulting compact expression for the vertex function reads
\begin{equation}
\label{eq: Vn definition short}
    \boxed{
    \mathcal{V}_{+, \bm{\mu}}^{(n)}(\bm{u}; p) = \C_+^{(n)}(p) \sum_{\bm{\alpha} = \pm i\bm{\mu}} \F_{\bm{\alpha}}^{(n)}(\bm{u}; p+\tfrac{n-2}{2}d) \,,
    }
\end{equation}
where each mode is given by
\begin{equation}
\label{eq: general mode Fn function}
    \boxed{
    \F_{\bm{\alpha}}^{(n)}(\bm{u}; p) \equiv \Gamma[p+|\bm{\alpha}|,-\bm{\alpha},\bm{1}+\bm{\alpha}] \left(\frac{\bm{u}}{2}\right)^{\bm{\alpha}} \bar{F}^{(n)}_C\left[\left.\begin{matrix} \frac{p+|\bm{\alpha}|}{2},\,  \frac{p+1+|\bm{\alpha}|}{2}\\ \bm{1}+\bm{\alpha} \end{matrix}\right\vert \bm{u}^2\right] \,.
    }
\end{equation}
The sum in~\eqref{eq: Vn definition short} runs over $2^n$ terms, each corresponding to a choice of positive- or negative-frequency mode for each of the $n$ massive legs. The $\aa=-$ vertex function is given as usual by the replacement $\C_+^{(n)}(p)\to \C_-^{(n)}(p)=\C_+^{(n)*}(p)$. It is now made manifest that the vertex functions belong to the class of generalised type-$C$ Lauricella functions of $n$ variables, $\bar{F}_C^{(n)}$. More details about these functions are given in App.~\ref{app: Lauricella Function}, although very little is known in general.

\subsubsection*{Conformally coupled limit}

For certain special values of the mass parameters $\mu_j$, $j=1, \ldots, n$, away from the principal series $\mu_j\in \mathbb{R}_{\geq0}$, the vertex function $\V^{(n)}$ simplifies. An example is given when all massive legs describe a conformally coupled scalar field, for which $\mu_j=-i/2$. In this case, the time integral~\eqref{eq: Vn definition} is straightforward to perform after simplifying the Hankel function
\begin{equation}
    H_{1/2}^{(2)}(z) = e^{+\frac{i\pi}{2}} \sqrt{\frac{2}{\pi z}} e^{-i z} \,.
\end{equation}
Using~\eqref{eq: elementary time integral}, we obtain
\begin{equation}
    \V^{(n)}_{+,\{-i/2\}}(\u; p) = \frac{e^{-\frac{i\pi}{2}(p+\frac{nd}{2}-\frac{n}{4}-d)}}{(u_1 \cdots u_n)^{1/2}} \left(\frac{2}{\pi}\right)^{n/2} \frac{\Gamma(p+n\frac{d-1}{2}-d)}{|\u|^{p+n\frac{d-1}{2}-d}} \,.
\end{equation}
This result can be recovered from the $n$-fold Mellin-Barnes integral representation~\eqref{eq: Vn MB rep} after using the Legendre duplication formula to simplify the pair of $\Gamma$-functions for each massive mode~\cite{Sleight:2019mgd}. Such integrals were also computed in e.g.~\cite{Creminelli:2011mw, Arkani-Hamed:2015bza, Belrhali:2026rkn}. 

\subsubsection{Annihilators} 

The vertex functions $\V^{(n)}$ satisfy a system of $n$ coupled partial differential equations in the $u_i$ variables, which is nothing but the Lauricella system~\cite{nakagawa2024appell, Matsumoto_2020, Slater1966}. For $n=1$, this system reduces to the standard hypergeometric differential equation, and for $n=2$, the solution is given by an Appell $F_4$ function. Let us explicitly write down the system for $n\geq1$. The type-$C$ Lauricella function $F_C^{(n)}$ in the $u_i^2$ variables ($i=1, \ldots, n$) is annihilated by the following differential operators (see App.~\ref{app: Lauricella Function} for more details):
\begin{equation}
    \tilde{\D}_i = \vartheta_i(\vartheta_i +2c_i-2) - u_i^2(\vartheta +2a)(\vartheta + 2b) \,,
\end{equation}
where $\vartheta_i \equiv u_i \,\partial/\partial u_i$, $\vartheta = \sum_{i=1}^n \vartheta_i$, $a = (p+\tfrac{n-2}{2}d+|\bm{\alpha}|)/2, b=(p+\tfrac{n-2}{2}d+1+|\bm{\alpha}|)/2$ and $c_i = 1+\alpha_i$. Now, we take into account the overall prefactor
\begin{equation}
    \phi(\u) \equiv \prod_{i=1}^n u_i^{\alpha_i} \,,
\end{equation}
which enters the vertex function, and construct the differential operator $\D_i$ that is related to $\tilde{\D}_i$ by the gauge transformation $\D_i = \phi \tilde{\D}_i \phi^{-1}$. Applying the gauge transformation to the Euler operators, we obtain
\begin{equation}
    \phi \vartheta_i \phi^{-1} = \vartheta_i - \alpha_i \,, \quad \phi \vartheta \phi^{-1} = \vartheta - |\bm\alpha| \,.
\end{equation}
Altogether, we find the following annihilator for the building block $\F_{\bm\alpha}^{(n)}(\bm{u};\tilde p)$,
\begin{equation}\label{eq: vertex function annihilator}
    \D_i = \vartheta_i^2 -\alpha_i^2 - u_i^2(\vartheta +  p+\tfrac{n-2}{2}d )(\vartheta +  p +\tfrac{n-2}{2}d +1)\,,
\end{equation}
for $i=1,\dots, n$. As we can see, this annihilator depends on $\alpha_i^2= -\mu_i^2$ and hence is the same for the positive and negative frequency mode. Consequently, $\D_i$ annihilates the full vertex function $\V^{(n)}$ given in~\eqref{eq: Vn definition short}. The solution to the corresponding system of partial differential equations, namely the vertex function $\V^{(n)}$, is completely fixed by imposing the Bunch-Davies vacuum in the far past, or equivalently regularity on the boundary hyperplanes of the kinematic polytope (see Sec.~\ref{subsec: Euclidean region}), i.e.~the absence of folded singularities.

\subsubsection{Partial soft limits}
\label{subsubsec: partial soft limits}

We now study the simplification of the vertex function $\V^{(n)}_{+,\bm{\mu}}(\u; p)$  when a subset $\mathcal{I}\subseteq\{1,\ldots,n\}$ of the kinematic ratios is taken to zero, which are partial internal soft limits. Let $k=|\mathcal{I}|$ and $\bar{\mathcal{I}}=\{1,\ldots,n\}\setminus\mathcal{I}$ denote the complementary set of $n-k$ ``hard'' legs. The soft limit amounts to sending $u_j\to 0^+$ for all $j\in\mathcal{I}$, while keeping $(u_j)_{j\in\bar{\mathcal{I}}}$ fixed.

\vskip 4pt
The simplification originates directly from the series definition~\eqref{eq: Lauricella Fc def} of the Lauricella function $F^{(n)}_C$: any term in the series with $k_j>0$ for a soft index $j\in\mathcal{I}$ is suppressed by $u_j^{2k_j}\to 0$, so only $k_j=0$ contributes for those indices. The $n$-variable Lauricella function therefore reduces to a Lauricella function of order $n-k$ in the surviving variables:
\begin{equation}
\label{eq: Lauricella soft reduction}
    F^{(n)}_C\left[\left.\begin{matrix} a,\,  b\\ c_1, \ldots, c_n \end{matrix}\right\vert \bm{0}_\I, (u_j^2)_{j\in \bar{\I}}\right] = F^{(n-k)}_C\left[\left.\begin{matrix} a,\,  b\\ (c_j)_{j\in \bar{\I}} \end{matrix}\right\vert (u_j^2)_{j\in \bar{\I}}\right] \,,
\end{equation}
where the $k$ soft variables simply drop out. In the regularised convention~\eqref{eq: Lauricella Fc regularised}, setting $u_j=0$ for $j\in\mathcal{I}$ produces a compensating factor $1/\prod_{j\in\mathcal{I}}\Gamma(1+\alpha_j)$, which cancels exactly against the corresponding $\Gamma[\mathbf{1}+\bm{\alpha}]$ factors in the numerator of 
$\F^{(n)}_{\bm{\alpha}}$ (see~\eqref{eq: general mode Fn function}). As a result, each mode factorises as
\begin{equation}
    \F^{(n)}_{\bm{\alpha}}(\u; p)\Big|_{u_j\to 0,\,j\in\I} = \left(\prod_{j\in\mathcal{I}}\Gamma(-\alpha_j) \left(\frac{u_j}{2}\right)^{\alpha_j}\right) \F^{(n-k)}_{\bm{\alpha}_{\bar{\I}}} \left(\u_{\bar{\I}};
    p+|\bm{\alpha}_\I|\right) \,,
\end{equation}
where $\bm{\alpha}_{\I}=(\alpha_j)_{j\in\I}$ and $|\bm{\alpha}_{\I}|=\sum_{j\in\I}\alpha_j$. The soft indices contribute only through a product of power-law prefactors $\Gamma(-\alpha_j)(u_j/2)^{\alpha_j}$, while the hard indices form a lower-order Lauricella mode $\F^{(n-k)}_{\bm{\alpha}_{\bar{\I}}}$ at a shifted effective twist $p\to p+|\bm{\alpha}_{\I}|$. The full vertex function in the partial soft limit therefore reads
\begin{equation}
    \V^{(n)}_{+,\bm{\mu}}(\u; p)\Big|_{u_j\to 0,\,j\in\mathcal{I}} = \C^{(n)}_+(p) \sum_{\bm{\alpha}=\pm i\bm{\mu}}
    \left(\prod_{j\in\I}\Gamma(-\alpha_j) \left(\frac{u_j}{2}\right)^{\alpha_j}\right) \F^{(n-k)}_{\bm{\alpha}_{\bar{\I}}} \left(\u_{\bar{\I}}; \tilde{p}+|\bm{\alpha}_{\I}|\right) \,,
\end{equation}
where $\tilde{p} = p+\tfrac{n-2}{2}d$. In the complete soft limit $\mathcal{I}=\{1,\ldots,n\}$, the Lauricella function of order $n-k=0$ evaluates to unity and the entire hypergeometric structure drops out, leaving only oscillatory power-law factors:
\begin{equation}
    \lim_{\u\to\bm{0}^+} \V^{(n)}_{+,\bm{\mu}}(\u; p) = \C^{(n)}_+(p) \sum_{\bm{\alpha}=\pm i\bm{\mu}} \Gamma[\tilde{p}+|\bm{\alpha}|, -\bm{\alpha}] \left(\frac{\u}{2}\right)^{\bm{\alpha}} \,.
\end{equation}
This is the natural $n$-leg generalisation of the single-massive-leg soft limit~\eqref{eq: single-massive-leg soft limit}: the Lauricella function evaluates trivially to $1$ at the origin, and the leading non-analytic behaviour of the vertex function is a sum of $2^n$ power laws $\prod_j u_j^{\alpha_j}$, with frequencies $\pm\mu_j$ for each massive leg $j$. For real twist $p\in\mathbb{R}$, the sum is manifestly real and generates multi-frequency oscillations in the logarithms $\log u_j$.

\subsection{Euclidean region}
\label{subsec: Euclidean region}

The vertex function $\V^{(n)}$ is a function of $n$ kinematic variables: $u_1, \ldots, u_n$. We now derive the complete set of inequalities that characterise the Euclidean (physically allowed) kinematic region for the corresponding vectors subject to a closure (momentum-conservation) constraint. 

\vskip 4pt
A kinematic configuration is given by $n+m$ Euclidean vectors in $\mathbb{R}^{d}$ ($d\ge 2$),
\begin{equation}
  \Y_1, \ldots, \Y_n\,, \quad \k_1, \ldots, \k_m \;\in\;\mathbb{R}^{d} \,,
\end{equation}
where the momenta $\Y_j$ are internal and flow through massive legs, with $j=1, \ldots, n$, and $\k_i$, with $i=1, \ldots, m$, are the external momenta. These vectors are subject to the closure condition
\begin{equation}
\label{eq: closure}
    \sum_{i=1}^{n} \Y_i + \sum_{j=1}^{m} \k_j = \bm{0} \,,
\end{equation}
which expresses momentum conservation or, equivalently, the requirement that the vectors form a closed polygon. An example of such configuration is given by:
\begin{center}
\begin{tikzpicture}[line width=1. pt, scale=2]
    \draw[black, postaction={decorate},
                        decoration={markings, mark=at position 0.5 with {\arrow{>}}}] 
                        (0, 0) -- node[left] {$\Y_1$} (-0.5, 1);
    \draw[black, postaction={decorate},
                        decoration={markings, mark=at position 0.5 with {\arrow{>}}}] 
                        (-0.5, 1) -- node[left] {$\Y_2$} (-0.2, 1.6);
    \node[black] at ($(-0.2, 1.6)!1/3!(0.5, 2)$) {$\cdot$};
    \node[black] at ($(-0.2, 1.6)!2/3!(0.5, 2)$) {$\cdot$};
    \draw[black, postaction={decorate},
                        decoration={markings, mark=at position 0.5 with {\arrow{>}}}] 
                        (0.5, 2) -- node[above] {$\Y_n$} (1.5, 1.5);

    \draw[gray!40, postaction={decorate},
                         decoration={markings, mark=at position 0.5 with {\arrow{>}}}] 
                         (1.5, 1.5) -- node[right] {$\k_1$} (1.5, 1);
    \draw[gray!40, postaction={decorate},
                         decoration={markings, mark=at position 0.5 with {\arrow{>}}}] 
                         (1.5, 1) -- node[below right] {$\k_2$} (1.2, 0.5);
    \node[gray!40] at ($(1.2, 0.5)!1/3!(0.6, 0.1)$) {$\cdot$};
    \node[gray!40] at ($(1.2, 0.5)!2/3!(0.6, 0.1)$) {$\cdot$};
    \draw[gray!40, postaction={decorate},
                         decoration={markings, mark=at position 0.5 with {\arrow{>}}}] 
                         (0.6, 0.1) -- node[below right] {$\k_m$} (0, 0);
\end{tikzpicture} 
\end{center}
Working with rescaled variables $u_j = |\Y_j|/X$, where $X=\sum_{j=1}^m |\k_i|>0$ is the vertex energy, we now show that the physical region is equivalent to the requirement that the Gram matrix of the $\u$-vectors is positive semi-definite and that its row-sum quadratic form is bounded above by unity. Combining this with the reverse triangle inequality yields a convex polytope in $\mathbb{R}_{\ge 0}^{n}$ whose facets correspond to degenerate, collinear configurations. 

\subsubsection*{Closure constraint} 

The closure condition~\eqref{eq: closure} for the rescaled kinematic variables read:
\begin{equation}
\label{eq: closure dimensionless}
    \sum_{i=1}^n \u_i = -\bm{\kappa} \equiv -\frac{1}{X}\sum_{j=1}^m \k_j \,.
\end{equation}
An immediate consequence of the triangle inequality applied to the $k$-sector is
\begin{equation}
\label{eq: kappa bound}
    |\bm{\kappa}| = \frac{1}{X} \Bigl\lvert\sum_{j=1}^m \k_j\Bigr\rvert \le \frac{1}{X} \sum_{j=1}^m |\k_j| = 1 \,.
\end{equation}
The physically allowed (Euclidean) kinematic region is the set of all tuples $(u_1, \ldots, u_n) \in \mathbb{R}_{\ge 0}^{n}$ for which there exist vectors $\u_1, \ldots, \u_n, \k_1, \ldots, \k_m \in \mathbb{R}^{d}$ with $|\u_i|=u_i$ satisfying the closure condition~\eqref{eq: closure}. We now characterise this region purely in terms of the~$u_i$.

\subsubsection*{Gram matrix} 

The central object of our analysis is the $n\times n$ symmetric Gram matrix of the rescaled $\u$-vectors, defined by~\cite{Eden:1966dnq, Cortes:2025ihs}
\begin{equation}
\label{eq: gram matrix}
    G_{ij} \equiv \u_i \cdot \u_j = u_i\,u_j\cos\theta_{ij},
  \qquad i,j=1,\ldots,n\,,
\end{equation}
where $\theta_{ij}$ denotes the angle between $\u_i$ and $\u_j$. In particular, the diagonal entries are $G_{ii}=u_i^{2}$.  The Gram matrix encodes all inner-product information about the configuration, and the following structural property is fundamental: for any collection of Euclidean vectors $\u_1, \ldots, \u_n \in \mathbb{R}^d$, the Gram matrix $G$ defined in~\eqref{eq: gram matrix} is positive semi-definite, $G\succeq 0$. Indeed, for every vector $\bm{c} = (c_1, \ldots, c_n)^T \in \mathbb{R}^d$ one computes
\begin{equation}
  \bm{c}^T G\,\bm{c} = \sum_{i,j=1}^n c_i\,(\u_i\cdot\u_j)\,c_j = \Bigl\lvert\sum_{i=1}^n c_i\,\u_i\Bigr\rvert^2 \ge 0\,.
\end{equation}
Since this holds for all $\bm{c}$, the matrix is positive semi-definite. In particular, its rank cannot exceed $d$, the ambient dimension, reflecting the fact that $n$ vectors in $\mathbb{R}^d$ span at most a $d$-dimensional subspace. Conversely, any positive semi-definite matrix $G$ of rank at most $d$ can be realised as the Gram matrix of some configuration of $n$ vectors in $\mathbb{R}^d$. This follows directly from a Cholesky (or eigenvalue) decomposition. Thus, positive semi-definiteness is not merely necessary but also sufficient for the matrix to correspond to a valid Euclidean configuration: {\it $G$ is the Gram matrix of some configuration of $n$ Euclidean vectors if and only if $G\succeq 0$}.

\subsubsection*{Quadratic constraint on $G$} 

The closure condition~\eqref{eq: closure dimensionless} together with the bound~\eqref{eq: kappa bound} impose a constraint on the Gram matrix that goes beyond positive semi-definiteness. Let $\bm{1} = (1, \ldots, 1)^T \in \mathbb{R}^n$ denote the all-ones vector. The closure condition allows us to evaluate the row-sum quadratic form
of $G$ as
\begin{equation}
\label{eq: quadratic}
    \bm{1}^T G\,\bm{1} = \sum_{i,j=1}^n G_{ij} = \sum_{i,j=1}^n (\u_i\cdot\u_j) = \Bigl\lvert\sum_{i=1}^n \u_i\Bigr\rvert^2 = |\bm{\kappa}|^{2} \,.
\end{equation}
Combined with the upper bound~\eqref{eq: kappa bound}, this yields
\begin{equation}
\label{eq:gram_upper}
    \bm{1}^T G\,\bm{1} \le 1 \,.
\end{equation}
This is the fundamental quantitative constraint imposed by the closure condition on the Gram matrix: the sum of all entries of $G$ is bounded above by unity. When the individual magnitudes $\kappa_j = |\k_j|/X$ are prescribed (with $\sum_j\kappa_j=1$ by construction), the same reverse triangle inequality applied to the $k$-sector gives a non-trivial lower bound,
\begin{equation}
\label{eq: kappa lower}
  |\bm{\kappa}| \ge \kappa_{\min} = \max\left(0, \max_j(\kappa_j) - \sum_{\ell\ne j}\kappa_\ell \right) \,,
\end{equation}
which is non-zero whenever one $k$-vector dominates all others. Together,~\eqref{eq: quadratic}--\eqref{eq: kappa lower} sandwich the row-sum quadratic form:
\begin{equation}
    \kappa_{\min}^2 \le \bm{1}^T G\, \bm{1} \le 1 \,.
\end{equation}

\subsubsection*{Characterisation of the kinematic region} 

We now combine the two ingredients---positive semi-definiteness and the closure-induced quadratic bound---to obtain a complete characterisation of the kinematic region. In particular, we will prove the following proposition:

\begin{proposition}[Kinematic region]
\label{prop: kinematic region}
A tuple $(u_1, \ldots, u_n) \in \mathbb{R}_{\ge 0}^n$ belongs to the Euclidean kinematic region if and only if the system 
\begin{equation}
\label{eq: system}
    G \succeq 0 \,, \quad G_{ii} = u_i^{2} \quad \text{for all} \quad i=1, \ldots, n \,, \quad \text{and} \quad \bm{1}^T G\,\bm{1} \le 1 \,,
\end{equation}
admits at least one solution $G\in\mathbb{R}^{n\times n}$. This feasibility condition is equivalent to the $n$ inequalities
\begin{equation}
\label{eq: kinematic region}
    \boxed{
    u_i \le 1 + \sum_{j\ne i} u_j \,, \quad  (i=1, \ldots, n) \,,
    }
\end{equation}
together with $u_i\ge 0$.
\end{proposition}
 
\begin{proof}
If a physical configuration exists, then it is guaranteed that $G\succeq 0$. The definition~\eqref{eq: gram matrix} gives $G_{ii}=u_i^2$, and we also have $\bm{1}^T G\,\bm{1}\le 1$ from previous arguments. Now, suppose a solution to~\eqref{eq: system} exists, so that in particular $|\sum_i\u_i|\le 1$.  For any index $j$, decompose the vector sum as $\sum_i\u_i = \u_j + \sum_{i\ne j}\u_i$ and apply the reverse triangle inequality:
\begin{equation}
  |\u_j| - \Bigl\lvert\sum_{i\ne j} \u_i\Bigr\rvert \le \Bigl\lvert\sum_{i=1}^n \u_i\Bigr\rvert \le 1 \,.
\end{equation}
Using $|\sum_{i\ne j}\u_i|\le\sum_{i\ne j}u_i$ and
$|\u_j|=u_j$ then yields $u_j\le 1+\sum_{i\ne j}u_i$, which is~\eqref{eq: kinematic region} for index $j$. Since $j$ was arbitrary, all $n$ inequalities follow. Now, suppose~\eqref{eq: kinematic region} holds. We exhibit an explicit solution to~\eqref{eq: system}. Let $j^*=\arg\max_i u_i$ be the index of the largest magnitude. Place all rescaled vectors along a single axis $\hat{\bm{e}}$ as follows: set $\u_i=u_i\hat{\bm{e}}$ for $i\ne j^*$ and $\u_{j^*}=-u_{j^*}\hat{\bm{e}}$. The corresponding Gram matrix has entries $G_{ij} = u_i u_j(\hat{\bm{e}}\cdot\hat{\bm{e}}) = \pm u_i u_j$ and is of rank one, hence positive semi-definite. The diagonal entries are $G_{ii}=u_i^2$ by construction. The row-sum quadratic form evaluates to
\begin{equation}
  \bm{1}^T G\,\bm{1} = \Bigl\lvert\sum_{i=1}^n \u_i\Bigr\rvert^2 = \Bigl\lvert u_{j^*} - \sum_{i\ne j^*}u_i\Bigr\rvert^2 \le 1 \,,
\end{equation}
where the last step uses~\eqref{eq: kinematic region} for $i=j^*$. This configuration thus provides a valid solution to~\eqref{eq: system}, completing the proof.
\end{proof}

\begin{figure}[h!]
    \centering
    \includegraphics[width=0.4\linewidth]{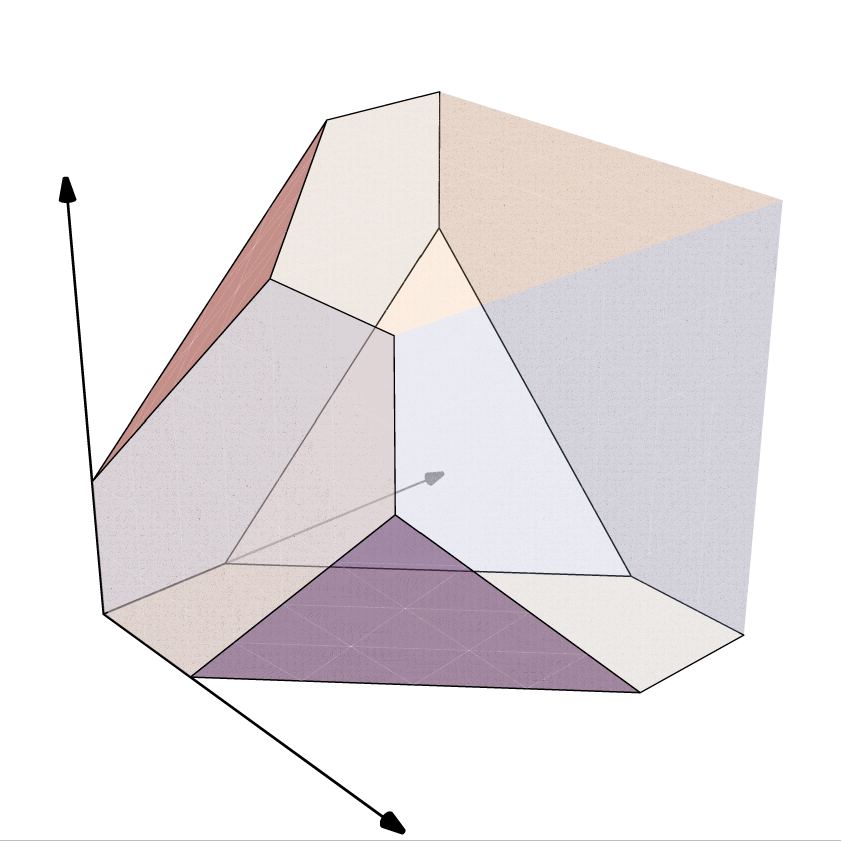}
    \caption{Kinematic polytope for $n=3$ in $(u_1, u_2, u_3)\in\mathbb{R}^3$ in the positive orthant.}
  \label{fig: kinematic region polytope}
\end{figure}
 
\subsubsection*{Geometric interpretation of the boundary} 

The hyperplane $u_i = 1+\sum_{j\ne i}u_j$ in $\mathbb{R}_{\ge 0}^n$ corresponds to a {\it degenerate} configuration in which all vectors are collinear and the $i$-th vector is anti-aligned with the sum of all others, with $|\bm{\kappa}|=1$, i.e.\ all $\k$-vectors are themselves collinear and parallel to $\bm{\kappa}$.  Points in the interior of the kinematic region correspond to generic (non-collinear) configurations, while the $n$ boundary hyperplanes of~\eqref{eq: kinematic region} are the extremal walls of the corresponding polytope. The intersection of any two boundary hyperplanes corresponds to a configuration in which two vectors are anti-aligned with the sum of all others, and so on. The unique vertex of the polytope is the origin, corresponding to the trivial configuration in which all vectors vanish. We show the kinematic region polytope for $n=3$ in Fig.~\ref{fig: kinematic region polytope}. Note that for $d\ge 2$ the physical region is the same polytope~\eqref{eq: kinematic region} regardless of the precise value of $d$, since increasing $d$ only gives more freedom to realise any given Gram matrix; it does not enlarge the set of achievable $(u_1,\ldots,u_n)$. Eventually, let us mention that the series representations of the Lauricella functions $F_C^{(n)}$ (for $n\ge2$) only converge in the simplex around the origin, and require analytic continuation elsewhere. 

\subsection{Marginal vertices}

We eventually consider the case of a marginal vertex for which the energy vanishes, $X=0$. For $n\geq2$ massive legs, using the shorthand notations from Sec.~\ref{subsec: generalisation to arbitrary massive legs}, the corresponding vertex function $\W_{\aa, \mu_1, \ldots, \mu_n}^{(n)}$ is defined by the following time integral:
\begin{equation}
\label{eq: Wn definition}
    \W^{(n)}_{+, \bm{\mu}}(\u; p) \equiv e^{+\frac{\pi}{2}\mu_{1\cdots n}} \int_{-\infty}^0 \frac{\d z}{(-z)^{d+1}} (-z)^{p+n\tfrac{d}{2}} \prod_{j=1}^n H_{i\mu_j}^{(2)}(-u_j z) \,,
\end{equation}
for the $\aa=+$ vertex, and is schematically given by
\begin{center}
\begin{tikzpicture}[line width=1. pt, scale=2]
    \draw[black] (0, 0) -- (-0.5, -0.5) node[below left] {$(Y_5, \mu_5)$};
    \draw[black] (0, 0) -- (-0.7, 0) node[left] {$(Y_4, \mu_4)$};
    \draw[black] (0, 0) -- (-0.5, 0.5) node[above left] {$(Y_3, \mu_3)$};
    \draw[black] (0, 0) -- (0, 0.7) node[above] {$(Y_2, \mu_2)$};
    \draw[black] (0, 0) -- (0.5, 0.5) node[above right] {$(Y_1, \mu_1)$};
    \draw[black] (0, 0) -- (0.7, 0) node[right] {$(Y_n, \mu_n)$};

    \node[black, rotate=25] at ($(-0.5, -0.5)!1/2!(1, -0.5)$) {$\cdots$};
    
    \draw[fill=black] (0, 0) circle (.05cm) node[below right] {$(0, p)$};
\end{tikzpicture} 
\end{center}
Notice the absence of the exponential term. Without loss of generality, we assume that $Y_n$ is the largest internal energy: $Y_n>Y_j$ for $j=1, \ldots, n-1$. The vertex function is therefore a function of the internal energy ratios: $u_j=Y_j/Y_n$ with the condition $u_n=1$.\footnote{Note that in the case of a marginal vertex, the corresponding dimensionless master integral is defined by stripping out the factor $(\pi/4)/(X_iY)^{d/2}$ for every bulk-to-bulk propagator where $Y$ is the largest internal energy of the massive leg connected to the vertex.} The Mellin-Barnes representation of the marginal vertex function is given by
\begin{equation}
    \begin{aligned}
        \W^{(n)}_{+, \bm{\mu}}(\u; p) = \frac{e^{\frac{i\pi n}{2}}}{\pi^n} \int_{-i\infty}^{+i\infty} &\prod_{j=1}^n \left[\frac{\d s_j}{2i\pi} \left(\frac{u_j}{2}\right)^{-2s_j} \Gamma[s_j\pm \tfrac{i\mu_j}{2}]\right] e^{-i\pi s_{1\cdots n}}\\
        &\times \int_{-\infty}^{z_0} \frac{\d z}{(-z)^{d+1}} (-z)^{p+n \tfrac{d}{2}-2s_{1\cdots n}} \,,
    \end{aligned}
\end{equation}
where $z_0$ is a late-time IR cutoff. The time integral evaluates to
\begin{equation}
    \int_{-\infty}^{z_0} \frac{\d z}{(-z)^{d+1}} (-z)^{p+n \tfrac{d}{2}-2s_{1\cdots n}} = -\frac{(-z_0)^{p+\tfrac{n-2}{2}d - 2s_{1\cdots n}}}{p+\tfrac{n-2}{2}d - 2s_{1\cdots n}} \,,
\end{equation}
with $\Re(p+\tfrac{n-2}{2}d - 2s_{1\cdots n})<0$. The leading contribution in the $z_0\to0$ limit is encoded in the residue of the single pole at $\tfrac{p}{2}+\tfrac{n-2}{4}d - s_{1\cdots n}=0$. The time integral is therefore encoded in the following $\delta$-function:
\begin{equation}
    i\pi \delta(\tfrac{p}{2}+\tfrac{n-2}{4}d - s_{1\cdots n}) = -\frac{(-z_0)^{p+\tfrac{n-2}{2}d - 2s_{1\cdots n}}}{p+\tfrac{n-2}{2}d - 2s_{1\cdots n}} \,.
\end{equation}
Evaluating the $s_n$ integration using the $\delta$-function yields the following expression:
\begin{equation}
    \begin{aligned}
        \W^{(n)}_{+, \bm{\mu}}(\u; p) = \tilde{\C}_+^{(n)}(p) \int_{-i\infty}^{+i\infty} \prod_{j=1}^{n-1} &\left[\frac{\d s_j}{2i\pi} \left(\frac{u_j}{2}\right)^{-2s_j} \Gamma[s_j\pm \tfrac{i\mu_j}{2}]\right] \\ 
        &\times\Gamma[\tfrac{p}{2}+\tfrac{n-2}{4}d - s_{1\cdots n-1} \pm \tfrac{i\mu_n}{2}] \,,
    \end{aligned}
\end{equation}
where we have defined
\begin{equation}
    \tilde{\C}_+^{(n)}(p) \equiv 2^{p+\tfrac{n-2}{2}d-1}\C_+^{(n)}(p) \,,
\end{equation}
with $\C_+^{(n)}$ defined in~\eqref{eq: Cn+ overall def}. Closing all integration contours to the left yields: 
\begin{equation}
\label{eq: Wn expression}
    \boxed{
    \W^{(n)}_{+, \bm{\mu}}(\u; p) = \tilde{\C}_+^{(n)}(p) \sum_{\bm{\alpha}=\pm i\bm{\mu}} \tilde{\F}^{(n)}_{\bm{\alpha}}(\u; p+\tfrac{n-2}{2}d) \,,
    }
\end{equation}
where each mode is defined as
\begin{equation}
    \boxed{
    \tilde{\F}^{(n)}_{\bm{\alpha}}(\u; p) \equiv \Gamma[\tfrac{p+|\bm{\alpha}|\pm i \mu_n}{2}, -\bm{\alpha}, \bm{1}+\bm{\alpha}] \, \u^{\bm{\alpha}} \, \bar{F}^{(n-1)}_C\left[\left.\begin{matrix} \frac{p+|\bm{\alpha}|+i\mu_n}{2},\,  \frac{p+|\bm{\alpha}|-i\mu_n}{2}\\ \bm{1}+\bm{\alpha} \end{matrix}\right\vert \u^2\right] \,.
    }
\end{equation}
Here, the multi-index vectors have size $n-1$: $\bm{\alpha}=(\alpha_1, \ldots, \alpha_{n-1}), \bm{\mu}=(\mu_1, \ldots, \mu_{n-1})$ and $\u=(u_1, \ldots, u_{n-1})$ (recall $u_n=1$), so that the sum in~\eqref{eq: Wn expression} has $n-1$ terms. An important feature is that setting the vertex energy to zero lowers the transcendental weight of the vertex function; an $n$-leg massive vertex function is described by a Lauricella function of order $n-1$, $F_C^{(n-1)}$. Setting $\mu_n=\pm i/2$ (one leg being conformally coupled), we recover the non-marginal vertex function using the duplication formula for the $\Gamma$-function.

\subsubsection*{Euclidean region}

Determining the physically allowed kinematic region for marginal vertex functions closely parallels the derivation presented in Sec.~\ref{subsec: Euclidean region}, with one subtle but important modification. The resulting constraints are
\begin{equation}
\label{eq: Euclidean region marginal vertex}
    u_i \leq 1+\sum_{j\neq i} u_j \,, \quad (i=1, \ldots, n-1)\,, \quad \text{and} \quad 1\leq \sum_{i=1}^{n-1} u_i \,.
\end{equation}
Unfortunately, the last condition---arising from imposing the generalised triangle inequality on the largest internal energy $\bm{Y}_n$---has a significant consequence: it excludes the simplex at the origin, thereby restricting the physically allowed kinematic region to a domain in which analytic continuations of type-$C$ Lauricella functions are currently not available. We will discuss this in a future work. As an example, in the marginal case $n=3$, the physical region is given by only (\textbf{IV}) in Fig.~\ref{fig: kinematic space for V2}.

\subsubsection*{Orthogonality relation for $n=2$}

A particularly illuminating special case of the marginal vertex function arises for $n=2$ massive legs. In this configuration, momentum conservation at the vertex with vanishing external energy, $X=0$, together with the requirement $Y_1, Y_2>0$, forces both internal momenta to be equal: $Y_1=Y_2$. In the dimensionless variables, this translates to $u_1=u_2=1$, which, as noted below~\eqref{eq: Euclidean region marginal vertex}, places the kinematic configuration precisely on the boundary of the physical region for marginal vertices. This corresponds to a maximally degenerate, collinear configuration. 

\vskip 4pt
For the special choice of twist $p=0$, the marginal vertex function evaluates to 
\begin{equation}
\label{eq: KL orthogonality relation}
    \W_{+, \mu_1, \mu_2}^{(2)}(\u; 0) = e^{+\frac{\pi}{2}\mu_{12}} \int_{-\infty}^0 \frac{\d z}{(-z)} H_{i\mu_1}^{(2)}(-z) H_{i\mu_2}^{(2)}(-z) = \tfrac{4|\mu_1|}{\pi} \Gamma[\pm i \mu_1]\delta(\mu_1^2 - \mu_2^2) \,.
\end{equation}
This is precisely the orthonormality condition of the Kontorovich-Lebedev transform~\cite{Kontorovich1938, Lebedev1946}, once the time axis has been Wick rotated on the imaginary axis, which turns the Hankel functions into modified Bessel ones. Notice that~\eqref{eq: KL orthogonality relation} is dimension independent. Physically, this identity expresses the fact that Hankel functions $H_{i\mu}^{(2)}$, with $\mu\in \mathbb{R}$, form a complete orthogonal set under the measure $\d z/(-z)$. We refer to~\cite{Belrhali:2026rkn} for the derivation using spectral theory; an independent derivation using Mellin-Barnes techniques is given in~\cite{Grafe:2026qsm}.

\newpage
\section{Massive Tree Graphs from Combinatorics}
\label{sec: spectral gluing}

Having studied the vertex functions, the remaining task is to evaluate the spectral integrals stemming from (anti-)time-ordered bulk-to-bulk propagators. As shown in~\eqref{eq: iepsilon distribution}, the associated spectral density naturally decomposes into two contributions: a factorised part, which trivialises the spectral integrals through a $\delta$-function and directly combines with the vertex functions, and an analytic part, for which the spectral integration must be carried out explicitly. 

\vskip 4pt
In this section, we develop a systematic algorithm for constructing arbitrary tree graphs via spectral gluing. We begin by examining the gluing procedure for two adjacent vertices in Sec.~\ref{subsec: gluing adjacent vertices}, and then extend this construction to more general graphs in Sec.~\ref{subsec: maximally-nested contributions from graph combinatorics}. We present an algorithm to write any {\it maximally-nested analytic} contribution of any massive tree correlator purely from combinatorial properties of the corresponding graph. 
The extraction of on-shell poles is discussed in Sec.~\ref{subsec: on-shell poles}. Finally, we summarise the procedure by formulating the {\it spectral gluing algorithm} in Sec.~\ref{subsec: spectral gluing algorithm}.

\subsection{Gluing adjacent vertices}
\label{subsec: gluing adjacent vertices}

Consider two vertices $i$ and $j$ of an arbitrary graph that carry the same Schwinger-Keldysh index and are connected by an edge. For concreteness, we focus on the case of two $(+)$-vertices with $n_i$ and $n_j$ massive legs, respectively. The analogous case of two $(-)$-vertices is obtained by complex conjugation. Both vertex functions depend on the kinematic variables $\bm{u}_i=(u_{i1}, \ldots, u_{in_i})$ and $\bm{u}_j=(u_{j1}, \ldots, u_{jn_j})$, with $0\leq u_{ik}$ ($k=1, \ldots, n_i$) and $0\leq u_{j\ell}$ ($\ell=1, \ldots, n_j$). The indices refer to the energies that the respective variable depends on. We set $u_{ik}\equiv Y_{ik}/X_i$ and $u_{j\ell}=Y_{j\ell}/X_j$. In particular, $u_{ij}\neq u_{ji}$. We now consider gluing the massive leg $u_{ij}$ of the left vertex to the massive leg $u_{ji}$ of the right vertex, which schematically reads
\begin{center}
\begin{tikzpicture}
[line width=1. pt, scale=2]
    \draw[fill=black] (0, 0) circle (.05cm) node[below right=1mm] {$\textcolor{pyblue}{u_{ij}}$};
    \draw[fill=black] (0, 0) circle (.05cm) node[above right] {$\textcolor{pyblue}{\nu}$};
    
    \draw[fill=black] (1, 0) circle (.05cm) node[below left=1mm] {$\textcolor{pyblue}{u_{ji}}$};
    \draw[fill=black] (1, 0) circle (.05cm) node[above left] {$\textcolor{pyblue}{\nu}$};
    
    \draw[black] (0, 0) -- (1, 0);

    \draw (0,0) -- ++(120:0.6) node[above] {$(u_{i1}, \mu_{i1})$};
    \draw (0,0) -- ++(150:0.6) node[above left] {$(u_{i2}, \mu_{i2})$};
    \draw (0,0) -- ++(180:0.6) node[left] {$(u_{i3}, \mu_{i3})$};
    \node at (-0.3,-0.15) {$\vdots$};
    \draw (0,0) -- ++(-120:0.6) node[below] {$(u_{in_i}, \mu_{in_i})$};

    \coordinate (top) at ($(-1.5,0.8)$);
    \coordinate (bottom) at ($(-1.5,-0.8)$);

    \draw[decorate,decoration={brace,amplitude=6pt}]
        (bottom) -- (top)
        node[midway,xshift=-1cm] {$n_i$-legs};

    \draw (1,0) -- ++(60:0.6) node[above] {$(u_{j1}, \mu_{j1})$};
    \draw (1,0) -- ++(30:0.6) node[above right] {$(u_{j2}, \mu_{j2})$};
    \draw (1,0) -- ++(0:0.6) node[right] {$(u_{j3}, \mu_{j3})$};
    \node at (1.3,-0.15) {$\vdots$};
    \draw (1,0) -- ++(-60:0.6) node[below] {$(u_{jn_j}, \mu_{jn_j})$};

    \coordinate (top) at ($(2.5,0.8)$);
    \coordinate (bottom) at ($(2.5,-0.8)$);

    \draw[decorate,decoration={brace,amplitude=6pt}]
        (top) -- (bottom)
        node[midway,xshift=+1cm] {$n_j$-legs.};
\end{tikzpicture}
\end{center}
For later convenience, the kinematic multi-index vector $\bm{u}_i$ where $u_{ij}$ is omitted from the list is denoted as
\begin{equation}
    \bm{u}_{i,(j)} = (u_{i1}, \ldots, u_{i,j-1}, u_{i,j+1}, \ldots, u_{in_i}) \,,
\end{equation}
where the subscript $(j)$ means ``all components except $j$", and similarly for $\bm{u}_{j,(i)}$. This contribution contains a spectral integral over the off-shell mass parameter $\nu$ running in the internal edge, and is given by
\begin{equation}
\label{eq: gluing adjacent vertices integral}
    \widehat{\I}^{\P_{ij}}(\bm{u}_i, \bm{u}_j) \equiv e^{-\frac{i\pi}{2}} \int\limits_{-\infty}^{+\infty} [\d\nu] \, \rho_{\nu, \P} \, \V_{+, (\nu, \bm{\mu}_{i,(j)})}^{(n_i)}(\bm{u}_i, p_i) \V_{+, (\nu, \bm{\mu}_{j,(i)})}^{(n_j)}(\bm{u}_j, p_j) \,.
\end{equation}
Here, we entirely focus on the analytic contribution from the spectral density $\rho_{\nu, \P}$, which we denote with $\P$ (referring to the Cauchy principal value, see Sec.~\ref{subsec: spectral propagators}). The overall phase $e^{-\frac{i\pi}{2}}$ comes from the spectral propagator~\eqref{eq: split representation ++}. Notice that this spectral function is not necessarily the spectral density of a free propagator, but can also be a resummed propagator or any exchange density. This way, the following gluing procedure captures not only tree-level contributions, but also loop-order exchanges, generalising bubble-loop resummations in~\cite{Grafe:2026qsm}. 

\subsubsection*{Dissecting the product of vertex functions}

Let us now study the analytic and kinematic structure of the product of two vertex functions. Each of the vertex functions in~\eqref{eq: gluing adjacent vertices integral} can be separated into a term proportional to $(u_{ij}/2)^{+i\nu}$ and one proportional to $(u_{ij}/2)^{-i\nu}$ (similarly for $u_{ji}$). These factors dominate the large-$\nu$ behaviour of the integrand and therefore control the choice of half plane in which to close the spectral integral contour. 

\vskip 4pt
In what follows, it will be useful to isolate the non-analytic dependence on $\nu$ in the vertex function by rewriting it as follows:
\begin{equation}
\label{eq: vertex function non-analytic nu decomposition}
    \begin{aligned}
        \V_{+, (\nu, \bm{\mu}_{i,(j)})}^{(n_i)}(\bm{u}_i; p_i) = &\C_+^{(n_i)} \left[\left(\frac{u_{ij}}{2}\right)^{+i\nu} \sum_{\bm{\alpha}=\pm i \bm{\mu}_{i,(j)}} \textcolor{pyblue}{\Gamma[-i\nu, 1+i\nu, \tilde{p}_i+i\nu+|\bm{\alpha}|]} \mathtt{F}_{+i\nu, \bm{\alpha}}^{(n_i)}(\bm{u}_i; \tilde{p}_i) \right.\\
        &\left.\hspace*{-0.2cm}+\left(\frac{u_{ij}}{2}\right)^{-i\nu} \sum_{\bm{\alpha}=\pm i \bm{\mu}_{i,(j)}} \textcolor{pyblue}{\Gamma[+i\nu, 1-i\nu, \tilde{p}_i-i\nu+|\bm{\alpha}|]} \mathtt{F}_{-i\nu, \bm{\alpha}}^{(n_i)}(\bm{u}_i; \tilde{p}_i)\right] \,,
    \end{aligned}
\end{equation}
with $\tilde{p}_i\equiv p_i+\tfrac{n_i-2}{2}d$, and where we have defined the auxiliary reduced mode functions
\begin{equation}
    \mathtt{F}_{\pm i\nu, \bm{\alpha}}^{(n_i)}(\bm{u}_i; \tilde{p}_i) \equiv \Gamma[-\bm{\alpha}, \bm{1}+\bm{\alpha}] \left(\frac{\bm{u}_{i,(j)}}{2}\right)^{\bm{\alpha}} \bar{F}^{(n_i)}_C\left[\left.\begin{matrix} \frac{\tilde{p}_i+|\bm{\alpha}|\pm i\nu}{2},\,  \frac{\tilde{p}_i+1+|\bm{\alpha}| \pm i\nu}{2}\\ 1\pm i\nu, \bm{1}+\bm{\alpha} \end{matrix}\right\vert \bm{u}_i^2\right] \,,
\end{equation}
that are entire functions of the complex variable $\nu$. In~\eqref{eq: vertex function non-analytic nu decomposition}, all poles in the complex $\nu$-plane are encoded in the \textcolor{pyblue}{blue} $\Gamma$-factors. The sum contains $2^{n-1}$ terms. We decompose $\V_{+, (\nu, \bm{\mu}'_{-j})}^{(n')}(\bm{u}'; p')$ in the same way. With this decomposition, the integral~\eqref{eq: gluing adjacent vertices integral} reads
\begin{equation}
    \begin{aligned}
        \widehat{\I}^{\P_{ij}}(\bm{u}_i, \bm{u}_j) = \sum_{\sigma_i, \sigma_j=\pm i\nu} \widehat{\I}^{\P_{ij}}_{\sigma_i, \sigma_j}(\bm{u}_i, \bm{u}_j) \,,
    \end{aligned}
\end{equation}
where the four terms in the sum are given by
\begin{equation}
    \begin{aligned}
        \widehat{\I}^{\P_{ij}}_{\sigma_i, \sigma_j}(\bm{u}_i, \bm{u}_j) &= e^{-\frac{i\pi}{2}} \C_+^{(n_i)}(p_i)\C_+^{(n_j)}(p_j) \int\limits_{-\infty}^{+\infty} \d\nu \, \rho_{\nu, \P} \, \left(\frac{u_{ij}}{2}\right)^{\sigma_i} \left(\frac{u_{ji}}{2}\right)^{\sigma_j} \\
        &\times \sum\limits_{\substack{\bm{\alpha}_i =\pm i\bm{\mu}_{i,(j)} \\ \bm{\alpha}_j = \pm i\bm{\mu}_{j,(i)}}} \Gamma\left[\begin{matrix} -\sigma_i, -\sigma_j, 1+\sigma_i, 1+\sigma_j, \tilde{p}_i+\sigma_i+|\bm{\alpha}_i|, \tilde{p}_j+\sigma_j+|\bm{\alpha}_j| \\ +i\nu, -i\nu \end{matrix}\right]  \\
        &\times \mathtt{F}_{\sigma_i, \bm{\alpha}_i}^{(n_i)}(\bm{u}_i; \tilde{p}_i) \mathtt{F}_{\sigma_j, \bm{\alpha}_j}^{(n_j)}(\bm{u}_j; \tilde{p}_j) \,,
    \end{aligned}
\end{equation}
where the factor $\Gamma[\pm i\nu]$ in the denominator stems from the integration measure $[\d\nu]\equiv \d\nu\N_\nu$. Each of the four integrals $\widehat{\I}^{\P_{ij}}_{\sigma_i, \sigma_j}$ exhibits a different large-$\nu$ behaviour of the integrand. While $\widehat{\I}^{\P_{ij}}_{\pm i\nu, \pm i\nu}$ depends on $(u_{ij} u_{ji}/4)^{\pm i\nu}$, the opposite sign integrals $\widehat{\I}^{\P_{ij}}_{\pm i\nu, \mp i\nu}$ depend on $(u_{ij} /u_{ji})^{\pm i\nu}$. Therefore, depending on whether $u_{ij}u_{ji}<4$ and $u_{ij}<u_{ji}$, we close the contour for each integral either in the upper or in the lower half plane. Naively, this suggests four different kinematic regions for the full integral, as shown in Fig.~\ref{fig: kinematic space product vertex}. However, we will now see that the final result does not depend on poles stemming from the equal sign integrals $\widehat{\I}^{\P_{ij}}_{\pm i\nu, \pm i\nu}$. 

\begin{figure}[h!]
\centering
    \begin{tikzpicture}[line width=1. pt, scale=2]

    \draw[->] (-0.2,0) -- (3.2,0) node[right] {$u_{ij}$};
    \draw[->] (0,-0.2) -- (0, 3.2) node[above] {$u_{ji}$};

    \foreach \x in {2} {
        \draw (\x, 0.03) -- (\x, -0.03) node[below] {\small $\x$};
    }
    \node[below left] at (0,0) {\small $0$};
    
    \foreach \y in {2} {
        \draw (0.03, \y) -- (-0.03, \y) node[left] {\small $\y$};
    }

    \draw[pyblue, line width=1.2pt, name path=diag] (0,0) -- (3,3) node[above right] {$u_{ij} = u_{ji}$};

    \draw[pyred, line width=1.2pt, domain=1.33:3, name path=hyper] plot (\x,{4/\x}) node[right] {$u_{ij} u_{ji} = 4$};
    
    \draw[gray, dotted, line width=1.2pt] (2,0) -- (2,2);
    \draw[gray, dotted, line width=1.2pt] (0,2) -- (2,2);
    \fill (2,2) circle (1pt);

    \fill[pyblue!5, intersection segments={of=diag and hyper, sequence={L2--R2}}] -- (3, 0) -- (0,0) -- cycle;

    \fill[pyblue!20]
        (2,2) --
        plot[domain=2:3] (\x,{4/\x}) -- (3,3) -- (2,2) -- cycle;

    \fill[pyred!5]
        (0,0) -- (2,2) --
        plot[domain=2:1.33] (\x,{4/\x}) -- (0,3) -- cycle;

    \fill[pyred!20]
        (2,2) -- (3,3) -- (1.33,3) --
        plot[domain=1.33:2] (\x,{4/\x}) -- (2,2) -- cycle;

    \draw[->] (-0.2,0) -- (3.2,0) node[right] {$u_{ij}$};
    \draw[->] (0,-0.2) -- (0, 3.2) node[above] {$u_{ji}$};

    \foreach \x in {2} {
        \draw (\x, 0.03) -- (\x, -0.03) node[below] {\small $\x$};
    }
    \node[below left] at (0,0) {\small $0$};
    
    \foreach \y in {2} {
        \draw (0.03, \y) -- (-0.03, \y) node[left] {\small $\y$};
    }

    \draw[pyblue, line width=1.2pt, name path=diag] (0,0) -- (3,3) node[above right] {$u_{ij} = u_{ji}$};

    \draw[pyred, line width=1.2pt, domain=1.33:3, name path=hyper] plot (\x,{4/\x}) node[right] {$u_{ij} u_{ji} = 4$};
    
    \draw[gray, dotted, line width=1.2pt] (2,0) -- (2,2);
    \draw[gray, dotted, line width=1.2pt] (0,2) -- (2,2);
    \fill (2,2) circle (1pt);

    \end{tikzpicture}
    \caption{Kinematic regions defined by the curves $u_{ij}=u_{ji}$ and $u_{ij} u_{ji}=4$, corresponding to different choices of contour closure in the integrals $\widehat{\I}^{\P_{ij}}_{\sigma_i, \sigma_j}$. The boundary $u_{ij} u_{ji}=4$ is spurious, as it is removed by pole cancellation.}
    \label{fig: kinematic space product vertex}
\end{figure}
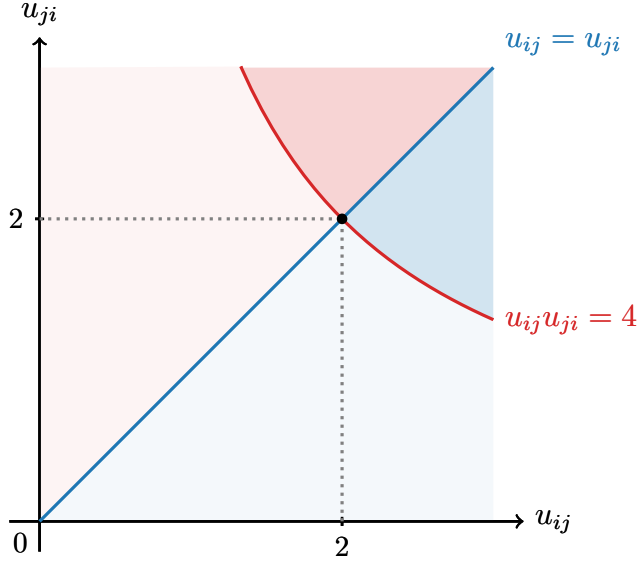

\subsubsection*{Pole cancellation \& physical poles}

The representation of $\V_{+, (\nu, \bm{\mu})}^{(n)}(\bm{u}; p)$ in~\eqref{eq: vertex function non-analytic nu decomposition} suggests the presence of poles at $\nu = \pm im$ with $m\in\mathbb{N}$ stemming from the factors $\Gamma[\mp i\nu,1\pm i\nu]$. However, it turns out that this is simply an effect of our choice of representation in terms of positive and negative frequency modes. We now prove that these poles are not physical, i.e.~they do not contribute to the spectral integration. This property is formulated for general vertex functions in the following lemma: 

\begin{lemma}
\label{lemma: regular vertex function}
    Any vertex function $\mathcal{V}_{+,\bm\mu}^{(n)}(\bm{u};p)$ is regular at $\mu_j =\pm im$ for all $m\in\mathbb{N}$ and $j=1,\dots,n$.
\end{lemma}

\begin{proof}
    Let us decompose the indices as $\bm\mu = (\mu_j, \bm\mu_{(j)})$ and the arguments as $\bm{u}=(u_j,\bm{u}_{(j)})$, as done previously. We can write
    \begin{equation}
        \mathcal{V}_{+,(\mu_j,\bm\mu_{(j)})}^{(n)}(u_j,\bm{u}_{(j)};p) = \mathcal{C}_+^{(n)}(p) \sum\limits_{\substack{\sigma =\pm i\mu_j \\ \bm\alpha = \pm i\bm\mu_{(j)}}} \Gamma[-\sigma,1+\sigma, \tilde p+\sigma +|\bm\alpha|] \left(\frac{u_j}{2}\right)^\sigma \mathtt{F}_{\sigma, \bm{\alpha}}^{(n)}(\bm{u}; \tilde{p})\,.
    \end{equation}
    We consider the (apparent) pole at $\mu_j =-im$ for $m\in\mathbb{N}$. For the $\sigma=+i\mu_j$ term this pole stems from the factor $\Gamma(-i\mu_j)$ which has residue $i(-1)^m/m!$. For the $\sigma=-i\mu_j$ term the pole stems from $\Gamma(1-i\mu_j)=-i\mu_j\Gamma(-i\mu_j)$ with residue $-mi(-1)^m/m!$. For the full vertex function we find
    \begin{equation}
    \label{eq: residue of Vn at -im}
        \begin{aligned}
        \underset{\mu_j=-im}{\Res} \mathcal{V}_{+,(\mu_j,\bm\mu_{(j)})}^{(n)} &= i\mathcal{C}_{+}^{(n)}(p) \sum\limits_{\bm\alpha=\pm i\bm\mu_{(j)}}(-1)^m \left\{\Gamma(\tilde p+|\bm\alpha|+m)\left(\frac{u_j}{2}\right)^m \mathtt{F}_{m,\bm\alpha}^{(n)}(\bm{u}; \tilde{p}) \right.\\
        &\left.- (m\leftrightarrow-m)\right\} \,.
        \end{aligned}
    \end{equation}
    Using a reflection identity for the Lauricella function that we prove in~\ref{app: Lauricella Function}, we obtain
    \begin{equation}
        \mathtt{F}_{-m,\bm\alpha}^{(n)}(\bm{u}; \tilde{p}) = \Gamma\left[\begin{matrix}
            \tilde p +|\bm\alpha|+m \\ \tilde p +|\bm\alpha|-m
        \end{matrix}\right] \left(\frac{u_j}{2}\right)^{2m} \mathtt{F}_{m,\bm\alpha}^{(n)} (\bm{u}; \tilde{p})\,,
    \end{equation}
    and hence~\eqref{eq: residue of Vn at -im} vanishes exactly. Similarly, the residue at $\mu_j=+im$ vanishes and hence $\mathcal{V}_{+,\bm\mu}^{(n)}(\bm{u};p)$ is regular at $\mu_j=\pm im$.
\end{proof}

As a consequence, when performing the spectral integrals over vertex functions, we can always ignore these spurious poles that stem from unphysical $\Gamma$-factors even though we might separate the integral into positive and negative frequency modes. Alternatively, we can realise that these apparent poles cancel among poles in mixing terms.

\subsubsection*{Spectral integration of adjacent vertices}

Let us now compute the spectral integrals $\widehat{\I}^{\P_{ij}}_{\sigma_i, \sigma_j}(\bm{u}_i, \bm{u}_j)$. 

\begin{itemize}
    \item $\widehat{\I}^{\P_{ij}}_{++}$: This integral depends on $(u_{ij} u_{ji}/4)^{i\nu}$ in its integrand. We first choose the kinematic region $u_{ij} u_{ji}<4$ where we must close the contour in the lower half plane to ensure that the integrand vanishes along the additional arc. In this case, only the factor $\Gamma(-i\nu)^2$ has poles in the lower half plane. One of these factors is cancelled by $\mathcal{N}_\nu$, so we encounter simple poles at $\nu = -i m$ for $m\in\mathbb{N}$. However, using our previous lemma, these are only apparent poles and cannot contribute to the final result, hence we can ignore them (alternatively we can pick up the residues and see later that they cancel with a contribution from $\widehat{\I}^{\P_{ij}}_{\pm \mp}$). Hence, in this kinematic region $\widehat{\I}^{\P_{ij}}_{++} = 0$, and by the identity theorem from complex analysis this can be analytically continued to the remaining kinematic region.

    \vskip 4pt
    One might worry that in the case $u_{ij} u_{ji}>4$ we need to close the contour in the upper half plane and pick up poles from both $\Gamma$-functions $\Gamma(\tilde p_i +|\bm\alpha_i|+i\nu)$ and $\Gamma(\tilde p_j +|\bm\alpha_j|+i\nu)$ leading to two additional residue series. However, for $u_{ij} u_{ji}>4$ at least one of the kinematic variables $u_{ij},u_{ji}$ lies outside the convergence domain of the respective Lauricella function $F_C^{(n)}[\dots\vert u_{ij}^2,\dots]$ which requires analytic continuation of the integrand. It turns out that this procedure introduces a factor $u_{ij}^{-i\nu}$ which exactly cancels the diverging kinematic term and allows us to close the contour in the lower half plane again. For the case of $F_C^{(1)}={}_2F_1$ one can see this by using a connection formula~\eqref{eq: 2F1 analytic continuation formula} for the hypergeometric function.

    \item $\widehat{\I}^{\P_{ij}}_{--}$: Likewise, this contribution must vanish and we are only left with the mixed integrals $\widehat{\I}^{\P_{ij}}_{\pm\mp}$.

    \item $\widehat{\I}^{\P_{ij}}_{+-}$: This integral depends on the kinematic ratio $(u_{ij}/u_{ji})^{i\nu}$. In this case poles stem from $\Gamma(1+i\nu,1-i\nu)$ as well as $\Gamma[\tilde p_i +|\bm\alpha_i|+i\nu,\tilde p_j +|\bm\alpha_j|-i\nu]$. The first set is only an apparent set of poles again and cancels the other poles from the $\widehat{\I}^{\P_{ij}}_{++}$ integral. Hence there are exactly two sets of \emph{physical} poles which depend on the twists $\tilde p_i, \tilde p_j$ as well as the other spectral indices. One of these sets lies in the upper and one in the lower half plane, and the choice of direction in which to close the contour depends on the ratio $u_{ij}/u_{ji}$. If $u_{ij}<u_{ji}$, we close in the lower half plane and pick up the poles from $\Gamma(\tilde p_j+|\bm\alpha_j|-i\nu)$ at $\nu =-i(m+\tilde p_j+|\bm\alpha_j|),\,m\in\mathbb{N}$. Similarly, if $u_{ij}>u_{ji}$, then we close in the upper half plane and pick up the poles from $\Gamma(\tilde p_i+|\bm\alpha_i|-i\nu)$ at $\nu =i(m+\tilde p_i+|\bm\alpha_i|),\,m\in\mathbb{N}$.

    \item $\widehat{\I}^{\P_{ij}}_{-+}$: Computing this integral is analogous to $\widehat{\I}^{\P_{ij}}_{+-}$.
\end{itemize}

In the kinematic region $u_{ij}<u_{ji}$, since $\widehat{\I}^{\P_{ij}}_{+-}=\widehat{\I}^{\P_{ij}}_{-+}$ by shadow symmetry, the result for the full integral $\widehat{\I}^{\P_{ij}}$ is given by
\begin{equation}
\label{eq: elementary gluing formula}
    \boxed{
    \begin{aligned}
        \widehat{\I}^{\P_{ij}}(\bm{u}_i, \bm{u}_j) &= 4\pi\, e^{-\frac{i\pi}{2}} \, \C_+^{(n_i)}(p_i) \C_+^{(n_j)}(p_j)\sum\limits_{m=0}^\infty \frac{(-1)^m}{m!} \\
        &\times\sum\limits_{\substack{\bm\alpha_i=\pm i\bm\mu_{i,(j)} \\ \bm\alpha_j=\pm i \bm\mu_{j,(i)}}} \rho_{-i\xi_m', \P} \Gamma[\xi_m', 1-\xi_m', \tilde{p}_i+|\bm\alpha_i|+\xi_m']\left(\frac{u_{ij}}{u_{ji}}\right)^{\xi_m'} \\
        &\times \mathtt{F}_{\xi_m', \bm\alpha_i}^{(n_i)}(\bm{u}_i; \tilde{p}_i) \mathtt{F}_{-\xi_m',\bm\alpha_j}^{(n_j)}(\bm{u}_j; \tilde{p}_j)\,,
    \end{aligned}
    }
\end{equation}
where $\xi_m' \equiv m+\tilde{p}_j+|\bm\alpha_j|$ shifts the parameter in the vertex functions running in the edge. The overall factor $4\pi e^{-\frac{i\pi}{2}}$ is ever present when gluing two adjacent vertex functions. It cancels against the prefactor that we pulled out when going from the graph $\G$ to the dimensionless seed integral $\widehat{I}$.

\vskip 4pt
The region $u_{ji}<u_{ij}$ is accessed by simply exchanging primed and un-primed variables. Since the integrals $\widehat{\I}^{\P_{ij}}_{\pm\pm}$ vanish, the naive four kinematic regions of Fig.~\ref{fig: kinematic space product vertex} collapse to only two regions that simply depend on the ratio of the two relevant kinematic variables on this edge.

\subsection{Maximally-nested contributions from graph combinatorics}
\label{subsec: maximally-nested contributions from graph combinatorics}

Having established the spectral gluing of adjacent pairs in Sec.~\ref{subsec: gluing adjacent vertices}, we now present a complete algorithm for constructing the maximally-nested analytic contribution to the all-$(+)$ master integral of any tree-level massive cosmological correlator. The key observation is that the output of one gluing inherits precisely the mode structure needed to apply the gluing formula again, so that one can proceed iteratively. Eventually, the full solution can be directly read off from graph combinatorics. 

\subsubsection*{Iterative procedure}

The reason for that is the following observation: the pole structure of the vertex functions is fully encoded in the factors $\Gamma[\tilde p +|\bm \alpha|]$ which depend on all the spectral parameters of the edges that are connected to a given vertex. Evaluating one of the spectral integrals amounts to picking up the residues at the poles of one of these factors which pushes the parametric dependence of these poles to all the $\Gamma$-factors of the vertex functions that depend on the same spectral variables. When solving the next spectral integral, the location of the poles now depends on the previously picked up poles and so on. Instead of going through this procedure iteratively, we can solve all the spectral integrals independently from each other (setting only one spectral index on-shell at each time) and solve a system of linear equations that relates the poles to each other. For more details, see App.~\ref{app: three-site chain} where we explicitly go through the iterative computation of the three-site chain.

\subsubsection*{Combinatorial algorithm}

Let us now summarise the combinatorial algorithm presented below, introducing the precise definitions only when they become necessary. In a nutshell, evaluating all spectral integrals amounts to solving a linear system whose coefficient matrix is the {\it reduced incidence matrix} of the correlator graph after fixing a root vertex: a purely combinatorial object encoding the graph topology. The unique solution to the system is read off via the {\it inverse reduced incidence matrix}, which is given by the {\it path matrix} and can be determined purely graphically. Solving the system sets all spectral variables on-shell in a previously-determined off-shell ansatz and produces a multi-fold series representation of the correlator valid in a specific kinematic region, that can also be determined fully graphically. We summarise this algorithm in the insert~\ref{box: MNA algo}, and provide a detailed example immediately afterwards.

\subsubsection*{Tree graphs}

We begin by setting up and recalling the relevant graph-theoretic definitions and notations. Let $\G$ be a connected tree graph with $V$ labelled vertices and $I=V-1$ internal edges.\footnote{Although we use the same notation here, the connected tree graph $\G$ is distinct from the full-correlator graph defined in Sec.~\ref{subsec: master integrals for tree correlators}, which represents the sum of all Schwinger-Keldysh contributions.} We denote the vertex set by $\{i\}$ and the edge set by $\{(i,j)\}$. The {\it degree} $n_i\geq1$ of a vertex $i$ is the number of edges incident to it; vertices of degree one are called {\it external} or {\it leaf} vertices, and all others are {\it internal}. We recall that to each vertex $i$ we assign an external energy $X_i>0$ (and a modified twist $\tilde{p}_i\in\mathbb{C}$), and to edge $(i,j)$ we assign an internal energy $Y_{ij}>0$ and a mass parameter $\mu_{ij}\in\mathbb{R}_{>0}$. For each incident vertex-edge pair $(i, (i,j))$, we define the dimensionless kinematic ratio $u_{ij}\equiv Y_{ij}/X_i$. For a leaf vertex $i$ with unique adjacent edge $(i,j)$ we simply write $u_i\equiv u_{ij}$; the physical region then requires $0\leq u_i \leq 1$. For internal vertices, the $u_{ij}$ satisfy the Euclidean inequalities of Sec.~\ref{subsec: Euclidean region}. Eventually, for each edge $(i,j)$, we introduce the off-shell spectral variable $x_{ij}\equiv i\nu_{ij}$, where $\nu_{ij}\in\mathbb{C}$ is the corresponding off-shell mass parameter.

\subsubsection{Rooted tree \& incidence matrix}

The first step is to choose a rooted tree. A {\it rooted tree} (or {\it family tree}) of $\G$ is obtained by designating one vertex $r\in\{1, \ldots, V\}$ as the {\it root} and orienting every edge away from the root $r$ and towards the leaves. Every non-root vertex $i$ then has a unique {\it parent edge} $(i,\pi(i))$: the last edge on the unique path from $r$ to $i$ in $\G$. The map $i\mapsto \pi(i)$ is a bijection from the $V-1$ non-root vertices to the $I=V-1$ edges. Each rooted tree uniquely fixes a graph decoration. We denote the rooted tree $\G_r$. A {\it rooted decoration} of $\G_r$ is an assignment of orientations, marked by arrows, to each edge, pointing away from the root vertex. 

\vskip 4pt
The (vertex-edge) {\it incidence matrix} of $\G_r$, denoted $Q(\G_r)$, is a $V\times I$ matrix defined as follows. The rows and columns of $Q(\G_r)$ are indexed by vertices $\{1, \ldots, V\}$ and edges $\{1, \ldots, I\}$, respectively. The $(i, j)$ entry of $Q(\G_r)$ is $0$ if vertex $i$ and edge $(i,j)$ are not incident, and otherwise it is $-1$ or $+1$ according as $(i,j)$ originates or terminates at $i$, respectively. Each column of $Q(\G_r)$ contains exactly one entry $+1$ and one entry $-1$ at the positions of the two endpoints of the respective edge, with all other entries zero. We often denote $Q(\G_r)$ simply by $Q$. An example is given by
\begin{equation}
\vcenter{\hbox{
\begin{tikzpicture}
[line width=1. pt, scale=2]
    [line width=1. pt, scale=2, baseline={(current bounding box.center)}]
    
    \draw[black, postaction={decorate}, decoration={markings, mark=at position 0.6 with {\arrow[pyred]{>}}}] (0, 0) -- node[left] {$\textcolor{pyblue}{x_{14}}$} (0, 0.7);
    \draw[black, postaction={decorate}, decoration={markings, mark=at position 0.6 with {\arrow[pyred]{<}}}] (0, 0) -- node[above left] {$\textcolor{pyblue}{x_{24}}$} (-0.5, -0.5);
    \draw[black, postaction={decorate}, decoration={markings, mark=at position 0.6 with {\arrow[pyred]{>}}}] (0, 0) -- node[above right] {$\textcolor{pyblue}{x_{34}}$} (0.5, -0.5);
        
    \draw[fill=black] (0, 0.7) circle (.05cm) node[above] {$\textcolor{black}{(1)}$};
    \draw[pyred, fill=pyred] (-0.5, -0.5) circle (.05cm) node[below left] {$\textcolor{black}{(2)}$};
    \draw[fill=black] (0.5, -0.5) circle (.05cm) node[below right] {$\textcolor{black}{(3)}$};
    \draw[fill=black] (0, 0) circle (.05cm) node[above right] {$\textcolor{black}{(4)}$};
\end{tikzpicture} 
    }}\,, \quad Q = 
    \begin{bNiceMatrix}[first-row,code-for-first-row=\scriptstyle, first-col,code-for-first-col=\scriptstyle,]
        & \textcolor{gray}{x_{14}} & \textcolor{gray}{x_{24}} & \textcolor{gray}{x_{34}} \\
        \textcolor{gray}{(1)} & +1 & 0 & 0 \\
        \textcolor{gray}{(2)} & 0 & -1 & 0 \\
        \textcolor{gray}{(3)} & 0 & 0 & +1 \\
        \textcolor{gray}{(4)} & -1 & +1 & -1 
    \end{bNiceMatrix}\,,
\end{equation}
where we have coloured in red the chosen root vertex and the edge arrows associated with the corresponding unique rooted decoration, and in blue the edge off-shell parameters which we use to label the edges.

\vskip 4pt
We now collect several structural properties of the incidence matrix of a tree graph that are essential for the algorithm. Their proofs are standard results of algebraic graph theory, see e.g.~\cite{Biggs_1974}.

\begin{proposition}[Properties of $Q(\G_r)$]
\label{prop: incidence matrix}
    Let $\G_r$ be a connected rooted tree with $V$ vertices and $I=V-1$ edges. Then:
    \begin{enumerate}
        \item[(i)] $\operatorname{rank}(Q) = V-1$;
        \item[(ii)] The row vectors of $Q(\G_r)$ sum to zero: $\sum_{i=1}^V Q_{ij}=0$ for all $j=1,\ldots, I$.
        \item[(iii)] $Q(\G_r)$ is {\it totally unimodular}: the determinant of every minor (square submatrix) is $0$ or $\pm1$;
        \item[(iv)] Any $I\times I$ minor of $Q(\G_r)$, obtained by deleting a row, is nonsingular and has determinant $\pm1$.
    \end{enumerate}
\end{proposition}
\noindent Property \textit{(iv)} follows from Kirchhoff's matrix-tree theorem, and will be crucial for what follows.

\subsubsection*{Off-shell ansatz}

Once a rooted tree is chosen, one writes an {\it off-shell ansatz} for the maximally-nested analytic contribution of the corresponding oriented graph. First, for every vertex $i$, we define the following linear combinations of off-shell edge variables
\begin{equation}
    \xi_i \equiv \sum_{j=1}^I Q_{ij} x_j \,,
\end{equation}
or $\bm{\xi}=Q\bm{x}$ in matrix notation. Iterating the single-edge gluing formula of Sec.~\ref{subsec: gluing adjacent vertices} over all $I$ edges yields the following off-shell ansatz: 
\begin{equation}
\label{eq: off-shell ansatz}
    \boxed{
    \begin{aligned}
        \widehat{\I}^{\P}(\G_r) = C(\G) \, &\left[\prod_{i=1}^V \Gamma\left(\tilde{p}_i+\xi_i\right) F^{(n_i)}_C\left[\left.\begin{matrix} \frac{\tilde{p}_i+\xi_i}{2},\,  \frac{\tilde{p}_i+1+\xi_i}{2}\\ 1+\epsilon_{i1} x_{i1}, \ldots, 1+\epsilon_{i n_i}x_{i n_i} \end{matrix}\right\vert u_{i1}^2, \ldots, u_{in_i}^2\right]\right] \\
        &\times \left[\prod_{(i,j)}\frac{1}{x_{ij}^2+\mu_{ij}^2} \, \left(u_{ij}\right)^{\epsilon_{ij} x_{ij}}\left(u_{ji}\right)^{\epsilon_{ji} x_{ji}}\right] \,,
    \end{aligned}
    }
\end{equation}
where the rooted tree property implies $\epsilon_{ij} =-\epsilon_{ji}$ for any edge $(i,j)$, and the overall coefficient is independent of the rooted decoration and is given by 
\begin{equation}
    C(\G) = \left(-4\pi e^{-\frac{i\pi}{2}}\right)^I \, \prod_{i=1}^V \C_+^{(n_i)}(p_i) \,.
\end{equation}
For each vertex $i$, we have defined $\epsilon_{ij}=+1$ (resp.~$\epsilon_{ij}=-1$), for $j=1, \ldots, n_i$ running over all incident edges to $i$, if the arrow on $(i,j)$ points towards (resp.~outwards) vertex $i$. We stress that the signs $\epsilon_{ij}$ are precisely the entries of the incidence matrix, i.e. $Q_{ij}=\epsilon_{ij}$. In particular, we have
\begin{equation}\label{eq: incidence matrix identity off shell ansatz}
    \xi_i = \sum\limits_{\ell=1}^{n_i} \epsilon_{i\ell}x_{i\ell}\,.
\end{equation}
The ansatz is a meromorphic function of the off-shell spectral variables $x_1, \ldots, x_I$, and its structure reflects the graph topology: the contribution from vertex $i$ depends only on the spectral variables of its adjacent edges, while the propagator for edge $(i,j)$ depends only on $x_j$. Of course, the ansatz is decoration dependent and is fully fixed once a root vertex is chosen.

\subsubsection{On-shell condition \& path matrix}

The last step is to set the mass parameters $x_j$ ($j=1, \ldots, I$) on shell. This is done by solving a linear system of algebraic equations. For a given root vertex $r$ (and therefore a given oriented decoration), define the {\it reduced incidence matrix}
\begin{equation}
    Q_r \equiv Q(\G_r)\big|_{\text{row } r \text{ deleted}} \,.
\end{equation}
By proposition~\ref{prop: incidence matrix} (iv), $Q_r$ is nonsingular with $|\det(Q_r)|=1$. The physical result is obtained by setting the spectral variables $x_j$ on shell, which is equivalent to the linear system
\begin{equation}
    Q_r \, \bm{x} = -(\tilde{\bm{p}}_{(r)} + \bm{m}) \,,
\end{equation}
where $\tilde{\bm{p}}_{(r)}^T \equiv (\tilde{p}_1, \ldots, \tilde{p}_{r-1}, \tilde{p}_{r+1}, \ldots, \tilde{p}_V)$ collects the modified twists in the row-order of $Q_r$ with the root twist deleted, and $\bm{m}^T \equiv (m_1, \ldots, m_I) \in \mathbb{N}_{\geq0}^I$ are summation indices. Since $|\det(Q_r)|=1$, the solution is unique:
\begin{equation}
\label{eq: on-shell condition}
    \boxed{
    \bm{x}^\ast = -Q_r^{-1}(\tilde{\bm{p}}_{(r)} + \bm{m})\,.
    }
\end{equation}
Remarkably, the inverse reduced incidence matrix admits a direct graphical construction, as it is nothing but the path matrix.

\begin{definition}
    The {\it path matrix} of the rooted tree $\G_r$ is
    \begin{equation}
        P_r(\G_r) \equiv Q_r^{-1} \in \mathbb{Z}_2^{I\times (V-1)}\,.
    \end{equation}
\end{definition}
\noindent The terminology is justified by the following proposition, which provides an explicit graphical formula for $P_r(\G_r)$ and makes clear that it encodes the path structure of the rooted tree.

\begin{proposition}[Graphical rule for $P_r(\G_r)$]
    Let $\G_r$ be a rooted tree with an oriented decoration. Given a path $\P_{r\to i}$ from the root $r$ to a vertex $i\neq r$ in $\G_r$, the incidence vector of $\P_{r\to i}$ is an $I\times 1$ vector defined as follows. The entries of the vector are indexed by the edges $\{e_1, \ldots, e_I\}$. The $j$th element of the vector is $0$ if the path does not contain $e_j$, and $+1$ otherwise. The path matrix $P_r(\G_r)$, simply denoted $P_r$, is defined by
    \begin{equation}
        (P_r)_{ij} \equiv \P_{r\to i}\,, \quad \text{for } j=1,\ldots, I\,.
    \end{equation}
    The $j$th column of $P_r$ is the incidence vector of the (unique, since $\G_r$ is a tree) path from the root $r$ to the vertex $i\neq r$.
\end{proposition}
\noindent The graphical procedure for reading off $P_r$ is therefore as follows. From the oriented decoration associated with a chosen root, trace the unique path $r\to i$ following the edge orientations and record each traversed edge. The resulting integer matrix is $P_r$. The previous example gives
\begin{equation}
\vcenter{\hbox{
\begin{tikzpicture}
[line width=1. pt, scale=2]
    [line width=1. pt, scale=2, baseline={(current bounding box.center)}]
    
    \draw[black, postaction={decorate}, decoration={markings, mark=at position 0.6 with {\arrow[pyred]{>}}}] (0, 0) -- node[left] {$\textcolor{pyblue}{x_{14}}$} (0, 0.7);
    \draw[black, postaction={decorate}, decoration={markings, mark=at position 0.6 with {\arrow[pyred]{<}}}] (0, 0) -- node[above left] {$\textcolor{pyblue}{x_{24}}$} (-0.5, -0.5);
    \draw[black, postaction={decorate}, decoration={markings, mark=at position 0.6 with {\arrow[pyred]{>}}}] (0, 0) -- node[above right] {$\textcolor{pyblue}{x_{34}}$} (0.5, -0.5);
        
    \draw[fill=black] (0, 0.7) circle (.05cm) node[above] {$\textcolor{black}{(1)}$};
    \draw[pyred, fill=pyred] (-0.5, -0.5) circle (.05cm) node[below left] {$\textcolor{black}{(2)}$};
    \draw[fill=black] (0.5, -0.5) circle (.05cm) node[below right] {$\textcolor{black}{(3)}$};
    \draw[fill=black] (0, 0) circle (.05cm) node[above right] {$\textcolor{black}{(4)}$};
\end{tikzpicture} 
    }}\,, \quad P_2 = 
    \begin{bNiceMatrix}[first-row,code-for-first-row=\scriptstyle, first-col,code-for-first-col=\scriptstyle,]
        & \textcolor{gray}{(1)} & \textcolor{gray}{(3)} & \textcolor{gray}{(4)} \\
        \textcolor{gray}{x_{12}} & +1 & 0 & 0 \\
        \textcolor{gray}{x_{24}} & +1 & +1 & +1 \\
        \textcolor{gray}{x_{34}} & 0 & +1 & 0
    \end{bNiceMatrix}\,.
\end{equation}
It can be explicitly checked that $P_2 = Q_2^{-1}$, where $Q_2$ is obtained from $Q$ after deleting the second row. Notice that by construction, the path matrix has only $0$ or $+1$ entries.

\subsubsection*{On-shell series}

The final result is found by replacing the off-shell edge variables $x_1, \ldots, x_I$ by the on-shell ones $x_1^\ast, \ldots, x_I^\ast$ determined by the on-shell condition~\eqref{eq: on-shell condition} in the ansatz~\eqref{eq: off-shell ansatz}. The found series solution exhibits a universal structure.

\begin{proposition}[On-shell evaluations]
\label{prop: on-shell evaluations}
    Let $\bm{x}^\ast = -P_r (\tilde{\bm p}_{(r)}+\bm{m})$ be the on-shell solution~\eqref{eq: on-shell condition}. Then at every non-root vertex $i\neq r$:
    \begin{equation}
        \tilde{p}_i + \xi_i\big|_{\bm{x}^\ast} = -m_{\pi(i)} \,,
    \end{equation}
    and at the root vertex $r$:
    \begin{equation}
        \tilde{p}_r + \xi_r\big|_{\bm{x}^\ast} = \tilde{p}_{1\cdots V} + |\bm{m}| \,,
    \end{equation}
    where $\tilde{p}_{1\cdots V} \equiv \sum_{i=1}^V \tilde{p}_i$ and $|\bm{m}| \equiv \sum_{j=1}^I m_j$.
\end{proposition}
\noindent Proposition~\ref{prop: on-shell evaluations} has immediate consequences for the structure of the ansatz~\eqref{eq: off-shell ansatz} under the substituation $\bm{x}=\bm{x}^\ast$. At each non-root vertex $i\neq r$:
\begin{equation}
    \Gamma(\tilde{p}_i+\xi_i)\big|_{\bm{x}^\ast} = \Gamma(-m_{\pi(i)}) = \frac{(-1)^{m_{\pi(i)}}}{m_{\pi(i)}!} \,.
\end{equation}
Moreover, the lower parameter of $F_C^{(n_i)}$ at vertex $i$ corresponding to the parent edge becomes $1-m_{\pi(i)}$. Since $-m_{\pi(i)}/2$ is a non-positive half-integer, one of the sums defining the Lauricella function truncates to a polynomial (see App.~\ref{app: Lauricella Function}). At the root vertex $r$:
\begin{equation}
    \Gamma(\tilde{p}_r+\xi_r)\big|_{\bm{x}^\ast} = \Gamma(\tilde{p}_{1\cdots V}+|\bm{m}|) \,.
\end{equation}
The Lauricella function $F_C^{(n_r)}$ retains its full hypergeometric structure. Eventually, summing over $\bm{m}\in \mathbb{N}_{\geq0}^I$ yields the {\it on-shell series} for root $r$.

\begin{proposition}[Simplification of kinematic ratios]\label{prop: simplification kinemtic ratios}
Let $\bm{x}^\ast = -P_r (\tilde{\bm p}_{(r)}+\bm{m})$ be the on-shell solution~\eqref{eq: on-shell condition} and $\bm \epsilon$ be the set of incidence signs. Then the kinematic factor in the off-shell ansatz simplifies to
\begin{equation}
    \prod_{(i,j)} (u_{ij})^{\epsilon_{ij}x_{ij}}(u_{ji})^{\epsilon_{ji}x_{ji}} = \prod_{\substack{i=1\\i\neq r}}^V \left(\frac{X_i}{X_r}\right)^{\tilde p_i+m_i}\,.
\end{equation}
\end{proposition}
\begin{proof}
    To see why this holds, let us first rewrite the product on the left-hand side that we used for the off-shell ansatz in a slightly more systematic way. For each edge $j$, the signs of the incidence matrix $\epsilon_{ij}$ and $\epsilon_{ji}$ of the two vertices attached to $(i,j)$ are opposites of each other and furthermore $x_{ij}=x_{ji} =i\nu_{ij}$ by definition. Therefore, $u_{ij}^{\epsilon_{ij}x_{ij}}u_{ji}^{\epsilon_{ji}x_{ji}} = (u_{ij}/u_{ji})^{\epsilon_{ij}x_{ij}} = (X_{j}/X_{i})^{\epsilon_{ij}x_{ij}}$. This ratio does no longer depend on the exchanged momentum and hence the whole expression will be a function of vertex energies only. Now we see that the vertex energy of a given vertex $i\neq r$ appears with a total power $-\sum_{\ell=1}^{n_i} \epsilon_{i\ell} x_{i\ell} = -\xi_i|_{\bm x^*}$, where the sum runs over all the edges attached to the vertex $i$ of degree $n_i$. But this sum is simply the $i$-th entry of the expression $-Q_r\bm x$, because the incidence matrix $Q_r$ has the $\epsilon_{i\ell}$ as entries. Therefore, from the defining equation $Q_r\bm x = -(\tilde{\bm p}_{(r)}+\bm m)$ we can immediately conclude that the power of $X_i$ in the full expression is simply $\tilde p_i +m_i$ (in agreement with Proposition \ref{prop: on-shell evaluations}). It only remains to find the respective power of $X_r$ for the root vertex $r$. But this is also straightforward: Since by definition of the root $\epsilon_{r\ell}=+1$ for all edges $\ell$ connected to $r$, the power of $X_r$ is given by $-\sum_{l=1}^{n_r} x_{r\ell} =-\xi_r|_{x^*}$. Using Proposition \ref{prop: on-shell evaluations} again, we obtain that this is equal to $-\tilde p_{1\cdots V}-|\bm m|+\tilde p_r$ and the right-hand side of the equation follows immediately. 
\end{proof}

\subsubsection{Maximally-nested analytic algorithm}

For the reader's convenience, we summarise the steps to construct the all-$(+)$ maximally-nested analytic contribution to any massive tree graph in the following box:

\begin{tcolorbox}
\label{box: MNA algo}
\begin{enumerate}
    \item {\bf Rooted tree.} From a tree graph $\G$, choose a root vertex $r$, construct the rooted decoration $\G_r$, and determine the corresponding incidence matrix $Q(\G_r)$;
    \item {\bf Off-shell ansatz.} Write the off-shell ansatz~\eqref{eq: off-shell ansatz} in terms of off-shell variables $\bm{x}$;
    \item {\bf On-shell condition.} Determine the path matrix $P_r(\G_r)$, which uniquely fixes the on-shell condition $\bm{x}^\ast$ given by~\eqref{eq: on-shell condition};
    \item {\bf Kinematic region.} Set $\bm{x}=\bm{x}^\ast$ in the off-shell ansatz and sum over all integers. This defines a solution that converges in the kinematic region determined by the oriented decoration;
\end{enumerate}
\end{tcolorbox}

A few comments are in order:

\paragraph{Kinematic region.} The obtained series converges absolutely in a specific kinematic region determined by the kinematic variable ratios. In the final sum, the factors $(X_i/X_r)^{\tilde p_i+m_i}$ will decay for $m_i\to \infty $ if $X_i <X_r$. To assess the radius of convergence, we can consider the soft limit of the internal energies $Y_{ij}\to 0$ in which all the hypergeometric functions trivialize to unity and the sum becomes a geometric type series that converges for $\sum_{i\neq r} X_i < X_r$. In particular, the choice of root makes it clear that we are expanding in terms of ratios of the largest energy $X_r$ which is chosen to be the root energy. Importantly, the internal energies play no role for the convergence of the series as the dependence on them is already resummed in terms of the hypergeometric functions.

\paragraph{Uniform transcendental weight.} The obtained series have manifestly uniform transcendental weight, {\it independently of the choice of root $r$}. Indeed, each child Lauricella function has a non-positive half-integer upper parameter, reducing the transcendental weight by unity (see App.~\ref{app: Lauricella Function}). The root Lauricella function $F_C^{(n_r)}$ contributes weight $n_r$. The $I$-fold sum contributes weight $I$. The total weight is
\begin{equation}
    n_r + \sum_{i\neq r} (n_i-1) + I = \sum_{i=1}^V n_i \,,
\end{equation}
independent of $r$ and equal to the transcendental weight of the fully-factorised diagram. This confirms that massive cosmological tree graphs have uniform transcendental weight.

\subsubsection{Example: five-site graph}

 As is often the case, an algorithm is more effectively conveyed through a concrete example than through purely abstract description. Let us illustrate the maximally-nested spectral gluing algorithm for a four-site chain with an additional leaf vertex attached to a central vertex:
\begin{center}
\begin{tikzpicture}
    [line width=1. pt, scale=2, baseline={(current bounding box.center)}]
    
    \draw[black] (0, 0) -- node[above=1mm] {$\textcolor{pyblue}{x_{12}}$} (0.7, 0);
    \draw[black] (0.7, 0) -- node[above=1mm] {$\textcolor{pyblue}{x_{23}}$} (1.4, 0);
    \draw[black] (1.4, 0) -- node[above=1mm] {$\textcolor{pyblue}{x_{34}}$} (2.1, 0);
    \draw[black] (1.4, 0) -- node[left=1mm] {$\textcolor{pyblue}{x_{35}}$} (1.4, 0.7);
        
    \draw[fill=black] (0, 0) circle (.05cm) node[below=2mm] {$\textcolor{black}{(1)}$};
    \draw[fill=black] (0.7, 0) circle (.05cm) node[below=2mm] {$\textcolor{black}{(2)}$};
    \draw[fill=black] (1.4, 0) circle (.05cm) node[below=2mm] {$\textcolor{black}{(3)}$};
    \draw[fill=black] (2.1, 0) circle (.05cm) node[below=2mm] {$\textcolor{black}{(4)}$};
    \draw[fill=black] (1.4, 0.7) circle (.05cm) node[above=2mm] {$\textcolor{black}{(5)}$};
\end{tikzpicture}
\end{center}
This graph topology is the simplest one beyond simple chain or star graphs. We have labelled the five vertices $1, \ldots, 5$, and the four edges with off-shell variables $x_{12}, x_{23}, x_{34}$ and $x_{35}$. We now choose vertex $(2)$ to be the root. The incidence matrix of the corresponding unique rooted decoration is given by:
\begin{equation}
\vcenter{\hbox{
\begin{tikzpicture}
[line width=1. pt, scale=2]
    [line width=1. pt, scale=2, baseline={(current bounding box.center)}]
    
    \draw[black, postaction={decorate}, decoration={markings, mark=at position 0.5 with {\arrow[pyred]{<}}}] (0, 0) -- node[above=1mm] {$\textcolor{pyblue}{x_{12}}$} (0.7, 0);
    \draw[black, postaction={decorate}, decoration={markings, mark=at position 0.5 with {\arrow[pyred]{>}}}] (0.7, 0) -- node[above=1mm] {$\textcolor{pyblue}{x_{23}}$} (1.4, 0);
    \draw[black, postaction={decorate}, decoration={markings, mark=at position 0.5 with {\arrow[pyred]{>}}}] (1.4, 0) -- node[above=1mm] {$\textcolor{pyblue}{x_{34}}$} (2.1, 0);
    \draw[black, postaction={decorate}, decoration={markings, mark=at position 0.5 with {\arrow[pyred]{>}}}] (1.4, 0) -- node[left=1mm] {$\textcolor{pyblue}{x_{35}}$} (1.4, 0.7);
        
    \draw[fill=black] (0, 0) circle (.05cm) node[below=2mm] {$\textcolor{black}{(1)}$};
    \draw[pyred, fill=pyred] (0.7, 0) circle (.05cm) node[below=2mm] {$\textcolor{black}{(2)}$};
    \draw[fill=black] (1.4, 0) circle (.05cm) node[below=2mm] {$\textcolor{black}{(3)}$};
    \draw[fill=black] (2.1, 0) circle (.05cm) node[below=2mm] {$\textcolor{black}{(4)}$};
    \draw[fill=black] (1.4, 0.7) circle (.05cm) node[above=2mm] {$\textcolor{black}{(5)}$};
\end{tikzpicture} 
    }}\,, \quad Q = 
    \begin{bNiceMatrix}[first-row,code-for-first-row=\scriptstyle, first-col,code-for-first-col=\scriptstyle,]
        & \textcolor{gray}{x_{12}} & \textcolor{gray}{x_{23}} & \textcolor{gray}{x_{34}} & \textcolor{gray}{x_{35}} \\
        \textcolor{gray}{(1)} & +1 & 0 & 0 & 0 \\
        \textcolor{gray}{(2)} & -1 & -1 & 0 & 0 \\
        \textcolor{gray}{(3)} & 0 & +1 & -1 & -1 \\
        \textcolor{gray}{(4)} & 0 & 0 & +1 & 0 \\
        \textcolor{gray}{(5)} & 0 & 0 & 0 & +1 
    \end{bNiceMatrix}\,.
\end{equation}
It is straightforward to verify that the sum of the entries in every column vanishes, or equivalently that the row vectors add up to zero. From the decorated graph, the on-shell ansatz for the analytic master integral reads
\begin{equation}
    \begin{aligned}
        \widehat{\I}^\P&(\begin{tikzpicture}[scale=0.8,       baseline=-0.2ex]
            \coordinate (A) at (0,0);
            \coordinate (B) at (1/2,0);
            \coordinate (C) at (1,0);
            \coordinate (D) at (3/2,0);
            \coordinate (E) at (1,1/2);
            \draw[thick, postaction={decorate}, decoration={markings, mark=at position 0.6 with {\arrow[pyred]{<}}}] (A) -- (B);
            \draw[thick, postaction={decorate}, decoration={markings, mark=at position 0.6 with {\arrow[pyred]{>}}}] (B) -- (C);
            \draw[thick, postaction={decorate}, decoration={markings, mark=at position 0.6 with {\arrow[pyred]{>}}}] (C) -- (D);
            \draw[thick, postaction={decorate}, decoration={markings, mark=at position 0.6 with {\arrow[pyred]{>}}}] (C) -- (E);
            \fill[black] (A) circle (2pt);
            \fill[pyred] (B) circle (2pt);
            \fill[black] (C) circle (2pt);
            \fill[black] (D) circle (2pt);
            \fill[black] (E) circle (2pt);
        \end{tikzpicture})
        = C(\begin{tikzpicture}[scale=0.8,       baseline=-0.2ex]
            \coordinate (A) at (0,0);
            \coordinate (B) at (1/2,0);
            \coordinate (C) at (1,0);
            \coordinate (D) at (3/2,0);
            \coordinate (E) at (1,1/2);
            \draw[thick] (A) -- (B);
            \draw[thick] (B) -- (C);
            \draw[thick] (C) -- (D);
            \draw[thick] (C) -- (E);
            \fill[black] (A) circle (2pt);
            \fill[black] (B) circle (2pt);
            \fill[black] (C) circle (2pt);
            \fill[black] (D) circle (2pt);
            \fill[black] (E) circle (2pt);
        \end{tikzpicture}) \\
        &\times \Gamma[\tilde{p}_1+x_{12}, \tilde{p}_2-x_{12}-x_{23}, \tilde{p}_3+x_{23}-x_{34}-x_{35}, \tilde{p}_4+x_{34}, \tilde{p}_5+x_{35}] \\
        &\times \frac{\left(u_{12}/u_{21}\right)^{x_{12}}}{x_{12}^2+\mu_{12}^2} \times  \frac{\left(u_{32}/u_{23}\right)^{x_{23}}}{x_{23}^2+\mu_{23}^2} \times \frac{\left(u_{43}/u_{34}\right)^{x_{34}}}{x_{34}^2+\mu_{34}^2} \times \frac{\left(u_{53}/u_{35}\right)^{x_{35}}}{x_{35}^2+\mu_{35}^2}  \\
        &\times F^{(1)}_C\left[\left.\begin{matrix} \frac{\tilde{p}_1+x_{12}}{2},\,  \frac{\tilde{p}_1+1+x_{12}}{2}\\ 1+x_{12} \end{matrix}\right\vert u_{12}^2\right] \, F^{(2)}_C\left[\left.\begin{matrix} \frac{\tilde{p}_2-x_{12}-x_{23}}{2},\,  \frac{\tilde{p}_2+1-x_{12}-x_{23}}{2}\\ 1-x_{12}, 1-x_{23} \end{matrix}\right\vert u_{21}^2, u_{23}^2\right] \\
        &\times F^{(3)}_C\left[\left.\begin{matrix} \frac{\tilde{p}_3+x_{23}-x_{34}-x_{35}}{2},\,  \frac{\tilde{p}_3+1+x_{23}-x_{34}-x_{35}}{2}\\ 1+x_{23}, 1-x_{34}, 1-x_{35}\end{matrix}\right\vert u_{32}^2, u_{34}^2, u_{35}^2\right] \\
        &\times F^{(1)}_C\left[\left.\begin{matrix} \frac{\tilde{p}_4+x_{34}}{2},\,  \frac{\tilde{p}_4+1+x_{34}}{2}\\ 1+x_{34} \end{matrix}\right\vert u_{43}^2\right] \, F^{(1)}_C\left[\left.\begin{matrix} \frac{\tilde{p}_5+x_{35}}{2},\,  \frac{\tilde{p}_5+1+x_{35}}{2}\\ 1+x_{35} \end{matrix}\right\vert u_{53}^2\right] \,,
    \end{aligned}
\end{equation}
where we have collected $\Gamma$-factors in the second line, kinematic factors and propagators in the third line, and Lauricella functions in the remaining lines. The various $\pm$ signs in the above expressions are given by the fixed oriented decoration. The overall coefficient is independent of the decoration and is given by
\begin{equation}
    C(\begin{tikzpicture}[scale=0.8,       baseline=-0.2ex]
            \coordinate (A) at (0,0);
            \coordinate (B) at (1/2,0);
            \coordinate (C) at (1,0);
            \coordinate (D) at (3/2,0);
            \coordinate (E) at (1,1/2);
            \draw[thick] (A) -- (B);
            \draw[thick] (B) -- (C);
            \draw[thick] (C) -- (D);
            \draw[thick] (C) -- (E);
            \fill[black] (A) circle (2pt);
            \fill[black] (B) circle (2pt);
            \fill[black] (C) circle (2pt);
            \fill[black] (D) circle (2pt);
            \fill[black] (E) circle (2pt);
        \end{tikzpicture}) = (-4\pi e^{-\frac{i\pi}{2}})^4 \, \C_+^{(1)}(p_1) \, \C_+^{(2)}(p_2) \, \C_+^{(3)}(p_3) \, \C_+^{(1)}(p_4) \C_+^{(1)}(p_5) \,.
\end{equation}
We now need to determine the on-shell condition for the off-shell variables $x_{12}, x_{23}, x_{34}$ and $x_{35}$. The choice $r=2$ for the root vertex uniquely fixes the corresponding reduced incidence matrix and its inverse, i.e.~the path matrix:
\begin{equation}
    Q_2 = 
    \begin{bNiceMatrix}[first-row,code-for-first-row=\scriptstyle, first-col,code-for-first-col=\scriptstyle,]
        & \textcolor{gray}{x_{12}} & \textcolor{gray}{x_{23}} & \textcolor{gray}{x_{34}} & \textcolor{gray}{x_{35}} \\
        \textcolor{gray}{(1)} & +1 & 0 & 0 & 0 \\
        \textcolor{gray}{(3)} & 0 & +1 & -1 & -1 \\
        \textcolor{gray}{(4)} & 0 & 0 & +1 & 0 \\
        \textcolor{gray}{(5)} & 0 & 0 & 0 & +1 
    \end{bNiceMatrix}\,, \quad 
    P_2 = 
    \begin{bNiceMatrix}[first-row,code-for-first-row=\scriptstyle, first-col,code-for-first-col=\scriptstyle,]
        & \textcolor{gray}{(1)} & \textcolor{gray}{(3)} & \textcolor{gray}{(4)} & \textcolor{gray}{(5)} \\
        \textcolor{gray}{x_{12}} & +1 & 0 & 0 & 0 \\
        \textcolor{gray}{x_{23}} & 0 & +1 & +1 & +1 \\
        \textcolor{gray}{x_{34}} & 0 & 0 & +1 & 0 \\
        \textcolor{gray}{x_{35}} & 0 & 0 & 0 & +1 
    \end{bNiceMatrix}\,.
\end{equation}
It is read off from the oriented graph as follows: for instance, the path from the root $r=2$ to vertex $4$ traverses $x_{23}$ and $x_{34}$, following the direction of the arrows. Therefore, the third column (associated to vertex 4 because the root has been removed) reads $(0, +1, +1, 0)$, and so on. From the fundamental property $P_2=Q_2^{-1}$, the on-shell condition is found by $\bm{x}^\ast = -P_2(\tilde{\bm{p}}_{(2)}+\bm{m})$. Explicitly, we obtain
\begin{equation}
    \left\{
    \begin{aligned}
        x_{12}^\ast &= -\tilde{p}_1-m_1 \,, \\
        x_{23}^\ast &= -\tilde{p}_{345}-m_{234} \,, \\
        x_{34}^\ast &= -\tilde{p}_4-m_3 \,, \\
        x_{35}^\ast &= -\tilde{p}_5-m_4 \,.
    \end{aligned}
    \right.
\end{equation}
Plugging back this solution in the original off-shell ansatz, and summing over all unconstrained integers yields
\begin{equation}
    \begin{aligned}
        \widehat{\I}^\P&(\begin{tikzpicture}[scale=0.8,       baseline=-0.2ex]
            \coordinate (A) at (0,0);
            \coordinate (B) at (1/2,0);
            \coordinate (C) at (1,0);
            \coordinate (D) at (3/2,0);
            \coordinate (E) at (1,1/2);
            \draw[thick, postaction={decorate}, decoration={markings, mark=at position 0.6 with {\arrow[pyred]{<}}}] (A) -- (B);
            \draw[thick, postaction={decorate}, decoration={markings, mark=at position 0.6 with {\arrow[pyred]{>}}}] (B) -- (C);
            \draw[thick, postaction={decorate}, decoration={markings, mark=at position 0.6 with {\arrow[pyred]{>}}}] (C) -- (D);
            \draw[thick, postaction={decorate}, decoration={markings, mark=at position 0.6 with {\arrow[pyred]{>}}}] (C) -- (E);
            \fill[black] (A) circle (2pt);
            \fill[pyred] (B) circle (2pt);
            \fill[black] (C) circle (2pt);
            \fill[black] (D) circle (2pt);
            \fill[black] (E) circle (2pt);
        \end{tikzpicture})
        = C(\begin{tikzpicture}[scale=0.8,       baseline=-0.2ex]
            \coordinate (A) at (0,0);
            \coordinate (B) at (1/2,0);
            \coordinate (C) at (1,0);
            \coordinate (D) at (3/2,0);
            \coordinate (E) at (1,1/2);
            \draw[thick] (A) -- (B);
            \draw[thick] (B) -- (C);
            \draw[thick] (C) -- (D);
            \draw[thick] (C) -- (E);
            \fill[black] (A) circle (2pt);
            \fill[black] (B) circle (2pt);
            \fill[black] (C) circle (2pt);
            \fill[black] (D) circle (2pt);
            \fill[black] (E) circle (2pt);
        \end{tikzpicture}) \sum_{\bm{m}\in\mathbb{N}_{\geq0}^4} \, \frac{(-1)^{|\bm{m}|}}{\bm{m}!} \, \Gamma[\tilde{p}_{12345}+|\bm{m}|] \, \frac{\left(X_1/X_2\right)^{\tilde{p}_1+m_1}}{(\tilde{p}_1+m_1)^2+\mu_{12}^2}\\
        &\times \frac{\left(X_3/X_2\right)^{\tilde{p}_{3}+m_{2}}}{(\tilde{p}_{345}+m_{234})^2+\mu_{23}^2} \times \frac{\left(X_4/X_2\right)^{\tilde{p}_4+m_3}}{(\tilde{p}_4+m_3)^2+\mu_{34}^2} \times \frac{\left(X_5/X_2\right)^{\tilde{p}_5+m_4}}{(\tilde{p}_5+m_4)^2+\mu_{35}^2}  \\
        &\times F^{(1)}_C\left[\left.\begin{matrix} \frac{-m_1}{2},\,  \frac{1-m_1}{2}\\ 1-\tilde{p}_1-m_1 \end{matrix}\right\vert u_{12}^2\right] \, F^{(2)}_C\left[\left.\begin{matrix} \frac{\tilde{p}_{12345}+|\bm{m}|}{2},\,  \frac{\tilde{p}_{12345}+1+|\bm{m}|}{2}\\ 1+\tilde{p}_1+m_1, 1+\tilde{p}_{345}+m_{234} \end{matrix}\right\vert u_{21}^2, u_{23}^2\right] \\
        &\times F^{(3)}_C\left[\left.\begin{matrix} \frac{-m_2}{2},\,  \frac{1-m_2}{2}\\ 1-\tilde{p}_{345}-m_{234}, 1+\tilde{p}_4+m_3, 1+\tilde{p}_5+m_4\end{matrix}\right\vert u_{32}^2, u_{34}^2, u_{35}^2\right] \\
        &\times F^{(1)}_C\left[\left.\begin{matrix} \frac{-m_3}{2},\,  \frac{1-m_3}{2}\\ 1-\tilde{p}_4-m_3 \end{matrix}\right\vert u_{43}^2\right] \, F^{(1)}_C\left[\left.\begin{matrix} \frac{-m_4}{2},\,  \frac{1-m_4}{2}\\ 1-\tilde{p}_5-m_4 \end{matrix}\right\vert u_{53}^2\right] \,.
    \end{aligned}
\end{equation}
Notice that the Lauricella function associated with the root vertex carries maximal transcendental weight, whereas all other Lauricella functions are either reduced or truncated to polynomials. The choice $r=2$ corresponds to the following kinematic region $\sum_{i\neq 2}X_i < X_2$.
Other choices of the root vertex lead to other series representations that converge in different kinematic regions. 

\subsection{On-shell poles}
\label{subsec: on-shell poles}

We now evaluate the contribution to the spectral gluing integral arising from the on-shell part of the spectral density, which localises the spectral integral at $\nu=\pm\mu$.

\subsubsection*{Frequency projectors}

Any meromorphic function $f(\nu)$ appearing inside a spectral integral inherits a natural decomposition into positive- and negative-frequency parts, dictated by the large-$|\nu|$ behaviour of the integrand. Concretely, given the vertex function decomposition~\eqref{eq: Vn definition short}, we define the positive- and negative-frequency projection of $f(\nu)$ as
\begin{equation}
    f(\nu) = f_+(\nu) + f_-(\nu) \,, \quad f_-(\nu) = f_+(-\nu) \,,
\end{equation}
where $f_+(\nu)$ (resp.~$f_-(\nu)$) collect all terms whose leading kinematic dependence goes as $(u_i/2)^{+i\nu}$ (resp.~$(u_i/2)^{-i\nu}$) for large $|\nu|$, assuming $1<u_i$, otherwise the contours should be closed in the opposite direction. The shadow symmetry $\nu\to-\nu$ exchanges the two sectors. We introduce the distributional projectors $\delta_\pm$ that extract each frequency component at the on-shell locus:
\begin{equation}
    \int_{-\infty}^{+\infty} \d\nu \, f(\nu) \, \delta_\pm(\nu-\mu) = f_\pm(\mu) \,,
\end{equation}
so that $\delta_+(\nu-\mu) + \delta_-(\nu-\mu) = \delta(\nu-\mu) = |\mu|\,\delta(\nu^2-\mu^2)$. Using the measure $[\d\nu]=\d\nu\,\N_\nu$ with $\N_\nu=1/\Gamma[\pm i\nu]$, the full on-shell spectral integral then evaluates to
\begin{equation}
\label{eq: on-shell pole integral}
    \begin{aligned}
        \int_{-\infty}^{+\infty} [\d\nu]\, \frac{f(\nu)}{(\nu^2-\mu^2)_{i\epsilon}} &= -i\sum_\pm e^{\pm\pi\mu} \int_{-\infty}^{+\infty} \d\nu \, f(\nu) \delta_\mp(\nu-\mu) \\
        &= -i \left[e^{+\pi\mu}\, f_-(\mu) + e^{-\pi\mu} \, f_+(\mu)\right] \,,
    \end{aligned}
\end{equation}
where we only collect the two terms corresponding to the two  on-shell poles $\nu=\pm\mu$, each dressed by the Boltzmann weight $e^{\pm \pi\mu}$ from the spectral density~\eqref{eq: iepsilon prescription}. Eq.~\eqref{eq: on-shell pole integral} is the master formula for evaluating on-shell contributions in the spectral gluing algorithm. Notice that, by construction, the resulting integration is manifestly shadow symmetric $\mu \leftrightarrow -\mu$.

\subsubsection*{On-shell gluing of adjacent vertices}

Consider the same setting as in Sec.~\ref{subsec: gluing adjacent vertices}: two adjacent $(+)$-vertices with vertex functions $\V^{(n_i)}_{+, (\nu, \bm{\mu}_{i,(j)})}(\bm{u}_i; p_i)$ and $\V^{(n_j)}_{+, (\nu, \bm{\mu}_{j,(i)})}(\bm{u}_j; p_j)$, glued through the common edge $(i,j)$ with mass parameter $\mu_{ij}$. The on-shell contribution to the spectral integral~\eqref{eq: gluing adjacent vertices integral} reads
\begin{equation}
    \widehat{\I}^{\delta_{ij}}(\bm{u}_i, \bm{u}_j) \equiv e^{-\frac{i\pi}{2}} \int\limits_{-\infty}^{+\infty} [\d\nu] \, \rho_{\nu, \delta} \, \V^{(n_i)}_{+, (\nu, \bm{\mu}_{i,(j)})}(\bm{u}_i; p_i) \V^{(n_j)}_{+, (\nu, \bm{\mu}_{j,(i)})}(\bm{u}_j; p_j) \,.
\end{equation}
Applying the projector identity~\eqref{eq: on-shell pole integral}, and using the frequency decomposition of each vertex function given in Sec.~\ref{subsec: gluing adjacent vertices}, we obtain
\begin{equation}
    \begin{aligned}
        &\widehat{\I}^{\delta_{ij}}(\bm{u}_i, \bm{u}_j) = -\C_+^{(n_i)}(p_i) \, \V^{(n_j)}_{+, (\mu_{ij}, \bm{\mu}_{j,(i)})}(\bm{u}_j; p_j) \\
        &\times\left[e^{+\pi\mu}\left(\frac{u_{ij}}{2}\right)^{+i\mu} \sum_{\bm{\alpha}=\pm i\bm{\mu}_{i,(j)}} \Gamma[-i\mu_{ij}, 1+i\mu_{ij}, \tilde{p}_i+i\mu_{ij}+|\bm{\alpha}|]\mathtt{F}_{+i\mu_{ij}, \bm\alpha}(\bm{u}_i; \tilde{p}_i) \right.\\
        &\left.+ e^{-\pi\mu}\left(\frac{u_{ij}}{2}\right)^{-i\mu} \sum_{\bm{\alpha}=\pm i\bm{\mu}_{i,(j)}} \Gamma[+i\mu_{ij}, 1-i\mu_{ij}, \tilde{p}_i-i\mu_{ij}+|\bm{\alpha}|]\mathtt{F}_{-i\mu_{ij}, \bm\alpha}(\bm{u}_i; \tilde{p}_i)\right] \,,
    \end{aligned}
\end{equation}
for $u_{ij}<u_{ji}$, and analogously for the opposite kinematic region $u_{ji}<u_{ij}$; see Sec.~\ref{app: two-site chain} for the detailed two-site chain case. The result has a transparent physical interpretation. Setting $\nu=\mu_{ij}$ on shell reconstructs the physical massive mode functions, and the two terms correspond to the emission of a particle ($\nu=+\mu$) with Boltzmann weight $e^{+\pi\mu}$ and its conjugate ($\nu=-\mu$) with Boltzmann weight $e^{-\pi\mu}$, therefore accounting for spontaneous particle production. Note also that, like the analytic piece of Sec.~\ref{subsec: gluing adjacent vertices}, the on-shell result depends on the kinematic region. 

\subsection{Spectral gluing algorithm}
\label{subsec: spectral gluing algorithm}

We now summarise the systematic procedure for computing any tree-level massive cosmological correlator, collecting the results of Secs.~\ref{subsec: gluing adjacent vertices},~\ref{subsec: maximally-nested contributions from graph combinatorics} and~\ref{subsec: on-shell poles} into a self-contained algorithm. Given a tree graph $\G$ with $V$ vertices and $I=V-1$ internal lines, and after assigning each vertex a Schwinger-Keldysh index $\aa_i=\pm$, the full graph decomposes into a sum of $2^V$ dimensionless master integrals 
\begin{equation}
    \widehat{\G}(\{X_i\}, \{Y_j\}) = \sum_{\aa_1, \ldots, \aa_V=\pm} (i\aa_1) \cdots (i\aa_V) \, \widehat{\I}_{\aa_1 \cdots \aa_V}(\{X_i\}, \{Y_j\}) \,.
\end{equation}
The number of independent master integrals is reduced to $2^{V-1}$ by the complex conjugation relation $\widehat{\I}_{-\aa_1 \cdots -\aa_V} = \widehat{\I}_{\aa_1 \cdots \aa_V}^*$. The spectral gluing algorithm proceeds as follows: 

\subsubsection*{Step 1: Factorisation}

For each master integral $\widehat{\I}_{\aa_1 \cdots \aa_V}$, examine every internal line $(i,j)$ connecting vertices $i$ and $j$:
\begin{itemize}
    \item {\bf Different-colour line} ($\aa_i\neq \aa_j$): the corresponding bulk-to-bulk propagator $\widehat{D}^{(\mu)}_{\pm\mp}$ is already factorised in time. No spectral integral is introduced. 
    \item {\bf Same-colour line} ($\aa_i = \aa_j$): introduce a spectral integral over the off-shell mass parameter $\nu_{ij}\in \mathbb{R}$, replacing the on-shell propagator by the split representation~\eqref{eq: split representation ++} or~\eqref{eq: split representation --}, thereby setting $\mu_{ij}\to\nu_{ij}$.
\end{itemize}
After applying this procedure to all same-colour internal lines, every time integration factorises completely. Each individual time integral is identified as a vertex function $\V_{\aa, \bm{\mu}}^{(n)}(\bm{u}; p)$ (see Sec.~\ref{sec: vertex functions}), where $n$ is the number of off-shell (spectral) or on-shell massive legs attached to the vertex, $p$ is the vertex twist, and $\bm{u} = \{Y_{ij}/X_i\}$ are the dimensionless kinematic ratios. Each vertex function is a finite linear combination of type-$C$ Lauricella functions $F_C^{(n)}$ in the $n$ kinematic variables, as given in~\eqref{eq: Vn definition short} and~\eqref{eq: general mode Fn function}. The master integral is then expressed as a product of vertex functions integrated over the spectral variables.

\subsubsection*{Step 2: Decomposition of the spectral density}

At tree level, each spectral density $\rho_{\nu_{ij}}(\mu_{ij})$ is decomposed into two contributions according to~\eqref{eq: iepsilon distribution}:
\begin{equation}
    \rho_{\nu_{ij}}(\mu_{ij}) = \underbrace{\delta\left(\frac{1}{\nu_{ij}^2-\mu_{ij}^2}\right)}_{\rm on-shell} + \underbrace{\P\left(\frac{1}{\nu_{ij}^2-\mu_{ij}^2}\right)}_{\rm analytic} \,.
\end{equation}
For a master integral involving $I_s$ spectral integrals, this decomposition generates $2^{I_s}$ contributions, according to whether each spectral integral is treated via its on-shell or analytic part. These contributions are labelled 
\begin{equation}
    \widehat{\I}_{\aa_1 \cdots \aa_V}^{\delta_{j_1} \cdots \delta_{j_m} \P_{k_1} \cdots \P_{k_n}} \,,
\end{equation}
where $j_1, \ldots, j_m$ denote lines treated on-shell and $k_1, \ldots, k_n$ those treated analytically ($m+n=s$). 

\subsubsection*{Step 3: Analytic \& on-shell gluing}

For each purely analytic contribution, evaluate the $n$ remaining spectral integrals using the maximally-nested analytic algorithm of Sec.~\ref{subsec: maximally-nested contributions from graph combinatorics}, applied to the subgraph induced by the same-colour vertices connected through the analytic lines. For each contribution involving at least one on-shell factor, evaluate the on-shell spectral integrals using the master formula~\eqref{eq: on-shell pole integral}. Any remaining analytic spectral integrals (lines $k_1, \ldots, k_n$) are subsequently evaluated using graph combinatorics, now applied to the reduced graph in which the on-shell lines have been restored to their physical mass. 

\vskip 4pt
The result is a finite sum of multi-fold series of products of Lauricella functions $F_C^{(n)}$, together with on-shell terms expressed directly as products of on-shell vertex functions. All contributions have uniform transcendental weight $\sum_i n_i$, independently of the Schwinger-Keldysh colouring and root choice. 

\newpage
\section{Selected Examples}
\label{sec: selected examples}

Building on the spectral gluing algorithm developed in the previous section, the complete solution for an arbitrary massive tree graph can be systematically constructed. In general, such solutions take the form of infinite, partially resummed series of products of multivariable hypergeometric functions of Lauricella type, whose structure is entirely determined by the topology of the graph. To illustrate this, we present several explicit examples: the two-site, three-site chains, the $N$-site chain, as well as the general $N$-site star graph, which we show in Fig.~\ref{fig: selected examples}. In these examples, we will focus on the maximally-nested analytic contributions, leaving all details in App.~\ref{sec: details on spectral gluing algorithm}.

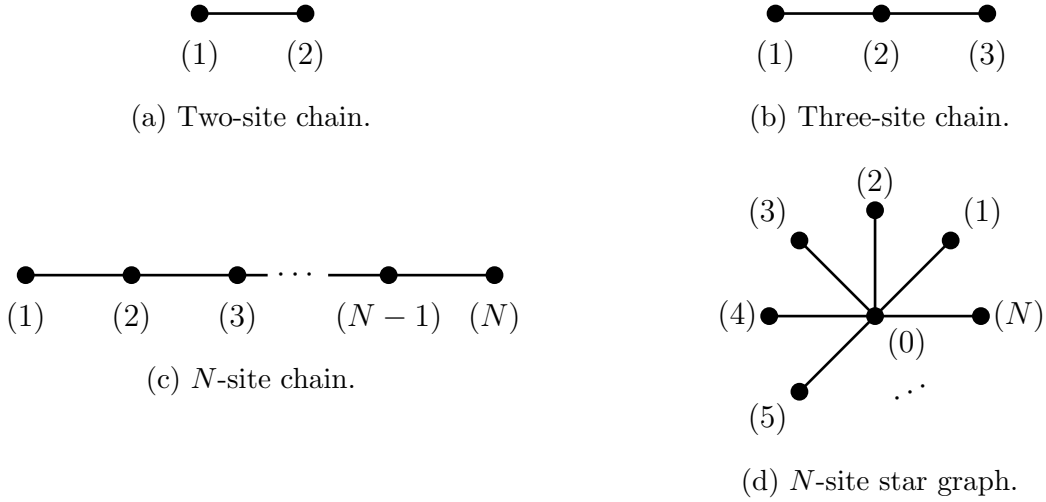
\begin{figure}[h!]
    \centering

    \begin{subfigure}[t]{0.45\textwidth}
        \centering
        \begin{tikzpicture}[line width=1pt, scale=2,
            baseline={(current bounding box.center)}]

            \draw[black] (0,0) -- (0.7,0);

            \draw[fill=black] (0,0) circle (.05cm) node[below=2mm] {$\textcolor{black}{(1)}$};
            \draw[fill=black] (0.7,0) circle (.05cm) node[below=2mm] {$\textcolor{black}{(2)}$};
        \end{tikzpicture}
        \caption{Two-site chain.}
        \label{fig:2site}
    \end{subfigure}
    \hfill
    \begin{subfigure}[t]{0.45\textwidth}
        \centering
        \begin{tikzpicture}[line width=1pt, scale=2,
            baseline={(current bounding box.center)}]

            \draw[black] (0,0) -- (0.7,0);
            \draw[black] (0.7,0) -- (1.4,0);

            \draw[fill=black] (0,0) circle (.05cm) node[below=2mm] {$\textcolor{black}{(1)}$};
            \draw[fill=black] (0.7,0) circle (.05cm) node[below=2mm] {$\textcolor{black}{(2)}$};
            \draw[fill=black] (1.4,0) circle (.05cm) node[below=2mm] {$\textcolor{black}{(3)}$};
        \end{tikzpicture}
        \caption{Three-site chain.}
        \label{fig:3site}
    \end{subfigure}

    \vspace{0.2cm}

    \begin{subfigure}[t]{0.45\textwidth}
        \centering
        \begin{tikzpicture}[line width=1pt, scale=2,
            baseline={(current bounding box.center)}]

            \draw[black] (0,0) -- (0.7,0);
            \draw[black] (0.7,0) -- (1.4,0);
            \draw[black] (1.4,0) -- (1.6,0);
            \draw[black] (2,0) -- (2.4,0);
            \draw[black] (2.4,0) -- (3.1,0);

            \node[black] at ($(1.6, 0)!1/2!(2, 0)$) {$\cdots$};

            \draw[fill=black] (0,0) circle (.05cm) node[below=2mm] {$\textcolor{black}{(1)}$};
            \draw[fill=black] (0.7,0) circle (.05cm) node[below=2mm] {$\textcolor{black}{(2)}$};
            \draw[fill=black] (1.4,0) circle (.05cm) node[below=2mm] {$\textcolor{black}{(3)}$};
            \draw[fill=black] (2.4,0) circle (.05cm) node[below=2mm] {$\textcolor{black}{(N-1)}$};
            \draw[fill=black] (3.1,0) circle (.05cm) node[below=2mm] {$\textcolor{black}{(N)}$};
        \end{tikzpicture}
        \caption{$N$-site chain.}
        \label{fig:Nsitechain}
    \end{subfigure}
    \hfill
    \begin{subfigure}[t]{0.45\textwidth}
        \centering
        \begin{tikzpicture}[line width=1pt, scale=2,
            baseline={(current bounding box.center)}]

            \draw[black] (0, 0) -- (-0.5, -0.5) node[below left] {$(5)$};
            \draw[black] (0, 0) -- (-0.7, 0) node[left] {$(4)$};
            \draw[black] (0, 0) -- (-0.5, 0.5) node[above left] {$(3)$};
            \draw[black] (0, 0) -- (0, 0.7) node[above] {$(2)$};
            \draw[black] (0, 0) -- (0.5, 0.5) node[above right] {$(1)$};
            \draw[black] (0, 0) -- (0.7, 0) node[right] {$(N)$};

            \node[black, rotate=25] at ($(-0.5, -0.5)!1/2!(1, -0.5)$) {$\cdots$};
    
            \draw[fill=black] (0, 0) circle (.05cm) node[below right] {$(0)$};
            \draw[fill=black] (-0.5, -0.5) circle (.05cm);
            \draw[fill=black] (-0.7, 0) circle (.05cm);
            \draw[fill=black] (-0.5, 0.5) circle (.05cm);
            \draw[fill=black] (0, 0.7) circle (.05cm);
            \draw[fill=black] (0.5, 0.5) circle (.05cm);
            \draw[fill=black] (0.7, 0) circle (.05cm);
        \end{tikzpicture}
        \caption{$N$-site star graph.}
        \label{fig:Nsitestar}
    \end{subfigure}

    \caption{Representative maximally-nested analytic cosmological tree graphs considered in Sec.~\ref{sec: selected examples}, illustrating the vertex labelling convention.}
    \label{fig: selected examples}
\end{figure}

\subsection{Two-site chain}
\label{subsec: two-site chain}

We first consider the simplest exchange graph: the two-site chain, which we treat in considerable details in App.~\ref{app: two-site chain}. Here, we only write the maximally-nested analytic contribution. Without loss of generality, we consider the kinematic region $u_{12}<u_{21}$, which corresponds to the following choice of root vertex:
\begin{equation}
\vcenter{\hbox{
\begin{tikzpicture}[line width=1pt, scale=2, baseline={(current bounding box.center)}]
    \draw[black, postaction={decorate}, decoration={markings, mark=at position 0.5 with {\arrow[pyred]{>}}}] (0,0) -- node[above=1mm] {$\textcolor{pyblue}{x_{12}}$} (0.7,0);

    \draw[pyred, fill=pyred] (0,0) circle (.05cm) node[below=2mm] {$\textcolor{black}{(1)}$};
    \draw[fill=black] (0.7,0) circle (.05cm) node[below=2mm] {$\textcolor{black}{(2)}$};
\end{tikzpicture}
    }}\,, \quad Q = 
    \begin{bNiceMatrix}[first-row,code-for-first-row=\scriptstyle, first-col,code-for-first-col=\scriptstyle,]
        & \textcolor{gray}{x_{12}} \\
        \textcolor{gray}{(1)} & -1 \\
        \textcolor{gray}{(2)} & +1 
    \end{bNiceMatrix}\,.
\end{equation}
After deleting the root vertex row $r=1$, the reduced incidence matrix and its inverse, i.e.~the path matrix, are trivial:
\begin{equation}
    Q_1 = 
    \begin{bmatrix}
        +1
    \end{bmatrix}\,, \quad 
    P_1 = 
    \begin{bmatrix}
        +1
    \end{bmatrix}\,.
\end{equation}
Writing the off-shell ansatz and setting the internal mass parameter on shell leads to the following series solution:
\begin{equation}
    \begin{aligned}
    \widehat{\I}^\P(\begin{tikzpicture}[scale=0.8, baseline=-0.5ex]
            \coordinate (A) at (0,0);
            \coordinate (B) at (0.6,0);
            \draw[thick, postaction={decorate}, decoration={markings, mark=at position 0.6 with {\arrow[pyred]{>}}}] (A) -- (B);
            \fill[pyred] (A) circle (2pt);
            \fill[black] (B) circle (2pt);
        \end{tikzpicture}) &= C(\begin{tikzpicture}[scale=0.8, baseline=-0.5ex]
            \coordinate (A) at (0,0);
            \coordinate (B) at (0.6,0);
            \draw[black, thick] (A) -- (B);
            \fill[black] (A) circle (2pt);
            \fill[black] (B) circle (2pt);
        \end{tikzpicture}) \sum_{m=0}^{\infty} \frac{(-1)^m}{m!} \frac{\Gamma[\tilde{p}_{12}+m]}{(\tilde{p}_2+m)^2+\mu^2}  \\
        &\times \left(\frac{X_2}{X_1}\right)^{\tilde{p}_2+m} F^{(1)}_C\left[\left.\begin{matrix} \frac{\tilde{p}_{12}+m}{2},\,  \frac{\tilde{p}_{12}+1+m}{2}\\ 1+\tilde{p}_2+m \end{matrix}\right\vert u_{12}^2\right] F^{(1)}_C\left[\left.\begin{matrix} \frac{-m}{2},\,  \frac{1-m}{2}\\ 1-\tilde{p}_2-m \end{matrix}\right\vert u_{21}^2\right] \,,
    \end{aligned}
\end{equation}
where $\tilde{p}_{1,2}=p_{1, 2}-d/2$ and $\tilde{p}_{12}=\tilde{p}_1+\tilde{p}_2$. The overall coefficient is root independent, and is given by
\begin{equation}
    C(\begin{tikzpicture}[scale=0.8, baseline=-0.5ex]
            \coordinate (A) at (0,0);
            \coordinate (B) at (0.6,0);
            \draw[black, thick] (A) -- (B);
            \fill[black] (A) circle (2pt);
            \fill[black] (B) circle (2pt);
        \end{tikzpicture}) = -4\pi e^{-\frac{i\pi}{2}} \C_+^{(1)}(p_1) \C_+^{(1)}(p_2) \,.
\end{equation}
The complementary region $X_1<X_2$ is found by swapping the kinematic variables $X_1 \leftrightarrow X_2$.

\subsection{Three-site chain}

Next, we turn our attention to the second-simplest exchange graph: the three-site chain. Again, we only consider the maximally-nested analytic contribution, leaving details in App.~\ref{app: three-site chain}. First, we consider the case for which the root vertex is the edge vertex $r=1$:
\begin{equation}
\vcenter{\hbox{
\begin{tikzpicture}[line width=1pt, scale=2, baseline={(current bounding box.center)}]
    \draw[black, postaction={decorate}, decoration={markings, mark=at position 0.5 with {\arrow[pyred]{>}}}] (0,0) -- node[above=1mm] {$\textcolor{pyblue}{x_{12}}$} (0.7,0);
    \draw[black, postaction={decorate}, decoration={markings, mark=at position 0.5 with {\arrow[pyred]{>}}}] (0.7,0) -- node[above=1mm] {$\textcolor{pyblue}{x_{23}}$} (1.4,0);

    \draw[pyred, fill=pyred] (0,0) circle (.05cm) node[below=2mm] {$\textcolor{black}{(1)}$};
    \draw[fill=black] (0.7,0) circle (.05cm) node[below=2mm] {$\textcolor{black}{(2)}$};
    \draw[fill=black] (1.4,0) circle (.05cm) node[below=2mm] {$\textcolor{black}{(3)}$};
\end{tikzpicture}
    }}\,, \quad Q = 
    \begin{bNiceMatrix}[first-row,code-for-first-row=\scriptstyle, first-col,code-for-first-col=\scriptstyle,]
        & \textcolor{gray}{x_{12}} & \textcolor{gray}{x_{23}} \\
        \textcolor{gray}{(1)} & -1 & 0 \\
        \textcolor{gray}{(2)} & +1 & -1 \\
        \textcolor{gray}{(3)} & 0 & +1
    \end{bNiceMatrix}\,.
\end{equation}
The reduced incidence and path matrices are
\begin{equation}
    Q_1 = 
    \begin{bmatrix}
        +1 & -1 \\
        0 & +1
    \end{bmatrix}\,, \quad 
    P_1 = 
    \begin{bmatrix}
        +1 & +1 \\
        0 & +1
    \end{bmatrix}\,.
\end{equation}
The off-shell vertex variables can be read off from the incidence matrix, and read:
\begin{equation}
    \xi_1 = -x_{12} \,, \quad \xi_2 = +x_{12}-x_{23} \,, \quad \xi_3 = +x_{23} \,.
\end{equation}
The on-shell condition is given by the path matrix, $\bm{x}^*=-P_1(\tilde{\bm{p}}_{(1)}+\bm{m})$:
\begin{equation}
    x_{12}^* = -(\tilde{p}_{23}+m_{12}) \,, \quad x_{23}^* = -(\tilde{p}_3+m_2) \,.
\end{equation}
This on-shell solution gives the following series solution:
\begin{equation}
    \begin{aligned}
    \widehat{\I}^\P(\begin{tikzpicture}[scale=0.8, baseline=-0.5ex]
            \coordinate (A) at (0,0);
            \coordinate (B) at (0.6,0);
            \coordinate (C) at (1.2,0);
            \draw[thick, postaction={decorate}, decoration={markings, mark=at position 0.6 with {\arrow[pyred]{>}}}] (A) -- (B);
            \draw[thick, postaction={decorate}, decoration={markings, mark=at position 0.6 with {\arrow[pyred]{>}}}] (B) -- (C);
            \fill[pyred] (A) circle (2pt);
            \fill[black] (B) circle (2pt);
            \fill[black] (C) circle (2pt);
        \end{tikzpicture}) &= C(\begin{tikzpicture}[scale=0.8, baseline=-0.5ex]
            \coordinate (A) at (0,0);
            \coordinate (B) at (0.6,0);
            \coordinate (C) at (1.2,0);
            \draw[black, thick] (A) -- (B);
            \draw[black, thick] (B) -- (C);
            \fill[black] (A) circle (2pt);
            \fill[black] (B) circle (2pt);
            \fill[black] (C) circle (2pt);
        \end{tikzpicture}) \sum_{m_1, m_2=0}^{\infty} \frac{(-1)^{m_{12}}}{m_1!m_2!}\frac{\Gamma[\tilde{p}_{123}+m_{12}]}{[(\tilde{p}_{23}+m_{12})^2+\mu_{12}^2][(\tilde{p}_3+m_2)^2+\mu_{23}^2]}  \\
        &\times \left(\frac{X_2}{X_1}\right)^{\tilde{p}_{2}+m_{1}} \left(\frac{X_3}{X_1}\right)^{\tilde{p}_3+m_2} F^{(1)}_C\left[\left.\begin{matrix} \frac{\tilde{p}_{123}+m_{12}}{2},\,  \frac{\tilde{p}_{123}+1+m_{12}}{2}\\ 1+\tilde{p}_{23}+m_{12} \end{matrix}\right\vert u_{12}^2\right]\\
        &\times F^{(2)}_C\left[\left.\begin{matrix} \frac{-m_1}{2},\,  \frac{1-m_1}{2}\\ 1-\tilde{p}_{23}-m_{12}, 1+\tilde{p}_3+m_2 \end{matrix}\right\vert u_{21}^2, u_{23}^2\right] F^{(1)}_C\left[\left.\begin{matrix} \frac{-m_2}{2},\,  \frac{1-m_2}{2}\\ 1-\tilde{p}_3-m_2 \end{matrix}\right\vert u_{32}^2\right]\,,
    \end{aligned}
\end{equation}
where 
\begin{equation}
    C(\begin{tikzpicture}[scale=0.8, baseline=-0.5ex]
            \coordinate (A) at (0,0);
            \coordinate (B) at (0.6,0);
            \coordinate (C) at (1.2,0);
            \draw[black, thick] (A) -- (B);
            \draw[black, thick] (B) -- (C);
            \fill[black] (A) circle (2pt);
            \fill[black] (B) circle (2pt);
            \fill[black] (C) circle (2pt);
        \end{tikzpicture}) = \left(-4\pi e^{-\frac{i\pi}{2}}\right)^2 \C_+^{(1)}(p_1) \C_+^{(2)}(p_2) \C_+^{(1)}(p_3) \,.
\end{equation}
This solution converges in the region $X_2+X_3< X_1$. The case $r=3$, considering the other edge vertex to be the root, can be recovered after relabelling the vertices $1\leftrightarrow3$. Now we consider the case for which the root vertex is the central one $r=2$:
\begin{equation}
\vcenter{\hbox{
\begin{tikzpicture}[line width=1pt, scale=2, baseline={(current bounding box.center)}]
    \draw[black, postaction={decorate}, decoration={markings, mark=at position 0.5 with {\arrow[pyred]{<}}}] (0,0) -- node[above=1mm] {$\textcolor{pyblue}{x_{12}}$} (0.7,0);
    \draw[black, postaction={decorate}, decoration={markings, mark=at position 0.5 with {\arrow[pyred]{>}}}] (0.7,0) -- node[above=1mm] {$\textcolor{pyblue}{x_{23}}$} (1.4,0);

    \draw[fill=black] (0,0) circle (.05cm) node[below=2mm] {$\textcolor{black}{(1)}$};
    \draw[pyred, fill=pyred] (0.7,0) circle (.05cm) node[below=2mm] {$\textcolor{black}{(2)}$};
    \draw[fill=black] (1.4,0) circle (.05cm) node[below=2mm] {$\textcolor{black}{(3)}$};
\end{tikzpicture}
    }}\,, \quad Q = 
    \begin{bNiceMatrix}[first-row,code-for-first-row=\scriptstyle, first-col,code-for-first-col=\scriptstyle,]
        & \textcolor{gray}{x_{12}} & \textcolor{gray}{x_{23}} \\
        \textcolor{gray}{(1)} & +1 & 0 \\
        \textcolor{gray}{(2)} & -1 & -1 \\
        \textcolor{gray}{(3)} & 0 & +1
    \end{bNiceMatrix}\,.
\end{equation}
The reduced incidence and path matrices are given by
\begin{equation}
    Q_2 = 
    \begin{bmatrix}
        +1 & 0 \\
        0 & +1
    \end{bmatrix}\,, \quad 
    P_2 = 
    \begin{bmatrix}
        +1 & 0 \\
        0 & +1
    \end{bmatrix}\,.
\end{equation}
Solving for the on-shell condition yields the following solution:
\begin{equation}
\label{eq: 3site chain mid root}
    \begin{aligned}
    \widehat{\I}^\P(\begin{tikzpicture}[scale=0.8, baseline=-0.5ex]
            \coordinate (A) at (0,0);
            \coordinate (B) at (0.6,0);
            \coordinate (C) at (1.2,0);
            \draw[thick, postaction={decorate}, decoration={markings, mark=at position 0.6 with {\arrow[pyred]{<}}}] (A) -- (B);
            \draw[thick, postaction={decorate}, decoration={markings, mark=at position 0.6 with {\arrow[pyred]{>}}}] (B) -- (C);
            \fill[black] (A) circle (2pt);
            \fill[pyred] (B) circle (2pt);
            \fill[black] (C) circle (2pt);
        \end{tikzpicture}) &= C(\begin{tikzpicture}[scale=0.8, baseline=-0.5ex]
            \coordinate (A) at (0,0);
            \coordinate (B) at (0.6,0);
            \coordinate (C) at (1.2,0);
            \draw[black, thick] (A) -- (B);
            \draw[black, thick] (B) -- (C);
            \fill[black] (A) circle (2pt);
            \fill[black] (B) circle (2pt);
            \fill[black] (C) circle (2pt);
        \end{tikzpicture}) \sum_{m_1, m_2=0}^{\infty} \frac{(-1)^{m_{12}}}{m_1!m_2!}\frac{\Gamma[\tilde{p}_{123}+m_{12}]}{[(\tilde{p}_1+m_1)^2+\mu_{12}^2][(\tilde{p}_3+m_2)^2+\mu_{23}^2]}  \\
        &\times \left(\frac{X_1}{X_2}\right)^{\tilde{p}_1+m_1} \left(\frac{X_3}{X_2}\right)^{\tilde{p}_3+m_2} F^{(1)}_C\left[\left.\begin{matrix} \frac{-m_1}{2},\,  \frac{1-m_1}{2}\\ 1-\tilde{p}_1-m_1 \end{matrix}\right\vert u_{12}^2\right]\\
        &\times F^{(2)}_C\left[\left.\begin{matrix} \frac{\tilde{p}_{123}+m_{12}}{2},\,  \frac{\tilde{p}_{123}+1+m_{12}}{2}\\ 1+\tilde{p}_1+m_1, 1+\tilde{p}_3+m_2 \end{matrix}\right\vert u_{21}^2, u_{23}^2\right] F^{(1)}_C\left[\left.\begin{matrix} \frac{-m_2}{2},\,  \frac{1-m_2}{2}\\ 1-\tilde{p}_3-m_2 \end{matrix}\right\vert u_{32}^2\right]\,,
    \end{aligned}
\end{equation}
which converges in the region $X_1+X_3<X_2$, as can be read off from the graph tree decoration.

\subsection{$N$-site chain}

It is straightforward to generalise the previous results for the two- and three-site chains to the $N$-site chain. Here again, we focus on the maximally-nested contribution. All other contributions are obtained from factorisations of simpler graphs. Let us consider the case $1<r<N$, i.e.~the root is {\it not} one of the leaf vertices:
\begin{equation}
    \hbox{
    \begin{tikzpicture}[line width=1pt, scale=2,
            baseline={(current bounding box.center)}]

            \draw[black, postaction={decorate}, decoration={markings, mark=at position 0.5 with {\arrow[pyred]{<}}}] (0,0) -- node[above=1mm] {$\textcolor{pyblue}{x_{12}}$} (0.7,0);
            \draw[black, postaction={decorate}, decoration={markings, mark=at position 0.5 with {\arrow[pyred]{<}}}] (0.7,0) -- (1.1,0);
            \draw[black, postaction={decorate}, decoration={markings, mark=at position 0.5 with {\arrow[pyred]{<}}}] (1.5,0) -- (1.9,0);
            \draw[black, postaction={decorate}, decoration={markings, mark=at position 0.5 with {\arrow[pyred]{>}}}] (1.9,0) -- (2.3,0);
            \draw[black, postaction={decorate}, decoration={markings, mark=at position 0.5 with {\arrow[pyred]{>}}}] (2.7,0) -- (3.1,0);
            \draw[black, postaction={decorate}, decoration={markings, mark=at position 0.5 with {\arrow[pyred]{>}}}] (3.1,0) -- node[above=1mm] {$\textcolor{pyblue}{x_{N-1,N}}$} (3.8,0);

            \node[black] at ($(1.1, 0)!1/2!(1.5, 0)$) {$\cdots$};
            \node[black] at ($(2.3, 0)!1/2!(2.7, 0)$) {$\cdots$};

            \draw[fill=black] (0,0) circle (.05cm) node[below=2mm] {$\textcolor{black}{(1)}$};
            \draw[fill=black] (0.7,0) circle (.05cm) node[below=2mm] {$\textcolor{black}{(2)}$};
            \draw[pyred, fill=pyred] (1.9,0) circle (.05cm) node[below=2mm] {$\textcolor{black}{(r)}$};
            \draw[fill=black] (3.1,0) circle (.05cm) node[below=2mm] {$\textcolor{black}{(N-1)}$};
            \draw[fill=black] (3.8,0) circle (.05cm) node[below=2mm] {$\textcolor{black}{(N)}$};
        \end{tikzpicture}
    } \,.
\end{equation}
The corresponding reduced incidence matrix acquires a block-diagonal form
\begin{equation}
    Q_r = \left(\begin{matrix}
        Q_r^{(1)} & \mathbf{0} \\
        \mathbf{0} & Q_r^{(2)}
    \end{matrix}\right)\,,
\end{equation}
where 
\begin{equation}
    Q_{r}^{(1)}=\begin{bNiceMatrix}[first-row,code-for-first-row=\scriptstyle, first-col,code-for-first-col=\scriptstyle,]
        & \textcolor{gray}{x_{12}} & \textcolor{gray}{x_{23}} &
        \textcolor{gray}{x_{34}} &\textcolor{gray}{\dots} & \textcolor{gray}{x_{r-1,r}} \\
        \textcolor{gray}{(1)} & +1 & 0 & 0 & \dots &0 \\
        \textcolor{gray}{(2)} & -1 & +1 & 0 & \dots & 0 \\
        \textcolor{gray}{(3)} & 0 & -1 & +1 & \dots & 0 \\
        \textcolor{gray}{\vdots} & \vdots & \vdots & \vdots & \ddots & \vdots \\
        \textcolor{gray}{(r-1)} & 0 & 0 & 0 & \dots & +1 
    \end{bNiceMatrix}\,, \,\,\,\,
    Q_r^{(2)}=\begin{bNiceMatrix}[first-row,code-for-first-row=\scriptstyle, first-col,code-for-first-col=\scriptstyle,]
        & \textcolor{gray}{x_{r,r+1}} & \textcolor{gray}{x_{r+1,r+2}} &
        \textcolor{gray}{x_{r+2,r+3}} &\textcolor{gray}{\dots} & \textcolor{gray}{x_{N-1,N}} \\
        \textcolor{gray}{(r+1)} & +1 & -1 & 0 & \dots &0 \\
        \textcolor{gray}{(r+2)} & 0 & +1 & -1 & \dots & 0 \\
        \textcolor{gray}{(r+3)} & 0 & 0 & +1 & \dots & 0 \\
        \textcolor{gray}{\vdots} & \vdots & \vdots & \vdots & \ddots & \vdots \\
        \textcolor{gray}{(N)} & 0 & 0 & 0 & \dots & +1 
    \end{bNiceMatrix}\,.
\end{equation}
Likewise, the path matrix $P_r = Q_r^{-1}$ is given by the block-diagonal matrix
\begin{equation}
    P = \left(\begin{matrix}
        P_{r}^{(1)} & \mathbf{0} \\
        \mathbf{0} & P_{r}^{(2)}
    \end{matrix}\right)\,,
\end{equation}
with 
\begin{equation}
    P_{r}^{(1)}=\begin{bNiceMatrix}[first-row,code-for-first-row=\scriptstyle, first-col,code-for-first-col=\scriptstyle,]
        & \textcolor{gray}{(1)} & \textcolor{gray}{(2)} &
        \textcolor{gray}{(3)} &\textcolor{gray}{\dots} & \textcolor{gray}{(r-1)} \\
        \textcolor{gray}{x_{12}} & +1 & 0 & 0 & \dots &0 \\
        \textcolor{gray}{x_{23}} & +1 & +1 & 0 & \dots & 0 \\
        \textcolor{gray}{x_{34}} & +1 & +1 & +1 & \dots & 0 \\
        \textcolor{gray}{\vdots} & \vdots & \vdots & \vdots & \ddots & \vdots \\
        \textcolor{gray}{x_{r-1,r}} & +1 & +1 & +1 & \dots & +1 
    \end{bNiceMatrix}\,, \,\,\,\,
    P_r^{(2)}=\begin{bNiceMatrix}[first-row,code-for-first-row=\scriptstyle, first-col,code-for-first-col=\scriptstyle,]
        & \textcolor{gray}{(r+1)} & \textcolor{gray}{(r+2)} &
        \textcolor{gray}{(r+3)} &\textcolor{gray}{\dots} & \textcolor{gray}{(N)} \\
        \textcolor{gray}{x_{r,r+1}} & +1 & +1 & +1 & \dots & +1 \\
        \textcolor{gray}{x_{r+1,r+2}} & 0 & +1 & +1 & \dots & +1 \\
        \textcolor{gray}{x_{r+2,r+3}} & 0 & 0 & +1 & \dots & +1 \\
        \textcolor{gray}{\vdots} & \vdots & \vdots & \vdots & \ddots & \vdots \\
        \textcolor{gray}{x_{N-1,N}} & 0 & 0 & 0 & \dots & +1 
    \end{bNiceMatrix}\,.
\end{equation}
The resulting series solution is found to be:
\begin{equation}
    \begin{aligned}
    &\widehat{\I}^\P(\begin{tikzpicture}[scale=0.8, baseline=-0.5ex]
            \coordinate (A) at (0,0);
            \coordinate (B1) at (0.3,0);
            \coordinate (B2) at (0.6,0);
            \coordinate (C) at (0.9,0);
            \coordinate (C1) at (1.2,0);
            \coordinate (C2) at (1.5,0);
            \coordinate (D) at (1.8,0);

            \node[black] at ($(0.3, 0)!1/2!(0.6, 0)$) {$\cdot$};
            \node[black] at ($(1.2, 0)!1/2!(1.5, 0)$) {$\cdot$};

            \draw[thick, postaction={decorate}, decoration={markings, mark=at position 0.6 with {\arrow[pyred]{<}}}] (A) -- (B1);
            \draw[thick, postaction={decorate}, decoration={markings, mark=at position 0.6 with {\arrow[pyred]{<}}}] (B2) -- (C);
            \draw[thick, postaction={decorate}, decoration={markings, mark=at position 0.6 with {\arrow[pyred]{>}}}] (C) -- (C1);
            \draw[thick, postaction={decorate}, decoration={markings, mark=at position 0.6 with {\arrow[pyred]{>}}}] (C2) -- (D);
            \fill[black] (A) circle (2pt);
            \fill[pyred] (C) circle (2pt);
            \fill[black] (D) circle (2pt);
        \end{tikzpicture}) = C(\begin{tikzpicture}[scale=0.8, baseline=-0.5ex]
            \coordinate (A) at (0,0);
            \coordinate (B1) at (0.3,0);
            \coordinate (B2) at (0.6,0);
            \coordinate (C) at (0.9,0);
            \coordinate (C1) at (1.2,0);
            \coordinate (C2) at (1.5,0);
            \coordinate (D) at (1.8,0);

            \node[black] at ($(0.3, 0)!1/2!(0.6, 0)$) {$\cdot$};
            \node[black] at ($(1.2, 0)!1/2!(1.5, 0)$) {$\cdot$};

            \draw[thick] (A) -- (B1);
            \draw[thick] (B2) -- (C);
            \draw[thick] (C) -- (C1);
            \draw[thick] (C2) -- (D);
            \fill[black] (A) circle (2pt);
            \fill[black] (C) circle (2pt);
            \fill[black] (D) circle (2pt);
        \end{tikzpicture}) \sum_{\bm{m}\in \mathbb{N}^{N-1}_{\geq0}}^{\infty} \frac{(-1)^{m_{1\cdots N-1}}}{m_1!\cdots m_{N-1}!} \Gamma[\tilde{p}_{1\cdots N}+m_{1\cdots N-1}]  \\
        &\times \frac{1}{[(\tilde{p}_1+m_1)^2+\mu_{12}^2][(\tilde{p}_{12}+m_{12})^2+\mu_{23}^2]\cdots[(\tilde{p}_{1\cdots r-1}+m_{1\cdots r-1})^2+\mu_{r-1, r}^2]} \\
        &\times \frac{1}{[(\tilde{p}_{r+1\cdots N}+m_{r\cdots N-1})^2+\mu_{r, r+1}^2]\cdots[(\tilde{p}_N+m_{N-1})^2+\mu_{N-1, N}^2]} \\
        &\times \left(\frac{X_1}{X_r}\right)^{\tilde{p}_1+m_1} \cdots \left(\frac{X_{r-1}}{X_r}\right)^{\tilde{p}_{r-1}+m_{r-1}} \left(\frac{X_{r+1}}{X_{r}}\right)^{\tilde{p}_{r+1}+m_{r}} \cdots \left(\frac{X_{N}}{X_r}\right)^{\tilde{p}_N+m_{N-1}}\\
        &\times F^{(1)}_C\left[\left.\begin{matrix} \frac{-m_1}{2},\,  \frac{1-m_1}{2}\\ 1-\tilde{p}_1-m_1 \end{matrix}\right\vert u_{12}^2\right] F^{(2)}_C\left[\left.\begin{matrix} \frac{-m_2}{2},\,  \frac{1-m_2}{2}\\ 1+\tilde{p}_1+m_1, 1-\tilde{p}_{12}-m_{12} \end{matrix}\right\vert u_{21}^2, u_{23}^2\right] \cdots\\
        &\cdots F^{(2)}_C\left[\left.\begin{matrix} \frac{\tilde{p}_{1\cdots N}+m_{1\cdots N-1}}{2},\,  \frac{\tilde{p}_{1\cdots N}+1+m_{1\cdots N-1}}{2}\\ 1+\tilde{p}_{1\cdots r-1}+m_{1\cdots r-1}, 1+\tilde{p}_{r+1\cdots N}+m_{r\cdots N-1} \end{matrix}\right\vert u_{r, r-1}^2, u_{r, r+1}^2\right] \cdots \\
        &\cdots F^{(2)}_C\left[\left.\begin{matrix} \frac{-m_{N-2}}{2},\,  \frac{1-m_{N-2}}{2}\\ 1-\tilde{p}_{N-1N}-m_{N-2N-1}, 1+\tilde{p}_N-m_{N-1} \end{matrix}\right\vert u_{N-1, N-2}^2, u_{N-1, N}^2\right] \\
        &\times F^{(1)}_C\left[\left.\begin{matrix} \frac{-m_{N-1}}{2},\,  \frac{1-m_{N-1}}{2}\\ 1-\tilde{p}_N-m_{N-1} \end{matrix}\right\vert u_{N,N-1}^2\right]\,,
    \end{aligned}
\end{equation}
where the overall coefficient reads
\begin{equation}
    C(\begin{tikzpicture}[scale=0.8, baseline=-0.5ex]
            \coordinate (A) at (0,0);
            \coordinate (B1) at (0.3,0);
            \coordinate (B2) at (0.6,0);
            \coordinate (C) at (0.9,0);
            \coordinate (C1) at (1.2,0);
            \coordinate (C2) at (1.5,0);
            \coordinate (D) at (1.8,0);

            \node[black] at ($(0.3, 0)!1/2!(0.6, 0)$) {$\cdot$};
            \node[black] at ($(1.2, 0)!1/2!(1.5, 0)$) {$\cdot$};

            \draw[thick] (A) -- (B1);
            \draw[thick] (B2) -- (C);
            \draw[thick] (C) -- (C1);
            \draw[thick] (C2) -- (D);
            \fill[black] (A) circle (2pt);
            \fill[black] (C) circle (2pt);
            \fill[black] (D) circle (2pt);
        \end{tikzpicture}) = \left(-4\pi e^{-\frac{i\pi}{2}}\right)^{N-1} \C_+^{(1)}(p_1) \left[\prod_{j=2}^{N-1}\C_+^{(2)}(p_j)\right] \C_+^{(1)}(p_N) \,.
\end{equation}
Notice that for $N=3$ and $r=2$, we recover the three-site chain result~\eqref{eq: 3site chain mid root}. The cases $r=1$ or $r=N$, i.e.~the root is chosen to be a leaf vertex, can be obtained straightforwardly.

\subsection{$N$-site star}
\label{subsec: Nsite star}

The last example we treat is the $N$-site star which consists of a central vertex of degree $n_0=N$ from which $N$ edges are connected to vertices of degree $n_i=1$ ($i=1, \ldots, N$). We first consider a leaf vertex to be the root, which without loss of generality we choose to be $r=1$:
\begin{equation}
\vcenter{\hbox{
\begin{tikzpicture}[line width=1pt, scale=2,
            baseline={(current bounding box.center)}]

            \draw[black, postaction={decorate}, decoration={markings, mark=at position 0.5 with {\arrow[pyred]{>}}}] (0, 0) -- (-0.5, -0.5) node[below left] {$(5)$};
            \draw[black, postaction={decorate}, decoration={markings, mark=at position 0.5 with {\arrow[pyred]{>}}}] (0, 0) -- (-0.7, 0) node[left] {$(4)$};
            \draw[black, postaction={decorate}, decoration={markings, mark=at position 0.5 with {\arrow[pyred]{>}}}] (0, 0) -- (-0.5, 0.5) node[above left] {$(3)$};
            \draw[black, postaction={decorate}, decoration={markings, mark=at position 0.5 with {\arrow[pyred]{>}}}] (0, 0) -- (0, 0.7) node[above] {$(2)$};
            \draw[black, postaction={decorate}, decoration={markings, mark=at position 0.5 with {\arrow[pyred]{<}}}] (0, 0) -- (0.5, 0.5) node[above right] {$(1)$};
            \draw[black, postaction={decorate}, decoration={markings, mark=at position 0.5 with {\arrow[pyred]{>}}}] (0, 0) -- (0.7, 0) node[right] {$(N)$};

            \node[black, rotate=25] at ($(-0.5, -0.5)!1/2!(1, -0.5)$) {$\cdots$};
    
            \draw[fill=black] (0, 0) circle (.05cm) node[below right] {$(0)$};
            \draw[fill=black] (-0.5, -0.5) circle (.05cm);
            \draw[fill=black] (-0.7, 0) circle (.05cm);
            \draw[fill=black] (-0.5, 0.5) circle (.05cm);
            \draw[fill=black] (0, 0.7) circle (.05cm);
            \draw[pyred, fill=pyred] (0.5, 0.5) circle (.05cm);
            \draw[fill=black] (0.7, 0) circle (.05cm);
        \end{tikzpicture}
    }}\,, \quad Q = \begin{bNiceMatrix}[first-row,code-for-first-row=\scriptstyle, first-col,code-for-first-col=\scriptstyle,]
        & \textcolor{gray}{x_{01}} & \textcolor{gray}{x_{02}} &
        \textcolor{gray}{x_{03}} &\textcolor{gray}{\dots} & \textcolor{gray}{x_{0N}} \\
        \textcolor{gray}{(0)} & +1 & -1 & -1 & \dots & -1 \\
        \textcolor{gray}{(1)} & -1 & 0 & 0 & \dots & 0 \\
        \textcolor{gray}{(2)} & 0 & +1 & 0 & \dots & 0 \\
        \textcolor{gray}{\vdots} & \vdots & \vdots & \vdots & \ddots & \vdots \\
        \textcolor{gray}{(N)} & 0 & 0 & 0 & \dots & +1 
    \end{bNiceMatrix}\,.
\end{equation}
The reduced incidence and path matrices are given by:
\begin{equation}
    Q_{1}=\begin{bNiceMatrix}[first-row,code-for-first-row=\scriptstyle, first-col,code-for-first-col=\scriptstyle,]
        & \textcolor{gray}{x_{01}} & \textcolor{gray}{x_{02}} &
        \textcolor{gray}{x_{03}} &\textcolor{gray}{\dots} & \textcolor{gray}{x_{0N}} \\
        \textcolor{gray}{(0)} & +1 & -1 & -1 & \dots &-1 \\
        \textcolor{gray}{(2)} & 0 & +1 & 0 & \dots & 0 \\
        \textcolor{gray}{(3)} & 0 & 0 & +1 & \dots & 0 \\
        \textcolor{gray}{\vdots} & \vdots & \vdots & \vdots & \ddots & \vdots \\
        \textcolor{gray}{(N)} & 0 & 0 & 0 & \dots & +1 
    \end{bNiceMatrix}\,, \,\,\,\,
    P_1=\begin{bNiceMatrix}[first-row,code-for-first-row=\scriptstyle, first-col,code-for-first-col=\scriptstyle,]
        & \textcolor{gray}{(0)} & \textcolor{gray}{(2)} &
        \textcolor{gray}{(3)} &\textcolor{gray}{\dots} & \textcolor{gray}{(N)} \\
        \textcolor{gray}{x_{01}} & +1 & +1 & +1 & \dots & +1 \\
        \textcolor{gray}{x_{02}} & 0 & +1 & 0 & \dots & 0 \\
        \textcolor{gray}{x_{03}} & 0 & 0 & +1 & \dots & 0 \\
        \textcolor{gray}{\vdots} & \vdots & \vdots & \vdots & \ddots & \vdots \\
        \textcolor{gray}{x_{0N}} & 0 & 0 & 0 & \dots & +1 
    \end{bNiceMatrix}\,.
\end{equation}
The corresponding series solution reads:
\begin{equation}
    \begin{aligned}
    &\widehat{\I}^\P\left(\begin{tikzpicture}[scale=0.8, baseline=-0.5ex]
            \coordinate (0) at (0,0);
            \coordinate (1) at (0.3,0.3);
            \coordinate (2) at (0,0.4);
            \coordinate (3) at (-0.3,0.3);
            \coordinate (4) at (-0.4,0);
            \coordinate (5) at (-0.3,-0.3);
            \coordinate (N) at (0.4,0);

            \draw[thick, postaction={decorate}, decoration={markings, mark=at position 0.6 with {\arrow[pyred, scale=0.6]{<}}}] (0) -- (1);
            \draw[thick, postaction={decorate}, decoration={markings, mark=at position 0.6 with {\arrow[pyred, scale=0.6]{>}}}] (0) -- (2);
            \draw[thick, postaction={decorate}, decoration={markings, mark=at position 0.6 with {\arrow[pyred, scale=0.6]{>}}}] (0) -- (3);
            \draw[thick, postaction={decorate}, decoration={markings, mark=at position 0.6 with {\arrow[pyred, scale=0.6]{>}}}] (0) -- (4);
            \draw[thick, postaction={decorate}, decoration={markings, mark=at position 0.6 with {\arrow[pyred, scale=0.6]{>}}}] (0) -- (5);
            \draw[thick, postaction={decorate}, decoration={markings, mark=at position 0.6 with {\arrow[pyred, scale=0.6]{>}}}] (0) -- (N);
            \fill[black] (0) circle (2pt);
            \fill[pyred] (1) circle (2pt);
            \fill[black] (2) circle (2pt);
            \fill[black] (3) circle (2pt);
            \fill[black] (4) circle (2pt);
            \fill[black] (5) circle (2pt);
            \fill[black] (N) circle (2pt);
        \end{tikzpicture}\right) = C\left(\begin{tikzpicture}[scale=0.8, baseline=-0.5ex]
            \coordinate (0) at (0,0);
            \coordinate (1) at (0.3,0.3);
            \coordinate (2) at (0,0.4);
            \coordinate (3) at (-0.3,0.3);
            \coordinate (4) at (-0.4,0);
            \coordinate (5) at (-0.3,-0.3);
            \coordinate (N) at (0.4,0);

            \draw[thick] (0) -- (1);
            \draw[thick] (0) -- (2);
            \draw[thick] (0) -- (3);
            \draw[thick] (0) -- (4);
            \draw[thick] (0) -- (5);
            \draw[thick] (0) -- (N);
            \fill[black] (0) circle (2pt);
            \fill[black] (1) circle (2pt);
            \fill[black] (2) circle (2pt);
            \fill[black] (3) circle (2pt);
            \fill[black] (4) circle (2pt);
            \fill[black] (5) circle (2pt);
            \fill[black] (N) circle (2pt);
        \end{tikzpicture}\right) \sum_{\bm{m}\in \mathbb{N}^{N}_{\geq0}}^{\infty} \frac{(-1)^{m_{1\cdots N}}}{m_1!\cdots m_N!} \frac{\Gamma[\tilde{p}_{0\cdots N}+m_{1\cdots N}]}{(\tilde{p}_{02\cdots N}+m_{1\cdots N})^2+\mu_1^2} \left(\frac{X_0}{X_1}\right)^{\tilde{p}_{0}+m_{1}} \\
        &\times \prod_{j=2}^N \left\{\frac{1}{(\tilde{p}_j+m_j)^2+\mu_j^2} \left(\frac{X_j}{X_1}\right)^{\tilde{p}_j+m_j} F^{(1)}_C\left[\left.\begin{matrix} \frac{-m_j}{2},\,  \frac{1-m_j}{2}\\ 1-\tilde{p}_j-m_j \end{matrix}\right\vert u_{j0}^2\right]\right\} \\
        &\times F^{(1)}_C\left[\left.\begin{matrix} \frac{\tilde{p}_{0\cdots N}+m_{1\cdots N}}{2},\,  \frac{\tilde{p}_{0\cdots N}+1+m_{1\cdots N}}{2}\\ 1+\tilde{p}_{02\cdots N}+m_{1\cdots N} \end{matrix}\right\vert u_{10}^2\right] \\
        &\times F^{(N)}_C\left[\left.\begin{matrix} \frac{-m_1}{2},\,  \frac{1-m_1}{2}\\ 1-\tilde{p}_{02\cdots N}-m_{1\cdots N}, 1+\tilde{p}_2-m_2, \ldots, 1+\tilde{p}_N+m_N \end{matrix}\right\vert u_{01}^2, \ldots, u_{0N}^2\right] \,,
    \end{aligned}
\end{equation}
with
\begin{equation}
    C\left(\begin{tikzpicture}[scale=0.8, baseline=-0.5ex]
            \coordinate (0) at (0,0);
            \coordinate (1) at (0.3,0.3);
            \coordinate (2) at (0,0.4);
            \coordinate (3) at (-0.3,0.3);
            \coordinate (4) at (-0.4,0);
            \coordinate (5) at (-0.3,-0.3);
            \coordinate (N) at (0.4,0);

            \draw[thick] (0) -- (1);
            \draw[thick] (0) -- (2);
            \draw[thick] (0) -- (3);
            \draw[thick] (0) -- (4);
            \draw[thick] (0) -- (5);
            \draw[thick] (0) -- (N);
            \fill[black] (0) circle (2pt);
            \fill[black] (1) circle (2pt);
            \fill[black] (2) circle (2pt);
            \fill[black] (3) circle (2pt);
            \fill[black] (4) circle (2pt);
            \fill[black] (5) circle (2pt);
            \fill[black] (N) circle (2pt);
        \end{tikzpicture}\right) = \left(-4\pi e^{-\frac{i\pi}{2}}\right)^N \C_+^{(N)}(p_0) \left[\prod_{j=1}^N \C_+^{(1)}(p_j)\right] \,.
\end{equation}
It is interesting to observe that this star topology leads to a structurally simpler series solution than the chain topology, as the on-shell condition is almost trivial. We now consider the central vertex to be the root, i.e.~$r=0$:
\begin{equation}
\vcenter{\hbox{
\begin{tikzpicture}[line width=1pt, scale=2,
            baseline={(current bounding box.center)}]

            \draw[black, postaction={decorate}, decoration={markings, mark=at position 0.5 with {\arrow[pyred]{>}}}] (0, 0) -- (-0.5, -0.5) node[below left] {$(5)$};
            \draw[black, postaction={decorate}, decoration={markings, mark=at position 0.5 with {\arrow[pyred]{>}}}] (0, 0) -- (-0.7, 0) node[left] {$(4)$};
            \draw[black, postaction={decorate}, decoration={markings, mark=at position 0.5 with {\arrow[pyred]{>}}}] (0, 0) -- (-0.5, 0.5) node[above left] {$(3)$};
            \draw[black, postaction={decorate}, decoration={markings, mark=at position 0.5 with {\arrow[pyred]{>}}}] (0, 0) -- (0, 0.7) node[above] {$(2)$};
            \draw[black, postaction={decorate}, decoration={markings, mark=at position 0.5 with {\arrow[pyred]{>}}}] (0, 0) -- (0.5, 0.5) node[above right] {$(1)$};
            \draw[black, postaction={decorate}, decoration={markings, mark=at position 0.5 with {\arrow[pyred]{>}}}] (0, 0) -- (0.7, 0) node[right] {$(N)$};

            \node[black, rotate=25] at ($(-0.5, -0.5)!1/2!(1, -0.5)$) {$\cdots$};
    
            \draw[pyred, fill=pyred] (0, 0) circle (.05cm) node[below right] {\textcolor{black}{$(0)$}};
            \draw[fill=black] (-0.5, -0.5) circle (.05cm);
            \draw[fill=black] (-0.7, 0) circle (.05cm);
            \draw[fill=black] (-0.5, 0.5) circle (.05cm);
            \draw[fill=black] (0, 0.7) circle (.05cm);
            \draw[fill=black] (0.5, 0.5) circle (.05cm);
            \draw[fill=black] (0.7, 0) circle (.05cm);
        \end{tikzpicture}
    }}\,, \quad Q = \begin{bNiceMatrix}[first-row,code-for-first-row=\scriptstyle, first-col,code-for-first-col=\scriptstyle,]
        & \textcolor{gray}{x_{01}} & \textcolor{gray}{x_{02}} &
        \textcolor{gray}{x_{03}} &\textcolor{gray}{\dots} & \textcolor{gray}{x_{0N}} \\
        \textcolor{gray}{(0)} & -1 & -1 & -1 & \dots & -1 \\
        \textcolor{gray}{(1)} & +1 & 0 & 0 & \dots & 0 \\
        \textcolor{gray}{(2)} & 0 & +1 & 0 & \dots & 0 \\
        \textcolor{gray}{\vdots} & \vdots & \vdots & \vdots & \ddots & \vdots \\
        \textcolor{gray}{(N)} & 0 & 0 & 0 & \dots & +1 
    \end{bNiceMatrix}\,.
\end{equation}
The corresponding reduced incidence and path matrices are nothing but the $N\times N$ identity matrix:
\begin{equation}
    Q_{0}=\begin{bNiceMatrix}[first-row,code-for-first-row=\scriptstyle, first-col,code-for-first-col=\scriptstyle,]
        & \textcolor{gray}{x_{01}} & \textcolor{gray}{x_{02}} &
        \textcolor{gray}{x_{03}} &\textcolor{gray}{\dots} & \textcolor{gray}{x_{0N}} \\
        \textcolor{gray}{(1)} & +1 & 0 & 0 & \dots & 0 \\
        \textcolor{gray}{(2)} & 0 & +1 & 0 & \dots & 0 \\
        \textcolor{gray}{(3)} & 0 & 0 & +1 & \dots & 0 \\
        \textcolor{gray}{\vdots} & \vdots & \vdots & \vdots & \ddots & \vdots \\
        \textcolor{gray}{(N)} & 0 & 0 & 0 & \dots & +1 
    \end{bNiceMatrix}\,, \,\,\,\,
    P_0=\begin{bNiceMatrix}[first-row,code-for-first-row=\scriptstyle, first-col,code-for-first-col=\scriptstyle,]
        & \textcolor{gray}{(1)} & \textcolor{gray}{(2)} &
        \textcolor{gray}{(3)} &\textcolor{gray}{\dots} & \textcolor{gray}{(N)} \\
        \textcolor{gray}{x_{01}} & +1 & 0 & 0 & \dots & 0 \\
        \textcolor{gray}{x_{02}} & 0 & +1 & 0 & \dots & 0 \\
        \textcolor{gray}{x_{03}} & 0 & 0 & +1 & \dots & 0 \\
        \textcolor{gray}{\vdots} & \vdots & \vdots & \vdots & \ddots & \vdots \\
        \textcolor{gray}{x_{0N}} & 0 & 0 & 0 & \dots & +1 
    \end{bNiceMatrix}\,.
\end{equation}
The on-shell condition is therefore completely trivial, and the series solution is given by:
\begin{equation}
    \begin{aligned}
    &\widehat{\I}^\P\left(\begin{tikzpicture}[scale=0.8, baseline=-0.5ex]
            \coordinate (0) at (0,0);
            \coordinate (1) at (0.3,0.3);
            \coordinate (2) at (0,0.4);
            \coordinate (3) at (-0.3,0.3);
            \coordinate (4) at (-0.4,0);
            \coordinate (5) at (-0.3,-0.3);
            \coordinate (N) at (0.4,0);

            \draw[thick, postaction={decorate}, decoration={markings, mark=at position 0.6 with {\arrow[pyred, scale=0.6]{>}}}] (0) -- (1);
            \draw[thick, postaction={decorate}, decoration={markings, mark=at position 0.6 with {\arrow[pyred, scale=0.6]{>}}}] (0) -- (2);
            \draw[thick, postaction={decorate}, decoration={markings, mark=at position 0.6 with {\arrow[pyred, scale=0.6]{>}}}] (0) -- (3);
            \draw[thick, postaction={decorate}, decoration={markings, mark=at position 0.6 with {\arrow[pyred, scale=0.6]{>}}}] (0) -- (4);
            \draw[thick, postaction={decorate}, decoration={markings, mark=at position 0.6 with {\arrow[pyred, scale=0.6]{>}}}] (0) -- (5);
            \draw[thick, postaction={decorate}, decoration={markings, mark=at position 0.6 with {\arrow[pyred, scale=0.6]{>}}}] (0) -- (N);
            \fill[pyred] (0) circle (2pt);
            \fill[black] (1) circle (2pt);
            \fill[black] (2) circle (2pt);
            \fill[black] (3) circle (2pt);
            \fill[black] (4) circle (2pt);
            \fill[black] (5) circle (2pt);
            \fill[black] (N) circle (2pt);
        \end{tikzpicture}\right) = C\left(\begin{tikzpicture}[scale=0.8, baseline=-0.5ex]
            \coordinate (0) at (0,0);
            \coordinate (1) at (0.3,0.3);
            \coordinate (2) at (0,0.4);
            \coordinate (3) at (-0.3,0.3);
            \coordinate (4) at (-0.4,0);
            \coordinate (5) at (-0.3,-0.3);
            \coordinate (N) at (0.4,0);

            \draw[thick] (0) -- (1);
            \draw[thick] (0) -- (2);
            \draw[thick] (0) -- (3);
            \draw[thick] (0) -- (4);
            \draw[thick] (0) -- (5);
            \draw[thick] (0) -- (N);
            \fill[black] (0) circle (2pt);
            \fill[black] (1) circle (2pt);
            \fill[black] (2) circle (2pt);
            \fill[black] (3) circle (2pt);
            \fill[black] (4) circle (2pt);
            \fill[black] (5) circle (2pt);
            \fill[black] (N) circle (2pt);
        \end{tikzpicture}\right) \sum_{\bm{m}\in \mathbb{N}^{N}_{\geq0}}^{\infty} \frac{(-1)^{m_{1\cdots N}}}{m_1!\cdots m_N!} \, \Gamma[\tilde{p}_{0\cdots N}+m_{1\cdots N}] \\
        &\times \prod_{j=1}^N \left\{\frac{1}{(\tilde{p}_j+m_j)^2+\mu_j^2} \left(\frac{X_j}{X_0}\right)^{\tilde{p}_j+m_j} F^{(1)}_C\left[\left.\begin{matrix} \frac{-m_j}{2},\,  \frac{1-m_j}{2}\\ 1-\tilde{p}_j-m_j \end{matrix}\right\vert u_{j0}^2\right]\right\} \\
        &\times F^{(N)}_C\left[\left.\begin{matrix} \frac{\tilde{p}_{0\cdots N}+m_{1\cdots N}}{2},\,  \frac{\tilde{p}_{0\cdots N}+1+m_{1\cdots N}}{2}\\ 1+\tilde{p}_1+m_1, \ldots, 1+\tilde{p}_N+m_N \end{matrix}\right\vert u_{01}^2, \ldots, u_{0N}^2\right] \,.
    \end{aligned}
\end{equation}
The number of kinematic regions, and hence the number of different possible series solutions for a given graph, appears to be closely connected to the transcendental weight of the correlator. As we have stated before, every cosmological graph has uniform transcendental weight. In the case of the $N$-site star, the fully factorised diagram consists of a product of $N$ hypergeometric functions $F_C^{(1)}={}_2F_1$ and one type-$C$ Lauricella function $F_C^{(N)}$ of order $N$. Therefore, the graph has transcendental weight $2N$. The fully analytic piece on the other hand is a $N$-fold sum over $N$ hypergeometric functions ${}_2F_1$ and one Lauricella function $F_C^{(N)}$. At first sight, this would correspond to a transcendental weight of $3N$. Notice though, that exactly $N$ of these functions come with upper parameters $-m_j/2$ and $(1-m_j)/2$ which always contains a negative integer. If a ${}_2F_1$ has a negative integer as an upper parameter, it becomes a finite polynomial with transcendental weight zero. If $F_C^{(N)}$ has a negative-integer upper parameter, one of the $N$ sums truncates at a finite order, also reducing the transcendental weight to $N-1$. Therefore, the analytic piece still has uniform transcendental weight $2N$. 

\vskip 4pt 
Since negative integers as upper parameters in $F_C^{(N)}$ can reduce the transcendental weight only by one, it is always required that at least $N-1$ of the $N$ hypergeometric functions ${}_2F_1$ appear with negative-integer upper parameters too (meaning that at most one of them can have arbitrary upper parameters). Hence there are exactly $N+1$ choices for where to put arbitrary upper parameters and the structure of the remaining functions is then already fixed by consistency conditions. This leads to the $N+1$ different series representations that can be obtained from closing the contour in different directions (and not $2^N$). 

\newpage
\section{Hidden Simplicity}
\label{sec: hidden simplicity}

In this section, we demonstrate that the spectral gluing algorithm naturally constructs cosmological correlators as expansions in eigenfunctions of the corresponding graph annihilators. Acting with graph annihilators on these series solutions then admits two complementary interpretations. On the one hand, it contracts an internal edge of the associated diagram. On the other hand, it removes the propagator contribution from the series representation. Equating both expressions gives rise to a remarkable new class of identities among generalised hypergeometric functions, which we refer to as {\it magical identities}. After illustrating the analogy with the simpler and familiar case of the Helmholtz equation, we present a general recipe for deriving an infinite family of such identities and conclude with several representative examples. In passing, we also discuss the large-mass expansion of the constructed series solutions.

\subsection{Helmholtz analogy}

For pedagogical purposes, it is instructive to revisit a simplified toy example that will later make the analogy with cosmological massive tree graphs transparent. Consider the Helmholtz equation with a generic source:
\begin{equation}
\label{eq: Helmholtz equation}
    \D F(\bm{x}) = S(\bm{x}) \,, \quad \text{with} \quad \D \equiv -\Delta+\mu^2 \,,
\end{equation}
where $\mu^2\in\mathbb{R}_{\geq0}$, that we solve in a $d$-dimensional box $\Omega = \{\bm{x}\in\mathbb{R}^d\, | \, 0\leq x_i\leq L\,, \,i=1, \ldots, d\}$. We require fixed boundary conditions $F(\bm{x})|_{\partial\Omega}=0$. The Dirichlet eigenfunctions of the differential operator $\D$ are given by
\begin{equation}
    \phi_{\bm{m}}(\bm{x}) = \prod_{i=1}^d \sin\left(\frac{m_i\pi x_i}{L}\right)\,,
\end{equation}
with $\bm{m} \equiv (m_1, \ldots, m_d)$, which satisfy the eigenvalue equation
\begin{equation}
    \D \phi_{\bm{m}}(\bm{x}) = (\lambda_{\bm{m}} + \mu^2) \phi_{\bm{m}}(\bm{x})\,, \quad \lambda_{\bm{m}} \equiv \sum_{\bm{m}\in \mathbb{N}^d_{\geq0}} \left(\frac{m_i \pi}{L}\right)^2\,.
\end{equation}
Suppose we can expand the source term in terms of eigenfunctions $S(\bm{x}) = \sum \rho_{\bm{m}} \phi_{\bm{m}}(\bm{x})$. Using the superposition principle, the general solution to~\eqref{eq: Helmholtz equation} can be found by acting with the differential operator:
\begin{equation}
    F(\bm{x}) = \sum_{\bm{m}\in \mathbb{N}^d_{\geq0}} \frac{\rho_{\bm{m}}}{\lambda_{\bm{m}}+\mu^2} \, \phi_{\bm{m}}(\bm{x}) \,.
\end{equation}
This expression is nothing but the mode expansion of the solution $F(\bm{x})$ in terms of eigenfunctions $\phi_{\bm{m}}(\bm{x})$ where the dynamical part is encoded in the propagator $1/(\lambda_{\bm{m}}+\mu^2)$ and where the source fixes the spectral density $\rho_{\bm{m}}$.

\subsubsection*{Collapse procedure}

We now illustrate how solutions ``collapse'' under repeated action of Helmholtz operators. To this end, we introduce the family of functions
\begin{equation}
    F^{(\mu_1, \ldots, \mu_n)}(\bm{x}) = \sum_{\bm{m}\in \mathbb{N}^d_{\geq0}} \frac{\rho_{\bm{m}} \, \phi_{\bm{m}}(\bm{x})}{[\lambda_{\bm{m}}+\mu_1^2] \cdots [\lambda_{\bm{m}}+\mu_n^2]} \,,
\end{equation}
with $n\geq 1$. It is straightforward to see that the action of the differential operator $\D_j$ with mass parameter $\mu_j^2$ removes the corresponding denominator:
\begin{equation}
    \begin{aligned}
        \D_j F^{(\mu_1, \ldots, \mu_n)}(\bm{x}) &= (-\Delta + \mu_j^2) F^{(\mu_1, \ldots, \mu_n)}(\bm{x}) \\
        &= \sum_{\bm{m}\in \mathbb{N}^d_{\geq0}} \frac{\rho_{\bm{m}} \, \phi_{\bm{m}}(\bm{x})}{[\lambda_{\bm{m}}+\mu_1^2] \cdots [\lambda_{\bm{m}}+\mu_{j-1}^2] [\lambda_{\bm{m}}+\mu_{j+1}^2] \cdots [\lambda_{\bm{m}}+\mu_n^2]} \,.
    \end{aligned}
\end{equation}
Iterating this procedure successively removes all propagator factors and eventually yields the source:
\begin{equation}
    F^{(\mu_1, \ldots, \mu_n)}(\bm{x}) \xrightarrow[\D_1]{} F^{(\mu_2, \ldots, \mu_n)}(\bm{x}) \xrightarrow[\D_2]{} \cdots \xrightarrow[\D_n]{} S(\bm{x}) = \sum_{\bm{m}\in \mathbb{N}^d_{\geq0}} \rho_{\bm{m}} \phi_{\bm{m}}(\bm{x})\,.
\end{equation}
As we will see below, this collapse mechanism admits a direct analogue in the context of cosmological massive tree graphs.

\subsubsection*{From Helmholtz to cosmological massive tree graphs}

In what follows, we make the correspondence explicit by showing that the spectral gluing algorithm constructs solutions for cosmological massive tree graphs via a direct application of the superposition principle. The building blocks of this construction are the eigenfunctions of the graph annihilators, which are precisely the vertex functions. These eigenfunctions are fixed by imposing the Bunch-Davies vacuum condition in the infinite past, which uniquely determines their form and ensures the absence of folded (collinear) singularities. This requirement plays a role analogous to Dirichlet boundary conditions for the Helmholtz eigenfunctions. In this correspondence, the eigenvalues are identified with the off-shell mass parameters associated with internal propagators. The analogue of the source term arises after successive applications of the graph annihilators: one obtains contracted graphs in which selected internal edges have been systematically removed.

\vskip 4pt
These and further parallels between the Helmholtz equation and cosmological tree graphs are summarised in the table below.

\renewcommand{\arraystretch}{1.2}
\begin{center}
\begin{tabular}{c | c} 
    \hline
    \rowcolor{lightgray}
    Helmholtz equation & Cosmological massive tree graphs  \\ [0.5ex] 
    \hline\hline
    \begin{tabular}{@{}c@{}}
        {\it Differential operators}\\
        $\D_j = -\Delta+\mu_j^2$
    \end{tabular} & 
    \begin{tabular}{@{}c@{}}
        {\it Graph annihilators}\\
        $\D_{ij} = \vartheta_{ij}^2 -u_{ij}^2\left(\vartheta_{\{i\}}+\tilde p_{i}\right)\left(\vartheta_{\{i\}}+\tilde p_i+1\right) +\mu_{ij}^2$
    \end{tabular}  \\
 \hline
 \begin{tabular}{@{}c@{}}
        {\it Eigenfunctions}\\
        $\phi_{\bm{m}}(\bm{x})$
    \end{tabular} & 
    \begin{tabular}{@{}c@{}}
        {\it Lauricella vertex functions}\\
        $\F^{(n)}(u_{ij})$
    \end{tabular}  \\
 \hline
 \begin{tabular}{@{}c@{}}
        {\it Eigenvalues}\\
        $\lambda_{\bm{m}}+\mu_j^2$
    \end{tabular} & 
    \begin{tabular}{@{}c@{}}
        {\it Off-shell mass parameters}\\
        $x_{ij}^2+\mu_{ij}^2$
    \end{tabular}  \\
 \hline
 \begin{tabular}{@{}c@{}}
        {\it General solution}\\
        $F^{(\mu_1, \ldots, \mu_n)}(\bm{x})$
    \end{tabular} & 
    \begin{tabular}{@{}c@{}}
        {\it Cosmological correlator}\\
        $\G(\{\bm{u}\})$
    \end{tabular}  \\
 \hline
 \begin{tabular}{@{}c@{}}
        {\it Boundary conditions}\\
        $F(\bm{x})|_{\partial\Omega}=0$
    \end{tabular} & 
    \begin{tabular}{@{}c@{}}
        {\it Bunch-Davies vacuum} \\
        Absence of folded singularities
    \end{tabular}  \\
 \hline
 \begin{tabular}{@{}c@{}}
        {\it Source} \\
        $S(\bm{x})$
    \end{tabular} & 
    \begin{tabular}{@{}c@{}}
        {\it Contracted graph}\\
        $C_{ij}[\G]$
    \end{tabular}  \\[1ex] 
 \hline
\end{tabular}
\end{center}

\subsection{Graph annihilators}

We begin by constructing the system of (partial) differential equations satisfied by cosmological graphs. The source terms are contracted graphs, obtained from the original graphs after collapsing an internal edge. 

\subsubsection*{Contracting graphs with differential equations}

A given tree-level cosmological graph $\G$ satisfies the following differential system
\begin{equation}
\label{eq: differential equations}
    \boxed{
    \begin{aligned}
        \D_{ij} \G =  C_{ij}[\G] \,, \quad \text{with} \quad 
        \D_{ij} \equiv \vartheta_{ij}^2 -u_{ij}^2\left(\vartheta_{\{i\}}+\tilde p_i\right)\left(\vartheta_{\{i\}}+\tilde p_i+1\right) +\mu_{ij}^2\,,
    \end{aligned}
    }
\end{equation}
where we used the Euler operators $\vartheta_{ij} = u_{ij}\partial_{u_{ij}}$ and $\vartheta_{\{i\}}=\sum_{\ell=1}^{n_i}\vartheta_{i\ell}$ (the sum runs over all the edges connected to the vertex $i$, of degree $n_i$, and the $\vartheta_{i\ell}$ are with respect to the momentum ratios on these edges, e.g.~$u_{i1}=Y_{i1}/X_i$). The contraction $C_{ij}[\G]$ of the graph $\G$ along the edge between $i$ and $j$ is defined by the following procedure:
\begin{enumerate}
    \item Remove the edge between $i$ and $j$ from $\G$ and pinch the two vertices together so that they form a new effective vertex;
    \item Assign the vertex energy $X_i+X_j$ as well as the twist $p_i+p_j$ to this vertex.
\end{enumerate}
Notice that $C_{ij}[\G] = C_{ji}[\G]$. The system of (partial) differential equations~\eqref{eq: differential equations} was fully derived in~\cite{Liu:2024str} (see also~\cite{Gasparotto:2024bku, Baumann:2026atn} for a twisted cohomology approach, and~\cite{Chen:2026dqp} at loop level). For completeness, we provide the derivation in App.~\ref{sec: diff eq derivation}, adjusting the original proof according to our notations. Schematically, an example of such a system is represented by:
\begin{equation}
    \textcolor{pyblue}{\D_{ij}}
    \vcenter{\hbox{
    \begin{tikzpicture}[line width=1pt, scale=2, baseline={(0,0)}]
        \draw[black] (0, 0) -- (0.7, 0);
        \draw[pyblue] (0.7, 0) -- (1.4, 0) node[midway, above=1mm] {$Y_{ij}$};
        \draw[black] (0, 0) -- (-0.5, 0.5);
        \draw[black] (0, 0) -- (-0.5, -0.5);
    
        \draw[fill=black] (0, 0) circle (.05cm);
        \draw[pyblue, fill=pyblue] (0.7, 0) circle (.05cm) node[below=1mm] {$(X_i, p_i)$};
        \draw[pyblue, fill=pyblue] (1.4, 0) circle (.05cm) node[below=1mm] {$(X_j, p_j)$};
        \draw[fill=black] (-0.5, 0.5) circle (.05cm);
        \draw[fill=black] (-0.5, -0.5) circle (.05cm);
    \end{tikzpicture}
    }}
    =
    \vcenter{\hbox{
    \begin{tikzpicture}[line width=1pt, scale=2, baseline={(0,0)}]
        \draw[black] (0, 0) -- (0.7, 0);
        \draw[black] (0, 0) -- (-0.5, 0.5);
        \draw[black] (0, 0) -- (-0.5, -0.5);
    
        \draw[fill=black] (0, 0) circle (.05cm);
        \draw[pyblue, fill=pyblue] (0.7, 0) circle (.05cm) node[above=1mm] {$(X_i+X_j, p_i+p_j)$};
        \draw[fill=black] (-0.5, 0.5) circle (.05cm);
        \draw[fill=black] (-0.5, -0.5) circle (.05cm);
    \end{tikzpicture}
    }} \,.
\end{equation}
It comes as no surprise that the differential operators that are present in this system of coupled differential equations are exactly the annihilators of the vertex functions that we derived in~\eqref{eq: vertex function annihilator}. Thus, we recover the well-known fact that the factorised contributions of a correlator correspond to homogeneous solutions to the underlying differential system. 

\subsubsection*{Expansion in eigenfunctions}

Our off-shell ansatz for the nested parts of the correlator turns out to be a very natural form for the inhomogeneous solution of this system. The reason is that it is constructed from \textit{eigenfunctions} of the differential operator~\eqref{eq: vertex function annihilator}. To see this, recall that the off-shell ansatz of the spectral gluing algorithm has the following form: 
\begin{align}
    \widehat\I^{\P} = C(\G) \, &\left[\prod_{i=1}^V \Gamma\left(\tilde{p}_i+\xi_i\right) F^{(n_i)}_C\left[\left.\begin{matrix} \frac{\tilde{p}_i+\xi_i}{2},\,  \frac{\tilde{p}_i+1+\xi_i}{2}\\ 1+\epsilon_{i1} x_{i1}, \ldots, 1+\epsilon_{i n_i}x_{i n_i} \end{matrix}\right\vert u_{i1}^2, \ldots, u_{in_i}^2\right]\right] \nonumber \\
        &\times \left[\prod_{(i,j)}\frac{1}{x_{ij}^2+\mu_{ij}^2} \, \left(\frac{u_{ij}}{u_{ji}}\right)^{\epsilon_{ij} x_{ij}}\right]\,,
\end{align}
where the values of the $x_i$ are fixed by solving the on-shell condition. Let us now act with the differential operator $\D_{ij}$ on this ansatz (meaning we look at vertex $i$ and the energy ratio $u_{ij}=Y_{ij}/X_i$ that is assigned to the edge connecting $i$ to $j$). The dependence on $u_{ij}$ is encoded in the term
\begin{equation}
    \F^{(n_i)}(u_{ij}) = u_{ij}^{\epsilon_{ij}x_{ij}} F^{(n_i)}_C\left[\left.\begin{matrix} \frac{\tilde{p}_i+\xi_i}{2},\,  \frac{\tilde{p}_i+1+\xi_i}{2}\\ 1+\epsilon_{i1} x_{i1}, \ldots, 1+\epsilon_{i n_i}x_{i n_i} \end{matrix}\right\vert u_{i1}^2, \ldots, u_{in_i}^2\right]\,.
\end{equation}
We now proceed the same way as when we derived the annihilator of the vertex functions and construct the annihilator of the function $\F^{(n_i)}(u_{ij})$. First, we use the differential equation for the Lauricella function, which in this case is annihilated by
\begin{equation}
    \tilde \D_{ij}^\F  = \vartheta_{ij}(\vartheta_{ij}+2\epsilon_{ij}x_{ij}) -u_{ij}^2(\vartheta_{\{i\}}+\tilde p_i +\xi_i)(\vartheta_{\{i\}}+\tilde p_i +\xi_i+1)\,,
\end{equation}
and construct the annihilator of $\F^{(n_i)}(u_{ij})$ by applying the the gauge transformation $\D_{ij}^\F = u_{ij}^{\epsilon_{ij}x_{ij}} \tilde\D_{ij}^\F u_{ij}^{-\epsilon_{ij}x_{ij}}$,
\begin{equation}
    \D_{ij}^\F  = \vartheta_{ij}^2 -x_{ij}^2 -u_{ij}^2\left(\vartheta_{\{i\}}+\tilde p_i +\xi_i -\sum\limits_{\ell=1}^{n_i}\epsilon_{i\ell}x_{i\ell}\right)\left(\vartheta_{\{i\}}+\tilde p_i +\xi_i -\sum\limits_{\ell=1}^{n_i}\epsilon_{i\ell}x_{i\ell}+1\right)\,,
\end{equation}
where we used $\epsilon_{i\ell}^2=1$ to simplify the second term. Now, recall that $\xi_i$ defines the linear system $\bm\xi = Q \bm{x}$ and that the entries of the incidence matrix $Q$ are precisely given by the signs $\epsilon_{i\ell}$, so that $\xi_i = \sum_{\ell=1}^{n_i} \epsilon_{i\ell}x_{i\ell}$ and the respective terms in the annihilator vanish identically. Therefore, we finally obtain
\begin{equation}
    \left[\vartheta_{ij}^2 -x_{ij}^2 -u_{ij}^2(\vartheta_{\{i\}} +\tilde p_i)(\vartheta_{\{i\}} +\tilde p_i +1)\right] \F^{(n_i)}(u_{ij})=0 \,.
\end{equation}
By comparing this expression to $\D_{ij}$, we discover that the two differential operators only differ by a constant,
\begin{equation}
    \D_{ij} -\D_{ij}^\F = x_{ij}^2+ \mu_{ij}^2 \,,
\end{equation}
which means that $\F^{(n_i)}(u_{ij})$ (and therefore also $\widehat\I^{\P}$) is an eigenfunction of $\D_{ij}$ with eigenvalue $x_{ij}^2 + \mu_{ij}^2$. This cancels exactly the propagator factor $1/( x_{ij}^2+\mu_{ij}^2)$ in the off-shell ansatz, and the result does not depend on $\mu_{ij}$ anymore, which is consistent with the contraction of the graph that removes exactly this edge. 

\subsubsection*{A spectral perspective on graph contractions}

In the spectral representation, the mechanism by which differential operators contract graphs becomes particularly transparent. In fact, the above observations can already be understood directly at the level of the graph integrands. To see this, consider the nested contribution associated with a given internal edge, which is precisely the building block responsible for generating the source term on the right-hand side of the differential equation. Suppose the vertices $i$ and $j$ are connected by an edge, and hence the correlator contains the spectral integral
\begin{equation}
    \G \supset e^{-i\frac{\pi}{2}}\int_{-\infty}^{+\infty} \frac{[\mathrm{d}\nu_{ij}]}{\nu_{ij}^2-\mu_{ij}^2} \V_{+,(\bm\nu_i,\nu_{ij})}^{(n_i)}(\bm{u}_i;p_i) \,\V_{+,(\bm\nu_j,\nu_{ij})}^{(n_j)}(\bm u_j;p_j)\,.
\end{equation}
The vertex function $\V_{+,(\bm\nu_i,\nu_{ij})}^{(n_i)}(\bm{u}_i;p_i)$ contains all the terms that depend on the momentum ratio $u_{ij}=Y_{ij}/X_i$ that is related to this edge. Acting with the differential operator $\D_{ij}= \vartheta_{ij}^2 -u_{ij}^2\left(\vartheta_{\{i\}}+\tilde p_{i}\right)\left(\vartheta_{\{i\}}+\tilde p_i+1\right) +\mu_{ij}^2$ on the graph then boils down to acting with it on this vertex function. We know that the vertex function is annihilated by the operator $\D_{ij}^{\V}= \vartheta_{ij}^2 -u_{ij}^2\left(\vartheta_{\{i\}}+\tilde p_{i}\right)\left(\vartheta_{\{i\}}+\tilde p_i+1\right) +\nu_{ij}^2$, so that
\begin{equation}
    \D_{ij} \V_{+,(\bm\nu_i,\nu_{ij})}^{(n_i)}(\bm{u}_i;p_i) = (\mu_{ij}^2-\nu_{ij}^2) \V_{+,(\bm\nu_i,\nu_{ij})}^{(n_i)}(\bm{u}_i;p_i)\,,
\end{equation}
and hence
\begin{equation}
    \D_{ij} \G \supset e^{i\frac{\pi}{2}}\int_{-\infty}^{+\infty} [\mathrm{d}\nu_{ij}]\, \V_{+,(\bm\nu_i,\nu_{ij})}^{(n_i)}(\bm{u}_i;p_i) \,\V_{+,(\bm\nu_j,\nu_{ij})}^{(n_j)}(\bm u_j;p_j)\,.
\end{equation}
Thus, we are left with a spectral integral without an additional spectral function. We can still go through the gluing algorithm and evaluate the integral in exactly the same way as before. This will result in a series expansion where the terms related to the spectral function are simply set to unity. This is what we observed from the eigenfunction argument above. And, since the right-hand side of the differential equation can be interpreted as a graph with the edge between $i$ and $j$ removed, we can view this as a generalisation of the completeness relation among vertex functions, i.e. schematically we have relations like
\begin{equation}
    \int_{-\infty}^{+\infty} [\mathrm{d}\nu_{ij}]\, \V_{+,(\bm\nu_i,\nu_{ij})}^{(n_i)}(\bm{u}_i;p_i) \,\V_{+,(\bm\nu_j,\nu_{ij})}^{(n_j)}(\bm u_j;p_j) \sim \V_{+,(\bm\nu_i,\bm\nu_j)}^{(n_i+n_j-2)}(\bm{u}_i,\bm{u}_j;p_i+p_j)\,.
\end{equation}
We discuss these identities in more details in the next section.

\subsection{Hypergeometric magical identities}

We now further explore the implications of the existence of these differential equations and how they relate to new \textit{magical identities} for (generalised) hypergeometric functions. 

\subsubsection{Warm-up example: two-site chain} 

Let us first take a look at the simplest case again: the two-site chain. According to the derived differential equations, the two-site chain satisfies the relation
\begin{equation}
    \D_{12}  \vcenter{\hbox{
\begin{tikzpicture}[line width=1pt, scale=2, baseline={(current bounding box.center)}]
    \draw[black] (0,0) -- node[above=1mm] {$\textcolor{black}{Y_{12}}$} (0.7,0);

    \draw[black, fill=black] (0,0) circle (.05cm) node[below=1mm] {$\textcolor{black}{(X_1, p_1)}$};
    \draw[fill=black] (0.7,0) circle (.05cm) node[below=1mm] {$\textcolor{black}{(X_2, p_2)}$};
\end{tikzpicture}
    }} =  \vcenter{\hbox{
\begin{tikzpicture}[line width=1pt, scale=2, baseline={(current bounding box.center)}]
    \draw[fill=black] (0,0) circle (.05cm) node[above=1mm] {$\textcolor{black}{(X_1+X_2, p_1+p_2)}$};
\end{tikzpicture}
    }}\,.
\end{equation}
Since acting with the differential operator on any factorised  contribution vanishes, only the fully nested contribution produces a source term on the right-hand side. Hence, we can view this equation as an equation for the fully nested part only. We therefore can restrict to the $++$-diagram only. The single-vertex graph on the right-hand side is given by
\begin{equation}
    \vcenter{\hbox{
\begin{tikzpicture}[line width=1pt, scale=2, baseline={(current bounding box.center)}]
    \draw[fill=black] (0,0) circle (.05cm) node[below=1mm] {$\textcolor{white}{p_1+p_2}$} node[above=1mm] {$\textcolor{black}{(X_1+X_2, p_1+p_2)}$};
\end{tikzpicture}
    }} = i\frac{e^{-i\frac{\pi}{2}(p_1+p_2-d)}\Gamma[p_1+p_2-d]}{(X_1+X_2)^{p_1+p_2-d}} \,,
\end{equation}
where we have included the kinematic prefactor for the all-plus graph. Notice that this is a simple function of the kinematic variables $X_1$ and $X_2$ (purely rational for integer twists). On the other hand, we can use our previously derived result for the fully nested part of the two-site chain and apply the differential operator $\D_{12}$, which cancels the spectral function in the series representation. Including all kinematic and numerical prefactors for the graph, we find
\begin{equation}
    \begin{aligned}
    \D_{12}  \vcenter{\hbox{
\begin{tikzpicture}[line width=1pt, scale=2, baseline={(current bounding box.center)}]
    \draw[black] (0,0) -- node[above=1mm] {$\textcolor{black}{Y_{12}}$} (0.7,0);

    \draw[black, fill=black] (0,0) circle (.05cm) node[below=1mm] {$\textcolor{black}{(X_1, p_1)}$};
    \draw[fill=black] (0.7,0) circle (.05cm) node[below=1mm] {$\textcolor{black}{(X_2, p_2)}$};
\end{tikzpicture}
    }}  &= i\frac{e^{-i\frac{\pi}{2}(\tilde p_1+\tilde p_2)}}{X_1^{\tilde p_1} X_2^{\tilde p_2}}\sum\limits_{m=0}^\infty \frac{(-1)^m}{m!} \Gamma[\tilde p_1+\tilde p_2+m] \left(\frac{X_2}{X_1}\right)^{\tilde p_2+m} \\ &\times F^{(1)}_C\left[\left.\begin{matrix} \frac{\tilde{p}_{12}+m}{2},\,  \frac{\tilde{p}_{12}+1+m}{2}\\ 1+\tilde{p}_2+m \end{matrix}\right\vert \left(\frac{Y_{12}}{X_1}\right)^2\right] F^{(1)}_C\left[\left.\begin{matrix} \frac{-m}{2},\,  \frac{1-m}{2}\\ 1-\tilde{p}_2-m \end{matrix}\right\vert \left(\frac{Y_{12}}{X_2}\right)^2\right]\,,
    \end{aligned}
\end{equation}
where we used the usual short-hand notation $\tilde p_i = p_i -\frac{d}{2}$. By comparing both sides, and noting that $F_C^{(1)}={}_2F_1$, we find the following \textit{magical identity} between a rational function and an infinite sum over hypergeometric functions,
\begin{equation}
\label{eq: two-site chain magic identity}
    \boxed{
    \begin{aligned}
    \frac{X_1^{\tilde p_1} X_2^{\tilde p_2}}{(X_1+X_2)^{\tilde p_1+\tilde p_2}} &= \sum\limits_{m=0}^\infty \frac{(-1)^m}{m!} \frac{\Gamma[\tilde p_1+\tilde p_2+m]}{\Gamma[\tilde p_1+\tilde p_2]} \left(\frac{X_2}{X_1}\right)^{\tilde p_2+m} \\ &\hspace*{-0.5cm}\times {}_2F_1\left[\left.\begin{matrix} \frac{\tilde{p}_{12}+m}{2},\,  \frac{\tilde{p}_{12}+1+m}{2}\\ 1+\tilde{p}_2+m \end{matrix}\right\vert \left(\frac{Y_{12}}{X_1}\right)^2\right]
    {}_2F_1\left[\left.\begin{matrix} \frac{-m}{2},\,  \frac{1-m}{2}\\ 1-\tilde{p}_2-m \end{matrix}\right\vert \left(\frac{Y_{12}}{X_2}\right)^2\right]\,.
    \end{aligned}
    }
\end{equation}
This identity was first derived in~\cite{Grafe:2026qsm} where it was also noted that the equation holds even if one removes the two hypergeometric functions on the right-hand side:
\begin{equation}
\label{eq: 2site chain geometric}
    \boxed{
    \frac{X_1^{\tilde p_1} X_2^{\tilde p_2}}{(X_1+X_2)^{\tilde p_1+\tilde p_2}}  = \sum\limits_{m=0}^\infty \frac{(-1)^m}{m!} \frac{\Gamma[\tilde p_1+\tilde p_2+m]}{\Gamma[\tilde p_1+\tilde p_2]} \left(\frac{X_2}{X_1}\right)^{\tilde p_2+m} \,,
    }
\end{equation}
which is nothing but the corresponding geometric series in $X_2/X_1$. This is not surprising either as~\eqref{eq: 2site chain geometric} corresponds to the internal soft limit of the the two-site graph, $Y_{12}\to0$, for which the hypergeometric functions drop off completely, see Sec.~\ref{subsubsec: partial soft limits}. It can be checked numerically that the above hypergeometric identity \eqref{eq: two-site chain magic identity} is indeed true for $X_2 < X_1$ and that the appearing hypergeometric functions are not trivially one. Instead, there is some very intriguing resummation happening, which also eliminates the $Y_{12}$-dependence completely. One can check numerically that~\eqref{eq: two-site chain magic identity} holds for any $Y_{12}<X_1, X_2$.

\begin{figure}[t!]
    \centering
    \begin{subfigure}{.5\textwidth}
        \centering
        \includegraphics[width=1\linewidth]{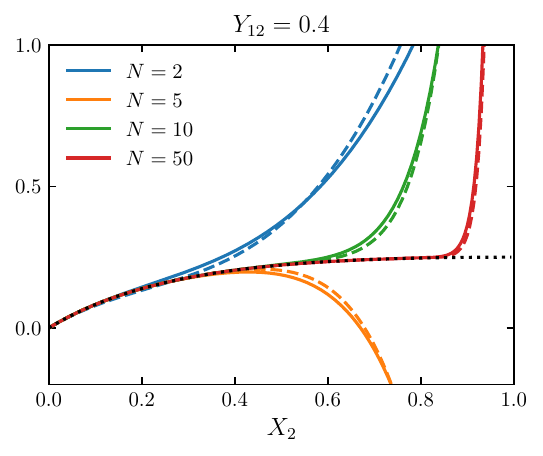}
    \end{subfigure}%
    \begin{subfigure}{.5\textwidth}
        \centering
        \includegraphics[width=1\linewidth]{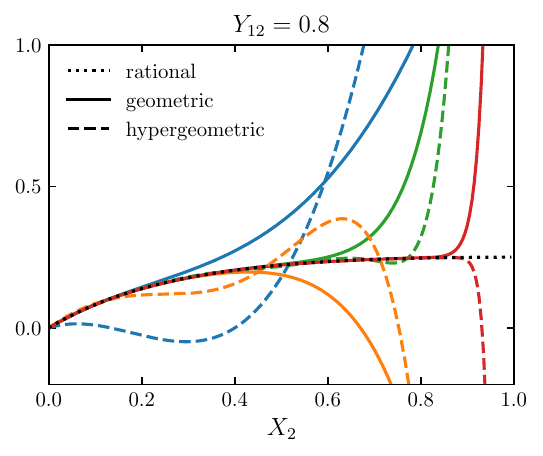}
    \end{subfigure}
   \caption{The hypergeometric magical identity~\eqref{eq: two-site chain magic identity} for the two-site chain as a function of the vertex energy $X_2$ with $X_1=1$, for fixed internal energy $Y_{12}=0.4$ ({\it left panel}) and $Y_{12}=0.8$ ({\it right panel}), varying the number of terms $N=2, 5, 10, 50$ kept in the series. The dotted black line corresponds to the rational function on the left-hand side of~\eqref{eq: two-site chain magic identity}, the solid lines represent the geometric series in $X_2/X_1$, i.e.~the right-hand side of~\eqref{eq: 2site chain geometric}, and the dashed lines show the hypergeometric series, i.e.~right-hand side of~\eqref{eq: two-site chain magic identity}, which are the only $Y_{12}$-dependent representations. We have set $\tilde{p}_1=\tilde{p}_2=1$ for concreteness, but stress that the derived magical identities hold for general twists.}
  \label{fig: 2site chain magical identity}
\end{figure}

\vskip 4pt
We illustrate this magical identity for the two-site chain in Fig.~\ref{fig: 2site chain magical identity} as a function of the external energy $X_2$ (for a fixed root vertex energy $X_1$), when varying the truncation order $N$ of the series representations. As expected, in the limit $Y_{12}\to0$, the right-hand side of~\eqref{eq: two-site chain magic identity} smoothly reduces to the simple geometric expansion of the contracted graph in $X_2/X_1$ given in~\eqref{eq: 2site chain geometric}. For a finite value of $Y_{12}\neq 0$, though, both the geometric and hypergeometric series representations are genuinely distinct at finite truncation order, and only match in the limit $N\to\infty$. This agreement therefore relies on a highly non-trivial resummation of the series, the mathematical origin of which remains to be understood.

\subsubsection{Identities from full contractions} 

It is of course straightforward to generalise the above observation to graphs of arbitrary topology. To be precise, we can contract any graph to a single vertex by applying a differential operator for every edge of the graph. This way, we can find very similar relations between purely rational functions and sums of hypergeometric functions. Let us outline the respective algorithm below:

\vskip 4pt
Consider any graph $\G$ with a set of vertices $\{i\}$ and edges $\{(i,j)\}$. For any edge $(i,j)$ act with a corresponding differential operator on the graph, either with $\D_{ij}$ or $\D_{ji}$ (the result will be the same). This will lead to a differential equation of the form
\begin{equation}
\label{eq: diff eq magical identities full contractions}
   \left[ \prod_{(i,j)} \D_{ij} \right] \G = \vcenter{\hbox{
\begin{tikzpicture}[line width=1pt, scale=2, baseline={(current bounding box.center)}]
    \draw[fill=black] (0,0) circle (.05cm)  node[above=2mm] {$\textcolor{black}{(\sum X_i,\sum p_i)}$};
\end{tikzpicture}
    }} \,,
\end{equation}
which means that the graph is collapsed to a single vertex. The right-hand side can be written down immediately as
\begin{equation}
    \vcenter{\hbox{
\begin{tikzpicture}[line width=1pt, scale=2, baseline={(current bounding box.center)}]
    \draw[fill=black] (0,0) circle (.05cm)  node[above=2mm] {$\textcolor{black}{(\sum X_i,\sum p_i)}$};
\end{tikzpicture}
    }} = i\frac{e^{-i\frac{\pi}{2}(\sum_i p_i -d)} \Gamma[\sum_i p_i -d]}{\left(\sum_i X_i\right)^{\sum_i p_i -d}} \,.
\end{equation}
Note that $\G$ contains factorised and (partially) nested terms. Whenever there is a factorised term, the action of the respective differential operator will eliminate this term. Therefore, only the fully nested contribution produces the source term on the right-hand side. We can therefore view this equation as an equation for the fully nested part, for which the series solution is constructed from the spectral gluing algorithm: it is found by solving a system of linear equations for the coefficients of the off-shell ansatz~\eqref{eq: off-shell ansatz}.
Applying the differential operators simply removes the propagators $\prod 1/(x_{ij}^2+\mu_{ij}^2)$ from this off-shell ansatz. This yields the left-hand side of~\eqref{eq: diff eq magical identities full contractions}:
\begin{align}
    \left[ \prod_{(i,j)} \D_{ij} \right] \G =& i^{V} \left[\prod_{i=1}^V \frac{1}{X_i^{p_i-d}}\right]\left[ \prod_{j=1}^I \frac{\pi/4}{(X_jX_j')^\frac{d}{2}}\right] \left(4\pi i\right)^I \left[\prod_{i=1}^V \C_+^{(n_i)}(p_i)\right] \nonumber\\
    &\times\sum\limits_{\bm m \geq 0} \frac{(-1)^{|\bm m|}}{\bm m!} \Gamma\left[ \tilde p_{1\cdots V} +|\bm m|\right] \left[\prod_{\substack{i=1 \\i\neq r}}^V \left(\frac{X_i}{X_r}\right)^{\tilde p_i+m_i}\right] \nonumber\\ 
    &\times\left[ \prod_{i=1}^VF^{(n_i)}_C\left[\left.\begin{matrix} \frac{\tilde{p}_i+\xi_i}{2},\,  \frac{\tilde{p}_i+1+\xi_i}{2}\\ 1+\epsilon_{i1} x_{i1}, \ldots, 1+\epsilon_{i n_i}x_{i n_i} \end{matrix}\right\vert u_{i1}^2, \ldots, u_{in_i}^2\right]\right] \,,
\end{align}
where the $x_{ij}$ are the solutions to the linear system defined by the respective path matrix. Comparing this to the right-hand side, simplifying the numerical prefactors with the help of the relations $V=I+1$ and $\sum_i n_i =2I$, and rearranging the kinematic factors  leads to the following generalised magical identity:
\begin{equation}
    \boxed{
    \begin{aligned}
        &\frac{X_1^{\tilde p_1} \cdots X_V^{\tilde p_V}}{(X_1+\dots+X_V)^{\tilde p_1+\dots+\tilde p_V}} = \sum\limits_{\bm m \geq 0} \frac{(-1)^{|\bm m|}}{\bm m!} \left(\tilde p_{1\cdots V}\right)_{|\bm m|} \left[\prod_{\substack{i=1\\i\neq r}}^V \left(\frac{X_{i}}{X_{r}}\right)^{\tilde p_i +m_i}\right]\\ 
        &\times\left[ \prod_{i=1}^VF^{(n_i)}_C\left[\left.\begin{matrix} \frac{\tilde{p}_i+\xi_i}{2},\,  \frac{\tilde{p}_i+1+\xi_i}{2}\\ 1+\epsilon_{i1} x_{i1}, \ldots, 1+\epsilon_{i n_i}x_{i n_i} \end{matrix}\right\vert \left(\frac{Y_{i1}}{X_i}\right)^2, \ldots, \left(\frac{Y_{in_i}}{X_i}\right)^2\right]\right] \,.
    \end{aligned}
    }
\end{equation}
What is so interesting and \textit{magical} about this identity is not only that it reduces multiple infinite sums over multivariate hypergeometric functions to a simple rational function, but that it also has the peculiar property that the above equation stays true if we remove the Lauricella functions entirely:
\begin{equation}
    \boxed{
    \frac{X_1^{\tilde p_1} \cdots X_V^{\tilde p_V}}{(X_1+\dots+X_V)^{\tilde p_1+\dots+\tilde p_V}} = \sum\limits_{\bm m \geq 0} \frac{(-1)^{|\bm m|}}{\bm m!} \left(\tilde p_{1\cdots V}\right)_{|\bm m|} \left[\prod_{\substack{i=1\\i\neq r}}^V \left(\frac{X_{i}}{X_{r}}\right)^{\tilde p_i +m_i}\right]\,.
    }
\end{equation}
This is just the multi-series expansion of $(1+\sum_{i\neq r} X_i/X_r)^{-\tilde p_{1\cdots V}}$. One can show this directly or by noting that our magical identity depends on the internal momenta $Y_{ij}$ only through the Lauricella functions and must of course stay true in the soft limit $Y_{ij}\to 0$, which sends the Lauricella functions to unity.

\subsubsection*{Explicit examples}

We can use the explicit series representations for the fully nested contribution to derive magical identities for any graph topology we have considered. Doing this is straightforward: leave out the factors $1/(x_{ij}^2+\mu_{ij}^2)$, replace the factor $\Gamma[\sum_i \tilde p_i +|\bm m|]$ by the Pochhammer symbol $(\sum_i \tilde p_i)_{|\bm m|}$, leave out all the numerical and kinematic prefactors and equate everything to the respective rational function in the kinematic variables. 

\vskip 4pt
Since we already derived the identity for the two-site chain, let us continue with the three-site chain. If we take the root vertex to be vertex $r=1$, then the magical identity reads
\begin{equation}
    \boxed{
    \begin{aligned}
        \frac{X_1^{\tilde p_1} X_2^{\tilde p_2} X_3^{\tilde p_3}}{(X_1+X_2+X_3)^{\tilde p_{123}}} &= \sum\limits_{m_1,m_2=0}^\infty \frac{(-1)^{m_{12}}}{m_1!m_2!} (\tilde p_{123})_{m_{12}} \left(\frac{X_2}{X_1}\right)^{\tilde p_{2}+m_{1}} \left(\frac{X_3}{X_1}\right)^{\tilde p_3 +m_2} \\
        &\hspace*{-2.5cm}\times F^{(1)}_C\left[\left.\begin{matrix} \frac{\tilde{p}_{123}+m_{12}}{2},\,  \frac{\tilde{p}_{123}+1+m_{12}}{2}\\ 1+\tilde{p}_{23}+m_{12} \end{matrix}\right\vert \left(\frac{Y_{12}}{X_1}\right)^2\right]
    F^{(1)}_C\left[\left.\begin{matrix} \frac{-m_2}{2},\,  \frac{1-m_2}{2}\\ 1-\tilde{p}_3-m_2 \end{matrix}\right\vert \left(\frac{Y_{23}}{X_3}\right)^2\right]
          \\
        &\hspace*{-2.5cm}\times F^{(2)}_C\left[\left.\begin{matrix} \frac{-m_1}{2},\,  \frac{1-m_1}{2}\\ 1-\tilde{p}_{23}-m_{12}, 1+\tilde{p}_3+m_2 \end{matrix}\right\vert \left(\frac{Y_{12}}{X_2}\right)^2, \left(\frac{Y_{23}}{X_2}\right)^2\right]\,.
    \end{aligned}
    }
\end{equation}
Here, we are expanding in the region $X_2,X_3< X_1$ (consistent with the choice of the root vertex). If, instead, we choose $r=2$ as the root, the magical identity becomes
\begin{equation}
    \boxed{
    \begin{aligned}
        \frac{X_1^{\tilde p_1} X_2^{\tilde p_2} X_3^{\tilde p_3}}{(X_1+X_2+X_3)^{\tilde p_{123}}} &= \sum\limits_{m_1,m_2=0}^\infty \frac{(-1)^{m_{12}}}{m_1!m_2!} (\tilde p_{123})_{m_{12}} \left(\frac{X_1}{X_2}\right)^{\tilde p_1+m_{1}} \left(\frac{X_3}{X_2}\right)^{\tilde p_3 +m_2} \\
        &\hspace*{-1cm}\times F^{(1)}_C\left[\left.\begin{matrix} \frac{-m_{1}}{2},\,  \frac{1-m_{1}}{2}\\ 1-\tilde{p}_{1}-m_{1} \end{matrix}\right\vert \left(\frac{Y_{12}}{X_1}\right)^2\right]
    F^{(1)}_C\left[\left.\begin{matrix} \frac{-m_2}{2},\,  \frac{1-m_2}{2}\\ 1-\tilde{p}_3-m_2 \end{matrix}\right\vert \left(\frac{Y_{23}}{X_3}\right)^2\right] \\
        &\hspace*{-1cm}\times F^{(2)}_C\left[\left.\begin{matrix} \frac{\tilde p_{123}+m_{12}}{2},\,  \frac{\tilde p_{123}+1+m_{12}}{2}\\ 1+\tilde{p}_{1}+m_{1}, 1+\tilde{p}_3+m_2 \end{matrix}\right\vert \left(\frac{Y_{12}}{X_2}\right)^2, \left(\frac{Y_{23}}{X_2}\right)^2\right]\,.
    \end{aligned}
    }
\end{equation}
This is expanding in the region $X_1,X_3<X_2$. Notice that these magical identities hold even in partial soft limits $Y_{12}\to0$ or $Y_{23}\to0$ for which {\it some} Lauricella functions are set to one, and the double-soft limit $Y_{12}, Y_{23}\to0$, for which we recover the multi-geometric series.

\vskip 4pt
As another example, let us take a look at the $N$-site star where we take the root to be the central vertex, $r=0$. Using the previously derived result in Sec.~\ref{subsec: Nsite star}, we find the following magical identity
\begin{equation}
    \boxed{
    \begin{aligned}
        \frac{X_0^{\tilde p_0}X_1^{\tilde p_1}\cdots X_N^{\tilde p_N}}{(X_0+\cdots+X_N)^{\tilde p_{01\cdots N}}} &=\sum_{\bm{m}\in \mathbb{N}^{N}_{\geq0}}^{\infty} \frac{(-1)^{m_{1\cdots N}}}{m_1!\cdots m_N!} \, \left(\tilde{p}_{0\cdots N}\right)_{m_{1\cdots N}}\\
        &\hspace*{-2.5cm}\times \prod_{j=1}^N \left\{ \left(\frac{X_j}{X_0}\right)^{\tilde{p}_j+m_j} F^{(1)}_C\left[\left.\begin{matrix} \frac{-m_j}{2},\,  \frac{1-m_j}{2}\\ 1-\tilde{p}_j-m_j \end{matrix}\right\vert \left(\frac{Y_{0j}}{X_j}\right)^2\right]\right\} \\
        &\hspace*{-2.5cm}\times F^{(N)}_C\left[\left.\begin{matrix} \frac{\tilde{p}_{0\cdots N}+m_{1\cdots N}}{2},\,  \frac{\tilde{p}_{0\cdots N}+1+m_{1\cdots N}}{2}\\ 1+\tilde{p}_1+m_1, \ldots, 1+\tilde{p}_N+m_N \end{matrix}\right\vert \left(\frac{Y_{01}}{X_0}\right)^2, \ldots, \left(\frac{Y_{0N}}{X_0}\right)^2\right]\,.
    \end{aligned}
    }
\end{equation}
Again, it is magical that such an identity holds even in partial soft limits, for which some of the internal energies vanish. 

\subsubsection{Identities from partial contractions} 

There is of course also the possibility to contract graphs only partially and to derive identities between Lauricella functions of different orders. This is done by applying differential operators to only some of the edges of a given graph. Let us discuss this procedure for the example of the three-site chain where we fix the root vertex to be the middle one ($r=2$). We apply the differential operator $\D_{23}$ to our solution~\eqref{eq: 3site chain mid root} which removes the factor $1/[(\tilde p_3+m_2)^2+\mu_{23}^2]$:
\begin{equation}
    \begin{aligned}
        \D_{23} \G(\begin{tikzpicture}[scale=0.8, baseline=-0.5ex]
            \coordinate (A) at (0,0);
            \coordinate (B) at (0.6,0);
            \coordinate (C) at (1.2,0);
            \draw[thick, postaction={decorate}, decoration={markings, mark=at position 0.6 with {\arrow[pyred]{<}}}] (A) -- (B);
            \draw[thick, postaction={decorate}, decoration={markings, mark=at position 0.6 with {\arrow[pyred]{>}}}] (B) -- (C);
            \fill[black] (A) circle (2pt);
            \fill[pyred] (B) circle (2pt);
            \fill[black] (C) circle (2pt);
        \end{tikzpicture}) &= \frac{ie^{-i\frac{\pi}{2}\tilde p_{123}}}{X_1^{\tilde p_1} X_2^{\tilde p_2} X_3^{\tilde p_3}} \sum\limits_{m_1,m_2=0}^\infty \frac{(-1)^{m_{12}}}{m_1!m_2!}\frac{\Gamma[\tilde p_{123}+m_{12}]}{(\tilde p_1+m_1)^2+\mu_{12}^2} \left(\frac{X_1}{X_2}\right)^{\tilde p_1+m_1}\left(\frac{X_3}{X_2}\right)^{\tilde p_3+m_2}  \\
        &\times F^{(1)}_C\left[\left.\begin{matrix} \frac{-m_{1}}{2},\,  \frac{1-m_{1}}{2}\\ 1-\tilde{p}_{1}-m_{1} \end{matrix}\right\vert \left(\frac{Y_{12}}{X_1}\right)^2\right]
    F^{(1)}_C\left[\left.\begin{matrix} \frac{-m_2}{2},\,  \frac{1-m_2}{2}\\ 1-\tilde{p}_3-m_2 \end{matrix}\right\vert \left(\frac{Y_{23}}{X_3}\right)^2\right]
          \\
        &\times F^{(2)}_C\left[\left.\begin{matrix} \frac{\tilde p_{123}+m_{12}}{2},\,  \frac{\tilde p_{123}+1+m_{12}}{2}\\ 1+\tilde{p}_{1}+m_{1}, 1+\tilde{p}_3+m_2 \end{matrix}\right\vert \left(\frac{Y_{12}}{X_2}\right)^2, \left(\frac{Y_{23}}{X_2}\right)^2\right]\,.
    \end{aligned}
\end{equation}
On the right-hand side of the differential equation, we now have the contracted graph instead, i.e.~a two-site chain:
\begin{equation}
    \begin{aligned}
        C_{23}[\G(\begin{tikzpicture}[scale=0.8, baseline=-0.5ex]
            \coordinate (A) at (0,0);
            \coordinate (B) at (0.6,0);
            \coordinate (C) at (1.2,0);
            \draw[thick, postaction={decorate}, decoration={markings, mark=at position 0.6 with {\arrow[pyred]{<}}}] (A) -- (B);
            \draw[thick, postaction={decorate}, decoration={markings, mark=at position 0.6 with {\arrow[pyred]{>}}}] (B) -- (C);
            \fill[black] (A) circle (2pt);
            \fill[pyred] (B) circle (2pt);
            \fill[black] (C) circle (2pt);
        \end{tikzpicture})] = \G(\begin{tikzpicture}[scale=0.8, baseline=-0.5ex]
            \coordinate (A) at (0,0);
            \coordinate (B) at (0.6,0);
            \draw[thick, postaction={decorate}, decoration={markings, mark=at position 0.6 with {\arrow[pyred]{<}}}] (A) -- (B);
            \fill[black] (A) circle (2pt);
            \fill[pyred] (B) circle (2pt);
        \end{tikzpicture}) &= \frac{i e^{-i\frac{\pi}{2}\tilde p_{123}}}{X_1^{\tilde p_1}(X_2+X_3)^{\tilde p_{23}}} \sum\limits_{m_1=0}^\infty \frac{(-1)^{m_1}}{m_1!} \frac{\Gamma[\tilde p_{123}+m_1]}{(\tilde p_1+m_1)^2+\mu_{12}^2} \\
        &\times\left(\frac{X_1}{X_2+X_3}\right)^{\tilde p_1+m_1} F^{(1)}_C\left[\left.\begin{matrix} \frac{-m_{1}}{2},\,  \frac{1-m_{1}}{2}\\ 1-\tilde{p}_{1}-m_{1} \end{matrix}\right\vert \left(\frac{Y_{12}}{X_1}\right)^2\right]\\
        &\times F^{(1)}_C\left[\left.\begin{matrix} \frac{\tilde p_{123}+m_1}{2},\,  \frac{\tilde p_{123}+1+m_1}{2}\\ 1+\tilde{p}_1+m_1 \end{matrix}\right\vert \left(\frac{Y_{12}}{X_2+X_3}\right)^2\right]\,.
    \end{aligned}
\end{equation}
Comparing the double sum on the one side and the single sum on the other side term by term for the summation index $m_1$, we find the following magical identity:
\begin{equation}
    \boxed{
    \begin{aligned}
       &\frac{X_2^{\tilde p_{12}+m_1}X_3^{\tilde p_3}}{(X_2+X_3)^{\tilde p_{123}+m_1}}  F^{(1)}_C\left[\left.\begin{matrix} \frac{\tilde p_{123}+m_1}{2},\,  \frac{\tilde p_{123}+1+m_1}{2}\\ 1+\tilde{p}_1+m_1 \end{matrix}\right\vert \left(\frac{Y_{12}}{X_2+X_3}\right)^2\right] \\
       &= \sum\limits_{m_2=0}^\infty \frac{(-1)^{m_2}}{m_2!} \left(\tilde p_{123}+m_1\right)_{m_2} \left(\frac{X_3}{X_2}\right)^{\tilde p_3+m_2} F^{(1)}_C\left[\left.\begin{matrix} \frac{-m_2}{2},\,  \frac{1-m_2}{2}\\ 1-\tilde{p}_3-m_2 \end{matrix}\right\vert \left(\frac{Y_{23}}{X_3}\right)^2\right]\\
        &\times F^{(2)}_C\left[\left.\begin{matrix} \frac{\tilde p_{123}+m_{12}}{2},\,  \frac{\tilde p_{123}+1+m_{12}}{2}\\ 1+\tilde{p}_{1}+m_{1}, 1+\tilde{p}_3+m_2 \end{matrix}\right\vert \left(\frac{Y_{12}}{X_2}\right)^2, \left(\frac{Y_{23}}{X_2}\right)^2\right] \,.
    \end{aligned}
    }
\end{equation}
Once again, this identity appears very mysterious (and therefore magical) to us, as: (i) it holds for any integer $m_1$ and any twists $\tilde{p}_{1, 2}$, (ii) it holds for partial soft limits $Y_{12}\to0$ or $Y_{23}\to0$, and (iii) it holds for the double soft limit $Y_{12}, Y_{23}\to0$. One can easily generalise this to graphs of more complicated topologies where a contraction of an edge with vertices of degree $n_1$ and $n_2$ relates a sum over $F_C^{(n_1)} \times F_C^{(n_2)}$ to a product of a rational function and a Lauricella function $F_C^{(n_1+n_2-2)}$ of a different order. Similarly, contracting multiple edges leads to magical identities involving multiple sums. 

\subsection{Large-mass expansion}
\label{subsec: large-mass expansion}

The combination of differential equations with these magical identities allows us to also write down expressions for large-mass expansions of cosmological correlators. For this, we can use the following expansion of the propagator:
\begin{equation}
    \frac{1}{x^2+\mu^2}=\frac{1}{\mu^2}\sum\limits_{n=0}^\infty \left(-\frac{x^2}{\mu^2}\right)^n\,.
\end{equation}
Now, we know that the functions appearing in the off-shell ansatz are eigenfunctions of the differential operators 
\begin{equation}
    \Delta_{ij} \equiv \vartheta_{ij}^2 -u_{ij}^2\left(\vartheta_{\{i\}}+\tilde p_i\right)\left(\vartheta_{\{i\}}+\tilde p_i+1\right)\,,
\end{equation}
so that $\Delta_{ij}\F^{(n_i)} (u_{ij})= x_{ij}^2 \F^{(n_i)}(u_{ij})$. Therefore, we can write
\begin{equation}
    \G = \frac{1}{\mu_{ij}^2}\sum\limits_{n=0}^\infty \left(-\frac{\Delta_{ij}}{\mu_{ij}^2}\right)^n C_{ij}[\G]\,.
\end{equation}
This is simply the inversion of the differential equation \eqref{eq: differential equations}. In particular, we can contract all the edges of a given graph to arrive at a single vertex. Then, the large-mass expansion is equivalent to taking derivatives of a rational function in the external kinematics, similar to the simple case discussed in \cite{Arkani-Hamed:2018kmz}:
\begin{equation}
    \G = 2\sin\left(\tfrac{\pi}{2}\left(p_{1\cdots V}-d\right)\right) \Gamma[p_{1\cdots V}-d] \prod_{(i,j)}\frac{1}{\mu_{ij}^2}\sum\limits_{n=0}^\infty \left(-\frac{\Delta_{ij}}{\mu_{ij}^2}\right)^n \frac{1}{X_{1\cdots V}^{p_{1\cdots V}-d}}\,.
\end{equation}
This, of course, requires us to first express the total energy pole of the contracted graph in terms of energy ratios $u_{ij}$, since the differential operators are defined with respect to internal line energies $Y_{ij}$ and not the external vertex energies $X_i$. To this end, we expand the contracted graph in such a way that we recover the normalisation of the original graph:
\begin{equation}
    \frac{1}{X_{1\cdots V}^{p_{1\cdots V}-d}} = \prod_{i=1}^V \frac{1}{X_i^{\tilde p_i}\left(1+\frac{X_{1\dots (i)\dots V}}{X_i}\right)^{\tilde p_i}}\,.
\end{equation}
Then, every ratio $X_j/X_i$ can be expressed as a telescoping product in terms of internal energies by moving along the path $\P_{j\to i}$ from vertex $j$ to vertex $i$ as
\begin{equation}
    \frac{X_j}{X_i}=\prod_{(k,\ell)\in\P_{j\to i}} \frac{u_{\ell k}}{u_{k \ell }}\,.
\end{equation} 
Notice that these are exactly the products that appear in the series expansions of the fully nested part if we express everything in terms of energy ratios.

\subsubsection*{Examples}

Let us discuss two examples. We start with the two-site chain, where the expansion was already mapped out in \cite{Arkani-Hamed:2018kmz}. We restrict to the case $p_1=p_2=2$ in order to compare to the known results. Then $\tilde p_1=\tilde p_2=\frac{1}{2}$ and we can rewrite
\begin{equation}
    \frac{1}{X_1+X_2} = \left[\frac{u_{12}u_{21}}{X_1 X_2}\right]^\frac{1}{2}\frac{1}{u_{12}+u_{21}}\,.
\end{equation}
Acting with the differential operator $\Delta_{12}$ on this expression leads to the following first two terms in the large-mass expansion:
\begin{align}
    \G \sim \frac{1}{\mu^2}\left[\frac{1}{X_1+X_2} - \frac{1}{\mu^2}\frac{X_1^2-6 X_1 X_2+X_2^2-8 Y^2}{4  (X_1+X_2)^3}+\dots\right]\,.
\end{align}
Although we had the choice of either using $\Delta_{12}$ or $\Delta_{21}$ for the expansion, the result is perfectly symmetric in the two external energies. 

\vskip 4pt
As a second example, we consider the three-site chain again and fix $r=1$ as the root vertex. For simplicity, we choose $p_1=p_3=2$ and $p_2=1$. Then, the total energy pole can be written as
\begin{equation}
    \frac{1}{(X_1+X_2+X_3)^2} = \frac{X_1^{-\frac12} X_2^{-1} X_3^{-\frac12}}{\left(1+\frac{u_{12}}{u_{21}}+\frac{u_{12}u_{23}}{u_{21}u_{32}}\right)^\frac12\left(1+\frac{u_{21}}{u_{12}}+\frac{u_{32}}{u_{23}}\right)\left(1+\frac{u_{32}}{u_{23}}+\frac{u_{32}u_{21}}{u_{23}u_{12}}\right)^\frac12}\,,
\end{equation}
and we find the expansion
\begin{align}
    \G &\sim \frac{1}{\mu_{12}^2\mu_{23}^2} \left(\frac{1}{(X_1+X_2+X_3)^2}\right. \nonumber \\
    &+\frac{1}{\mu_{12}^2}\frac{1}{4 X_1^2 (X_1+X_2+X_3)^4}\left[X_1^2 (9 X_1^2 - 14 X_1 (X_2 + X_3) + (X_2 + X_3)^2) \right.\nonumber \\
    & -\left.3 (21 X_1^2 + 18 X_1 (X_2 + X_3) + 5 (X_2 + X_3)^2) Y_{12}^2\right]+(1\leftrightarrow 3)+\dots\Big)\,.
\end{align}
It would be interesting to match the found rational functions of energies to explicit bulk (partial) contact interactions. A natural question, beyond the scope of this work, is whether we can reconstruct the full massive correlator by resumming such large-mass expansions. 

\newpage
\section{Conclusions \& Outlook}
\label{sec: conclusions}

The analytic structure of tree-level scattering amplitudes in flat space is strikingly simple. Once locality and unitarity are imposed, observables are rational functions of the kinematic invariants, with singularities dictated by particle propagation and factorisation. Adding masses enriches the spectrum of exchanged particles but leaves this underlying structure essentially unchanged. Cosmological observables, by contrast, are fundamentally more intricate, and develop a hypergeometric structure already at tree level, due to the distorted propagation of massive particles in an expanding universe. Understanding whether this apparent complexity conceals a simpler underlying mathematical structure has become one of the central questions in the study of cosmological correlators.

\vskip 4pt
In this paper, we have shed light on the mathematical structure of all-tree massive cosmological correlators. Despite their apparent complexity, we have shown that these observables can be constructed systematically from fundamental building blocks---vertex functions belonging to the family of Lauricella generalised hypergeometric functions---that are glued together through spectral integrals. We have designed a spectral gluing algorithm for maximally-nested analytic contributions that streamlines spectral integrations, and that follows purely from elementary combinatorial properties of the corresponding graphs. For every choice of root vertex, the algorithm produces a series representation in the corresponding region of kinematic space. These representations reveal a remarkably simple and universal organisation of massive correlators. The non-trivial bulk dynamics associated with spontaneous production and subsequent evolution of massive modes is captured by simple rational spectral propagators, whose poles encode the exchange of on-shell particles, closely paralleling the role of propagators in flat-space perturbation theory. The remaining kinematic dependence naturally separates into a geometric series in ratios of external energies, valid within the chosen convergence domain, together with generalised hypergeometric functions that resum the dependence on internal energies as one moves away from internal soft limits. What initially appears as a highly complicated collection of nested time integrals is therefore reorganised into a surprisingly rigid hierarchy of familiar mathematical structures.

\vskip 4pt
Consistent bulk time evolution is dictated by a set of Lauricella-type partial differential equations in the external energies obeyed by the boundary massive correlators. We have shown that the spectral gluing algorithm naturally constructs solutions to these equations by expanding correlators in eigenfunctions of the corresponding differential operators. Remarkably, and indeed by construction, acting with these graph annihilators on our series representations simultaneously removes the dynamical on-shell propagators, and collapses the corresponding internal lines of the graph. This collapsing mechanism gives rise to a miraculous new class of ``magical'' identities among generalised hypergeometric functions, relating them to correlators of simpler graph topologies. This reveals a striking hidden simplicity of massive correlators: once the bulk dynamics has been stripped away, the intricate hypergeometric structure collapses onto simple rational functions. Whether the full hypergeometric structure of massive cosmological correlators can ultimately be reconstructed from these fully collapsed rational graphs remains an intriguing open question.

\vskip 4pt
The structure uncovered in this work points towards a richer mathematical framework that remains to be fully understood. Uncovering this structure remains an exciting challenge. We end by highlighting several promising directions for future exploration, spanning both physical and mathematical aspects.

\subsubsection*{Physics}

We have worked within an idealised setup of tree-level correlators of conformally coupled external scalars with arbitrary internal massive exchange. The structures uncovered here suggest that the same techniques should generalise naturally to deformed or broader settings.

\begin{itemize}
    \item {\it Degenerate limit \& marginal vertices}: Many correlators of phenomenological interest involve boost-breaking linear mixings. In tree-level graphs, such mixing can be implemented by taking a folded limit at the corresponding vertex, which renders it marginal. This is the case, for instance, for the well-known triple-exchange bispectrum~\cite{Chen:2009zp}, for which no closed-form expression that converges in the entire physical kinematic region is currently known. While imposing the Bunch-Davies vacuum ensures that individual vertex functions remain finite and free of folded singularities (see Sec.~\ref{subsec: single-massive vertex function}), gluing them can nevertheless generate spurious divergences that must cancel in the full result. At present, it is unclear how this cancellation can be implemented systematically. More generally, it would be interesting to understand how the spectral gluing algorithm can be deformed or extended to accommodate marginal vertices.
    
    \item {\it Conformally coupled limit}: Most of the neat properties and mathematical structure are known for the idealised and somewhat special physical setup of conformally coupled scalars in power-law cosmology. In particular, these correlators admit twisted integral representations with rational integrands, providing an ideal framework for twisted cohomology. At tree level, they are known to be expressible in terms of polylogarithms (in the case of a de Sitter twist). Our generalised hypergeometric series representations should therefore reduce to these simpler polylogarithmic functions in the conformally coupled limit. Establishing this reduction explicitly, together with understanding the cancellation of spurious infrared divergences and developing systematic resummation techniques from the hypergeometric series to polylogarithms, would further clarify the found mathematical structure.

    \item {\it Breaking de Sitter symmetry}: By nature, inflation necessarily breaks the de Sitter symmetry: it must eventually end, and the evolution of the inflaton singles out a preferred time direction. In the simplest slow-roll scenarios, this mild breaking of de Sitter boosts gives rise to primordial non-Gaussianities with small amplitudes, close to the gravitational floor. More generally, however, the amplitude of primordial non-Gaussianities is intimately tied to the degree of symmetry breaking. It would therefore be interesting to extend the spectral gluing algorithm to scenarios with broken de Sitter boosts, such as models with a reduced sound speed and/or with exchanged particles coupled to a helical chemical potential (see e.g.~\cite{Qin:2025xct}).

    \item {\it Loops}: A central challenge is to understand how to systematically perform loop integrals in cosmology. In the spirit of the tree theorem, loop integrands can be generated by shifting and integrating over external energies of tree-level correlators. The hypergeometric series representations obtained from the spectral gluing algorithm therefore naturally provide the fundamental building blocks for constructing massive loop integrands. This suggests that an extended version of the spectral gluing algorithm exists for massive loop integrands. More broadly, many of the combinatorial structures underlying the present construction are intrinsically tree-level, and it remains unclear how they should be generalised to loop diagrams.

    \item {\it Hidden simplicity}: We have seen that the complicated hypergeometric structure of massive correlators away from soft limits is far from arbitrary. Instead, it is completely determined by simple combinatorial rules, which ultimately originate from consistent bulk time evolution. Conversely, stripping away the dynamical propagators (by acting with graph annihilators) reduces the series representations to remarkably simple rational functions. This echoes the role of the Parke-Taylor formula~\cite{Parke:1986gb}, which revealed an unexpected simplicity hidden beneath the complexity of Feynman-diagram computations, and ultimately led to the discovery of previously hidden symmetry~\cite{Drummond:2008vq}  and novel geometric structures~\cite{Arkani-Hamed:2013jha} governing scattering amplitudes. We hope that a similar story unfolds in cosmology, with the spectral gluing framework providing a first step towards uncovering the underlying symmetries and geometric principles that organise massive correlators.

    \item {\it Conformal correlators}: After Wick rotating to Euclidean time, vertex functions are nothing but conformal contact correlators, whose kinematic structure is largely fixed by conformal symmetry. While these correlators are well understood in position space, their momentum-space description has only recently begun to be explored. In particular, solving the conformal Ward identities reveals their underlying hypergeometric structure (see, e.g.,~\cite{Bzowski:2013sza, Coriano:2019sth}). Notably, $n$-point conformal correlators admit a representation as Euclidean Feynman integrals with massless propagators, where the loop momenta run along the edges of an $(n-1)$-simplex~\cite{Bzowski:2019kwd, Bzowski:2020kfw}. As these conformal correlators constitute the fundamental building blocks of massive de Sitter correlator, a deeper understanding of their mathematical structure is likely to provide new insights into the properties of cosmological observables.
\end{itemize}

\subsubsection*{Mathematics}

We expect that a more complete understanding of the uncovered rich mathematical structure would illuminate new aspects of the underlying physics.

\begin{itemize}
    \item {\it Analytic continuation}: The series representations of vertex functions in terms of type-$C$ Lauricella functions are expansions around soft kinematic limits and do not converge throughout the entire physical region. The situation is even more subtle for marginal vertices, where the Lauricella series converge everywhere except in the physical kinematic domain, making analytic continuation indispensable. Standard approaches based on evaluating the Mellin-Barnes integral representations by collecting residues generate series expansions around singular loci of the associated system of differential equations, but do not provide expansions valid at generic points in kinematic space. A particularly promising alternative is the method of brackets~\cite{gonzalez2008definiteintegralsmethodbrackets, gonzalez2010methodbrackets2examples}, which extends Ramanujan’s master theorem to several variables (see~\cite{Raman:2025tsg} for a recent application to massive correlators). We will discuss these analytic continuation techniques in future work.
    
    \item {\it GKZ systems}: The building blocks of massive correlators can be expressed in terms of multivariable hypergeometric functions of the Gel'fand, Kapranov \& Zelevinsky (GKZ) type~\cite{Gelfand:1990bua, Gelfand1989HypergeometricFA, Gelfand1991HYPERGEOMETRICFT, alma991007079196806161}. This formalism provides a powerful and systematic language in which to study their structure: it allows for the construction of weight-shifting operators, relates analytic properties of the functions to the geometry of the associated Newton polytope, and determines their differential equations directly from the $\A$-matrix, which encodes both their combinatorial and singularity structure. A comprehensive understanding of contact conformal correlators within the GKZ formalism remains an open and promising direction for future research (see~\cite{Caloro:2023cep} for a recent study, and~\cite{Grimm:2024tbg, Grimm:2025zhv} for applications to correlators of conformally coupled scalars in power-law cosmology).

    \item {\it Geometry}: The emergence of a universal hypergeometric pattern suggests that the complicated integral representations of cosmological perturbation theory may not be the most fundamental description of these observables. Rather, they hint at the existence of an intrinsic mathematical object in which massive correlators are defined directly, with the spectral gluing combinatorics emerging as natural structural principles rather than consequences of explicit time integrations. Summing over graphs or channels often reveals an underlying geometric structure of the associahedron type~\cite{Arkani-Hamed:2017mur}. We saw that different series representations can be constructed from distinct choices of rooted trees associated with the same graph. This naturally suggests the existence of a richer geometric organisation, in which different choices of roots correspond to different facets, each providing a valid series representation in a specific region of kinematic space. More broadly, it is tempting to speculate that such a geometric picture could also shed light on hidden zeros of massive correlators, extending observations made in the conformally coupled case~\cite{De:2025bmf, Li:2026gns}.

    \item {\it Massive cosmological symbols}: Feynman integrals can often be recast as iterated integrals~\cite{Chen1971} (see~\cite{hain2001iteratedintegralsalgebraiccycles, Brown:2013qva, Bogner:2012dn} for modern reviews), with polylogarithmic functions being the primary example. Symbol techniques provide a compact way to encode their algebraic and analytic structure. Instead of working with complicated transcendental functions, the symbol maps an integral to an ordered tensor product of its ``letters'', making functional identities, branch-cut and sequential discontinuity structure, and differential relations much more transparent. In cosmology, since the Lauricella building blocks are themselves defined through iterated constructions and carry a natural notion of transcendental weight, it is tempting to generalise these ideas to define ``massive cosmological symbols''. Such a framework would provide a systematic way to manipulate generalised hypergeometric series representations and, ultimately turn functional equations---such as the discovered magical identities---into linear algebra problems at the symbol level.
\end{itemize}

\paragraph{Acknowledgments.} We are grateful to Daniel Baumann, Hadrien Brochet, Chandramouli Chowdhury, Saiei Matsubara, Joe Marshall, Matthias Nowinski, Prashanth Raman, Ivo Sachs, Cristian Vergu and Zhong-Zhi Xianyu for insightful discussions. DW is funded by the Deutsche Forschungsgemeinschaft (DFG, German Research Foundation) under Germany’s Excellence Strategy – EXC-2094/2–390783311. The research of DW is partially funded by the European Union (ERC, \raisebox{-2pt}{\includegraphics[height=0.9\baselineskip]{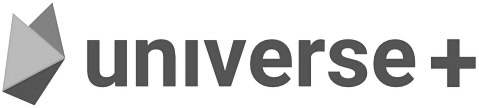}}, 101118787). {\small Views and opinions expressed are however those of the author(s) only and do not necessarily reflect those of the European Union or the European Research Council Executive Agency. Neither the European Union nor the granting authority can be held responsible for them.}

\newpage
\appendix
\section{Special Functions}
\label{sec: special functions}

As made manifest by the spectral gluing algorithm, all tree-level massive cosmological graphs can be represented as series of a plethora of multivariate generalised hypergeometric functions. In this appendix, we collect useful definitions, notations, and special functions used in the main text. 

\paragraph{$\Gamma$-function.} To avoid cluttered notations throughout the main text, we introduce the following shorthand notation for the product of $\Gamma$-functions:
\begin{equation}
    \Gamma[a_1, \ldots, a_n] \equiv \Gamma(a_1) \cdots \Gamma(a_n) \,.
\end{equation}
It proved useful to also define the following notation for the ratio of $\Gamma$-functions:
\begin{equation}
    \Gamma\left[\begin{matrix} a_1, \ldots, a_m \\ b_1, \ldots, b_n \end{matrix}\right] \equiv \frac{\Gamma[a_1, \ldots, a_m]}{\Gamma[b_1, \ldots, b_n]} \,.
\end{equation}
We have extensively used the Pochhammer symbol $(a)_n$ (sometimes called the rising factorial), defined as
\begin{equation}
    (a)_n \equiv \frac{\Gamma(a+n)}{\Gamma(n)} = a (a+1) (a+2) \cdots (a+n-1) \,,
\end{equation}
with $(a)_0=1$. Notice that when $a=-m$, with $m\in \mathbb{N}$, the Pochammer symbols becomes
\begin{equation}
    (-m)_n = \frac{(-1)^n m!}{(m-n)!}=(-m)(-m+1) \cdots (-m+n-1)\,,
\end{equation}
which vanishes for $n>m$, since one of the factors hits zero. This termination property is crucial as, e.g., a hypergeometric function with (at least one) negative-integer parameter truncates to a finite polynomial.

\subsection{Hypergeometric ${}_2F_1$ function}

Gau\ss's hypergeometric function is defined by
\begin{equation}
    {}_2F_1\left[\left.\begin{matrix} a,\,  b\\ c \end{matrix}\right\vert z \right] \equiv \sum_{n=0}^{\infty} \frac{(a)_n (b)_n}{(c)_n} \frac{z^n}{n!} \,,
\end{equation}
where $z \in \mathbb{C}$ is a variable, $a, b, c \in \mathbb{C}$ are parameters (with the condition $c\notin \mathbb{Z}_{\leq0}$). The series converges in the unit disc $\{z\in\mathbb{C} \, | \, |z|<1\}$. On the circle, the series converges absolutely when $\Re(c-a-b)>0$. We define the corresponding regularised function by the bar decoration:
\begin{equation}
    {}_2\bar{F}_1\left[\left.\begin{matrix} a,\,  b\\ c \end{matrix}\right\vert z \right] = \frac{1}{\Gamma(c)} \,{}_2F_1\left[\left.\begin{matrix} a,\,  b\\ c \end{matrix}\right\vert z \right] \,.
\end{equation}
The series terminates if either the parameter $a$ or $b$ is a non-positive integer, in which case the function reduces to a polynomial:
\begin{equation}
    {}_2F_1\left[\left.\begin{matrix} -m,\,  b\\ c \end{matrix}\right\vert z \right] = \sum_{n=0}^m (-1)^n \binom{m}{n} \frac{(b)_n}{(c)_n} z^n \,,
\end{equation}
with $m\in\mathbb{N}_{\geq0}$. A consequence of the above statement is:
\begin{equation}
    {}_2F_1\left[\left.\begin{matrix} \frac{-m}{2},\,  \frac{1-m}{2}\\ c \end{matrix}\right\vert z \right] = \sum_{n=0}^{\infty} \frac{(-m)_{2n}}{4^n (c)_n n!} z^n \,,
\end{equation}
where we have used the Legendre duplication formula for Pochhammer symbols
\begin{equation}
    \left(\frac{-m}{2}\right)_n \left(\frac{1-m}{2}\right)_n = \frac{(-m)_{2n}}{4^n} \,.
\end{equation}
For $m\geq0$ non-negative integer, $(-m)_{2n}$ vanishes as soon as $n\geq \lfloor m/2\rfloor+1$. The series therefore truncates to a polynomial of degree $\lfloor m/2\rfloor$. This truncation lowers the transcendental weight by unity and appears in the spectral gluing algorithm.

\subsubsection*{Differential equation}

It is well known that the ${}_2F_1$ function is annihilated by the following hypergeometric differential operator
\begin{equation}
    \D_z = z(1-z) \partial_z^2 + [c-(a+b+1)z]\partial_z - ab \,.
\end{equation}
The corresponding differential equation, $\D_z f=0$, has three singular loci: $0, 1$ and $\infty$, around which analytic continuations in the form of series representations can be found.

\subsubsection*{Analytic continuation}

Gau\ss's hypergeometric function satisfies the following linear transformations, which are used for analytic continuation:
\begin{equation}
\label{eq: 2F1 analytic continuation formula}
    \begin{aligned}
        \frac{\sin(\pi(b-a))}{\pi}{}_2\bar{F}_1\left[\left.\begin{matrix} a,\,  b\\ c \end{matrix}\right\vert z \right] &= \frac{(-z)^a}{\Gamma[b, c-a]} {}_2\bar{F}_1\left[\left.\begin{matrix} a,\,  a-c+1\\ a-b+1 \end{matrix}\right\vert \frac{1}{z} \right] - (a \leftrightarrow b) \,, \\
        \frac{\sin(\pi(c-a-b))}{\pi}{}_2\bar{F}_1\left[\left.\begin{matrix} a,\,  b\\ c \end{matrix}\right\vert z \right] &= \frac{1}{\Gamma[c-a, c-b]} {}_2\bar{F}_1\left[\left.\begin{matrix} a,\,  b\\ a+b-c+1 \end{matrix}\right\vert 1-z \right] \\
        &\,\,\,\,\,- \frac{(1-z)^{c-a-b}}{\Gamma[a, b]} {}_2\bar{F}_1\left[\left.\begin{matrix} c-a,\,  c-b\\ c-a-b+1 \end{matrix}\right\vert 1-z \right] \,.
    \end{aligned}
\end{equation}

\subsubsection*{Integral representations}

If $a, b \notin \mathbb{N}_{\leq0}$, then the regularised Gau\ss's hypergeometric function admits the following Mellin-Barnes integral representation:
\begin{equation}
    {}_2\bar{F}_1\left[\left.\begin{matrix} a,\,  b\\ c \end{matrix}\right\vert z \right] = \frac{1}{\Gamma[a, b]} \int_{-i\infty}^{+i\infty} \frac{\d s}{2i\pi} \, \Gamma\left[\begin{matrix} a+s, b+s, -s \\ c+s \end{matrix}\right](-z)^s \,.
\end{equation}
Other integral representations, especially of the Euler-type, can be found in e.g.~\cite{Olver:2010ouy}.

\subsection{Appell $F_4$ function}
\label{subsec: Appell F4}

The hypergeometric Appell $F_4$ has the following series representation around the origin $(z_1, z_2)=(0, 0)$:
\begin{equation}
\label{eq: Appell F4 series rep}
    F_4\left[\left.\begin{matrix} a,\,  b\\ c_1, \, c_2 \end{matrix}\right\vert z_1, z_2\right] \equiv \sum_{m_1, m_2=0}^\infty \frac{(a)_{m+n}(b)_{m+n}}{(c_1)_m (c_2)_n} \frac{z_1^m}{m!} \frac{z_2^n}{n!} \,,
\end{equation}
where the condition that $c_1, c_2 \notin \mathbb{N}_{\leq0}$ is assumed.
The Appell $F_4$ series is one of the four Appell series which are double-series generalisations of the Gau\ss's  hypergeometric series ${}_2F_1$, see~\cite{Srivastava1985MultipleGH} for more details. The convergence region of the $F_4$ series is $\sqrt{|z_1|}+\sqrt{|z_2|}<1$. For real positive values of $z_1$ and $z_2$, this is a quite restricted area in the $(z_1, z_2)$-plane. We define the regularised Appell $F_4$ function as
\begin{equation}
    \bar{F}_4\left[\left.\begin{matrix} a,\,  b\\ c_1, \, c_2 \end{matrix}\right\vert z_1, z_2\right] = \frac{1}{\Gamma[c_1, c_2]} F_4\left[\left.\begin{matrix} a,\,  b\\ c_1, \, c_2 \end{matrix}\right\vert z_1, z_2\right] \,,
\end{equation}
which is analytic in the entire complex planes of the parameters $a, b, c_1$ and $c_2$.

\subsubsection*{Differential equation}

The function $F_4$ is annihilated by the following differential operators~\cite{HExton_1995}
\begin{equation}
\label{eq: Appell F4 PDE system}
    \begin{aligned}
        \D_{z_1} &= z_1(1-z_1)\frac{\partial^2}{\partial z_1^2} - z_2^2\frac{\partial^2}{\partial z_2^2} - 2z_1z_2 \frac{\partial^2}{\partial z_1 \partial z_2} + (c_1-rz_1)\frac{\partial}{\partial z_1} - rz_2 \frac{\partial}{\partial z_2} - ab \,, \\
        \D_{z_2} &= z_2(1-z_2)\frac{\partial^2}{\partial z_2^2} - z_1^2\frac{\partial^2}{\partial z_1^2} - 2z_1z_2 \frac{\partial^2}{\partial z_1 \partial z_2} + (c_2-rz_2)\frac{\partial}{\partial z_2} - rz_1 \frac{\partial}{\partial z_1} - ab \,,
    \end{aligned}
\end{equation}
where $r\equiv a+b+1$. The corresponding (partial) differential equations are $\D_{z_1}f=0$ and $\D_{z_2}f=0$. This system can be written as a linear system of integrable total differential equations:
\begin{equation}
    \d \bm{f}(z_1, z_2) = \Omega(z_1, z_2) \bm{f}(z_1, z_2)\,,
\end{equation}
with $\bm{f} \equiv (f, \partial f/\partial z_1, \partial f/\partial z_2, \partial^2f/\partial z_1\partial z_2)^T$, and where $\Omega$ is a $4\times4$ matrix and each of its components $\Omega_{ij}$ ($1\leq i, j\leq4$) is a differential $1$-form with rational function coefficients which have poles only on $S=S_1\cup S_2\cup S_3$, where
\begin{equation}
    \begin{aligned}
        S_1 &= \{(z_1, z_2)\in \mathbb{C}^2|z_1=0\} \,, \\
        S_2 &= \{(z_1, z_2)\in \mathbb{C}^2|z_2=0\} \,, \\
        S_3 &= \{(z_1, z_2)\in \mathbb{C}^2|D(z_1, z_2)=0\} \,.
    \end{aligned}
\end{equation}
Here, $D(z_1, z_2)$ denotes the polynomial $1-2(z_1+z_2)+(z_1-z_2)^2$. As a consequence, any solution of the system~\eqref{eq: Appell F4 PDE system} that is holomorphic in an open set of $\mathbb{C}^2-S$ can be analytically continued along any curve in $\mathbb{C}^2-S$, and the family of all solutions form a four-dimensional vector space.

\subsubsection*{Analytic continuation}

The system~\eqref{eq: Appell F4 PDE system} has three singular loci: $(0, 0), (0, \infty)$ and $(\infty, 0)$, around which simple analytic continuations can be found. For instance, near $(0, \infty)$, we have~\cite{AppellFonctionsHE}
\begin{equation}
    \begin{aligned}
        F_4\left[\left.\begin{matrix} a,\,  b\\ c_1, \, c_2 \end{matrix}\right\vert z_1, z_2\right] &= \Gamma\left[\begin{matrix} b-a, c_2 \\ b, c_2-a \end{matrix}\right] (-z_2)^{-a} F_4\left[\left.\begin{matrix} a,\,  a-c_2+1\\ c_1, \, a-b+1 \end{matrix}\right\vert \frac{z_1}{z_2}, \frac{1}{z_2}\right] \\
        &+ \Gamma\left[\begin{matrix} a-b, c_2 \\ a, c_2-b \end{matrix}\right] (-z_2)^{-b} F_4\left[\left.\begin{matrix} b,\,  b-c_2+1\\ c_1, \, b-a+1 \end{matrix}\right\vert \frac{z_1}{z_2}, \frac{1}{z_2}\right] \,,
    \end{aligned}
\end{equation}
where the series on the right-hand side converge in the region $\sqrt{|z_1/z_2|}+\sqrt{|1/z_2|}<1$, corresponding to the region (\textbf{II}) of Fig.~\ref{fig: kinematic space for V2}, with $z_1=u_1^2$ and $z_2=u_2^2$. Naturally, the analytic continuation which converges in the region (\textbf{III}) is found by exchanging $z_1\leftrightarrow z_2$ as well as $c_1 \leftrightarrow c_2$. These analytic continuations can be obtained by collecting residues of the following Mellin-Barnes representation:
\begin{equation}
    F_4\left[\left.\begin{matrix} a,\,  b\\ c_1, \, c_2 \end{matrix}\right\vert z_1, z_2\right] = \Gamma\left[\begin{matrix} c_1, c_2 \\ a, b \end{matrix}\right] \int\limits_{-i\infty}^{+i\infty} \frac{\d s_1}{2i\pi} \frac{\d s_2}{2i\pi} \Gamma\left[\begin{matrix} a+s_{12}, b+s_{12}, -s_1, -s_2 \\ c_1+s_1, c_2+s_2 \end{matrix}\right] (-z_1)^{s_1} (-z_2)^{s_2} \,,
\end{equation}
where the integration contour is defined as usual and separates left from right poles. The following reduction formula also provides the neat analytic continuation used in the main text:
\begin{equation}
    F_4\left[\left.\begin{matrix} a,\,  b\\ c, \, a+b-c+1 \end{matrix}\right\vert z_1(1-z_2), z_2(1-z_1)\right] = {}_2F_1\left[\left.\begin{matrix} a,\,  b\\ c \end{matrix}\right\vert z_1\right] {}_2F_1\left[\left.\begin{matrix} a,\,  b\\ a+b-c+1 \end{matrix}\right\vert z_2\right] \,,
\end{equation}
although it holds for a special combination of parameters only. An extension of this reduction formula is given by~\cite{Burchnall1940EXPANSIONSOA}
\begin{equation}
    \begin{aligned}
        F_4\left[\left.\begin{matrix} a,\,  b\\ c_1, \, c_2 \end{matrix}\right\vert z_1(1-z_2), z_2(1-z_1)\right] &= \sum_{k=0}^{+\infty} \frac{(a)_k (b)_k (a+b-c_1-c_2+1)_k}{(c_1)_n (c_2)_n k!} z_1^k z_2^k \\
        &\times {}_2F_1\left[\left.\begin{matrix} a+k,\,  b+k\\ c_1+k \end{matrix}\right\vert z_1\right] {}_2F_1\left[\left.\begin{matrix} a+k,\,  b+k\\ c_2+k \end{matrix}\right\vert z_2\right] \,.
    \end{aligned}
\end{equation}
These formula are related to the following symbolic form~\cite{bateman_1953_cnd32-h9x80}
\begin{equation}
    F_4\left[\left.\begin{matrix} a,\,  b\\ c_1, \, c_2 \end{matrix}\right\vert z_1, z_2\right] = \nabla(a)\nabla(b) \,\, {}_2F_1\left[\left.\begin{matrix} a,\,  b\\ c_1 \end{matrix}\right\vert z_1\right] {}_2F_1\left[\left.\begin{matrix} a,\,  b\\ c_2 \end{matrix}\right\vert z_2\right] \,,
\end{equation}
where $\nabla$ is the differential operator defined by
\begin{equation}
    \nabla(h) \equiv \Gamma\left[\begin{matrix} h, \theta_1+\theta_2+h \\ \theta_1+h, \theta_2+h \end{matrix}\right] \,,
\end{equation}
with $\theta_i \equiv z_i \partial/\partial z_i$ ($i=1, 2$) being the Euler operator. Additional analytic continuations have been obtained in~\cite{Alkofer:2008dt, Ananthanarayan:2020xut}.

\subsection{Lauricella $F_C^{(n)}$ functions}
\label{app: Lauricella Function}

Lauricella functions are multi-variate generalisations of Gau\ss's hypergeometric function, originally introduced by Lauricella in~\cite{Lauricella1893SulleFI}. One distinguishes four different types, labelled as $A$, $B$, $C$ and $D$. For our purpose, we are only interested in the type-$C$ Lauricella function $F_C^{(n)}$ which contains ${}_2F_1 = F_C^{(1)}$ and $F_4 = F_C^{(2)}$ as special cases. Using multi-index notation $\bm{k}=(k_1,\dots,k_n)$, $\bm{z}=(z_1,\dots,z_n)$ as well as $\bm{k}!= k_1!\cdots k_n!$ and $|\bm{k}|=k_1+\cdots+k_n$, the Lauricella function $F_C^{(n)}$ of type-$C$ and order $n$ is defined via the following series representation
\begin{equation}
\label{eq: Lauricella Fc def}
    F_C^{(n)}\left[\left.\begin{matrix}
            a,b \\c_1, \dots,c_n
        \end{matrix}\right\vert z_1, \dots, z_n\right] \equiv \sum\limits_{\bm k \geq 0} \frac{(a)_{|\bm{k}|}(b)_{|\bm{k}|}}{(c_1)_{k_1} \cdots (c_n)_{k_n}} \frac{\bm{z}^{\bm{k}}}{\bm{k}!} \,,
\end{equation}
for $c_i \notin \mathbb{N}_{\leq0}$ ($i=1, \ldots, n$). Each $k_i$ runs over positive integers so that the series has $n$ layers. The series $F_C^{(n)}$ converges in the domain
\begin{equation}
    D_C = \{ \bm{z} \in \mathbb{C}^n \, | \, \sum_{1\leq i \leq n}\sqrt{|z_i|} <1 \}\,,
\end{equation}
which is nothing but the simplex around in the origin. The series representation is naturally invariant under the symmetric group $S_n$ by simultaneous permutation of the variables and the parameters in $\mathbb{C}^n$. In the main text, we are mostly working with the regularised Lauricella function
\begin{equation}
\label{eq: Lauricella Fc regularised}
    \bar{F}_C^{(n)}\left[\left.\begin{matrix}
            a,b \\c_1, \dots,c_n
        \end{matrix}\right\vert z_1, \dots, z_n\right] = \frac{1}{\Gamma[c_1, \ldots, c_n]}F_C^{(n)}\left[\left.\begin{matrix}
            a,b \\c_1, \dots,c_n
        \end{matrix}\right\vert z_1, \dots, z_n\right] \,,
\end{equation}
which is analytic in $\mathbb{C}^n$ in the space of parameters $a, b, c_1, \ldots, c_n$.

\subsubsection*{Differential equation}

We can derive the system of differential equations satisfied by $F_C^{(n)}$ directly from its series representation. For that purpose, note that the coefficients of the series fulfill
\begin{equation}
    \frac{A(k_1,\dots, k_i+1,\dots, k_n)}{A(k_1,\dots,k_i,\dots,k_n)} = \frac{(a+|\bm{k}|)(b+|\bm{k}|)}{(c_i+k_i)(1+k_i)}\,.
\end{equation}
Using the Euler operators $\theta_i=z_i \partial/\partial z_i$ and $\theta = \sum_{i=1}^n\theta_i$ which act on the expansion arguments as $\theta_i \bm{z}^{\bm{k}} = k_i \bm{z}^{\bm k}$ and $\theta \bm{z}^{\bm k} = |\bm k| \bm{z}^{\bm k}$, we immediately find the Euler form of the differential system,
\begin{equation}
    \left[\theta_i(\theta_i+c_i-1) -z_i(\theta +a)(\theta+b)\right] F_C^{(n)} \left[\left.\begin{matrix}
            a,b \\c_1, \dots,c_n
        \end{matrix}\right\vert z_1, \dots, z_n\right] =0\,,
\end{equation}
for $i=1,\dots,n$. When dealing with quadratic arguments $z_i=u_i^2$ as in the vertex functions in the main text, we can simply replace the Euler operators $\theta_i = \frac12 \vartheta_i$ where $\vartheta_i \equiv u_i\partial/\partial u_i$. Then, we find the annihilators
\begin{equation}
    \D_i = \vartheta_i(\vartheta_i +2c_i-2) -u_i^2 (\vartheta +2a)(\vartheta+2b)\,,
\end{equation}
with $i=1, \ldots, n$.

\vskip 4pt
For completeness, let us explicitly write down annihilators without using Euler operators. Let $\mathbb{C}[\bm{z}, \bm{\partial}]$ be the ring generated by the polynomial ring $\mathbb{C}[\bm{z}]$ in $n$ variables $z_1, \ldots, z_n$ and the partial differential operators $\partial_1 \equiv \partial/\partial z_1, \ldots, \partial_n \equiv \partial/\partial z_n$, with relations $\partial_i\partial_j=\partial_j\partial_i$ ($1\leq i, j\leq n$) and 
\begin{equation}
    \partial_i f(\bm{z}) = \frac{\partial f(\bm{z})}{\partial z_i} + f(\bm{z}) \partial_i\,, \quad f(\bm{z})\in \mathbb{C}[\bm{z}] \,.
\end{equation}
Then, the following elements of $\mathbb{C}[\bm{z}, \bm{\partial}]$ annihilate Lauricella's hypergeometric series $F_C^{(n)}$:
\begin{equation}
    \begin{aligned}
        \D_i &= z_i(1-z_i)\partial_i^2 - \sum\limits_{\substack{1\leq j\leq n \\ j\neq i}} z_iz_j \partial_i\partial_j - \sum\limits_{\substack{1\leq j, k\leq n \\ j\neq i}} z_j z_k \partial_j \partial_k \\
        &+ [c_i-(a+b+1)z_i]\partial_i - (a+b+1) \sum\limits_{\substack{1\leq j\leq n \\ j\neq i}} z_j \partial_j - ab \,,
    \end{aligned}
\end{equation}
for $i=1, \ldots, n$. The left ideal $E_C$ of $\mathbb{C}[\bm{z}, \bm{\partial}]$ generated by them is called Lauricella's system of hypergeometric differential equations (of type $C$). It is known that the rank of the corresponding system of partial differential equations is $2^n$, and the singular locus is
\begin{equation}
    S_C = \left(\prod_{i=1}^n z_i \prod_{\epsilon_1, \ldots, \epsilon_n=\pm 1} \left(1+\sum_{i=1}^n \epsilon_i \sqrt{z_i}\right)=0\right) \subset \mathbb{C}^n \,.
\end{equation}
For example, $S_C = (z_1z_2(z_1^2+z_2^2-2z_1z_2-2z_1-2z_2+1)=0)$ for $n=2$. More details are given in~\cite{goto2017fundamentalgroupcomplementsingular}. Although the series $F_C^{(n)}$ is defined under the condition $c_1, \ldots, c_n \notin \mathbb{N}_{\leq0}$, the system $E_C$ is defined for any parameter values $c_1, \ldots, c_n\in \mathbb{C}$. Notice that the corresponding solutions around $z_i=\infty$ can be accessed through that around $z'_i=0$ by the variable change $z_i'=1/z_i$ and the relation $z_i\partial/\partial z_i = -z_i'\partial/\partial z'_i$. The system $E_C$ can thus be regarded as defined on $(\mathbb{P}^1)^n$, where $\mathbb{P}^1$ denotes the complex projective line. Finally, let us mention that the system $E_C$ can be transformed into a Pfaffian system
\begin{equation}
    \d \bm{f}(\bm{z}) = \Omega(\bm{z}) \bm{f}(\bm{z}) \,,
\end{equation}
with the integrability condition $\d\Omega(\bm{z})=\Omega(\bm{z})\wedge \Omega(\bm{z})$ where the entries of the connection matrix $\Omega(\bm{z})$ are rational $1$-forms in $\bm{z}$, and
\begin{equation}
    \bm{f}(\bm{z}) = (F, (\partial_i F), \ldots, (\partial_{i_1} \cdots \partial_{i_r} F)_{i_1<i_2<\cdots<i_r}, \ldots, \partial_{i_1}\cdots \partial_{i_n}F)^T \,.
\end{equation}
More details can be found in~\cite{Matsumoto_2020, Goto:2022uo}.

\subsubsection*{Reflexion formula}

An important feature of the spectral gluing algorithm is the pole cancellation among diagonal terms when multiplying individual modes of vertex functions. This is made possible thanks to the following reflexion formula for Lauricella functions, which, to our knowledge, was not known before.

\begin{lemma}[Lauricella reflexion formula]\label{lemma-Fc}
    Let $m\in\mathbb{N}_{>0}$ be a positive integer and let $F_C^{(n)}$ be the Lauricella function of type C. Then the following limit holds:
    \begin{align}
        \lim\limits_{c\to 1-m} &\frac{1}{\Gamma(c)} \, F_C^{(n)}\left[\left.\begin{matrix}
            p-\tfrac{m}{2}, p+\tfrac12-\tfrac{m}{2} \\c, \alpha_2, \dots,\alpha_n
        \end{matrix}\right\vert z_1, \dots, z_n\right] \nonumber \\
        &= \frac{\mathrm{\Gamma}(2p+m)}{\mathrm{\Gamma}(1+m)\mathrm{\Gamma}(2p-m)}\left(\frac{z_1}{4}\right)^m\, F_C^{(n)}\left[\left.\begin{matrix}
            p+\tfrac{m}{2}, p+\tfrac12+\tfrac{m}{2} \\1+m, \alpha_2, \dots,\alpha_n
        \end{matrix}\right\vert z_1, \dots, z_n\right] \,.
    \end{align}
\end{lemma}

\begin{proof}
    By definition of the Lauricella function, we have
    \begin{align}
        \lim\limits_{c\to 1-m} \frac{1}{\mathrm{\Gamma}(c)} F_C^{(n)}\left[\left.\begin{matrix}
            p-\tfrac{m}{2}, p+\tfrac12-\tfrac{m}{2} \\c, \alpha_2, \dots,\alpha_n
        \end{matrix}\right\vert z_1, \dots, z_n\right] \nonumber \\
        = \sum\limits_{\bm{k}\geq 0} \frac{\bm{z}^{\bm{k}}}{\bm{k}!} \frac{(p-\tfrac{m}{2})_{|\bm{k}|}(p+\tfrac12-\tfrac{m}{2})_{|\bm{k}|}}{\mathrm{\Gamma}(1-m+k_1)(\alpha_2)_{k_2}\cdots(\alpha_n)_{k_n}} \,.
    \end{align}
    Here, $\bm{z}=(z_1,\dots,z_n)\in\mathbb{R}^n$, $\bm{k}=(k_1,\dots,k_n)\in\mathbb{N}_{\geq}^n$, $|\bm{k}|=k_1+\dots+k_n$ and $\bm{k}! = k_1!\cdots k_n!$. Since $m$ is a positive integer, we have $1/\Gamma(1-m-k_1) = 0$ for $k_1\in\{0,\dots,m-1\}$. Thus, the sum can actually be restricted to $k_1\geq m$ and we can perform a shift of the index $k_1\to k_1' = k_1-m$. First, we immediately recognise $\Gamma(1-m+k_1)=\Gamma(1+k_1')=k_1'!$. Also, $z_1^{k_1} = z_1^{m+k_1'}$ and $k_1! = (m+k_1')!=(1+m)_{k_1'} \Gamma(1+m)$. Lastly, we can use identities for the Pochhammer symbols to write
    \begin{equation}
        (p-\tfrac{m}{2})_{k_1+\cdots +k_n} = (p-\tfrac{m}{2})_m(p+\tfrac{m}{2})_{k_1'+k_2+\cdots +k_n} \,,
    \end{equation}
    and similarly for $(p+\tfrac12-\tfrac{m}{2})_{|\bm{k}|}$. Finally, we relabel the index $k_1'\to k_1$ which now runs from 0 to $\infty$ again and obtain 
    \begin{equation}
        \frac{z_1^m}{m!}(p-\tfrac{m}{2})_m(p+\tfrac12-\tfrac{m}{2})_m \sum\limits_{\bm{k}\geq 0} \frac{\bm{z}^{\bm{k}}}{\bm{k}!}\frac{(p+\tfrac{m}{2})_{|\bm{k}|}(p+\tfrac12+\tfrac{m}{2})_{|\bm{k}|}}{(1+m)_{k_1}(\alpha_2)_{k_2}\cdots(\alpha_n)_{k_n}}.
    \end{equation}
    Using the duplication formula for the $\Gamma$-function then provides the r.h.s. of the equation.
\end{proof}

\begin{corollary}
Let $m\in\mathbb{N}_>$ be a positive integer and $\bar F_C^{(n)}$ be the regularised Lauricella function of type-$C$ and order $n$. Then,
\begin{align}
        \bar{F}_C^{(n)}\left[\left.\begin{matrix}
            p-\tfrac{m}{2}, p+\tfrac12-\tfrac{m}{2} \\1-m, \alpha_2, \dots,\alpha_n
        \end{matrix}\right\vert \bm{z}\right] 
        = \Gamma\left[\begin{matrix}
            2p+m \\2p-m
        \end{matrix}\right]\left(\frac{z_1}{4}\right)^m\, \bar{F}_C^{(n)}\left[\left.\begin{matrix}
            p+\tfrac{m}{2}, p+\tfrac12+\tfrac{m}{2} \\1+m, \alpha_2, \dots,\alpha_n
        \end{matrix}\right\vert \bm{z}\right].
    \end{align}
\end{corollary}
\noindent This result follows immediately from the previous lemma and is the master formula to make the pole cancellation within the spectral gluing algorithm manifest.

\section{Details on Spectral Gluing Algorithm}
\label{sec: details on spectral gluing algorithm}

In this appendix, we provide additional details on the spectral gluing algorithm for the two- and three-site chains. In particular, we explicitly write all the (partially) factorised contributions, as well as the on-shell pieces. 

\subsection{Two-site chain}
\label{app: two-site chain}

This well-studied graph has been the subject of several works, and various closed-form solution are available in the literature, see e.g.~\cite{Arkani-Hamed:2018kmz, Qin:2022fbv}. The full dimensionless graph receives $2^2=4$ contributions
\begin{equation}
    \widehat{\G}(X_1, X_2, Y) = \sum_{\aa, \bb=\pm} (i\aa)(i\bb) \, \widehat{\I}_{\aa\bb}(X_1, X_2, Y) \,,
\end{equation}
where only half of the dimensionless master integrals $\widehat{\I}_{\aa\bb}$ need to be evaluated explicitly, since the remaining ones follow from complex conjugation: $\widehat{\I}_{--} = \widehat{\I}_{++}^*$ and $\widehat{\I}_{+-} = \widehat{\I}_{-+}^*$. These integrals depend on the two energy ratios
\begin{equation}
    u_1 \equiv \frac{Y}{X_1} \,, \quad u_2 \equiv \frac{Y}{X_2} \,,
\end{equation}
which we consider positive but not necessarily bounded by unity. This choice serves a dual purpose: it demonstrates that our solution naturally accommodates reduced sound speeds, and provides a first example of internal momentum ratios that are not subject to any generalised triangle inequality. 

\paragraph{Factorised contribution.} The (fully) factorised contribution is simply given by the product of vertex functions:
\begin{equation}
        \widehat{\I}_{-+}(u_1, u_2) = 
        \raisebox{-5pt}{
        \begin{tikzpicture}[line width=1. pt, scale=2]
            \draw[black] (0, 0) -- (1, 0);
            \draw[white, line width=2pt] 
  ($(0,0)!1/3!(1,0)$) -- ($(0,0)!2/3!(1,0)$);
            \node at ($(0,0)!1/2!(1,0)$) {$\times$};

            \draw[fill=white] (0, 0) circle (.05cm);
            \draw[fill=black] (1, 0) circle (.05cm);
        \end{tikzpicture} 
        } 
        = \V_{-, \mu}^{(1)}(u_1; p_1) \, \V_{+, \mu}^{(1)}(u_2; p_2) \,.
\end{equation}

\paragraph{Nested contribution.} The (fully) nested contribution is given by the following spectral integral
\begin{equation}
    \begin{aligned}
        \widehat{\I}_{++}(u_1, u_2) &= e^{-\frac{i\pi}{2}} \int\limits_{-\infty}^{+\infty} [\d\nu]\rho_\nu(\mu) \left(
        \raisebox{-5pt}{
        \begin{tikzpicture}[line width=1. pt, scale=2]
            \draw[black] (0, 0) -- (1, 0);
            \draw[white, line width=2pt] 
  ($(0,0)!1/3!(1,0)$) -- ($(0,0)!2/3!(1,0)$);
            \node at ($(0,0)!1/2!(1,0)$) {$\times$};

            \draw[fill=black] (0, 0) circle (.05cm) node[above right=0.1mm] {$\nu$};
            \draw[fill=black] (1, 0) circle (.05cm) node[above left=0.1mm] {$\nu$};
        \end{tikzpicture} 
        }\right) \\
        &= e^{-\frac{i\pi}{2}} \int\limits_{-\infty}^{+\infty} [\d\nu]\rho_\nu(\mu) \, \V_{+, \nu}^{(1)}(u_1; p_1) \, \V_{+, \nu}^{(1)}(u_2; p_2) \,,
    \end{aligned}
\end{equation}
after factorising the integrand and setting the internal leg off shell $\mu \to \nu$. Recall that the (tree-level) spectral density is defined as $\rho_\nu(\mu) \equiv \frac{1}{(\nu^2-\mu^2)_{i\epsilon}}$. Performing the spectral integration reduces to collecting on-shell and analytic poles of the integrand. Without loss of generality, we restrict to the kinematic region $u_1 < u_2$; the complementary region follows immediately by exchanging the two labels.

\paragraph{On-shell poles.} As explained in Sec.~\ref{subsec: on-shell poles}, collecting on-shell poles amounts to projecting vertex functions onto positive- and negative-frequency modes with appropriate Boltzmann weights. Explicitly, we obtain
\begin{equation}
    \begin{aligned}
        &\widehat{\I}_{++}^{(\delta_{12})}(u_1, u_2) = e^{-\frac{i\pi}{2}} \int\limits_{-\infty}^{+\infty} [\d\nu]\rho_\nu(\mu) \, \C_+^{(1)}(p_1)\C_+^{(1)}(p_2)\\
        &\times\left\{\F^{(1)}_{+i\nu}(u_1; \tilde{p}_1)\F^{(1)}_{+i\nu}(u_2; \tilde{p}_2) + \F^{(1)}_{+i\nu}(u_1; \tilde{p}_1)\F^{(1)}_{-i\nu}(u_2; \tilde{p}_2) + (\nu \leftrightarrow -\nu)\right\} \,,
    \end{aligned}
\end{equation}
where the notation $\widehat{\I}_{++}^{(\delta_{12})}$ means that we only consider the on-shell ($\delta$) poles in the spectral integral gluing the first two ($12$) Schwinger-Keldysch indices, and where we have defined the modified twist $\tilde{p}_{1,2} \equiv p_{1,2}-d/2$. This is a general feature of the spectral gluing algorithm. In the kinematic region $u_1<u_2$, the second term is dominated by the positive-frequency mode $\F^{(1)}_{+i\nu}$ and the spectral integration contour is therefore closed to the lower-half complex $\nu$-plane. Projecting on the on-shell poles yields
\begin{equation}
    \begin{aligned}
        &\widehat{\I}_{++}^{(\delta_{12})}(u_1, u_2) = -\C_+^{(1)}(p_1)\C_+^{(1)}(p_2) \\
        &\times \left\{ \textcolor{pyblue}{e^{+\pi\mu}} \left[\F^{(1)}_{+i\mu}(u_1; \tilde{p}_1)\F^{(1)}_{+i\mu}(u_2; \tilde{p}_2) + \F^{(1)}_{+i\mu}(u_1; \tilde{p}_1)\F^{(1)}_{-i\mu}(u_2; \tilde{p}_2)\right] \right.\\
        &\left.\,\,\,\,\,+\textcolor{pyblue}{e^{-\pi\mu}} \left[\F^{(1)}_{-i\mu}(u_1; \tilde{p}_1)\F^{(1)}_{-i\mu}(u_2; \tilde{p}_2) + \F^{(1)}_{-i\mu}(u_1; \tilde{p}_1)\F^{(1)}_{+i\mu}(u_2; \tilde{p}_2)\right]\right\} \\
        &= - \C_+^{(1)}(p_1) \left[\textcolor{pyblue}{e^{+\pi\mu}}\F^{(1)}_{+i\mu}(u_1; \tilde{p}_1) + \textcolor{pyblue}{e^{-\pi\mu}}\F^{(1)}_{-i\mu}(u_1; \tilde{p}_1)\right] \V_{+, \mu}^{(1)}(u_2; p_2) \,.
    \end{aligned}
\end{equation}
We observe that the right vertex has been fully reconstructed, and that the left vertex acquired Boltzmann weights, coloured in \textcolor{pyblue}{blue}, on the positive- and negative-frequency modes. 

\paragraph{Analytic poles.} Collecting the analytic poles is more delicate and we will treat this case in full detail to illustrate the general pole cancellation mechanism discussed in Sec.~\ref{subsec: gluing adjacent vertices}. This contribution reads
\begin{equation}
    \begin{aligned}
        \widehat{\I}_{++}^{(\P_{12})}(u_1, u_2) = e^{-\frac{i\pi}{2}} &\int\limits_{-\infty}^{+\infty} [\d\nu]\rho_\nu(\mu) \C_+^{(1)}(p_1) \C_+^{(1)}(p_2) \\
        &\times ( \F^{(1)}_+\F^{(1)}_+ + \F^{(1)}_+\F^{(1)}_- + \F^{(1)}_-\F^{(1)}_+ + \F^{(1)}_-\F^{(1)}_-) \,,
    \end{aligned}
\end{equation}
where the integrands, using the short notation $\F^{(1)}_\pm$ to mean individual-mode vertex functions, read
\begin{equation}
    \begin{aligned}
        \F^{(1)}_\pm \F^{(1)}_\pm &= \Gamma[\textcolor{pyblue}{\mp i\nu}, \textcolor{pyred}{\mp i\nu}, \tilde{p}_1\pm i\nu, \tilde{p}_2\pm i\nu, \textcolor{pyred}{1\pm i\nu}, 1\pm i\nu] \\
        &\times\left(\frac{u_1 u_2}{2}\right)^{\pm i\nu} \bar{F}^{(1)}_C\left[\left.\begin{matrix} \frac{\tilde{p}_1\pm i\nu}{2},\,  \frac{\tilde{p}_1+1\pm i\nu}{2}\\ 1\pm i\nu \end{matrix}\right\vert u_1^2\right] \bar{F}^{(1)}_C\left[\left.\begin{matrix} \frac{\tilde{p}_2\pm i\nu}{2},\,  \frac{\tilde{p}_2+1\pm i\nu}{2}\\ 1\pm i\nu \end{matrix}\right\vert u_2^2\right] \,, \\
        \F^{(1)}_\pm \F^{(1)}_\mp &= \Gamma[\textcolor{pyred}{\mp i\nu}, \textcolor{pyred}{\pm i\nu}, \tilde{p}_1\pm i\nu, \textcolor{pyblue}{\tilde{p}_2\mp i\nu}, 1\pm i\nu, \textcolor{pyblue}{1\mp i\nu}] \\
        &\times\left(\frac{u_1}{u_2}\right)^{\pm i\nu} \bar{F}^{(1)}_C\left[\left.\begin{matrix} \frac{\tilde{p}_1\pm i\nu}{2},\,  \frac{\tilde{p}_1+1\pm i\nu}{2}\\ 1\pm i\nu \end{matrix}\right\vert u_1^2\right] \bar{F}^{(1)}_C\left[\left.\begin{matrix} \frac{\tilde{p}_2\mp i\nu}{2},\,  \frac{\tilde{p}_2+1\mp i\nu}{2}\\ 1\mp i\nu \end{matrix}\right\vert u_2^2\right] \,.
    \end{aligned}   
\end{equation}
We use the notation $\widehat{\I}_{++}^{(\P_{12})}$ to mean that we only consider the analytic ($\P$) poles in the spectral integral gluing the first two ($12$) Schwinger-Keldysch indices. We henceforth focus exclusively on the contributions $\F^{(1)}_+ \F^{(1)}_+$ and $\F^{(1)}_+ \F^{(1)}_-$, as the remaining two are related to these by shadow symmetry and yield identical contributions. The regularised Lauricella $\bar{F}_C^{(1)} \equiv {}_2\bar{F}_1$ is analytic in the complex $\nu$-plane, so that all the analytic poles are given by the $\Gamma$-functions. In the regime where $u_1 u_2<4$, we close the contour in the lower-half plane. We have coloured in \textcolor{pyblue}{blue} terms that contain analytic poles. Other contributions in \textcolor{pyred}{red} cancel against the spectral density $\N_\nu = 1/\Gamma[\pm i \nu]$. Collecting the infinite tower of poles located at $\nu=-im$, with $m\in \mathbb{N}$, in the $\F^{(1)}_+ \F^{(1)}_+$ term yields the following series:
\begin{equation}
\label{eq: 2site chain series 1}
    \begin{aligned}
        2\pi &e^{-\frac{i\pi}{2}} \C_+^{(1)}(p_1) \C_+^{(1)}(p_2) \sum_{m=0}^\infty \frac{m^2}{m^2 + \mu^2} \frac{(-1)^m}{m!} \Gamma[\tilde{p}_1+m, \tilde{p}_2+m, m] \\
        &\times \left(\frac{u_1 u_2}{4}\right)^m \bar{F}^{(1)}_C\left[\left.\begin{matrix} \frac{\tilde{p}_1+m}{2},\,  \frac{\tilde{p}_1+1+m}{2}\\ 1+m \end{matrix}\right\vert u_1^2\right] \bar{F}^{(1)}_C\left[\left.\begin{matrix} \frac{\tilde{p}_2+m}{2},\,  \frac{\tilde{p}_2+1+m}{2}\\ 1+m \end{matrix}\right\vert u_2^2\right] \,.
    \end{aligned}
\end{equation}
The infinite towers of poles located at $\nu=-i(\tilde{p}_2+m)$ and $\nu=-im$, with $m\in \mathbb{N}$, in the $\F^{(1)}_+ \F^{(1)}_-$ term give the following series:
\begin{equation}
\label{eq: 2site chain series 2}
    \begin{aligned}
        -2\pi &e^{-\frac{i\pi}{2}} \C_+^{(1)}(p_1) \C_+^{(1)}(p_2) \sum_{m=0}^\infty \frac{m^2}{m^2 + \mu^2} \frac{(-1)^m}{m!}\Gamma[m, \tilde{p}_1+m, \tilde{p}_2-m]  \\
        &\times \left(\frac{u_1}{u_2}\right)^m \bar{F}^{(1)}_C\left[\left.\begin{matrix} \frac{\tilde{p}_1+m}{2},\,  \frac{\tilde{p}_1+1+m}{2}\\ 1+m \end{matrix}\right\vert u_1^2\right] \bar{F}^{(1)}_C\left[\left.\begin{matrix} \frac{\tilde{p}_2-m}{2},\,  \frac{\tilde{p}_2+1-m}{2}\\ 1-m \end{matrix}\right\vert u_2^2\right] \,,
    \end{aligned}
\end{equation}
and
\begin{equation}
    \begin{aligned}
        -2\pi &e^{-\frac{i\pi}{2}} \C_+^{(1)}(p_1) \C_+^{(1)}(p_2) \sum_{m=0}^\infty \frac{1}{(\tilde{p}_2+m)^2 + \mu^2} \frac{(-1)^m}{m!}\Gamma[\tilde{p}_1+\tilde{p}_2+m, 1+\tilde{p}_2+m, 1-\tilde{p}_2-m]  \\
        &\times \left(\frac{u_1}{u_2}\right)^{\tilde{p}_2+m} \bar{F}^{(1)}_C\left[\left.\begin{matrix} \frac{\tilde{p}_1+\tilde{p}_2+m}{2},\,  \frac{\tilde{p}_1+\tilde{p}_2+1+m}{2}\\ 1+\tilde{p}_2+m \end{matrix}\right\vert u_1^2\right] \bar{F}^{(1)}_C\left[\left.\begin{matrix} \frac{-m}{2},\,  \frac{1-m}{2}\\ 1-\tilde{p}_2-m \end{matrix}\right\vert u_2^2\right] \,.
    \end{aligned}
\end{equation}
Both series~\eqref{eq: 2site chain series 1} and~\eqref{eq: 2site chain series 2} cancel against each other due to the identity
\begin{equation}
    \bar{F}^{(1)}_C\left[\left.\begin{matrix} \frac{p-m}{2},\,  \frac{p+1-m}{2}\\ 1-m \end{matrix}\right\vert z\right] = \Gamma\left[\begin{matrix} p+m \\ p-m \end{matrix}\right] \left(\frac{z}{4}\right)^m \bar{F}^{(1)}_C\left[\left.\begin{matrix} \frac{p+m}{2},\,  \frac{p+1+m}{2}\\ 1+m \end{matrix}\right\vert z\right] \,,
\end{equation}
where we recall that $\bar{F}_C^{(1)}$ is nothing but the (regularised) ${}_2\bar{F}_1$ hypergeometric function. This identity is the direct application of the lemma~\ref{lemma-Fc} for $n=1$. In the end, only a single tower of poles, coming from the cross term $\F^{(1)}_+ \F^{(1)}_-$, contributes to the final solution. 

\vskip 4pt
One might still worry that in the regime where $4<u_1u_2$, the contribution $\F^{(1)}_+ \F^{(1)}_+$ naively requires closing the contour in the opposite direction, picking other analytic poles which do not cancel those from $\F^{(1)}_+ \F^{(1)}_-$. However, assuming $4<u_1u_2$ necessarily implies that one of the momentum ratios $u_1$ or $u_2$ (or both) is greater than unity. Without loss of generality, considering $u_1<4$ requires analytically continuing the $\bar{F}_C^{(1)}$ function with argument $u_2$ outside the unit disc using the identity:
\begin{equation}
    \begin{aligned}
        \tfrac{1}{\pi} \bar{F}^{(1)}_C\left[\left.\begin{matrix} \frac{\tilde{p}+i\nu}{2},\,  \frac{\tilde{p}+1+i\nu}{2}\\ 1+i\nu \end{matrix}\right\vert z\right] &= \frac{(-z)^{\frac{-\tilde{p}-i\nu}{2}}}{\Gamma[\tfrac{\tilde{p}+1+i\nu}{2}, \tfrac{2-\tilde{p}+i\nu}{2}]} \bar{F}^{(1)}_C\left[\left.\begin{matrix} \frac{\tilde{p}+i\nu}{2},\,  \frac{\tilde{p}-i\nu}{2}\\ 1/2 \end{matrix}\right\vert \frac{1}{z}\right] \\
        &- \frac{(-z)^{\frac{-\tilde{p}-1-i\nu}{2}}}{\Gamma[\tfrac{\tilde{p}+i\nu}{2}, \tfrac{1-\tilde{p}+i\nu}{2}]} \bar{F}^{(1)}_C\left[\left.\begin{matrix} \frac{\tilde{p}+1+i\nu}{2},\,  \frac{\tilde{p}+1-i\nu}{2}\\ 3/2 \end{matrix}\right\vert \frac{1}{z}\right] \,,
    \end{aligned}
\end{equation}
which maps $\bar{F}^{(1)}_C$ at infinity ($|z|\to\infty$) to $\bar{F}^{(1)}_C$ around the origin ($|z|\to 0$). The overall exponential behaviour at infinity in the complex $\nu$-plane is therefore dictated by $(\tfrac{u_1 u_2}{4})^{+i\nu} u_2^{-i\nu} = (\tfrac{u_1}{4})^{+i\nu}$, and we again close the contour in the lower half-plane. The symmetric case $u_2<4$ follows analogously. In the marginal situation where both $\bar{F}^{(1)}_C$ functions in $u_1$ and $u_2$ require analytic continuation outside the unit circle, the arc at infinity no longer decays exponentially but rather as a power law, as can be verified explicitly using Stirling's formula. We expect the resulting series solution to exhibit slow convergence in this region of kinematic space. Eventually, as explained in Sec.~\ref{subsec: gluing adjacent vertices} in the general case, the boundary $u_1 u_2=4$ in the kinematic region is spurious.

\paragraph{Full result.} The full closed-form solution for the two-site chain graph is found by summing all contributions. For $u_1<u_2$, the full result for the dimensionless graph reads
\begin{equation}
\label{eq: full 2-site chain result}
    \begin{aligned}
        \widehat{\G}(u_1, u_2) &= \sum_{\aa, \bb=\pm} (i\aa)(i\bb) \, \widehat{\I}_{\aa\bb}(u_1, u_2) \\
        &= \V_{-, \mu}^{(1)}(u_1; p_1) \, \V_{+, \mu}^{(1)}(u_2; p_2) + \text{\,c.c.} \\
        &+\C_+^{(1)}(p_1) \left[e^{+\pi\mu}\F^{(1)}_{+i\mu}(u_1; \tilde{p}_1) + e^{-\pi\mu}\F^{(1)}_{-i\mu}(u_1; \tilde{p}_1)\right] \V_{+, \mu}^{(1)}(u_2; p_2) + \text{\,c.c.} \\
        &+4\pi e^{-\frac{i\pi}{2}} \C_+^{(1)}(p_1) \C_+^{(1)}(p_2) \sum_{m=0}^\infty \frac{\Gamma[\tilde{p}_1+\tilde{p}_2+m]}{(\tilde{p}_2+m)^2 + \mu^2} \frac{(-1)^m}{m!}  \\
        &\times \left(\frac{u_1}{u_2}\right)^{\tilde{p}_2+m} F^{(1)}_C\left[\left.\begin{matrix} \frac{\tilde{p}_1+\tilde{p}_2+m}{2},\,  \frac{\tilde{p}_1+\tilde{p}_2+1+m}{2}\\ 1+\tilde{p}_2+m \end{matrix}\right\vert u_1^2\right] F^{(1)}_C\left[\left.\begin{matrix} \frac{-m}{2},\,  \frac{1-m}{2}\\ 1-\tilde{p}_2-m \end{matrix}\right\vert u_2^2\right] + \text{\,c.c.} \,,
    \end{aligned}
\end{equation}
where $\tilde{p}_{1, 2} \equiv p_{1, 2}-d/2$. The vertex functions $\V_{\pm, \mu}^{(1)}$ are defined in~\eqref{eq: V1 expression}, the individual frequency modes $\F_{\pm i\mu}^{(1)}$ are given in~\eqref{eq: F1 function}, and the prefactors $\C_{\pm}^{(1)}$ including a phase are given in~\eqref{eq: overall C1+}. Notice that we have added a factor $2$ in the series to account for the shadow symmetric term, and have removed some $\Gamma$-functions at the cost of replacing the (regularised) $\bar{F}_C^{(1)}$ function by $F_C^{(1)} \equiv {}_2F_1$. For $p_1=p_2=2$ and $d=3$, we recover the series solution derived in~\cite{Werth:2024mjg}. It is straightforward to check that the full solution satisfies the largest-time equation~\cite{Werth:2024mjg}, which provides a non-trivial consistency check.

\begin{figure}[h!]
    \centering
    \begin{subfigure}{.5\textwidth}
        \centering
        \includegraphics[width=1\linewidth]{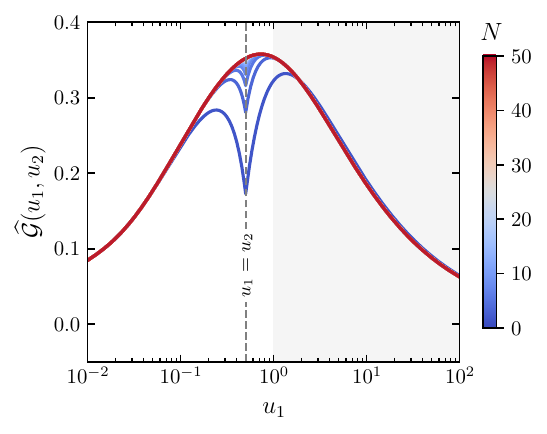}
    \end{subfigure}%
    \begin{subfigure}{.5\textwidth}
        \centering
        \includegraphics[width=1\linewidth]{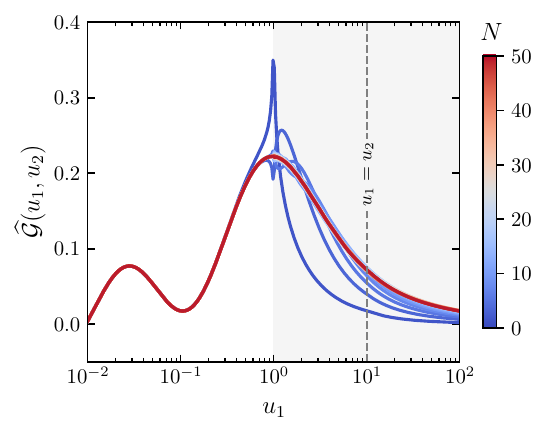}
    \end{subfigure}
   \caption{The full dimensionless two-site chain graph $\widehat{\G}(u_1, u_2)$ as function of the dimensionless kinematic variable $u_1$ with fixed $u_2=0.5$ ({\it left panel}) and $u_2=10$ ({\it right panel}), for $\mu=2$, $p_1=p_2=2$ and $d=3$. We include terms in the series from $m=1, \ldots, N$. The shaded gray region corresponds to kinematic configurations outside the unit disk in the complex $u_1$-plane.}
  \label{fig: 2-site chain}
\end{figure}

\vskip 4pt
Quite remarkably, the second hypergeometric function in $u_2$ in the series solution is actually just a polynomial. Indeed, since $m$ is a non-negative integer, one of the two upper parameters of the hypergeometric function always reduces to a non-positive integer; specifically $-m/2$ when $m$ is even, and $(1-m)/2$ when $m$ is odd. The Pochhammer symbol in the numerator therefore vanishes at finite order, terminating the series and reducing the ${}_2F_1$ to a polynomial of degree $\lfloor m/2 \rfloor$ in $u_2$, independently of the value of $\tilde{p}_2$. In the end, all individual contributions have {\it uniform transcendental weight}. In Fig.~\ref{fig: 2-site chain}, we show the dimensionless graph $\widehat{\G}(u_1, u_2)$ as function of the energy ratio $u_1$, fixing $u_2=0.5$ and $u_2=10$, from the soft limit $u_1\ll 1$ and extending outside the unit disk $u_1\gg 1$. We observe that the solution converges well and stabilises when the series contains $N \sim \O(50)$ terms. As expected, larger internal masses require keeping more terms in the series to achieve the same level of accuracy. Outside the unit disk, corresponding to the shaded region, the solution can be evaluated via the standard analytic continuation of the hypergeometric ${}_2F_1$ function. We notice that when $u_2<1$, the transition $u_1=u_2$ appears discontinuous, whereas when $1<u_2$, the solution is entirely smooth and only the transition $u_1=1$ exhibits a visible discontinuity when too few terms are retained in the series. The solution~\eqref{eq: full 2-site chain result} converges in the entire physical region and also beyond the unit disk: $0\leq u_1$ and $0\leq u_2$. 

\subsection{Three-site chain}
\label{app: three-site chain}

Here, we outline the application of the gluing algorithm for the analytic part of the diagram and provide full analytic solutions also for the on-shell contributions. In contrast to the main text, we will show how to obtain solutions from iterated integration, which is equivalent to solving the linear system for the off-shell ansatz.\footnote{We also refer to~\cite{Belrhali:2026act} for this computation.} Since there are three vertices, the full dimensionless graph contains $2^3=8$ contributions
\begin{equation}
    \widehat{\G}(X_1,X_2,X_3,Y_1,Y_2) = \sum\limits_{\aa,\bb,\cc=\pm} (i\aa)(i\bb)(i\cc)\widehat{\I}_{\aa\bb\cc}(X_1,X_2,X_3,Y_1,Y_2)\,,
\end{equation}
where again half of the dimensionless master integrals $\widehat{\I}_{\aa\bb\cc}$ are related via complex conjugation. The master integrals depend only on the kinematic ratios
\begin{equation}
    u_{12} \equiv \frac{Y_{12}}{X_1} \,, \quad u_{21} \equiv \frac{Y_{12}}{X_2} \,, \quad u_{23} \equiv \frac{Y_{23}}{X_2}\,, \quad u_{32} \equiv \frac{Y_{23}}{X_3}\,,
\end{equation}
which are again positive but not necessarily bound by unity. In fact, even without reduced sound speeds there is no kinematic bound on the variables $u_{21}$ and $u_{23}$ due to the arbitrariness of $X_2$. 

\vskip 4pt
\paragraph{Fully nested contribution.} Since there are two edges this time, there is not only a fully factorised and a fully nested, but also partially nested contribution to the graph. Since we are mostly interested in applications of our gluing algorithm, we start with the fully nested contribution $\widehat{\I}_{+++}$, which is given by the double spectral integral
\begin{equation}
    \begin{aligned}
        \widehat{\I}_{+++}(\{u_i\},\{v_i\}) = e^{-i\pi} \int\limits_{-\infty}^{+\infty} [\d\nu_1][\d\nu_2] \rho_{\nu_1}(\mu_1)\rho_{\nu_2}(\mu_2) \left(
        \raisebox{-5pt}{
        \begin{tikzpicture}[line width=1. pt, scale=2]
            \draw[black] (0, 0) -- (2, 0);
            \draw[white, line width=2pt] 
  ($(0,0)!1/3!(1,0)$) -- ($(0,0)!2/3!(1,0)$);
            \node at ($(0,0)!1/2!(1,0)$) {$\times$};
            \draw[white, line width=2pt] 
  ($(1,0)!1/3!(2,0)$) -- ($(1,0)!2/3!(2,0)$);
            \node at ($(1,0)!1/2!(2,0)$) {$\times$};

            \draw[fill=black] (0, 0) circle (.05cm) node[above right=0.1mm] {$\nu_1$};
            \draw[fill=black] (1, 0) circle (.05cm) node[above left=0.1mm] {$\nu_1$};
            \draw[fill=black] (1, 0) circle (.05cm) node[above right=0.1mm] {$\nu_2$};
            \draw[fill=black] (2, 0) circle (.05cm) node[above left=0.1mm] {$\nu_2$};
        \end{tikzpicture} 
        }\right) \\
        = e^{-i\pi} \int\limits_{-\infty}^{+\infty} \{[\d\nu_i]\rho_{\nu_i}(\mu_i)\} \, \V_{+, \nu_1}^{(1)}(u_{12}; p_1) \, \V_{+,\nu_1,\nu_2}^{(2)}(u_{21},u_{23};p_2)\,\V_{+, \nu_2}^{(1)}(u_{32}; p_3) \,.
    \end{aligned}
\end{equation}
We will discuss the fully analytic piece $\widehat{\I}_{+++}^{\P}$ first. Following the gluing algorithm, we start by separating into positive and negative frequency modes. Our spectral integral can be decomposed into
\begin{align}
    \widehat{\I}_{+++}(\{u_i\}) = &e^{-i\pi} \C_+^{(1)}\C_+^{(2)}\C_+^{(1)}\int\limits_{-\infty}^{+\infty} [\d\nu_1][\d\nu_2] \rho_{\nu_1}(\mu_1)\rho_{\nu_2}(\mu_2) \nonumber\\
    &\times\left(\F_+^{(1)}+ \F_-^{(1)}\right) \left(\F_{++}^{(2)}+\F_{+-}^{(2)}+\F_{-+}^{(2)}+ \F_{--}^{(2)}\right)\left(\F_+^{(1)}+ \F_-^{(1)}\right)\,.
\end{align}
First, we evaluate the spectral integral over $\nu_1$ where we can look at the mixed term $\F_+^{(1)}\F_{-\pm}^{(2)}$ so that the relevant vertex pole factors are given by $\Gamma[p_1-\tfrac{d}{2}+i\nu_1]$ and $\Gamma[p_2-i\nu_1\pm i\nu_2]$ and the kinematic dependence is dominated by the term $(u_{12}/u_{21})^{i\nu_1}$. If we pick up the poles from the second vertex factor (corresponding to the kinematic region $u_{12}<u_{21}$), the result will depend on $(u_{12}/u_{21})^{\pm i\nu_2}$. When we then consider the second integral and again look at a mixed term $\propto \F_{\pm,+}^{(2)}\F_-^{(1)}$, the integrand will depend on the kinematic ratio $(u_{12}u_{23}/u_{21}u_{32})^{i\nu_2}$. Therefore, the four kinematic regions have a slightly more subtle dependence on the kinematic variables and we cannot only compare the kinematic ratios assigned to the respective edge. 

\vskip 4pt
Here, we will focus on the kinematic region $u_{12}>u_{21}$ and $u_{32}>u_{23}$, and provide results for the remaining regions at the end of this section. In this region, we close the contour for the first integral in the upper half plane, picking up the poles from $\Gamma[p_1-\tfrac{d}{2}+i\nu_1]$. For the second integral we look at the contribution from $\F_{\cdots +}^{(2)}\F_{-}^{(1)}$ which depends on $(u_{23}/u_{32})^{i\nu_2}$ and has vertex pole factors $\Gamma[p_{12}-\tfrac{d}{2}+m+i\nu_2]$ and $\Gamma[p_3-\tfrac{d}{2}-i\nu_2]$. This time, we have to close in the lower half plane and pick up the poles from $\Gamma[p_3-\tfrac{d}{2}-i\nu_2]$. In total, we find:
\begin{align}
\label{eq: 3site chain 1}
    \widehat{\I}_{+++}^{\P} = -&16\pi^2\C_+^{(1)}\C_+^{(2)}\C_+^{(1)} \sum\limits_{m_1,m_2=0}^\infty \frac{(-1)^{m_{12}}}{m_1!m_2!}\frac{\Gamma[\tilde p_{123}+m_{12}]}{\left[(\tilde p_1+m_1)^2+\mu_{12}^2\right]\left[(\tilde p_{3}+m_{2})^2+\mu_{23}^2\right]}  \nonumber\\
    &\times \left(\frac{u_{21}}{u_{12}}\right)^{m_1+\tilde p_1}\left(\frac{u_{23}}{u_{32}}\right)^{m_{2}+\tilde p_{3}}  F^{(1)}_C\left[\left.\begin{matrix} \frac{-m_1}{2},\,  \frac{1-m_1}{2}\\ 1-\tilde{p}_1-m_1 \end{matrix}\right\vert u_{12}^2\right] \nonumber\\
    &\times F_C^{(2)}\left[\left.\begin{matrix} \frac{\tilde p_{123}+m_{12}}{2},\,  \frac{\tilde p_{123}+1+m_{12}}{2}\\ 1+\tilde{p}_1+m_1, 1+\tilde p_{3}+m_{2} \end{matrix}\right\vert u_{21}^2, u_{23}^2\right]
    F^{(1)}_C\left[\left.\begin{matrix} \frac{-m_2}{2},\,  \frac{1-m_2}{2}\\ 1-\tilde{p}_{3}-m_{2} \end{matrix}\right\vert u_{32}^2\right]\,.
\end{align}
In addition, there are also terms stemming from evaluating the spectral integrals on-shell. If both are evaluated on-shell, we find
\begin{align}
    \widehat{\I}_{+++}^{\delta_{12}\delta_{23}} = \C_+^{(1)}\left(e^{+\pi\mu_1}\F_{+i\mu_1}^{(1)} + e^{-\pi\mu_1}\F_{-i\mu_1}^{(1)}\right)\,\V_{+,\mu_1\mu_2}^{(2)}\,\C_+^{(1)}\left(e^{+\pi\mu_2}\F_{+i\mu_2}^{(1)} + e^{-\pi\mu_2}\F_{-i\mu_2}^{(1)}\right)\,,
\end{align}
where we suppressed the arguments of the coefficients and vertex functions. The remaining two terms are obtained by setting only one of the integrands on-shell and picking up the analytic poles for the second. They are given by
\begin{align}
    \widehat{\I}_{+++}^{\delta_{12}\P_{23}} = 4\pi i\, \C_+^{(1)}\C_+^{(2)} \C_+^{(1)} \left(e^{+\pi\mu_1}\F_{+i\mu_1}^{(1)} + e^{-\pi\mu_1}\F_{-i\mu_1}^{(1)}\right) \sum\limits_{\alpha=\pm i \mu_1} \sum\limits_{m=0}^\infty  \nonumber \\
    \times \frac{(-1)^m}{m!} \frac{\Gamma[\tilde p_{23}+m+\alpha]}{(\tilde p_3+m)^2+\mu_{23}^2}  \left(\frac{u_{23}}{u_{32}}\right)^{m+\tilde p_3} F^{(1)}_C\left[\left.\begin{matrix} \frac{-m}{2},\,  \frac{1-m}{2}\\ 1-\tilde{p}_3-m \end{matrix}\right\vert u_{32}^2\right]  \nonumber 
    \\
    \times F_C^{(2)}\left[\left.\begin{matrix} \frac{\tilde p_{23}+m+\alpha}{2},\,  \frac{\tilde p_{23}+1+m+\alpha}{2}\\ 1+\alpha, 1+\tilde p_3+m \end{matrix}\right\vert u_{21}^2, u_{23}^2\right]\,,
\end{align}
as well as 
\begin{align}
    \widehat{\I}_{+++}^{\P_{12}\delta_{23}} = 4\pi i\, \C_+^{(1)}\C_+^{(2)} \C_+^{(1)} \left(e^{+\pi\mu_2}\F_{+i\mu_2}^{(1)} + e^{-\pi\mu_2}\F_{-i\mu_2}^{(1)}\right) \sum\limits_{\alpha=\pm i \mu_2} \sum\limits_{m=0}^\infty  \nonumber \\
    \times \frac{(-1)^m}{m!} \frac{\Gamma[\tilde p_{12}+m+\alpha]}{(\tilde p_1+m)^2+\mu_{12}^2}  \left(\frac{u_{21}}{u_{12}}\right)^{m+\tilde p_1} F^{(1)}_C\left[\left.\begin{matrix} \frac{-m}{2},\,  \frac{1-m}{2}\\ 1-\tilde{p}_1-m \end{matrix}\right\vert u_{12}^2\right]  \nonumber 
    \\
    \times F_C^{(2)}\left[\left.\begin{matrix} \frac{\tilde p_{12}+m+\alpha}{2},\,  \frac{\tilde p_{12}+1+m+\alpha}{2}\\  1+\tilde p_1+m, 1+\alpha \end{matrix}\right\vert u_{21}^2, u_{23}^2\right]\,.
\end{align}

\paragraph{(Partially) factorised contributions.} The fully factorised diagram $\widehat{\I}_{+-+}$ is given by a simple product of vertex functions:
\begin{equation}
    \begin{aligned}
        \widehat{\I}_{+-+}(\{u_i\}) &= \raisebox{-5pt}{
        \begin{tikzpicture}[line width=1. pt, scale=2]
            \draw[black] (0, 0) -- (2, 0);
            \draw[white, line width=2pt] 
  ($(0,0)!1/3!(1,0)$) -- ($(0,0)!2/3!(1,0)$);
            \node at ($(0,0)!1/2!(1,0)$) {$\times$};
            \draw[white, line width=2pt] 
  ($(1,0)!1/3!(2,0)$) -- ($(1,0)!2/3!(2,0)$);
            \node at ($(1,0)!1/2!(2,0)$) {$\times$};

            \draw[fill=black] (0, 0) circle (.05cm) node[above right=0.1mm] {$\nu_1$};
            \draw[fill=white] (1, 0) circle (.05cm) node[above left=0.1mm] {$\nu_1$};
            \draw[fill=white] (1, 0) circle (.05cm) node[above right=0.1mm] {$\nu_2$};
            \draw[fill=black] (2, 0) circle (.05cm) node[above left=0.1mm] {$\nu_2$};
        \end{tikzpicture} 
        } \nonumber \\
        &= \V_{+,\mu_1}^{(1)}(u_{12};p_1)\V_{-,\mu_1\mu_2}^{(2)}(u_{21},u_{23};p_2)\V_{+,\mu_2}^{(1)}(u_{32};p_3)\,.
    \end{aligned}
\end{equation}
The partially factorised diagrams $\widehat{\I}_{++-}$ and $\widehat{\I}_{-++}$ each contain a time-ordered propagator and hence one layer of spectral integration. For example, the contribution $\widehat{\I}_{-++}$ can be represented by,
\begin{equation}
    \begin{aligned}
       \widehat{\I}_{-++}(\{u_i\}) &= e^{-i\pi/2} \int\limits_{-\infty}^{+\infty} [\d\nu]\rho_{\nu}(\mu_2) \left(
        \raisebox{-5pt}{
        \begin{tikzpicture}[line width=1. pt, scale=2]
            \draw[black] (0, 0) -- (2, 0);
            \draw[white, line width=2pt] 
  ($(0,0)!1/3!(1,0)$) -- ($(0,0)!2/3!(1,0)$);
            \node at ($(0,0)!1/2!(1,0)$) {$\times$};
            \draw[white, line width=2pt] 
  ($(1,0)!1/3!(2,0)$) -- ($(1,0)!2/3!(2,0)$);
            \node at ($(1,0)!1/2!(2,0)$) {$\times$};

            \draw[fill=white] (0, 0) circle (.05cm) node[above right=0.1mm] {$\mu_1$};
            \draw[fill=black] (1, 0) circle (.05cm) node[above left=0.1mm] {$\mu_1$};
            \draw[fill=black] (1, 0) circle (.05cm) node[above right=0.1mm] {$\nu$};
            \draw[fill=black] (2, 0) circle (.05cm) node[above left=0.1mm] {$\nu$};
        \end{tikzpicture} 
        }\right) \\
        &= \V_{+, \nu_1}^{(1)}(u_{12}; p_1) \, e^{-i\pi/2} \int\limits_{-\infty}^{+\infty} [\d\nu]\rho_{\nu}(\mu_2) \,  \V_{+,\mu_1,\nu}^{(2)}(u_{21},u_{23};p_2)\,\V_{+, \nu}^{(1)}(u_{32}; p_3) \,.
    \end{aligned}
\end{equation}
Again, we have to separate the on-shell poles from the analytic poles. Paying attention to the kinematic region $u_{i2} > u_{2i}$, we find the contributions
\begin{align}
    \widehat{\I}_{++-}^{\P} &= -4\pi i\,\C_+^{(1)}(p_1)\C_+^{(2)}(p_2)\V_{-,\mu_2}^{(1)}(u_{32};p_3)\sum\limits_{\alpha=\pm i\mu_1}\sum\limits_{m=0}^\infty  \frac{(-1)^m}{m!} \frac{\Gamma[\tilde p_{12}+m+\alpha]}{(\tilde p_1+m)^2+\mu_{12}^2}   \nonumber \\
    &\times \left(\frac{u_{21}}{u_{12}}\right)^{m+\tilde p_1} F^{(1)}_C\left[\left.\begin{matrix} \frac{-m}{2},\,  \frac{1-m}{2}\\ 1-\tilde{p}_1-m \end{matrix}\right\vert u_{12}^2\right]  F_C^{(2)}\left[\left.\begin{matrix} \frac{\tilde p_{12}+m+\alpha}{2},\,  \frac{\tilde p_{12}+1+m+\alpha}{2}\\ 1+\tilde{p}_1+m, 1+\alpha \end{matrix}\right\vert u_{21}^2, u_{23}^2\right]\,,
\end{align}
as well as
\begin{align}
    \widehat{\I}_{-++}^{\P} &= -4\pi i\,\V_{-,\mu_1}^{(1)}(u_{12};p_1)\C_+^{(2)}(p_2)\C_+^{(1)}(p_3)\sum\limits_{\alpha=\pm i\mu_2}\sum\limits_{m=0}^\infty  \frac{(-1)^m}{m!} \frac{\Gamma[\tilde p_{23}+m+\alpha]}{(\tilde p_3+m)^2+\mu_{23}^2}   \nonumber \\
    &\times \left(\frac{u_{23}}{u_{32}}\right)^{m+\tilde p_3} F^{(1)}_C\left[\left.\begin{matrix} \frac{-m}{2},\,  \frac{1-m}{2}\\ 1-\tilde{p}_3-m \end{matrix}\right\vert u_{32}^2\right]  F_C^{(2)}\left[\left.\begin{matrix} \frac{\tilde p_{23}+m+\alpha}{2},\,  \frac{\tilde p_{23}+1+m+\alpha}{2}\\ 1+\alpha, 1+\tilde p_3+m \end{matrix}\right\vert u_{21}^2, u_{23}^2\right]\,,
\end{align}
for the analytic poles. The on-shell poles contribute 
\begin{align}
    \widehat{\I}_{++-}^{\delta_{12}} =& -\C_+^{(1)}(p_1)\left(e^{+\pi\mu_1}\F_{+i\mu_1}^{(1)} + e^{-\pi\mu_1}\F_{-i\mu_1}^{(1)}\right)\,\V_{+,\mu_1\mu_2}^{(2)}(u_{21},u_{23};p_2)\,\V_{-,\mu_2}^{(1)}(u_{32};p_3) \,,\nonumber\\
    \widehat{\I}_{-++}^{\delta_{23}} =& -\V_{-,\mu_1}^{(1)}(u_{12};p_1)\,\V_{+,\mu_1\mu_2}^{(2)}(u_{21},u_{23};p_2)\,\C_+^{(1)}(p_3)\left(e^{+\pi\mu_2}\F_{+i\mu_2}^{(1)} + e^{-\pi\mu_2}\F_{-i\mu_2}^{(1)}\right)\,.
\end{align}
The remaining four integrals $\widehat{\I}_{---}$, $\widehat{\I}_{--+}$, $\widehat{\I}_{+--}$ and $\widehat{\I}_{-+-}$ are simply the complex conjugates of all the above terms. Collecting all these different terms finishes our discussion of the double exchange diagram. 

\vskip 4pt
\paragraph{Other kinematic regions.}
As mentioned before, switching to different kinematic regions corresponds to closing the integration contour in different half planes which leads to different poles that are picked up in the analytic pieces (including the analytic parts of partially factorised contributions). The structure of the on-shell terms will also change slightly, since the projection on positive and negative frequency modes is altered in different regions. As an example, we provide the two other series solutions that can be obtained from the spectral integral by closing in different directions. Interestingly there are only two more of them and not three, like one would expect from the $2^2=4$ combinations of closing the contour in the upper and lower half plane. In addition to the $u_{i2} > u_{2i} $ solution from above, we find
\begin{align}
\label{eq: 3site chain 2}
    \widehat{\I}_{+++}^{\P} = -&16\pi^2\C_+^{(1)}\C_+^{(2)}\C_+^{(1)} \sum\limits_{m_1,m_2=0}^\infty \frac{(-1)^{m_{12}}}{m_1!m_2!} \frac{\Gamma[\tilde p_{123}+m_{12}]}{\left[(\tilde p_{23}+m_{12})^2+\mu_{12}^2\right]\left[(\tilde p_{3}+m_{2})^2+\mu_{23}^2\right]}  \nonumber\\
    &\times \left(\frac{u_{12}}{u_{21}}\right)^{m_{12}+\tilde p_{23}}\left(\frac{u_{23}}{u_{32}}\right)^{m_{2}+\tilde p_{3}}  F^{(1)}_C\left[\left.\begin{matrix} \frac{\tilde p_{123}+m_{12}}{2},\,  \frac{\tilde p_{123}+1+m_{12}}{2}\\ 1+\tilde{p}_{23}+m_{12} \end{matrix}\right\vert u_{12}^2\right] \nonumber\\
    &\times F_C^{(2)}\left[\left.\begin{matrix} \frac{-m_1}{2},\,  \frac{1-m_1}{2}\\ 1-\tilde{p}_{23}-m_{12},1+\tilde p_3+m_2 \end{matrix}\right\vert u_{21}^2, u_{23}^2\right]
    F^{(1)}_C\left[\left.\begin{matrix} \frac{-m_2}{2},\,  \frac{1-m_{2}}{2}\\ 1-\tilde{p}_{3}-m_{2} \end{matrix}\right\vert u_{32}^2\right]\,,
\end{align}
in the kinematic region where $u_{12}<u_{21}$ and $u_{23} < u_{32}$ (meaning we close both integrals in the lower half plane). For the two remaining regions with $u_{12}>u_{21}$ and $u_{32}<u_{23}$ as well as $u_{12}< u_{21}$ and $u_{12}u_{23}>u_{21}u_{32}$ we close one integral in the upper and one in the lower half plane, which both results in
\begin{align}
\label{eq: 3site chain 3}
    \widehat{\I}_{+++}^{\P} &= -16\pi^2\C_+^{(1)}\C_+^{(2)}\C_+^{(1)} \sum\limits_{m_1,m_2=0}^\infty \frac{(-1)^{m_{12}}}{m_1!m_2!} \frac{\Gamma[\tilde p_{123}+m_{12}]}{\left[(\tilde p_1+m_1)^2+\mu_{12}^2\right]\left[(\tilde p_{12}+m_{12})^2+\mu_{23}^2\right]}  \nonumber\\
    &\times \left(\frac{u_{21}}{u_{12}}\right)^{m_1+\tilde p_1}\left(\frac{ u_{32}}{u_{23}}\right)^{m_{12}+\tilde p_{12}}  F^{(1)}_C\left[\left.\begin{matrix} \frac{-m_1}{2},\,  \frac{1-m_1}{2}\\ 1-\tilde{p}_1-m_1 \end{matrix}\right\vert u_{12}^2\right] \nonumber\\
    &\times F_C^{(2)}\left[\left.\begin{matrix} \frac{-m_2}{2},\,  \frac{1-m_2}{2}\\ 1+\tilde{p}_1+m_1, 1-\tilde p_{12}-m_{12} \end{matrix}\right\vert u_{21}^2, u_{23}^2\right]
    F^{(1)}_C\left[\left.\begin{matrix} \frac{\tilde{p}_{123}+m_{12}}{2},\,  \frac{\tilde{p}_{123}+1+m_{12}}{2}\\ 1+\tilde{p}_{12}+m_{12} \end{matrix}\right\vert u_{32}^2\right]\,.
\end{align}
When describing where to close the contour, we always consider a mixed term $\F_+\F_-$ where the left-most vertex contributes the positive-frequency mode. For the second mixed term, the direction of contour-closing would be reversed. 

\section{Derivation of Differential Equations}
\label{sec: diff eq derivation}

In this appendix, we provide the details of the derivation of the system of differential equations that describes a given tree-level correlator. We closely follow the original derivation in~\cite{Liu:2024str} and adjust accordingly to our definitions and notations.\footnote{In particular, our definition of the dimensionless seed integral $\widehat{\I}$ differs from the dimensionless graph as defined in~\cite{Liu:2024str} by a factor $\prod_{(i,j)} (u_{ij}u_{ji})^{d/2}$.}

\vskip 4pt
As a starting point, we consider any internal line of a given tree-level diagram $\widehat{\I}_{\aa_1\cdots \aa_N}(\{X_i\},\{Y_j\})$:
\begin{center}
\begin{tikzpicture}
[line width=1. pt, scale=2]
    \draw[fill=black] (0, 0) circle (.05cm) node[below right=1mm] {$\textcolor{pyblue}{u_{ij}}$};
    \draw[fill=black] (0, 0) circle (.05cm) node[above right] {$\textcolor{pyblue}{i}$};
    
    \draw[fill=black] (1, 0) circle (.05cm) node[below left=0.1mm] {$\textcolor{pyblue}{u_{ji}}$};
    \draw[fill=black] (1, 0) circle (.05cm) node[above left] {$\textcolor{pyblue}{j}$};
    
    \draw[black] (0, 0) -- (1, 0);

    \draw (0,0) -- ++(120:0.6) node[above] {$(u_{i2}, \mu_{i2})$};
    \draw (0,0) -- ++(150:0.6) node[above left] {$(u_{i3}, \mu_{i3})$};
    \draw (0,0) -- ++(180:0.6) node[left] {$(u_{i4}, \mu_{i4})$};
    \node at (-0.3,-0.15) {$\vdots$};
    \draw (0,0) -- ++(-120:0.6) node[below] {$(u_{iV}, \mu_{iV})$};

    \draw (1,0) -- ++(60:0.6) node[above] {$(u_{j2}, \mu_{j2})$};
    \draw (1,0) -- ++(30:0.6) node[above right] {$(u_{j3}, \mu_{j3})$};
    \draw (1,0) -- ++(0:0.6) node[right] {$(u_{j4}, \mu_{j4})$};
    \node at (1.3,-0.15) {$\vdots$};
    \draw (1,0) -- ++(-60:0.6) node[below] {$(u_{j{V'}}, \mu_{j{V'}})$};

\end{tikzpicture}
\end{center}
We label the two vertices attached to this internal line by indices $i$ and $j$. Furthermore, we label the edges attached to $i$ by  $i2,\dots,iV$ where $V=\text{deg} \,i$ is the degree of the vertex $i$ (and similarly for the edges attached to $j$). Let $X_i, X_j, X_{n}$ with $n=2,\dots,V$ be the respective vertex energies, $Y_{ij}$ the energy corresponding to the chosen internal line and $Y_{in}$ the energies along the other edges connected to $i$. Then, the diagram will contain the following time integral corresponding to the vertex $i$
\begin{equation}\label{eq: diff eq time integral}
    \widehat{\I}_{\aa_1\dots \aa_N} \supset \int_{-\infty}^0 \frac{\mathrm{d}z_i}{(-z_i)^{d+1}}(-z_i)^{p_i} e^{i\aa_iz_i} \widehat{D}_{\aa_i\aa_j}^{(\mu_{ij})}(u_{ij}z_i, u_{ji}z_j) \prod\limits_{n=2}^V \widehat{D}_{\aa_i\aa_{n}}^{(\mu_{in})}(u_{in}z_i,u_{ni}z_{n})\,,
\end{equation}
where the energy ratios are defined as $u_{ij}=Y_{ij}/X_i$, $u_{ji}=Y_{ij}/X_j$, $u_{in}=Y_{in}/X_i$ and $u_{ni}=Y_{in}/X_{n}$. To construct a system of differential equations, the idea is to insert a Klein-Gordon operator that acts on $\widehat{D}_{\aa_i\aa_j}^{(\mu_{ij})}(u_{ij}z_i, u_{ji}z_j)$ and to pull the differential operator out of the integral using integration by parts identities. The Klein-Gordon equation for the propagators in $\text{dS}_{d+1}$ is given by
\begin{equation}
    \left[\tau_i^2\partial_{\tau_i}^2-(d-1)\tau_i\partial_{\tau_i} + Y_{ij}^2\tau_i^2+m_{ij}^2\right]D_{\aa_i\aa_j}^{(\mu_{ij})}(Y_{ij};\tau_i,\tau_j) = -i\aa_i (\tau_i\tau_j)^{\frac{d+1}{2}}\delta(\tau_i-\tau_j)\delta_{\aa_i\aa_j}\,.
\end{equation}
As a first step, we will rewrite $\tau_i= z_i/X_i$ and $\tau_j=z_j/X_j$ and introduce another dimensionless propagator $\tilde D_{\aa_i\aa_j}^{(\mu_{ij})} (u_{ij} z_i,u_{ji}z_j) = Y_{ij}^d D_{\aa_i\aa_j}^{(\mu_{ij})} (Y_{ij};\tau_i,\tau_j)$ which has the advantage that it only depends on the products $u_{ij} z_i$ and $u_{ji} z_j$. Note that this propagator is related to the dimensionless propagator used throughout this work via
\begin{equation}\label{eq: relation dimensionless propagators}
    \tilde D_{\aa_i\aa_j}^{(\mu_{ij})} (u_{ij}z_i,u_{ji}z_j) = \frac{\pi}{4}(u_{ij}u_{ji})^{d/2} \widehat D_{\aa_i\aa_j}^{(\mu_{ij})} (u_{ij}z_i,u_{ji}z_j)\,.
\end{equation}
Rewriting the Klein-Gordon equation in the coordinate $z_i$ first and then using the fact that $(z_i\partial_{z_i})^n \tilde D_{\aa_i\aa_j}^{(\mu)} (u_{ij}z_i,u_{ji}z_j) = (u_{ij}\partial_{u_{ij}})^n \tilde D_{\aa_i\aa_j}^{(\mu_{ij})} (u_{ij}z_i,u_{ji}z_j)$ leads to
\begin{equation}
    \left[\vartheta_{ij}^2 -d\vartheta_{ij} +u_{ij}^2z_i^2 + m_{ij}^2\right] \tilde D_{\aa_i\aa_j}^{(\mu_{ij})} (u_{ij}z_i,u_{ji}z_j) = -i\aa_i (u_{ij}z_iu_{ji}z_j)^{\frac{d+1}{2}} \delta(u_{ij}z_i-u_{ji}z_j)\delta_{\aa_i\aa_j}\,,
\end{equation}
where we introduced the Euler operator $\vartheta_{ij} = u_{ij} \partial_{u_{ij}}$. As a last step, we plug in the relation \eqref{eq: relation dimensionless propagators} and commute the factor $u_{ij}^{d/2}$ with the differential operator. We end up with the following Klein-Gordon equation for the dimensionless propagator in this work:
\begin{equation}\label{eq: Klein Gordon dimensionless}
    \left[\vartheta_{ij}^2 + \mu_{ij}^2 + u_{ij}^2 z_i^2\right] \widehat{D}_{\aa_i\aa_j}^{(\mu_{ij})}(u_{ij}z_j,u_{ji}z_j) = -\tfrac{4i}{\pi} \aa_i (u_{ij}u_{ji})^\frac{1}{2} (z_iz_j)^\frac{d+1}{2}\delta(u_{ij}z_i-u_{ji}z_j)\delta_{\aa_i\aa_j}\,.
\end{equation}
We will now insert the Klein-Gordon operator into \eqref{eq: diff eq time integral} so that it is acting only on $\widehat{D}_{\aa_i\aa_j}^{(\mu_{ij})}(u_{ij}z_i, u_{ji}z_j)$. For the left-hand side of \eqref{eq: Klein Gordon dimensionless}, we will apply integration by parts to pull out the differential operators. In particular, we make use of
\begin{align}
    0 &= \int_{-\infty}^0 \mathrm{d}z_i\,\partial_{z_i}\left[(-z_i)^{p_i-d} e^{i\aa_iz_i} \widehat{D}_{\aa_i\aa_j}^{(\mu_{ij})}(u_{ij}z_i, u_{ji}z_j) \prod\limits_{n=2}^V \widehat{D}_{\aa_i\aa_{n}}^{(\mu_{in})}(u_{in}z_i,u_{ni}z_{n})\right] \nonumber \\
    &= \int_{-\infty}^0 \mathrm{d}z_i\,(-z_i)^{p_i-d-1} e^{i\aa_iz_i} \left[(p_i-d)+i\aa_iz_i+z_i\partial_{z_i}\right]\widehat{D}_{\aa_i\aa_j}^{(\mu_{ij})} \prod\limits_{n=2}^V \widehat{D}_{\aa_i\aa_{n}}^{(\mu_{in})} \nonumber \\
    &= \int_{-\infty}^0 \mathrm{d}z_i\,(-z_i)^{p_i-d-1} e^{i\aa_iz_i} \left[(p_i+\tfrac{V-2}{2}d)+i\aa_iz_i+\vartheta_{ij} +\sum\limits_{n=2}^V\vartheta_{{in}}\right]\widehat{D}_{\aa_i\aa_j}^{(\mu_{ij})} \prod\limits_{n=2}^V \widehat{D}_{\aa_i\aa_{n}}^{(\mu_{in})} \,,
\end{align}
where we used the relation $z_i\partial_{z_i} \widehat{D}_{\aa_i\aa_j}^{(\mu_{ij})} = (\vartheta_{ij} +\tfrac{d}{2})\widehat{D}_{\aa_i\aa_j}^{(\mu_{ij})}$ which immediately follows from the respective relation for $\tilde D_{\aa_i\aa_j}^{(\mu_{ij})}$. This result allows us to find a relation between the above time integral and one with a different twist $p_i$ which eventually should take care of the term involving $u_{ij}^2z_i^2$ in the Klein-Gordon operator. Schematically, we obtain the relation
\begin{equation}
    \int_{-\infty}^0 \mathrm{d}z_i \, z_i \bm A = -\frac{1}{i\aa_i} \left(\vartheta_{\{i\}} +p_i+\tfrac{V-2}{2}d\right)\int_{-\infty}^0\mathrm{d}z_i\,\bm A\,,
\end{equation}
where the integrand $\bm A$ is the one in \eqref{eq: diff eq time integral} and $\vartheta_{\{i\}} = \vartheta_{ij} +\sum_{n=2}^V \vartheta_{in}$. Thus, the differential operator that is pulled out of the integral is given by
\begin{equation}
    \D_{ij} = \vartheta_{ij}^2 + \mu_{ij}^2 - u_{ij}^2(\vartheta_{\{i\}}+p_i+\tfrac{V-2}{2}d)(\vartheta_{\{i\}}+p_i+\tfrac{V-2}{2}d+1)\,.
\end{equation}
It is of course no coincidence that this differential operator coincides exactly with the annihilator of the vertex functions in \eqref{eq: vertex function annihilator}, which means that factorised diagrams correspond to homogeneous solutions of the system of differential equations. 

\vskip 4pt
Our next goal will be to derive the right-hand side of the differential equation which stems from the source term in the Klein-Gordon equation. At the end, we will act with the operator $\D_{ij}$ on the full graph $\G$, which means that we have to keep track of all the prefactors of the seed integral as well. Remember that the full correlator is given by
\begin{equation}
    \G = \sum\limits_{\aa_1,\dots,\aa_V=\pm} i\aa_1\cdots i\aa_V \prod\limits_{i=1}^V\frac{1}{X_i^{p_i-d}} \prod\limits_{j=1}^I\frac{\pi/4}{(X_jX_j')^\frac{d}{2}} \,\widehat{\I}_{\aa_1\cdots\aa_V}\,,
\end{equation}
with the dimensionless master integral $\widehat{\I}_{\aa_1\cdots\aa_V}$ defined in~\eqref{eq: def dimensionless master integrals}. Acting with $\D_{ij}$ on the graph $\G$ is by construction equivalent to placing the Klein-Gordon operator $\bm D_{ij}=\vartheta_{ij}^2 +\mu_{ij}^2+u_{ij}^2z_i^2$ in front of the dimensionless propagators $\widehat{D}_{\aa_i\aa_j}^{(\mu_{ij})}(u_{ij}z_i,u_{ji}z_j)$. Using \eqref{eq: Klein Gordon dimensionless}, we can replace the combination of Klein-Gordon operator and propagator by the source term on the right-hand side of the equation. The $\delta$-function then trivialises the $z_i$ integral. Then, using the schematic notation for the integrand in \eqref{eq: diff eq time integral} as before, we find
\begin{align}
    \int_{-\infty}^0 \mathrm{d}z_i\, \bm D_{ij} \bm A = -i\aa_i\frac{4}{\pi} \left(\frac{X_i}{X_j}\right)^{p_i+\frac{V_i-2}{2}d} (-z_j)^{p_i} e^{i\aa_i\frac{X_i}{X_j}z_j} \delta_{\aa_i\aa_j} \prod\limits_{n=2}^{V_i} \widehat{D}_{\aa_j\aa_{n}}^{(\mu_{in})} \left(\tfrac{Y_{in}}{X_j}z_j, \tfrac{Y_{in}}{X_{n}}z_{n}\right)\,.
\end{align}
One has to be a bit careful because the notation for the arguments is slightly misleading, since the dimensionless propagator $\widehat{D}$ is not a function of two products (see the definition in \eqref{eq: def dimensionless prop}) and so evaluating the $\delta$-function needs to be done carefully. 

\vskip 4pt
Acting with the differential operator on the graph collapses one time integral ($z_i$ in this case) which means that we are removing a vertex from the graph. In fact, since in the above result all the propagators that were once connected to vertex $i$ have now acquired a dependence on the time $z_j$ instead, the interpretation of this procedure is that we are removing the edge between vertex $i$ and $j$ and merge the two vertices. Also, we see that the prefactor $-i\aa_i \frac{4}{\pi}$ already cancels the factor $i\aa_i$ in the definition of the correlator and removes the factor $\pi/4$ from the edge that collapsed. Also,  $\delta_{\aa_i\aa_j}$ removes one layer of summation over Schwinger-Keldysh indices which is also consistent with a graph of reduced vertex number. To better understand the properties of this new graph, let us take a look at the full $z_j$-integral after integrating out $z_i$. Ignoring the factor $-i\aa_i \frac{4}{\pi}\delta_{\aa_i\aa_j}$ (which we have already understood), this integral is given by
\begin{align}
    \left(\frac{X_i}{X_j}\right)^{p_i+\frac{V_i-2}{2}d}\int_{-\infty}^0 \frac{\mathrm{d}z_j}{(-z_j)^{d+1}} (-z_j)^{p_i+p_j} e^{i\aa_j\left(1+\frac{X_i}{X_j}\right)z_j} &\prod\limits_{n=2}^{V_i} \widehat{D}_{\aa_j\aa_{n}}^{(\mu_{in})} \left(\tfrac{Y_{in}}{X_j}z_j, \tfrac{Y_{in}}{X_{n}}z_{n}\right)\nonumber \\
    \times&\prod\limits_{m=2}^{V_j} \widehat{D}_{\aa_j\aa_{m}}^{(\mu_{jm})} \left(\tfrac{Y_{jm}}{X_j}z_j, \tfrac{Y_{jm}}{X_{m}}z_{m}\right)\,.
\end{align}
In order to obtain the usual normalisation of the time variable in the exponential function, we make the substitution $z_j = z_j /(1+\tfrac{X_i}{X_j})$ which leads to (again, we need to be a bit careful when applying this substitution to $\widehat{D}$ and take the correct time-dependence into account)
\begin{align}
    \left(\frac{X_i}{X_j}\right)^{p_i+\frac{V_i-2}{2}d} \int_{-\infty}^0 \frac{\mathrm{d}z_j'}{(-z_j')^{d+1}} (-z_j')^{p_i+p_j} e^{i\aa_jz_j'} &\prod\limits_{n=2}^{V_i} \widehat{D}_{\aa_j\aa_{n}}^{(\mu_{in})} \left(\tfrac{Y_{in}}{X_i+X_j}z_j', \tfrac{Y_{in}}{X_{n}}z_{n}\right)\nonumber \\
    \times \left(\frac{X_j}{X_i+X_j}\right)^{p_i+p_j-d +(V_i-1+V_j-1)\frac{d}{2}}&\prod\limits_{m=2}^{V_j} \widehat{D}_{\aa_j\aa_{m}}^{(\mu_{jm})} \left(\tfrac{Y_{jm}}{X_i+X_j}z_j', \tfrac{Y_{jm}}{X_{m}}z_{m}\right)\,.
\end{align}
The interpretation of the integral part is clear: This is simply a time integral over a vertex $j$ with new twist parameter $p_i+p_j$ and new vertex energy $X_i+X_j$ which is connected to all the vertices to which $i$ and $j$ had been connected via the inherited propagators. The final task is to make sense of the kinematic ratios in front. For this end, we turn our attention to the kinematic factors in the definition of $\G$ and combine them with the additional factors. After a bit of algebra, we find
\begin{align}
    \left(\frac{X_i}{X_j}\right)^{p_i+\frac{V_i-2}{2}d} \left(\frac{X_j}{X_i+X_j}\right)^{p_i+p_j-d +(V_i-1+V_j-1)\frac{d}{2}}\prod\limits_{l=1}^V\frac{1}{X_l^{p_l-d}} \prod\limits_{r=1}^I\frac{1}{(X_rX_r')^\frac{d}{2}} \nonumber \\
    =\frac{1}{(X_i+X_j)^{p_i+p_j-d}} \prod_{\substack{l=1 \\ l\neq i,j}}^V \frac{1}{X_l^{p_l-d}} \prod_{\substack{r=1 \\ r\neq ij, in, jm}}^I \frac{1}{(X_rX_r')^\frac{d}{2}} \prod_{n=2}^{V_i} \frac{1}{((X_i+X_j)X_{n})^\frac{d}{2}} \nonumber \\
    \times \prod_{m=2}^{V_j} \frac{1}{((X_i+X_j)X_{m})^\frac{d}{2}}\,.
\end{align}
This is precisely the kinematic prefactor of the new contracted graph $C_{ij}[\G]$ which is given by removing the edge between $i$ and $j$, merging the two vertices and defining the twist and vertex energy to be the sum of the individual ones. Therefore, the full differential equation is simply given by
\begin{equation}
    \D_{ij} \G = C_{ij}[\G]\,.
\end{equation}
This completes the derivation of the system of (partial) differential equations for cosmological correlators.

\addcontentsline{toc}{section}{References}
\bibliographystyle{JHEP}
\bibliography{references.bib}

\end{document}